\documentclass[bj,authoryear,noshowframe]{imsart}
\usepackage{url}
\usepackage{booktabs}
\usepackage{graphicx}
\graphicspath{{figures/}}
\usepackage{enumitem}
\usepackage{microtype}
\usepackage{tabularx}
\usepackage{placeins}
\usepackage{tikz}
\usetikzlibrary{positioning,arrows.meta,calc}
\usepackage[plain,noend]{algorithm2e}
\DontPrintSemicolon
\usepackage{hyperref}
\hypersetup{
  hidelinks,
  pdftitle={Federated multiple testing with quantized p-values under family-wise error control},
  pdfauthor={Prasanjit Dubey and Xiaoming Huo},
  pdflang={en-US},
  pdfsubject={Theory and methods for federated multiple testing under strong family-wise error control},
  pdfkeywords={combination test, distributed testing, family-wise error rate, federated inference, quantization, report resolution, sufficiency}
}

\startlocaldefs
\theoremstyle{plain}
\newtheorem{theorem}{Theorem}
\newtheorem{proposition}{Proposition}
\newtheorem{lemma}{Lemma}
\newtheorem{corollary}{Corollary}
\theoremstyle{definition}
\newtheorem{definition}{Definition}
\newtheorem{assumption}{Assumption}

\newcommand{\pr}{\mathrm{pr}}
\newcommand{\E}{E}
\newcommand{\fwer}{\mathrm{FWER}}
\makeatletter
\def\common@pub@types@hook{\def\betal{et al.}}
\makeatother
\allowdisplaybreaks
\endlocaldefs

\begin{document}

\begin{frontmatter}
\title{Federated multiple testing with quantized $p$-values under
family-wise error control}
\runtitle{Federated multiple testing with quantized $p$-values}

\begin{aug}
\author[A]{\inits{P.}\fnms{Prasanjit}~\snm{Dubey}\ead[label=e1]{pdubey31@gatech.edu}}
\author[A]{\inits{X.}\fnms{Xiaoming}~\snm{Huo}\ead[label=e2]{huo@gatech.edu}}
\address[A]{H. Milton Stewart School of Industrial and Systems Engineering,
Georgia Institute of Technology, Atlanta, Georgia 30332,
U.S.A.\printead[presep={,\ }]{e1,e2}}
\end{aug}

\begin{abstract}
Several institutions test one family of hypotheses under family-wise error
control but cannot pool their data, so each site releases only a quantized
report of its $p$-value for each hypothesis, such as whether the $p$-value is
significant.  The report's resolution is the number of values it can take.
We quantify the power lost to such reports relative to the most powerful
centralized procedure given all $p$-values (the oracle), with reports and
combining rule chosen optimally.  A report from which the $p$-value can be
recovered loses no power, so the loss is due to compression, not
decentralization.  Under the stated regularity and power-objective conditions,
every finite resolution loses power, and the optimal loss decays as the
inverse square of the resolution: equal-width intervals attain this order, and
no partition into as many cells improves it.  In two-site Beta examples,
optimized one-bit reports are nearly lossless for a single hypothesis but lose
several times as much power for two.  At level $0.05$ with two
$\mathrm{Beta}(1,2)$ hypotheses, optimally combined three-bit equal-width
reports retain $99.2\%$ of oracle power.  We give the power-optimal combining
rule for given reports, and a rule using only their exactly known null
distribution, valid under arbitrary dependence across hypotheses.
\end{abstract}


\end{frontmatter}

\section{Introduction}\label{sec:intro}

Several institutions test one family of hypotheses without pooling their
data.  Each releases a report of its own evidence, and a center must
decide the whole family from those reports while controlling the probability
of any false rejection.
Replication projects, health-data networks, genetic consortia, and the
federated evaluation of medical-AI models all have this
form~\citep{klein2014investigating,platt2018sentinel,evangelou2013meta,
zenk2025fets,karargyris2023medperf}.  The reports often take only a few
values.  A research synthesis that records only whether each
study was significant is a vote count~\citep{wilkinson1951statistical,hedges1980vote},
and Many Labs~1 tabulated this summary for each of its samples~\citep{klein2014investigating}.
Disclosure control rounds released values, groups them into bands, or
suppresses them~\citep{hundepool2012statistical}, and in decentralized detection each
sensor sends a local decision because communication is
limited~\citep{tenney1981detection,chair1986optimal,tsitsiklis1993decentralized}.
This paper asks how much power such a quantized report loses relative to the best
procedure that sees every $p$-value, and how the center should read the
reports it receives.

There are $K$ hypotheses and $S$ sites (the institutions); site $s$ computes from its own data
a $p$-value $u_{s,k}$ for each hypothesis $k$.  The $p$-values are
independent, uniform on $[0,1]$ when the null holds, and have a known
decreasing density under the alternative, which may differ across sites.
The alternative-to-null density ratio is called the \emph{likelihood ratio}.
Site $s$ sends the center only $T_s(u_{s,k})$, where the \emph{report map}
$T_s$ is fixed in advance, is the same for every hypothesis, and takes at
most $m$ values.  We call $m$ the \emph{resolution} of the reports.  Each
value identifies a \emph{report cell}: the set of $p$-values mapped to it.
A report is thus a \emph{quantized $p$-value}: it discloses only which cell
contains $u_{s,k}$.
The \emph{threshold report} is $\mathbf 1\{u_{s,k}\le t_s\}$ for a threshold
$t_s$, and the \emph{equal-width report} says which of the $m$ intervals of
length $1/m$ contains $u_{s,k}$.  The center sees the $SK$
reports and nothing else, and returns a set of rejected hypotheses.  All $K$ hypotheses are
decided at once, which we call \emph{multiplicity}, under strong control of
the family-wise error rate (FWER): whatever the configuration of true and false
hypotheses, the probability of rejecting any true hypothesis is at most the
prespecified level $\alpha\in(0,1)$.

The power $\Pi_K(T,\mathcal R)$ of a report design $T=(T_1,\ldots,T_S)$ and a
center rule $\mathcal R$ is the expected fraction rejected when all $K$
hypotheses are false (the \emph{full alternative}).  Let $\pi^*_K(\alpha)$
denote the centralized oracle power: the largest power attainable from all
$SK$ $p$-values~\citep{rosset2022optimal}.  Define
\[
 \delta_K(m)\;=\;\pi^*_K(\alpha)\;-\;
 \sup_{T}\;\sup_{\mathcal R}\;\Pi_K(T,\mathcal R),
\]
the outer supremum running over report designs of resolution at most $m$
and the inner one over center rules that control the family-wise error rate.
A center holding the $p$-values could form the reports itself, so
$\delta_K(m)\ge0$: this \emph{deficit} is the power lost to the reports when
everything else is done as well as possible.

Formally, we ask how $\delta_K(m)$ depends on the resolution $m$ and on the
number of hypotheses $K$.  Section~\ref{sec:setup} gives the full formulation, including
heterogeneous sites, arbitrary measurable report cells, and power averaged
under a prior on the number of false hypotheses.

Quantized reports lose power even for one hypothesis once there are two or
more sites.  With one site the oracle rejects when the $p$-value is below
$\alpha$, which one bit reproduces; with two or more it rejects when the
product of the sites' likelihood ratios exceeds a threshold, a curved region
that no finite union of Cartesian products of report cells reproduces.  Multiplicity adds
boundaries of this kind: a family requires the center to choose how many
hypotheses to reject as well as which, so the optimal centralized rejection
region is assembled from one smooth piece for each possible number of
rejections, and the pieces meet at corners (Figure~\ref{fig:oracle_kink}).
A report that discloses only the cell containing a $p$-value replaces the
site's likelihood ratio by its average over that cell, and a rule that reads only the
reports cannot follow these boundaries.  Section~\ref{sec:setup} makes this
precise through the oracle's optimization problem over rejection rules and
its dual.

\begin{figure}[!tbp]
\centering
\includegraphics[width=0.5\linewidth]{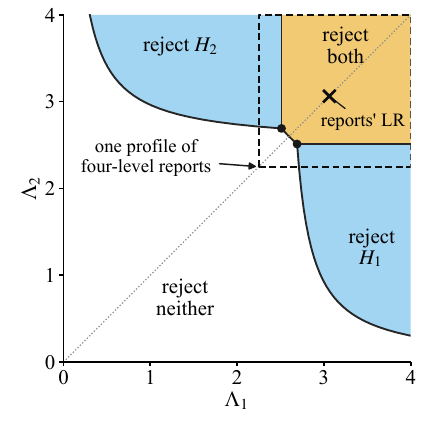}
\caption{Quantized reports lose power because one report profile can span all
four of the oracle's actions.  The plane of the two aggregate likelihood ratios
for $K=S=2$ homogeneous $\mathrm{Beta}(1,2)$ sites at $\alpha=0.05$ under the
point prior, with the oracle's action in each region: reject neither
hypothesis, only the one with the larger likelihood ratio, or both.  The
boundaries follow from the minimizing multiplier $\mu^*=(1.256,0.226)$ of the
continuous dual (Section~\ref{sec:kink}) and meet at corners (dots).  The dashed square holds every likelihood-ratio
pair compatible with one profile of equal-width four-level reports, in which
both sites report the most significant cell for both hypotheses; a rule that
reads only the reports must act alike whenever this profile occurs, although
the oracle's action changes across the square.  The cross marks the reports'
own likelihood ratios (LR), and the dotted line is the diagonal
$\Lambda_1=\Lambda_2$.}
\label{fig:oracle_kink}
\end{figure}

Three lines of existing work border this problem, and each leaves open the
question asked here.
Optimal testing under family-wise error control supplies the benchmark.
Following earlier optimality results for multiple comparison
procedures~\citep{spjotvoll1972optimality}, it has produced maximin stepwise
procedures, Bayes and minimax programs, and exact solutions for two
hypotheses~\citep{lehmann2005optimality,rosenblum2014optimal,
heller2023optimal}.  \citet{rosset2022optimal} formulate the optimization
problem whose value is $\pi^*_K(\alpha)$, for a model unchanged by relabeling
the hypotheses (exchangeable).  Later work computes it, for
three hypotheses from necessary optimality
conditions~\citep{dubey2025necessary} and in general through identities for
elementary symmetric polynomials (sums of products of distinct likelihood ratios; \citealp{dubey2026esp}), and extends the optimal test to
block-structured families~\citep{dubey2026boost}.
Related power-directed closed tests, built on the closure
principle~\citep{marcus1976closed}, provide practical
procedures~\citep{karmakar2025bottomup}.  In all of this work the center sees
the data.  That
literature tells us what $\pi^*_K(\alpha)$ is; it does not ask what a site
must transmit for a center to approach that value.  Our results do not depend
on which algorithm computes the oracle.

Distributed multiple testing has been studied mainly under the false
discovery rate (FDR), through one-shot aggregation, data-shielded
high-dimensional procedures, and network methods~\citep{su2015knockoffaggregation,
liu2021integrative,ramdas2017qute,pournaderi2023largescale,hu2026federated}.
\citet{ramdas2022quteexpanded} develop procedure-specific quantization for their network testing rule and the Benjamini--Hochberg procedure;
\citet{xiang2019quantization} combines quantized $p$-values with a discrete
FDR procedure, and \citet{golz2022spatial} quantize sensor $p$-values for
spatial FDR control.  Summary-statistic knockoffs give family-wise control for
feature selection~\citep{yu2024summary}.  Our question is the best strong-FWER
power attainable when each site reports on the same hypotheses using
finitely many values.

Decentralized detection supplies the $m^{-2}$ order and the
sufficiency argument, but for a single global test.  What is new here is that
the same order governs the strong-FWER deficit under multiplicity, where the
oracle solves a linear program and the loss can vanish for some priors
(Section~\ref{sec:kink}), with a matching lower bound over arbitrary cells.
We also measure how multiplicity increases the loss of optimized one-bit
threshold reports.  Related decentralized and federated testing work studies
communication or privacy constraints under Bayes-risk, error-exponent, divergence, separation, or
minimax-rate criteria~\citep{tenney1981detection,ahlswede1986hypothesis,han1987hypothesis,
tsitsiklis1993decentralized,berlinet2006asymptotic,
szabo2023distributed,cai2024federated,pensia2024communication}.  Losslessness under sufficiency
is an instance of Blackwell's comparison of
experiments~\citep{blackwell1953equivalent}, and $m^{-2}$ is
the classical many-level (high-rate) quantization order~\citep{poor1988fine,benitz1989asymptotically,
gray1998quantization,gupta2003highrate}, also at a fixed false-alarm (type~I
error) level~\citep{villard2011highrate}.

The rules that read the reports draw on classical work.  Combining discrete
$p$-values through Fisher's statistic is an old
problem~\citep{lancaster1949combination,kincaid1962combination}, and
meta-analyses that receive only the $p$-values below a threshold impute the
rest~\citep{tang2014imputation}.  The optimal fusion of independent local
decisions about one hypothesis is a weighted vote, under Bayes
risk~\citep{chair1986optimal} and at a fixed false-alarm
level~\citep{thomopoulos1987optimal}.  The default rule of this paper applies
these constructions hypothesis by hypothesis, with exact null calibration and
a Bonferroni step.

We control the family-wise error rate~\citep{lehmann2005generalizations,goeman2011exploratory}, as regulated multi-endpoint analyses require~\citep{dmitrienko2009multiple,fda2022endpoints}, because the motivating families are
confirmatory and each rejection is acted on individually; FDR suits
exploratory screening, in which false rejections are tolerated in
proportion~\citep{benjamini1995controlling,goeman2014multiple}.
Each hypothesis also has one status shared by all sites, so replicability analysis,
which asks at how many sites an effect is present through partial-conjunction
nulls, is a different problem~\citep{benjamini2008screening,
bogomolov2018testing,bogomolov2023replicability}.

The results are the five statements below; Sections~\ref{sec:setup}
and~\ref{sec:method} give their formal versions.  Statement~(i) is the
baseline that defines the problem, (ii)--(iv) are the main results, and~(v)
is a default that needs no model of the alternative.
Results (ii) and (iii) assume $S\ge2$ sites.  Section~\ref{sec:setup} also
states them for power averaged over the number of false hypotheses; for
$K\ge2$ they then need the average to put positive weight on the full
alternative, as every average used here does, or the oracle condition of
Section~\ref{sec:kink}.
\begin{enumerate}[label=(\roman*),leftmargin=*,itemsep=3pt]
\item \emph{Reporting the $p$-values themselves loses no power.}  If each
$T_s$ is injective, so that the center can recover the $p$-values, the
deficit is zero for every $K$ and $S$, as sufficiency implies
(Proposition~\ref{prop:distribution}).  Measured against the $p$-value
oracle, any loss is therefore due to compression, not to keeping the data at
the sites.
\item \emph{Every design of finite resolution loses power.}  If the alternative
densities are continuous and strictly decreasing, every design whose maps
take finitely many values has power strictly below $\pi^*_K(\alpha)$, and
threshold reports have a positive deficit even when their thresholds are
chosen optimally (Propositions~\ref{prop:distribution}--\ref{prop:converse} and
Theorem~\ref{thm:oracle_kink}(a)).
\item \emph{The loss decays as the inverse square of the resolution, and no
placement of the cells improves the exponent.}  If the alternative densities
are moreover continuously differentiable on $[0,1]$, with negative derivative
on $[0,1)$, there are constants $0<c<C<\infty$ and $m_0$ such that
\[
 c\,m^{-2}\;\le\;\delta_K(m)\;\le\;C\,m^{-2},\qquad m\ge m_0 .
\]
Equal-width cells attain the upper bound, and the lower bound holds over
every way of dividing $[0,1]$ into $m$ measurable cells, so concentrating
the cells near zero can improve the constant but not the order (Proposition~\ref{prop:converse}(b) and
Theorem~\ref{thm:oracle_kink}(b)).  With $b$ bits per report, $m=2^b$ and the
optimal deficit is therefore of order $4^{-b}$.
\item \emph{The optimal way to read reports of a given resolution.}  Once
the report maps are fixed, the reports have known distributions under the
null and under the modeled alternative, so they pose a centralized multiple
testing problem of their own.  Applied to that problem, with the likelihood
ratios of the reports in place of those of the $p$-values, the optimal
centralized procedure is computed by a finite linear program when the reports
take finitely many values; it controls the
family-wise error rate in finite samples and is the most powerful rule
that reads only the reports (Theorem~\ref{thm:aggregator}).  We call it the
\emph{model-aware} rule.
\item \emph{A default that needs no model for the alternative.}
The null distribution of the reports depends on the report maps alone.
For interval cells, sum the negative logarithms of the reported cells' upper
edges across sites (Fisher's statistic; \citealp{fisher1932statistical}),
calibrate the sum against its exact null distribution, and test each
hypothesis at level $\alpha/K$ (Bonferroni's correction;
\citealp{dunn1961multiple}).
We call this the \emph{null-calibrated} rule.  It controls the family-wise
error rate under arbitrary dependence among hypotheses, needing only
independence across sites within each hypothesis.  For a single hypothesis,
with sites sharing one alternative density and one threshold, the rule
randomized at its boundary is also most powerful
(Propositions~\ref{prop:marginal} and~\ref{prop:marginalK1}).
\end{enumerate}

To our knowledge, the strictness and rate results (ii) and (iii) have not
been established for the strong-FWER deficit under multiplicity with at
most $m$ report values per site.  The protocol is one-shot and
componentwise, each report depending on a single $p$-value; reports that
combine several $p$-values, multi-round exchanges, and reports of varying
length are outside its scope.

The examples of Section~\ref{sec:numerics}, which use $\mathrm{Beta}(1,2)$
alternatives (the $p$-value density $2(1-u)$) at $\alpha=0.05$ and two sites
unless stated otherwise, suggest three conclusions for practice.  First,
\emph{multiplicity amplifies the cost of one-bit threshold reports}: with
optimized thresholds they lose $1.1\%$ of oracle power for one hypothesis and $8.5\%$
for two, and, with a numerically selected common threshold, at least $11.6\%$
for eight, measured against two-bit power, which cannot exceed the
uncomputed eight-hypothesis oracle.  Second, \emph{a small increase in resolution preserves nearly all
oracle power}: for two hypotheses, optimally combined two- and three-bit
equal-width reports retain $97.1\%$ and $99.2\%$ of it.  Third, \emph{how the
center reads the reports matters at least as much as their resolution}, because in
these examples testing each hypothesis separately costs far more power than
quantizing its $p$-values.  When the alternative is modeled and the
hypotheses are independent, the model-aware rule on one-bit reports from two
or three sites therefore beats, under multiplicity, even the most powerful Bonferroni
rule on the full $p$-values, which tests each hypothesis by its likelihood
ratio at level $\alpha/K$.  It has $57$--$67\%$ more power for two hypotheses
and more than four times as much for eight, where at two sites it has power
$0.120$ against $0.023$.  Result (iii) establishes the optimal asymptotic order; the
finite-resolution figures depend on the model and the report design.

Reanalyses of released results show how the report design changes what a
valid rule can detect.  In a
federated tumor-segmentation benchmark (Section~\ref{sec:numerics_fets}) and in
the Many Labs~1 replication project (Section~\ref{sec:application}), a one-bit design that requires every site to agree loses findings that counting the
sites' one-bit reports recovers, and, in the segmentation data, refining the
report adds further findings.  Section~\ref{sec:design} explains why requiring agreement is
costly when there are many sites: for a single hypothesis with homogeneous
$\mathrm{Beta}(1,\theta)$ alternatives ($p$-value density
$\theta(1-u)^{\theta-1}$ at every site), the probability that the unanimity
rule misses a false hypothesis decays only polynomially in the number of
sites, whereas that of the optimal count rule, which rejects when enough
sites have $p$-values below their thresholds (local rejections), decays
exponentially.

Section~\ref{sec:setup} states the loss results, Section~\ref{sec:method} develops the
two rules and their use with estimated alternatives, Section~\ref{sec:numerics} reports the numerical study and the
reanalyses, and Section~\ref{sec:discussion} discusses extensions.

\section{Problem formulation and communication-loss theory}\label{sec:setup}

This section sets up the model and the report designs, defines the power
deficit against the centralized oracle, and shows when finite resolution
loses power and how fast refining the reports recovers it, for one hypothesis
(Proposition~\ref{prop:converse}) and under multiplicity
(Theorem~\ref{thm:oracle_kink}).

\subsection{Model}\label{sec:model}

We test $K$ hypotheses $H_1,\dots,H_K$, each true or false, from $p$-values
computed at $S$ sites that do not pool their data.  A configuration vector
$h \in \{0,1\}^K$ records which are false: $h_k = 1$ indicates that $H_k$ is
false, so that its alternative holds, and $\mathcal{H}_0(h) = \{k : h_k =
0\}$ denotes the true hypotheses. Throughout, $K \ge 1$ and $S \ge 1$ are integers and the level
$\alpha \in (0,1)$ is fixed.  Site $s\in\{1,\ldots,S\}$ holds a
$p$-value $u_{s,k} \in [0,1]$ for each hypothesis $k\in\{1,\ldots,K\}$,
computed from data local to that site.  Let $g_s$ be the density of
$u_{s,k}$ when $H_k$ is false, the same for every $k$ by
Assumption~\ref{ass:model} below; because the density of $u_{s,k}$ is one on
$[0,1]$ when $H_k$ is true, $g_s$ is also the alternative-to-null likelihood
ratio, and site $s$ forms the local likelihood ratio
$\Lambda_{s,k}=g_s(u_{s,k})$.

The constraint is \emph{federation} in the data-locality sense:
individual-level data remain at their sites, and the center sees each
site only through the prespecified per-hypothesis reports it discloses.  No
formal privacy guarantee, such as differential privacy, is imposed, and
leakage is not quantified beyond the choice of the report class.

\begin{assumption}[Distributed two-group model]\label{ass:model}
Conditionally on $h$, the variables $\{u_{s,k} : s = 1,\dots,S;\ k =
1,\dots,K\}$ are mutually independent. For every site $s$, $u_{s,k} \sim
\mathrm{Unif}[0,1]$ if $h_k = 0$ and has density $g_s$ if $h_k = 1$, where
each $g_s$ is a
\emph{nonincreasing} density on $[0,1]$.
\end{assumption}

\begin{assumption}[Known local alternatives]\label{ass:known}
Each $g_s$ is known to the sites and to the center.
\end{assumption}

\begin{assumption}[Regularity]\label{ass:reg}
Each $g_s$ is continuous and strictly decreasing on $[0,1]$, so $u \mapsto
g_s(u)$ is a bijection onto its range.
\end{assumption}

The assumptions are used unevenly: the model-aware rule needs the two-group
model and the known alternative, which let the center compute the exact
distribution of the reports; the null-calibrated rule needs neither the known
alternative nor independence across hypotheses; and
regularity, which makes $p$-values and likelihood ratios
information-equivalent, serves the strictness and rate results.
Assumption~\ref{ass:model} supplies independent replication across sites and
an exchangeable two-group structure across hypotheses, one unchanged by
relabeling them.  Heterogeneity across sites enters through $g_s$, but the
common $g_s$ across hypotheses is structural, because replacing it by
$g_{s,k}$ would remove the exchangeability that the oracle's optimization
problem uses (Section~\ref{sec:kink}).
Table~S4 of the supplement lists the assumptions of each result.

Throughout, $\pr_h$ and $\E_h$ denote probability and expectation at
configuration $h$, with $0$ and $1$ abbreviating the constant
configurations: $\pr_0$ and $\E_0$ refer to the global null $h=(0,\ldots,0)$,
where every $u_{s,k}$ is uniform, and $\pr_1$ and $\E_1$ to the full
alternative $h=(1,\ldots,1)$.  Under Assumption~\ref{ass:model},
$\E_h(\Lambda_{s,k}) = \int g_s(u)\,\mathrm{d}u = 1$ whenever $h_k = 0$,
and by cross-site independence the product
\begin{equation}\label{eq:aggLR}
\Lambda_k \;=\; \prod_{s=1}^S \Lambda_{s,k} \;=\; \prod_{s=1}^S g_s(u_{s,k})
\end{equation}
is the likelihood ratio of all the evidence about hypothesis $k$, with null
expectation one; we write $\Lambda=(\Lambda_1,\ldots,\Lambda_K)$ for the
vector of these ratios.

\subsection{Federated reports and their resolution}\label{sec:reports}

Each site summarizes each hypothesis from that hypothesis's own local
$p$-value, through a single report map shared across hypotheses; we call such
designs \emph{componentwise}.

\begin{definition}[Federated report design and procedure]\label{def:protocol}
A report design fixes, for each site $s\in\{1,\ldots,S\}$, a measurable
\emph{report map} $T_s : [0,1] \to \mathcal{M}_s$ that does not depend on
$k$; here $\mathcal M_s$ is the site's message space, a measurable space
whose elements are the reports the site may send.  The site transmits
$y_{s,k} = T_s(u_{s,k}) \in \mathcal{M}_s$ for every $k\in\{1,\ldots,K\}$,
and we write $\mathbf y =
\{y_{s,k}\}_{s=1,\dots,S;\,k=1,\dots,K}$ for the resulting report array. The
center applies a (randomized) rule $\mathcal{R} : (\mathbf y,\xi) \mapsto
\mathcal{R}(\mathbf y,\xi) \subseteq \{1,\dots,K\}$ (Section~\ref{sec:criterion}), where
$\xi$ is internal randomization. A federated procedure is the pair comprising
a report design and a center rule. Each site transmits once; the center
may announce the design to the sites beforehand.
\end{definition}

Refining a report never lowers the power a center can attain.  Reports are
ordered in the Blackwell sense~\citep{blackwell1953equivalent}: if one report
is a measurable function of another, as when each cell of one quantized report
is a union of cells of the other, no center rule attains higher power from the
first than the optimal rule attains from the second.  The constrained quantity
is the \emph{resolution} of the report maps; full and threshold reports,
defined next, anchor the spectrum.

\begin{definition}[Report classes]\label{def:classes}
A site uses a \emph{full report} if $T_s$ is injective on $[0,1]$ with
Borel range in a standard Borel space.  The inverse is then measurable by the
Lusin--Suslin theorem, so the center recovers $u_{s,k}$, equivalently
(Assumption~\ref{ass:reg}) $\Lambda_{s,k}$, and up to this relabeling
$T_s(u) = u$.  Write $\mathcal{T}^{\mathrm{full}}$ for the resulting class of
report designs. A site uses a \emph{threshold report} if $T_s(u) =
\mathbf{1}\{u \le t_s\}$ for a local threshold $t_s \in (0,1]$; write $\mathcal{T}^{\mathrm{thr}}$ for the
resulting class, indexed by the thresholds $t = (t_1,\dots,t_S)$.  A site
uses the \emph{equal-width $m$-level report}, the uniform quantizer of
quantization theory, if
$T_s(u)=\min\{m,1+\lfloor mu\rfloor\}$, which discloses which of $m$
equal-width cells contains $u$; write $\mathcal T^{(m)}$ for the singleton
class in which every site uses it.
\end{definition}

We call a report \emph{sufficient} when the center can recover
$\Lambda_{s,k}$ from it, as it can from a full report; the report array then
carries everything the $p$-value array says about $h$.  Finite-range maps
interpolate between the full and threshold classes;
Sections~\ref{sec:criterion}--\ref{sec:kink} give conditions for strict loss and
the rate at which refinement closes it.

The lower bounds below range over $\mathcal T_{\le m}$, the class of all
deterministic componentwise report designs of Definition~\ref{def:protocol}
in which every report map has at most $m$ output values,
$|\operatorname{range}(T_s)|\le m$, for an integer $m\ge1$. The report cells may be arbitrary
measurable sets, and the maps may differ across sites.

\subsection{Criterion, benchmark, and the power loss from quantized reports}\label{sec:criterion}

The criterion is finite-sample strong control of the family-wise error rate,
and the objective is average power weighted by a prior on the number of false
hypotheses.  With $\mathcal R=\mathcal R(\mathbf y,\xi)$ the rejection set of
Definition~\ref{def:protocol} and all probabilities taken over the data and
the internal randomization $\xi$, the requirement is
\begin{equation}\label{eq:fwer}
\fwer_h \;=\; \pr_h\{\mathcal{R} \cap \mathcal{H}_0(h) \neq \emptyset\}
\;\le\; \alpha \quad \text{for every } h \in \{0,1\}^K .
\end{equation}
How many hypotheses are false is not known when a procedure must be chosen,
so we do not fix it.  Which ones are false is immaterial, because Assumption~\ref{ass:model} is invariant under
relabeling the hypotheses, so we average over that as well.  For
$\gamma=1,\dots,K$ let
\[
 \Pi_\gamma \;=\; \binom{K}{\gamma}^{-1}\sum_{h:\,|h|=\gamma}\gamma^{-1}\,
 \E_h\bigl|\mathcal R\cap\{k:h_k=1\}\bigr|,
 \qquad |h|=\sum_k h_k,
\]
be the expected fraction of false hypotheses rejected, averaged over the
configurations with exactly $\gamma$ false hypotheses.  The
supplement shows that an optimal procedure can be taken to be equivariant
under relabeling of the hypotheses; such a procedure has the same power at
every configuration with $\gamma$ false hypotheses, and $\Pi_\gamma$ is that
common value.  Let $w$ be a probability
distribution on $\{1,\dots,K\}$, the prior on the number of false
hypotheses.  The objective is the \emph{prior-weighted average power}
\begin{equation}\label{eq:objective}
 \Pi_w \;=\; \sum_{\gamma=1}^{K}w(\gamma)\,\Pi_\gamma .
\end{equation}
The point prior, which puts
$w(K)=1$, gives average power at the full alternative,
$K^{-1}\E_1|\mathcal R|$, the objective of the centralized
theory~\citep{rosset2022optimal,dubey2026esp}; the flat prior
$w(\gamma)=1/K$ gives equal weight to each positive number of false
hypotheses.  Except where a prior is named, reported numbers use the point prior;
Section~\ref{sec:numerics} and the supplement recompute the principal quantities
at nondegenerate priors.

The benchmark is the centralized $p$-value oracle, which observes the full
array $\{u_{s,k}\}$. It is not an oracle that pools the patient- or
study-level data from which the local $p$-values were computed.  The
supplement proves attainment and strong duality for this benchmark's
power-maximization program (Section~\ref{sec:kink}) and for the same program
under any dominated report distribution.

\begin{proposition}[Centralized oracle]\label{prop:oracle}
Under Assumptions~\ref{ass:model}--\ref{ass:known}, an optimal procedure
satisfying~\eqref{eq:fwer} on the full $p$-value array can be chosen to depend on the
data only through the aggregate likelihood ratios~\eqref{eq:aggLR} and
auxiliary randomization.  It may be taken to be a likelihood-ratio-ordered
policy for the exchangeable two-group model~\citep{rosset2022optimal}, that
is, a rule that ranks the hypotheses by $\Lambda_k$ and rejects a number of
the top-ranked ones that depends on the ranked values.
Denote its average power by $\pi^{*}_w(\alpha)$.
\end{proposition}

The oracle thus summarizes each hypothesis by one number, the product
$\Lambda_k$ of the sites' likelihood ratios, and a report design determines
how much of each factor reaches the center.
The quantity of interest is the \emph{power deficit} of a report class
$\mathcal T$, a set of report designs $T=(T_1,\ldots,T_S)$ such as
$\mathcal T^{\mathrm{full}}$, $\mathcal T^{\mathrm{thr}}$ or
$\mathcal T_{\le m}$ of Section~\ref{sec:reports}:
\begin{equation}\label{eq:deficit}
\delta_w(\mathcal{T}) \;=\; \pi^{*}_w(\alpha) \;-\;
\sup_{T\in\mathcal T}\ \sup_{\mathcal R\text{ satisfying }\eqref{eq:fwer}}\,
\Pi_w(T,\mathcal R)
\;\ge\; 0 ,
\end{equation}
where $\Pi_w(T,\mathcal R)$ is the objective~\eqref{eq:objective} of design
$T$ and rule $\mathcal R$.  For a
fixed report design $T$, write
\[
d_w(T)=\pi^*_w(\alpha)-\sup_{\mathcal R}\Pi_w(T,\mathcal R),
\]
where the supremum is over valid center rules.  We call $d_w(T)$ the
\emph{fixed-design gap} of $T$, so that
$\delta_w(\mathcal T)=\inf_{T\in\mathcal T}d_w(T)$.  The gap and the deficit
differ: a positive gap at every design in a class need not give the class a
positive deficit, because the infimum over infinitely many designs can be
zero.

For the point prior we abbreviate $\Pi_w$, $\pi^*_w(\alpha)$, $\delta_w$ and $d_w$ to $\Pi_K$,
$\pi^*_K(\alpha)$, $\delta_K$ and $d_K$.  Thus $\Pi_K$ is the component $\Pi_\gamma$ at
$\gamma=K$.  On $\pi^*$, $\delta$ and $d$, an integer subscript
names the number of hypotheses under the full alternative and the subscript
$w$ a general prior; in this notation the quantity $\delta_K(m)$ of
Section~\ref{sec:intro} is $\delta_K(\mathcal T_{\le m})$.  The most powerful
center rule for a fixed report design, the model-aware rule of
Section~\ref{sec:aware}, gives $d_w(T)$ at that design; reaching
$\delta_w(\mathcal T^{\mathrm{thr}})$ requires optimizing the thresholds $t$
as well.

\begin{proposition}[Federation with sufficient reports loses no power]\label{prop:distribution}
Under Assumptions~\ref{ass:model}--\ref{ass:known} for part~(a), and
Assumptions~\ref{ass:model}--\ref{ass:reg} for part~(b):
\begin{itemize}
\item[(a)] sufficient reports lose no power: $d_w(T)=0$ for every design $T$
whose reports are sufficient, so $\delta_w(\mathcal{T}^{\mathrm{full}})=0$,
for every $K$ and $S$;
\item[(b)] at $K=S=1$, the threshold $t_1=\alpha$ gives
$\delta_1(\mathcal T^{\mathrm{thr}})=0$.  If $K=1$ and $S\ge2$, every
threshold vector has a positive fixed-design gap and compactness of
$[0,1]^S$ gives $\delta_1(\mathcal T^{\mathrm{thr}})>0$.
\end{itemize}
\end{proposition}

Proposition~\ref{prop:distribution} places the cost of federation in
compression.  Measured against the $p$-value oracle, keeping the data at the
sites costs nothing when the reports are sufficient, because a sufficient
report reproduces the oracle statistic; in estimation, summary statistics can
likewise match individual-level data~\citep{lin2010relative}.  A single site
testing a single hypothesis loses nothing by reporting only whether its
$p$-value is below $\alpha$, but with two or more sites even the best
thresholds lose power.  For fixed thresholds the optimal power of
threshold reports is the value of a finite linear program, and the outer
optimization over $t=(t_1,\ldots,t_S)$ is carried out numerically in
Section~\ref{sec:numerics}.  The next result settles the single-hypothesis case for $S\ge2$:
every finite-range report loses power, and, for smooth alternatives, refining the report recovers it
at the inverse-square rate.  Under multiplicity both questions, strictness and
rate, are answered in Section~\ref{sec:kink}, for every finite-range design at once.

\begin{proposition}[Single hypothesis: strict loss and the inverse-square rate]\label{prop:converse}
Suppose Assumptions~\ref{ass:model}--\ref{ass:reg} hold, $K=1$ and $S\ge2$.
\begin{itemize}
\item[(a)] Every report design in which all maps $T_s$ have finite ranges,
$|\operatorname{range}(T_s)| < \infty$, has a strictly positive fixed-design gap
$d_1(T)$.
\item[(b)] If in addition each $g_s\in C^1[0,1]$ with
$g_s'<0$ on $[0,1)$, there are constants $0<c<C<\infty$ and $m_0<\infty$,
depending on $S$, $\alpha$ and the densities, such that
\[
 c m^{-2}\le \delta_1(\mathcal T_{\le m})
 \le \delta_1(\mathcal T^{(m)})\le C m^{-2},\qquad m\ge m_0.
\]
\end{itemize}
\end{proposition}

Part~(b) says what resolution can buy.  With $b$ bits per report the
smallest attainable loss is of order $4^{-b}$, and because the lower bound
permits arbitrary measurable report cells, not only intervals, no placement of
the cells improves on equal-width cells by more than a constant factor.  The
upper bound holds for every $K$ and $S$ under the same smoothness
(supplement, Lemma~S1), and the conditions hold, for example, for
$\mathrm{Beta}(1,\theta)$ alternatives with $\theta\ge2$.

\subsection{Strict loss and the sharp rate under multiplicity}\label{sec:kink}

Strictness and the rate both come from the boundaries between the oracle's
actions; under multiplicity these boundaries also meet at corners.  By
Proposition~\ref{prop:oracle} the oracle ranks the hypotheses by
$\Lambda_k$ and chooses how many of the top-ranked ones to reject, so its
rejection region in the space of $(\Lambda_1,\ldots,\Lambda_K)$ is assembled
from one smooth piece per rejection count; Figure~\ref{fig:oracle_kink} draws
the pieces for two hypotheses.  A componentwise report gives the center only
the average of each factor $g_s(u_{s,k})$ of $\Lambda_k$ over the reported
cell, its null conditional expectation there.  A rule that reads only the
reports must therefore act alike on all data that produce one report profile,
that is, one array of reports, and the likelihood ratios of such data can
straddle a boundary between two of the oracle's actions.

The boundaries can be located through the dual of the oracle's optimization
problem.  By Proposition~\ref{prop:oracle} it suffices to choose, for each
vector $v=(v_1,\ldots,v_K)$ of aggregate likelihood ratios, with ordered
coordinates $v_{(1)}\ge\cdots\ge v_{(K)}$, the probability $r_l(v)$ of
rejecting the $l$ largest, $l=0,\ldots,K$.  Since every probability is a null
expectation weighted by the false hypotheses' likelihood ratios, the oracle
power is the value of the linear program
\[
 \text{maximize }\E_0\Bigl\{\sum_{l}r_l(\Lambda)A_l(\Lambda)\Bigr\}
 \text{ subject to }\E_0\Bigl\{\sum_{l}r_l(\Lambda)B_{j,l}(\Lambda)\Bigr\}
 \le\binom Kj\alpha\quad(j=1,\ldots,K)
\]
over $r_l\ge0$ with $r_0+\cdots+r_K=1$ (supplement, Section~S2).  Here $A_l$
is the \emph{objective coefficient} of rejecting the top $l$, equal under the
point prior to $l/K$ times the full alternative's density $v_1\cdots v_K$
relative to the global null, and given for general $w$ in the supplement, and
\[
 B_{j,l}(v)=e_{K-j}(v)-v_{(1)}\cdots v_{(l)}\,
 e_{K-l-j}\bigl(v_{(l+1)},\ldots,v_{(K)}\bigr)\;\ge\;0
\]
is its \emph{family-wise error coefficient} at $j$ true hypotheses, where
$e_d$ is the elementary symmetric polynomial of degree $d$, with $e_0=1$ and
$e_d=0$ when $d<0$ or $d$ exceeds the number of its arguments.  Thus
$e_{K-j}(v)$ sums the densities, relative to the global null, of the
configurations with $j$ true hypotheses, each the product of the likelihood
ratios of its false hypotheses, and the subtracted term removes those in
which the top $l$ are all false; each constraint sums the family-wise error
rate over these configurations, at each of which a symmetric rule has the
same rate.  The
\emph{continuous dual} of this program,
an equivalent minimization problem for the oracle power before any
quantization, has one multiplier $\mu_j\ge0$ for each constraint and the
objective
\[
 \mathcal D(\mu)=\alpha\sum_{j=1}^K\binom Kj\mu_j+\E_0\{F_\mu(\Lambda)\},
 \qquad
 F_\mu(v)=\max\{0,\psi_1^\mu(v),\ldots,\psi_K^\mu(v)\},
\]
where the ordered action branch $\psi_l^\mu(v)=A_l(v)-\sum_j\mu_jB_{j,l}(v)$
corresponds to rejecting the top $l$ likelihood ratios.  The oracle power is
$\pi_w^*(\alpha)=\min_{\mu\ge0}\mathcal D(\mu)$.  We call $F_\mu$ the
\emph{dual integrand}.  At a minimizing multiplier $\mu^*$ the oracle takes,
at almost every $v$, an action whose branch attains the maximum in
$F_{\mu^*}(v)$, so the boundaries in Figure~\ref{fig:oracle_kink}, which we
call \emph{decision boundaries}, are where the maximizing branch changes.  Let
$\lambda_-=\prod_s g_s(1)$ and $\lambda_+=\prod_s g_s(0)$ be the smallest and
largest values an aggregate likelihood ratio can take, and write
$\Lambda_{-a}=(\Lambda_k:k\ne a)$ for the coordinates other than $a$.

The next condition, which concerns the oracle alone and not the reports, asks
that some decision boundary be a genuine kink of the dual integrand.  It
holds automatically whenever the objective puts positive weight on the full
alternative (Lemma~\ref{lem:oak_auto} below), so it matters only for other
objectives.

\begin{definition}[Oracle active-kink condition]\label{def:oak}
The \emph{oracle active-kink condition} (OAK) holds if there exist a
minimizing continuous-dual multiplier $\mu^*$, a coordinate
$a\in\{1,\ldots,K\}$, a measurable set
$\mathcal B\subset\mathbb R_+^{K-1}$ with
$\pr_0(\Lambda_{-a}\in\mathcal B)>0$, and measurable functions
$q:\mathcal B\to(\lambda_-,\lambda_+)$ and $\Delta:\mathcal B\to(0,\infty)$
such that, at almost every $z\in\mathcal B$ under the null distribution of
$\Lambda_{-a}$, the convex piecewise-affine section
\[
 f_z(x)=F_{\mu^*}(z_1,\ldots,z_{a-1},x,z_{a+1},\ldots,z_K)
\]
has an upward derivative jump of at least $\Delta(z)$ at $q(z)$.  Any such
$(\mu^*,a,\mathcal B,q,\Delta)$ is a \emph{witness} for OAK.
\end{definition}

In words, OAK asks the dual integrand,
read along one likelihood-ratio coordinate, to have a kink (an upward jump in
derivative) at an interior point, on a positive-probability set of the remaining coordinates; we call such a kink \emph{active}.  In Figure~\ref{fig:oracle_kink}, crossing the
vertical boundary between the regions labeled \emph{reject $H_2$} and
\emph{reject both},
with $\Lambda_2$ held fixed, is such a kink (supplement, Lemma~S5).

\begin{lemma}[OAK at objectives charging the full alternative]\label{lem:oak_auto}
Suppose Assumptions~\ref{ass:model}--\ref{ass:reg} hold, $K\ge2$ and
$w(K)>0$.  Then OAK holds, and every minimizing continuous-dual multiplier
$\mu^*$ belongs to a witness.
\end{lemma}

Why Lemma~\ref{lem:oak_auto} holds is easiest to see with one hypothesis.  There
$F_\mu(v)=(v-\mu)_+=\max\{v-\mu,0\}$,
the oracle rejects when $\Lambda$ exceeds the null upper $\alpha$-quantile
$\mu^*$, and $F_{\mu^*}$ bends exactly at that threshold: it is zero where
the oracle accepts and rises where it rejects.  With several hypotheses the
same picture holds along each coordinate.  The set where $F_{\mu^*}$ vanishes
is where rejecting nothing is optimal, and its boundary, read along one
likelihood ratio with the others held fixed, carries a kink.  When
the objective puts positive weight on the full alternative, that is, when
$w(K)>0$, the set where $F_{\mu^*}$ is positive, where some rejection
strictly beats rejecting nothing, has null probability strictly between
zero and one.  That boundary is therefore crossed on a set of positive
probability, and a convex piecewise-affine function that is zero on one side
of a point and positive on the other bends there; the supplement gives the
proof.

The condition $w(K)>0$ cannot simply be dropped: with $K=3$, all prior
weight on a single false hypothesis, and signals so weak that
$\lambda_+\le2\lambda_-$, the dual integrand vanishes everywhere, so
rejecting all three hypotheses and rejecting none tie at every point.
Ignoring the data and rejecting all three with probability $\alpha$ is then
optimal, so no report loses any power, however few values it takes
(supplement, Section~S1.5).

What a report loses is how much averaging within its cells lowers the dual
objective.  Precisely, for a report design $T$ let $\mathcal G$ be the
sigma-field generated by the reports, and let
$\bar\Lambda=\E_0(\Lambda\mid\mathcal G)$ be the vector of cell-averaged
likelihood ratios, which is the likelihood-ratio vector of the reports
themselves.  Because $F_\mu$ is convex in each coordinate and the likelihood
ratios are conditionally independent given the reports, averaging can only
lower it, $\E_0\{F_\mu(\Lambda)\}\ge\E_0\{F_\mu(\bar\Lambda)\}$ for
every $\mu$, and the drop measures what the report loses at that multiplier.
Evaluating the drop at the oracle's multiplier $\mu^*$ gives a lower bound
on the fixed-design gap, and at the minimizer $\bar\mu$ of the reports'
dual, which replaces $\Lambda$ by $\bar\Lambda$, an upper bound:
\begin{equation}\label{eq:dual_pinch}
 \E_0\{F_{\mu^*}(\Lambda)-F_{\mu^*}(\bar\Lambda)\}
 \le d_w(T)\le
 \E_0\{F_{\bar\mu}(\Lambda)-F_{\bar\mu}(\bar\Lambda)\}.
\end{equation}

We call the left side the \emph{conditional-Jensen lower bound}; both
bounds follow by comparing the two duals at $\mu^*$ and $\bar\mu$
(supplement, Section~S1.4).  Averaging leaves the dual
integrand unchanged on report cells that lie within one action region, where
it is a single branch, affine in every coordinate, so both bounds are driven
by the cells that straddle a decision boundary.

\begin{theorem}[Finite-report loss and the optimal componentwise resolution rate]
\label{thm:oracle_kink}
Suppose Assumptions~\ref{ass:model}--\ref{ass:reg} hold, $K\ge2$ and
$S\ge2$, and that either $w(K)>0$ or OAK holds.  Then:
\begin{itemize}
\item[(a)] Every fixed report design whose report maps have finite ranges has
$d_w(T)>0$, and compactness of $[0,1]^S$ gives
$\delta_w(\mathcal T^{\mathrm{thr}})>0$ for the threshold-report class.
\item[(b)] If, in addition, every $g_s\in C^1[0,1]$ and every $g_s'$ is
strictly negative on $[0,1)$, then constants $0<c<C<\infty$ and
$m_0<\infty$ exist such that
\[
 c m^{-2}\le \delta_w(\mathcal T_{\le m})
 \le \delta_w(\mathcal T^{(m)})\le C m^{-2},
 \qquad m\ge m_0.
\]
The constants $c,C,m_0$ may depend on $K,S,\alpha,w$, the densities, and the
selected OAK witness, but not on $m$.
Thus $m^{-2}$ is the optimal rate over all deterministic one-shot
componentwise report maps with at most $m$ values per site, including
arbitrary measurable nonuniform cells, and equal-width $m$-level reports
attain it.
\end{itemize}
\end{theorem}

Theorem~\ref{thm:oracle_kink} shows that multiplicity leaves the
single-hypothesis picture of Proposition~\ref{prop:converse} intact.  Under
the stated conditions, every report with finitely many values loses power,
even with optimally chosen thresholds, and the smallest attainable loss falls
as $m^{-2}$, or $4^{-b}$ in bits, whatever the number of hypotheses.
Equal-width cells already achieve this order; concentrating the cells where
the evidence is strongest, or making them irregular, can improve only the
constant.  Multiplicity can therefore act only through the constants;
Section~\ref{sec:numerics} measures or bounds its effect at finite resolution
for up to eight hypotheses.

The rate has a simple source.  With equal-width cells, the cells that
straddle a decision boundary form a \emph{boundary layer} of total
probability $O(m^{-1})$, and within each of them the likelihood ratios vary by
$O(m^{-1})$, which gives the upper bound $O(m^{-2})$.  For the lower bound,
disclosing all but one site's likelihood ratio reduces the gap at an active
kink to the loss from averaging the kink function $x\mapsto(x-q)_+$ within a
report cell.  Integrated over the kink location $q$, this loss is half the
within-cell squared error, which no partition into $m$ measurable cells,
connected or not, reduces below order $m^{-2}$.

Strictness has the same source: no finite product partition of the sites' $p$-values follows the
oblique level sets of $\Lambda_k$, so some cell straddles an active kink with
positive probability.  At $K=1$ strictness parallels the known power loss of quantized
likelihood-ratio tests~\citep{tenney1981detection,tsitsiklis1993extremal}.
A general prior changes only the objective coefficients $A_l$, which keep
the branch of every fixed rejection set affine in each coordinate, so the same arguments
apply for every $w$ under OAK (supplement, Section~S1.9).

\section{Aggregation under a fixed report distribution}\label{sec:method}

The two rules of this section read the same reports but require different
inputs.  The null report distribution $\tilde\nu_{s,0}$ is obtained by applying
$T_s$ to a uniform $p$-value, so it is known exactly and does not depend on
$g_s$.  Applying $T_s$ to a $p$-value with density $g_s$ gives the alternative
report distribution $\tilde\nu_{s,1}$.  The model-aware rule uses both, enforcing the strong-FWER
constraints under the whole induced distribution, and is optimal for fixed maps.  The
null-calibrated rule uses only the first, gaining robustness without
general optimality.

\subsection{Model-aware aggregation}\label{sec:aware}

For fixed report maps the reports form another exchangeable two-group
problem, whose per-hypothesis likelihood ratio is the product over sites of
the report likelihood ratios.  The report likelihood ratio at a site, and its product across sites, are
\[
\ell_s \;=\; \frac{\mathrm{d}\tilde\nu_{s,1}}{\mathrm{d}\tilde\nu_{s,0}}, \qquad
L_k \;=\; \prod_{s=1}^S \ell_s(y_{s,k}),
\]
with $\ell_s(c)=\tilde\nu_{s,1}(\{c\})/\tilde\nu_{s,0}(\{c\})$ for a
report value $c$ of positive null mass; $L_k$ is the cell-averaged likelihood
ratio $\bar\Lambda_k$ of Section~\ref{sec:kink}.
Values outside the null support are immaterial because
$\tilde\nu_{s,1}\ll\tilde\nu_{s,0}$.  Algorithm~S1 of the supplement gives
complete pseudocode for threshold reports.

\begin{theorem}[Finite-sample validity and optimality for fixed reports]\label{thm:aggregator}
Suppose Assumptions~\ref{ass:model}--\ref{ass:known} hold and the sites use
any fixed report maps with the same $T_s$ across hypotheses. Let the center
apply the likelihood-ratio-ordered optimal policy of the report-induced
two-group model with per-coordinate statistic $L_k$, computed at level
$\alpha$ from the linear program of Section~\ref{sec:kink} with $L_k$ in place
of $\Lambda_k$, with random tie-breaking among equal likelihood ratios and with
any randomization between rejection counts that this primal program supplies. Then:
\begin{itemize}
\item[(a)] it controls the family-wise error rate~\eqref{eq:fwer} in finite
samples, because the program imposes its constraints, at every
configuration, under the true report distribution $(\tilde\nu_0, \tilde\nu_1)$,
$\tilde\nu_i = \otimes_s \tilde\nu_{s,i}$, both of which are known
(Assumption~\ref{ass:known});
\item[(b)] it is most powerful for $\Pi_w$ among all procedures, randomized
or not, satisfying~\eqref{eq:fwer} and measurable with respect to the report
array, for the given reports;
\item[(c)] its family-wise error coefficients are the nonnegative polynomials
$B_{j,l}$ of Section~\ref{sec:kink}, evaluated at the report likelihood ratios; for a
finite report distribution, the primal and dual are finite linear programs,
the dual certifies the value, and the primal supplies the cellwise action
probabilities.
\end{itemize}
The same conclusions hold for any specified dominated pair
$(\nu_0,\nu_1)$, common to all hypotheses and replicated independently across
$k$, using $\mathrm d\nu_1/\mathrm d\nu_0$, with all probabilities computed under that pair; this is the result for a \emph{working
distribution}, a specified pair that the center treats as the true one, and it
is what the plug-in rule of Section~\ref{sec:plugin} applies conditionally on the
training data.
\end{theorem}

Theorem~\ref{thm:aggregator} makes the center's problem centralized again.
Once the report maps are fixed, the most powerful rule is the oracle's own
program with the report likelihood ratios $L_k$ in place of $\Lambda_k$, so
the only power lost is what the reports themselves impose, the fixed-design
gap $d_w(T)$.  On atomic report distributions the optimal rule may randomize, both to
order tied likelihood ratios and between rejection counts; the supplement
gives the construction and a deterministic, conservative variant (Section~S4.1).  Optimality is conditional on
the maps: choosing them is the outer problem in~\eqref{eq:deficit}, and they
may be selected from the known model or from training data independent of
the tested array, provided their report distributions are known.

\subsection{Null-calibrated aggregation}\label{sec:marginal}

The reports support a family-wise guarantee that needs no model for the
alternative, since the null half of the report distribution is known without
any model, and no independence across hypotheses, since each hypothesis is
calibrated from its own reports and Bonferroni's correction holds under any
dependence.  This matters when the local alternative is unknown or
only estimated, when the $K$ hypotheses are dependent, or when a default
requiring no optimization is wanted.
The guarantee also permits conservative local tests: a null $p$-value $p$
is \emph{superuniform} if $\pr(p\le x)\le x$ for every $x\in[0,1]$.

Fix the report maps and suppose each map is ordered: its cells are intervals
of $[0,1]$, ordered by position from smallest to largest $p$-value.  Write
$\preceq_s$ for this significance order, so $1\preceq_s0$ for threshold
reports; equal-width and full reports use the usual order, up to the
relabeling in Definition~\ref{def:classes}.  Let
$\hat p_s(c)=\tilde\nu_{s,0}(\{c':c'\preceq_s c\})$ be the null probability
of reporting $c$ or a more significant cell, computable from $T_s$ alone.
The center forms
\begin{equation}\label{eq:fisher}
 W_k \;=\; -\sum_{s=1}^S \log \hat p_s(y_{s,k}),
\end{equation}
Fisher's combination of the reports' own $p$-values.  Let $W^0$ have the
reference distribution obtained by applying the report maps to independent
uniform $p$-values, and put
\[
 \bar F_0(\tau)=\pr_0(W^0\ge \tau),\qquad
 p_k^{\mathrm{nc}}=\bar F_0(W_k).
\]
The center rejects $H_k$ when $p_k^{\mathrm{nc}}\le\alpha/K$.  This Bonferroni step can be replaced by Holm's step-down
procedure~\citep{holm1979simple}.

For a threshold report at threshold $t_s$ the statistic $W_k$ reduces to a weighted count:
$\hat p_s$ equals $t_s$ on a local rejection and $1$
otherwise, so $W_k$ increases with every local rejection and its null
distribution is the convolution of $S$ independent two-point laws; under a
common threshold $W_k$ is a monotone function of the binomial count of local
rejections.  For the full report, $\hat p_s=u_{s,k}$ and $2W_k$ is
exactly $\chi^2_{2S}$ under the null.

Fisher's statistic is a default; whether it is optimal depends on the
alternative.  It is asymptotically
optimal in Bahadur's sense among methods of combining independent
tests~\citep{littell1971asymptotic}, but any reasonable combination of
independent $p$-values is optimal against some
alternative~\citep{birnbaum1954combining}, and the standard combinations are
likelihood-ratio tests against particular
alternatives~\citep{heard2018choosing}.  For the $\mathrm{Beta}(1,\theta)$
alternatives of Section~\ref{sec:numerics} the likelihood ratio
$\prod_s\theta(1-u_{s,k})^{\theta-1}$ orders the evidence by
$\prod_s(1-u_{s,k})$, whereas Fisher's statistic orders it by
$\prod_s u_{s,k}$.  When a model of the alternative is trusted for
ranking but not for calibration, the report likelihood ratio $L_k$ of
Section~\ref{sec:aware}, calibrated against its own exact null distribution, can
replace $W_k$; Proposition~\ref{prop:marginal} still applies when the
modeled alternative densities are nonincreasing.

\begin{proposition}[Null-calibrated aggregation]\label{prop:marginal}
Suppose, in place of Assumption~\ref{ass:model}, only that for each $k$ the
local $p$-values
$(u_{s,k})_{s=1}^S$ are independent across sites and are superuniform
whenever $h_k=0$, the joint distribution across hypotheses being
otherwise arbitrary.  Let the report maps be ordered and fixed, and let the
center reject $H_k$ when $p_k^{\mathrm{nc}}\le\alpha/K$, with $W_k$
as in~\eqref{eq:fisher}.  Then the family-wise error rate~\eqref{eq:fwer} is
at most $\alpha$ in finite samples at every configuration.  The same bound
holds, in place of~\eqref{eq:fisher}, for any fixed statistic of hypothesis $k$'s
reports $(y_{s,k})_{s=1}^S$ that is nonincreasing in each $\hat p_s(y_{s,k})$ and calibrated
against its own exact null distribution, and, in place of the Bonferroni
step, for any multiplicity correction valid under arbitrary dependence, such
as Holm's step-down procedure.
The maps and the statistic may be fixed in
advance or be measurable with respect to a training sigma-field jointly
independent of the entire tested $p$-value array.
\end{proposition}

Proposition~\ref{prop:marginal} needs neither a model for the alternative
nor independence across hypotheses, so within-site correlation across the
$K$ hypotheses is permitted.  Under a true hypothesis, superuniform local
$p$-values give terms of $W_k$ stochastically no larger than those from
uniform inputs, because, for an ordered map, each term is nonincreasing in
its input $p$-value; independence across sites preserves this ordering for their sum,
so $p_k^{\mathrm{nc}}$ is superuniform, and the Bonferroni step needs nothing
more.  No part of the argument uses $g_s$: the reference tail is a functional
of the report maps alone.  Control therefore survives a misspecified or
independently estimated local alternative.  Discreteness-aware refinements of
Bonferroni~\citep{tarone1990modified,westfall1997discrete} gain power when the
hypotheses' attainable $p$-values differ or their joint null distribution is
known; here all $p_k^{\mathrm{nc}}$ share one null support, so Tarone's
refinement reduces to the Bonferroni step.

\begin{proposition}[Optimal null calibration at a single hypothesis]\label{prop:marginalK1}
Under Assumption~\ref{ass:model}, let $K=1$, let the sites be homogeneous
with a common nonincreasing $g$,
and let every site send a threshold report at the common threshold $t$.  Write $N$ for the number of
local rejections, $N_0\sim\operatorname{Bin}(S,t)$ for an independent null
reference, and $V\sim\operatorname{Unif}[0,1]$ for a variable independent of
the reports.  The rule that rejects when
\[
 p_t^{\mathrm{rand}}
 =\pr(N_0>N\mid N)+V\,\pr(N_0=N\mid N)\le\alpha
\]
has exact null level $\alpha$ and attains the model-aware optimum for
every fixed $t$.  Its calibration does not use $g$.  Moreover, if the
model-aware optimum over common thresholds is attained at a threshold $t$
with $\pr(N_0\ge c)=\alpha$ for some integer $c$, the deterministic null-calibrated
count rule attains that optimum without randomization.
\end{proposition}

At a single hypothesis, then, the center loses nothing by ignoring the alternative
model: with homogeneous sites and a common threshold, counting local
rejections and calibrating the count exactly, with boundary randomization,
is already the most powerful reading of the reports.  The reason is that the report likelihood ratio is
nondecreasing in the rejection count, so the Neyman--Pearson rule is a count
test with boundary randomization, the classical counting rule for identical
sensors~\citep{thomopoulos1987optimal,viswanathan1989counting}.  Without randomization the count rule can
leave level unused at a fixed threshold, but it attains the same power at an
exact-level optimum.

Under multiplicity the two rules are no longer equivalent.  The
null-calibrated rule is \emph{marginal}, testing each hypothesis from its
own reports alone.  Reading the $K$ hypotheses \emph{jointly}, as the
model-aware rule does, can increase power even when each marginal count
boundary is randomized; Section~\ref{sec:numerics} measures the gain.

E-values, nonnegative statistics with expectation at most one under every
permitted null distribution, give a second route to the guarantee of
Proposition~\ref{prop:marginal}, never more powerful on the same statistic:
Markov's inequality applied to a product of site-level e-values replaces the
exact null tail~\citep{vovk2021evalues,ramdas2023game} (supplement,
Section~S4.4).

\subsection{Optimal local thresholds}\label{sec:design}

For threshold reports the design reduces to the thresholds $t =
(t_1,\dots,t_S)$.  The simplest null-calibrated design is the
\emph{unanimity} rule, the count rule with cutoff $S$, which rejects a
hypothesis only when every site rejects it locally.  Its exact null probability is
$\prod_s t_s$, so the rule is null-calibrated at level $\alpha/K$ if and only if
$\prod_s t_s\le\alpha/K$, and its power at the full alternative is then
$\prod_s G_s(t_s)$, where $G_s$ is the distribution function with density $g_s$.

\begin{proposition}[Optimal unanimity thresholds]\label{prop:alloc}
Suppose each $g_s$ is continuous on $(0,1)$ and $G_s(t)>0$ for every
$t\in(0,1)$. Any interior maximizer of
$\prod_s G_s(t_s)$ subject to $\prod_s t_s = \alpha/K$ equalizes the local
elasticity across sites,
\[
\frac{t_s\, g_s(t_s)}{G_s(t_s)} \;=\; \kappa \quad \text{for all } s,
\]
with $\kappa$ fixed by the constraint $\prod_s t_s = \alpha/K$. If in addition $t \mapsto t\, g_s(t)/G_s(t)$ is nonincreasing for each
$s$, equivalently $\log G_s(e^x)$ is concave, every such stationary point is
a global maximizer; and if each elasticity is a strictly decreasing bijection
of $(0,1)$ onto $(0,1)$, as for the $\mathrm{Beta}(1,\theta)$ alternatives
with $\theta > 1$, an interior stationary
point exists and is the unique maximizer. Under these concavity conditions,
homogeneous sites $g_s\equiv g$ have a symmetric optimum
$t_s=(\alpha/K)^{1/S}$.
\end{proposition}

Equal elasticity typically gives the more informative site the tighter
threshold: with nonincreasing elasticities, the site whose detection
probability falls proportionally less when its threshold is tightened is
given the tighter one.
For one hypothesis at $\alpha=0.05$ with one $\mathrm{Beta}(1,2)$ and one
$\mathrm{Beta}(1,4)$ site, the optimal thresholds are $0.354$ and $0.141$, with
$5\%$ more power than a common threshold (supplement, Section~S3.2).
Thresholds optimal for the model-aware rule, which also uses the alternative
report distribution, are computed numerically (Section~\ref{sec:numerics}).

Unanimity is a poor design when there are many sites.  For homogeneous sites,
the optimal unanimity thresholds $t_s=(\alpha/K)^{1/S}$ approach one as the
number of sites grows, allowing a single site with a weak result to veto
rejection.  For one hypothesis and homogeneous
$\mathrm{Beta}(1,\theta)$ sites with $\theta>1$, Proposition~S4 of the
supplement makes this precise.  As $S\to\infty$, the probability that the
oracle misses the false hypothesis decays exponentially in $S$, at rate
$\theta-1-\log\theta$, and so does that of the best threshold report, at a
strictly smaller rate, attained by a count of local rejections at a common
threshold.  The miss probability of the unanimity rule, by contrast, decays
only polynomially, as $\{\log(1/\alpha)\}^{\theta}S^{1-\theta}\{1+o(1)\}$.
Counting local rejections, rather than requiring all of them, therefore keeps
an exponential rate, in line with classical results on many sensors and
counting rules in distributed
detection~\citep{tsitsiklis1988large,viswanathan1989counting}.

\subsection{Unknown local alternatives}\label{sec:plugin}

When the local alternatives must be estimated, the model-aware rule can
still be made valid by lowering its level, but the correction can demand
large pilot samples, held-out $p$-values that each site draws under the
alternative.  Only the alternative half of the report
distribution is unknown, and for a quantized report it is a low-dimensional
object: one scalar $G_s(t_s)$ per site for a threshold report, an $m$-cell
multinomial for the equal-width $m$-level report.  Let the center form estimates
$\hat\nu_{s,1}$ of the alternative report distributions from training data
independent of the tested array, with $\hat\nu_{s,1}\ll\tilde\nu_{s,0}$ at
every site, and suppose that, with probability at least $1-\zeta$ over the
training data,
$\sum_{s=1}^S\mathrm{TV}(\hat\nu_{s,1},\tilde\nu_{s,1})\le\varepsilon$, where
$\mathrm{TV}$ is the total-variation distance.  If
$(K-1)\varepsilon+\zeta\le\alpha$, the program of Theorem~\ref{thm:aggregator},
run under the estimates at level $\alpha-(K-1)\varepsilon-\zeta$ (the
plug-in rule), controls
the family-wise error rate~\eqref{eq:fwer} at level $\alpha$ at every
configuration, making no rejections when the adjusted level is zero
(supplement, Proposition~S2).  The factor $K-1$ bounds the false hypotheses
in any configuration with a true one, so at $K=1$ no deflation is needed.  For threshold
reports, pilot samples of size $n$ per site and a
Dvoretzky--Kiefer--Wolfowitz bound give
$\varepsilon=O(Sn^{-1/2})$, up to logarithmic factors.  The correction
vanishes at the parametric rate if $\zeta$ does, but its constants can demand
large pilots, so at moderate pilot sizes the
null-calibrated rule of Section~\ref{sec:marginal}, whose critical value never
consults the alternative, is the practical default.

\section{Numerical study and data illustrations}\label{sec:numerics}

The numerical study targets four claims: finite-report strictness at $K=2$;
the inverse-square resolution rate; the growing benefit of joint aggregation
under multiplicity; and the different robustness scopes of the two rules.
Report-distribution values come from deterministic finite linear programs,
solved to numerical tolerance with primal--dual checks; threshold searches and
oracle values are numerical, and the supplement certifies the
two-site Beta threshold searches at $K\le2$ to within $2\times10^{-4}$ in power
by exact rational arithmetic.  Section~\ref{sec:numerics_summary}
lists the findings; the full study is in Section~S3 of the supplement, with
computational details and certificates in Section~S4.  The two data illustrations
that follow use only the null-calibrated rule, because their hypotheses
share cases or participants, which violates the independence across
hypotheses that the model-aware rule requires.

\subsection{Summary of the numerical study}\label{sec:numerics_summary}

\begin{figure}[!tbp]
\centering
\includegraphics[width=\linewidth]{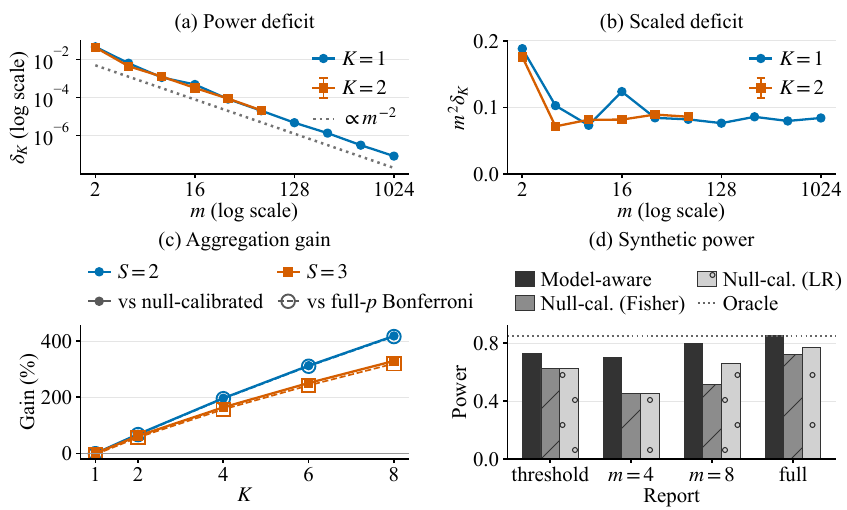}
\caption{Refining the report recovers power at the inverse-square rate,
joint aggregation gains more as $K$ grows, and the model-aware rule
approaches the oracle as the equal-width report is refined; $\alpha=0.05$,
point prior. (a,b) Absolute deficits $\delta_K=\delta_K(\mathcal T^{(m)})$ and
scaled deficits $m^2\delta_K$ of equal-width reports, upper bounds on
$\delta_K(m)$, for $\mathrm{Beta}(1,2)$ sites with $S=2$, with a dotted
$m^{-2}$ reference line in (a); $K=2$ error bars show the solver's bracket on
report power, and the oracle is computed numerically. (c) Percentage gains of
the model-aware rule on one-bit reports over the null-calibrated rule on
one-bit reports (filled markers, solid lines) and over the most powerful
Bonferroni rule on the full $p$-values (open markers, dashed lines), for
$\mathrm{Beta}(1,2)$ sites with $S=2,3$. (d) Power for $K=2$ at three
normal-shift sites, by report; the likelihood-ratio (LR) bars use the
null-calibrated rule with the report likelihood ratio in place of Fisher's
statistic. Threshold designs are optimized separately and are not nested in
the equal-width maps, so a threshold bar can exceed the $m=4$ bar.}
\label{fig:theory_inverse_square}
\end{figure}

Unless stated otherwise, the sites are homogeneous with $\mathrm{Beta}(1,2)$
alternatives, $g(u)=2(1-u)$, $S=2$, $\alpha=0.05$, and the objective is the
point prior.  Losses are reported as \emph{relative deficits}, the deficit as
a percentage of oracle power.  The model-aware rule reads the hypotheses
jointly; the marginal rules, which test each hypothesis separately, include
the null-calibrated rule (Fisher--Bonferroni) and, on the full $p$-values,
the most powerful Bonferroni rule.
\begin{enumerate}[label=(\alph*),leftmargin=*,itemsep=3pt]
\item \emph{One bit loses little at a single hypothesis.}  At $K=1$ the oracle is the
Neyman--Pearson test on the product likelihood ratio, with power
$\pi^*_1=0.15954$.  With the optimized threshold report, the center rejects when
both sites report significance at $t=0.224$ (the square root of $\alpha$),
with power $0.15778$: one optimized bit per site loses $1.1\%$ of $\pi^*_1$,
and $1.9\%$ under $\mathrm{Beta}(1,3)$ (Table~S5).
\item \emph{One bit loses far more under multiplicity.}  At $K=2$ the oracle
value is $\pi^*_2=0.15646$, and the optimized threshold report, read by the model-aware rule of Section~\ref{sec:aware}, attains $0.14312$: the optimized one-bit
report loses $8.5\%$ of $\pi^*_2$, and $6.2\%$ under
$\mathrm{Beta}(1,3)$.  For $\mathrm{Beta}(1,2)$ the search gives the two
identical sites slightly different thresholds, $0.288$ and $0.328$; a common
threshold loses $8.6\%$.  The continuous oracle is computed only up to $K=2$, where it reduces to
one-dimensional quadrature, but two-bit equal-width reports
cannot exceed oracle power, so at the numerically selected common threshold one bit loses at least $9.0\%$ of oracle power at
$K=4$ and $11.6\%$ at $K=8$ (Table~S8).
\item \emph{The inverse-square rate is visible at practical resolutions}
(Figure~\ref{fig:theory_inverse_square}(a,b)).  The fitted exponent of the
deficit of equal-width reports, the magnitude of the least-squares slope of
$\log\delta_K(\mathcal T^{(m)})$ against $\log m$, is $2.00$ at $K=1$ over
$m=32,64,\ldots,1024$ and $1.998$ at $K=2$ over seven resolutions from $8$ to $64$.  The scaled deficit oscillates at small $m$, where the oracle's decision boundary falls differently relative to the cell edges.  Three bits, $m=8$, leave relative deficits of $0.72\%$ at $K=1$ and
$0.81\%$ at $K=2$; at $K=2$, two bits leave $2.9\%$, and the equal-width
one-bit report, whose threshold is $1/2$ rather than optimized as in~(b),
leaves $28\%$.
\item \emph{How the center reads the reports matters at least as much as
their resolution.}  At $K=2$, optimizing the thresholds for
the null-calibrated rule of Section~\ref{sec:marginal} gives losses of $45.8\%$ and
$43.3\%$ for the two Beta models, against $8.5\%$ and $6.2\%$ for the
model-aware rule in~(b), so reading one-bit reports jointly recovers
$81$--$86\%$ of the marginal rule's loss.  Almost all of that loss is
the cost of testing each hypothesis separately rather than of quantization:
with full $p$-values the most powerful Bonferroni rule loses $45.3\%$ and
$42.5\%$, and the one-bit null-calibrated rule is itself a Bonferroni rule,
so one bit costs the best Bonferroni rule less than one percentage point of
oracle power.  Separate testing cannot pool evidence: in
Figure~\ref{fig:oracle_kink} the oracle rejects both hypotheses once both
likelihood ratios exceed $2.7$; the best Bonferroni rule needs each above
$3.14$.  On the full $p$-values Fisher--Bonferroni loses
more than on one bit, $55.4\%$ and $56.8\%$, because Fisher's statistic
orders the evidence differently from the likelihood ratio (Table~S1).
\item \emph{The gain from reading jointly grows with $K$}
(Figure~\ref{fig:theory_inverse_square}(c)).  Model-aware aggregation of
one-bit reports from two or three sites, each rule at its own optimized common
threshold, gains $62$--$69\%$ over the
null-calibrated rule at $K=2$ and $330$--$419\%$ at $K=8$ ($0.120$ versus $0.023$ at
$S=2$), and nearly as much against the most powerful Bonferroni rule applied
to the full $p$-values, $57$--$67\%$ at $K=2$ and $320$--$417\%$ at $K=8$.
Under the flat prior the gain at $S=2$ is $34\%$ at $K=2$, rising to $157\%$
at $K=8$; that prior also weights configurations with a single false
hypothesis, where there is no second alternative to pool.
\item \emph{The flat prior increases relative deficits in these examples.}
At $K=2$ the threshold-report relative deficits become $9.6\%$ and $9.3\%$
for $\mathrm{Beta}(1,2)$ and $\mathrm{Beta}(1,3)$, against $8.5\%$ and
$6.2\%$ at the point prior (Section~S3.4).
\item \emph{From two to eight sites, one bit loses more as sites are added.}
At $K=1$ the deficit of the best common-threshold report rises from $1.1\%$
at $S=2$ to $18.4\%$ at $S=8$, stays there through $S=10$ (Table~S6), and
tends to zero as $S\to\infty$, where all powers tend to one (supplement, Proposition~S4).  The
null-calibrated count rule and the model-aware
program agree to six digits at every $S$, and unanimity is the optimal count
cutoff only up to $S=7$, as the large-$S$ theory of Section~\ref{sec:design}
leads one to expect.
\item \emph{The model-aware rule is the more powerful, and the
null-calibrated rule the more robust.}  In a synthetic study with $K=2$ and
three heterogeneous normal-shift sites (one-sided $z$-tests with mean shifts
$1$, $1.5$, and $2$), the model-aware rule approaches
the oracle power $0.852$ as the equal-width report is refined, reaching
$0.797$ at three bits, against $0.728$ for its optimized one-bit threshold
(Figure~\ref{fig:theory_inverse_square}(d)); the two null-calibrated rules
share a separate threshold design, optimized among unanimity rules.  The
null-calibrated rule falls short partly because quantized reports
leave null level unused, Fisher's statistic orders the evidence differently
from the likelihood ratio, and the Bonferroni step is less powerful than
Holm's.  At $m=8$, replacing either of the last two recovers part of the
gap; the report likelihood ratio, calibrated against its exact null
distribution, raises power from $0.513$ to $0.661$ (Table~S9).
Underestimating the effects by $25\%$ raises the model-aware family-wise error rate to
$0.054$--$0.055$, and within-site dependence across hypotheses raises it to
$0.064$--$0.067$; the null-calibrated rate never exceeds $\alpha$, since its critical values use only each hypothesis's own null report distribution.
\end{enumerate}

\subsection{Decentralized medical-AI benchmark: FeTS reanalysis}
\label{sec:numerics_fets}

A retrospective reanalysis of the Federated Tumor Segmentation (FeTS) benchmark
shows how the report design changes the rejection set: at nominal family-wise
level $0.05$, one-bit reports at the unanimity threshold
reject five of the twelve null hypotheses, two-bit equal-width reports ten,
and full $p$-values twelve (Table~\ref{tab:fets}).  These counts compare report protocols on one
retrospectively selected family.  The data are
the released case-level metrics of the FeTS 2022
Task~2 challenge, which evaluated brain-tumor segmentation models at
institutions whose images and reference labels stayed
local~\citep{zenk2025fets,karargyris2023medperf}.
After the organizers' quality control, $2{,}178$ cases remain at the $25$
institutions that were unseen during training and at which all five models
below completed evaluation.

The family has $K=12$ claims.  Model 10, the only official submission whose documented method targets distribution shift, is compared with each of the four
other official submissions (8, 11, 12, 54) on Dice score, a measure of
segmentation overlap, for enhancing tumor, tumor core, and whole tumor.  With
$\omega_{s,c,r}=\pr\{\operatorname{Dice}_{10,r}>\operatorname{Dice}_{c,r}\}$
the probability at institution $s$ that model 10 beats comparator $c$ on
region $r$ for a case, the null is
$H_{c,r}=\bigcap_{s=1}^{25}\{\omega_{s,c,r}\le1/2\}$, and each site's
$p$-value is the exact upper binomial tail of its strict-win count.  Rejecting
$H_{c,r}$ therefore says that model~10 has a strict majority-win Dice
advantage over comparator $c$ on region $r$ at \emph{one or more} included
institutions; one institution alone can produce such a rejection.  Nested
regions, shared cases, and repeated comparisons make the twelve claims
dependent in an unrestricted way, so the model-aware program does not apply
and Proposition~\ref{prop:marginal} supplies the fixed-family guarantee.

The center needs nothing but the report maps to calibrate each protocol
exactly.  Each site sends, for each claim, a threshold report at the
unanimity threshold $t=(0.05/12)^{1/25}=0.803$ (the \emph{unanimity bit}),
an equal-width $m$-level report, or the full $p$-value.  The unanimity threshold is the design of
Proposition~\ref{prop:alloc} for homogeneous sites: the rule rejects only
when all $25$ sites reject locally, a requirement that Section~\ref{sec:design} shows to
be costly when there are many sites.  The center reads a report as an upper
bound on its compatible $p$-values (namely $t$ or $1$, the upper edge $j/m$
of the reported cell $j$, or $p$ itself), combines the $25$ site values by
Fisher's statistic, and rejects when the corresponding null-tail $p$-value is
at most $\alpha/K$.

\begin{table}[!tbp]
\caption{FeTS reanalysis: rejections rise from $5$ of the $12$ retrospectively
selected model--region nulls with the unanimity bit to all $12$ with full
$p$-values, at nominal family-wise level
$0.05$.  The primary column corrects over these $12$ claims; the
sensitivity column corrects over all $60$ ordered official-model--region
claims but counts rejections only among the same $12$.}
\label{tab:fets}
\centering
\begin{tabular}{lcc}
\toprule
report & primary, $K=12$ & sensitivity, $K=60$ \\
\midrule
unanimity bit & $5$ & $5$ \\
equal-width $m=2$ & $8$ & $8$ \\
equal-width $m=4$ & $10$ & $8$ \\
equal-width $m=8$ & $11$ & $11$ \\
full $p$-value & $12$ & $11$ \\
\bottomrule
\end{tabular}
\end{table}

The ordering of the protocols persists in larger and different families.
Correcting over all $60$ ordered comparisons among the five official models
and three regions, rather than over the twelve with model~10 first, keeps all
five unanimity-bit rejections and eleven of the twelve full-report ones.
A family of $18$ two-sided comparisons among the other four official models,
which avoids the model-10 condition discussed below, shows the same ordering: the unanimity bit rejects $4$,
the equal-width one-bit report $14$, and the full $p$-values $17$
(Table~S10).  The two one-bit protocols differ in where they place their
single cell boundary: in the primary family, the equal-width one-bit rule rejects eight nulls using
a cutoff of $20$ out of $25$ reports below $1/2$, whereas the unanimity rule
rejects five using all $25$ reports at or below $0.803$.  The observed
rejection sets, listed in Table~S10, are nested under equal-width refinement.

A separate full-report partial-conjunction analysis in the supplement bounds
how many institutions have such an advantage:
nominal $95\%$ simultaneous lower bounds on the number of included
institutions with $\omega_{s,c,r}>1/2$ are positive for all twelve claims and
at least four for ten of them (Table~S11), with arbitrary
dependence across institutions and claims permitted.

Neither analysis
establishes an advantage at every institution, in mean Dice, in clinical
benefit, or in deployment safety, and neither is a clinical validation.  Both
presuppose valid local $p$-values, which require the retained cases to be
i.i.d.\ draws from each institution's post-quality-control population and
the predictions to use evaluation data only from the current case; for
model~10, which recomputes its normalization statistics at prediction time,
that condition rests on our reading of its published method as doing so one
case at a time.  The
Fisher-combination analysis of Table~\ref{tab:fets} also requires independent
institutions, and both treat as
fixed in advance a family that was chosen retrospectively.  The supplement
gives the calibrations, provenance, exclusions, complete rejection sets, and
further sensitivity analyses.

\subsection{Replication laboratories: Many Labs 1}\label{sec:application}

Many Labs~1~\citep{klein2014investigating} supplies what the FeTS
reanalysis lacks, a family of hypotheses fixed before any data were
collected.  The project reran 13 classic and contemporary findings of
psychology in 36 independent samples totaling $6{,}344$ participants.  Its
protocol was peer reviewed and accepted before data collection, every
participant completed every task in a single session, and every sample's
results were reported, whatever their outcome, so no result entered the
analysis because it was significant.  The family is the project's 16
preregistered tests, the anchoring finding contributing four, and the sites
are the 35 samples that recorded every test.  Each sample computes, from its own participants alone,
a one-sided $p$-value in the direction the original study predicted.  The
null hypothesis $H_k$ for test $k$ is the intersection over the 35 samples of
the hypotheses of no effect in that direction, so a rejection means that the
effect is present in at least one sample's population, not that the finding
replicated in every sample.

The null-calibrated rule has nominal family-wise level $0.05$ for every
report below, because the sample $p$-values are valid and the dependence
among them is of the kind Proposition~\ref{prop:marginal} allows.  Fourteen
of the 16 tests compare randomly assigned conditions, so random assignment
within each sample makes their $p$-values valid under the null; the other
two, the sex difference in implicit math attitudes and the relation between
implicit and explicit math attitudes, rest on independent sampling of
participants.  The four two-by-two comparisons use Fisher's exact test, and
the Welch and correlation $t$ tests are valid to normal-theory accuracy at
these sample sizes.  The samples are disjoint by design, so the dependence
created by each participant completing every task runs across tests, not
across samples; the participants of the two open online pools, Amazon
Mechanical Turk and Project Implicit, cannot be checked against those of the
other samples, and a sensitivity analysis drops them.

Table~\ref{tab:manylabs} compares three one-bit reports that a sample could
send for each test, and the full $p$-value.  The \emph{unanimity bit} reports
whether the sample's $p$-value is at most a threshold of approximately $0.848$,
chosen so that unanimity has probability exactly $0.05/16$ under the
independent-uniform reference distribution; it is the design used for FeTS.
The \emph{5\%
bit} says whether the sample's result is significant at the conventional
two-sided 5\% level in the predicted direction, the summary the project
itself tabulated.  The \emph{direction bit}, which is the two-cell
equal-width report, says only whether the sample's $p$-value is below $1/2$,
which for a $t$ test means that its estimate points the predicted way.  Each
report is calibrated against its exact reference-null distribution, as in
Section~\ref{sec:marginal}.

\begin{table}[!tbp]
\caption{Many Labs 1: counting the samples whose one-sided $p$-value is below
$1/2$ (the direction bit) rejects all 14 nulls that the full $p$-values
reject, while the unanimity bit misses three of them; 16 preregistered nulls,
nominal family-wise level $0.05$, 35 samples.  ``Significant'' counts samples
significant at the two-sided 5\% level in the predicted direction; ``predicted
way'' counts samples whose one-sided $p$-value is below $1/2$; for the $t$
tests, these are the samples whose estimate points the predicted way.  Table~S14
lists the eleven tests that every report rejects.  Equal-width reports with
4 to 64 cells reject the same 14 nulls as the direction bit.}
\label{tab:manylabs}
\centering
\begin{tabular}{lcccccc}
\toprule
 & \multicolumn{2}{c}{samples (of 35)} & \multicolumn{4}{c}{rejected by} \\
\cmidrule(lr){2-3}\cmidrule(lr){4-7}
test & significant & predicted way & unanimity bit & 5\% bit & direction bit & full $p$ \\
\midrule
eleven other tests & 11--35 & 32--35 & $\checkmark$ & $\checkmark$ & $\checkmark$ & $\checkmark$ \\
low vs.\ high category scales & 20 & 34 & -- & $\checkmark$ & $\checkmark$ & $\checkmark$ \\
quote attribution & 17 & 30 & -- & $\checkmark$ & $\checkmark$ & $\checkmark$ \\
imagined contact & 4 & 26 & -- & -- & $\checkmark$ & $\checkmark$ \\
flag priming & 1 & 18 & -- & -- & -- & -- \\
currency priming & 1 & 17 & -- & -- & -- & -- \\
\midrule
rejections (of 16) & & & 11 & 13 & 14 & 14 \\
\bottomrule
\end{tabular}
\end{table}

As in the FeTS reanalysis and consistent with Section~\ref{sec:design}, a
one-bit design that requires unanimity loses findings that counting one-bit
reports recovers.  Eleven tests are rejected by every report, even the unanimity
bit, and the two findings the project did not replicate, flag priming and
currency priming, by none.  The reports differ only on the other three
findings: low versus high category scales, quote attribution, and imagined
contact.
The unanimity bit misses all three, because it needs every sample to agree
and one to four samples had $p$-values above $0.848$.  The 5\% bit misses imagined
contact, which was significant in only four samples.  The direction bit
recovers every finding that the full $p$-values detect: for imagined contact,
26 of the 35 samples pointed the predicted way, and the calibrated rule
rejects when at least 26 do.  That margin is thin.  Without the two online pools, imagined contact is
detected only by reports with at least 32 cells or by the full $p$-values.
With all 35 samples, counting the anchoring finding once, through any one of
its four questions, gives the same pattern:
$8$, $10$, $11$, and $11$ rejections of $13$ for the four reports.
Table~S14 gives the per-test results, and Section~S5.2 these
sensitivity analyses.

\section{Discussion}\label{sec:discussion}

Measured against the $p$-value oracle, keeping data at the sites costs
nothing, whereas disclosing only finitely many values costs power at every
resolution under the stated conditions; the optimal cost falls as the inverse
square of the resolution, a rate that equal-width cells attain and that no
placement of the cells, within our componentwise strong-FWER model, improves.
The separation holds because, under the two-group model, federated multiple
testing is a centralized problem on the distribution of the reports; for
optimal multiple testing more broadly, the deficit measures how the
centralized optimum degrades as its input is quantized.  Multiplicity
changes the constants rather than the rate.  In the two-site examples the
best one-bit threshold report loses several times more for two hypotheses than for one,
whereas equal-width reports lose similar fractions of oracle power for
one and two hypotheses, within a factor of $1.5$ at every computed
resolution.  With ten sites, one-bit reports at the best common threshold lose $18.4\%$
even for one hypothesis.

For a consortium choosing a protocol, the examples and the two reanalyses
suggest three lessons.  With model-aware reading, two or three bits can be
nearly as good as full $p$-values (Section~\ref{sec:numerics_summary}(c)).  Joint reading pays when the alternative can be modeled and the hypotheses
are independent: in the Beta examples, model-aware aggregation of one-bit
reports gains more as $K$ grows, even against the most powerful Bonferroni
rule on the full $p$-values.  And
with many sites, one-bit reports should be counted rather than required to
agree: unanimity missed findings in both data reanalyses, and for homogeneous
$\mathrm{Beta}(1,\theta)$ alternatives its miss probability decays only
polynomially in the number of sites.

When the local alternatives are correctly modeled and the hypotheses are
independent, the model-aware rule is the most powerful reading of the reports
and controls the family-wise error rate in finite samples.  When the
alternative is uncertain or the hypotheses are dependent, calibrating
Fisher's combination of the reports against its own null distribution still
controls the family-wise error rate and is the rule both data reanalyses
required; a likelihood-ratio statistic calibrated in the
same way keeps that guarantee and can be more powerful when a model of the
alternative is trusted for ranking.  For homogeneous threshold reports at a
single hypothesis, null calibration with boundary randomization is also most
powerful.

Several extensions would carry the analysis further.  Determining the sharp
constant in the $m^{-2}$ rate and the cell density that minimizes it would
turn the rate into a design rule for nonuniform quantizers.  Joint encoders
that share a total number of bits
across hypotheses, interactive protocols, and variable-length codes lie outside the
componentwise class and could change the constant or the exponent.  Letting
$K$ grow, alone or with the number of sites, would show whether the
multiplicity penalty saturates.  Model-aware constraints valid under dependence, and
plug-in corrections sharper than the supplement's total-variation
bound, would widen the regime in which the model-aware
rule can be used.  Conditions on the densities and the prior that imply OAK when $w(K)=0$
remain open.

Other error rates connect to the same reports.  Applying the
Benjamini--Hochberg step to independent null-calibrated
$p$-values gives an FDR counterpart~\citep{benjamini1995controlling}.
Under arbitrary dependence across hypotheses, the supplement's product
e-values instead support the e-BH procedure~\citep{wang2022false}.

\begin{acks}[Acknowledgments]
We used generative AI tools (ChatGPT, Google Gemini, and Claude) for language
editing and code formatting support only. All analyses, results, and
mathematical derivations are the authors' own work.
\end{acks}

\begin{funding}
Dubey acknowledges partial support from the Stewart Topper Fellowship at the
Georgia Institute of Technology. Huo was partially supported by the National
Science Foundation under Grant~2229876, the A.~Russell Chandler~III
Professorship, the NIH-sponsored Georgia Clinical and Translational Science
Alliance, and the Georgia Department of Transportation.
\end{funding}

\begin{supplement}
\stitle{Supplement to ``Federated multiple testing with quantized $p$-values
under family-wise error control''}
\sdescription{All proofs, a reader's guide to the results (Table~S4), the
general-$K$ dual and primal programs, the full numerical study, the
Beta-family and dual-enclosure certificates, the complete numerical tables and diagnostics,
and full details of the FeTS and Many Labs~1 analyses.  In this preprint
it follows the references; its sections, equations, results, figures, tables
and algorithm carry an S prefix.}
\end{supplement}

\begin{supplement}
\noindent\stitle{Code and data}%
\sdescription{All data analyzed are publicly released. The Many Labs~1
inputs are the per-sample summary statistics that
\citet{klein2014investigating} released on the Open Science Framework under the
CC0 public-domain dedication, extracted by a pinned script that verifies their
SHA-256 digest.
The FeTS analysis uses the released case-level performance
metrics in the Source Data archive (MOESM4) accompanying
\citet{zenk2025fets}, and the analysis script verifies the archive's SHA-256
digest before use; the protected institutional images and reference labels underlying that
challenge are not publicly available and are not redistributed here. Code and
the small input files reproducing every reported number accompany the
submission and will be archived with a digital object identifier before
publication.}%
\end{supplement}

\bibliographystyle{spr-ims-nameyear}

\clearpage
\setcounter{section}{0}
\setcounter{equation}{0}
\setcounter{table}{0}
\setcounter{figure}{0}
\setcounter{algocf}{0}
\setcounter{theorem}{0}
\setcounter{proposition}{0}
\setcounter{lemma}{0}
\setcounter{corollary}{0}
\setcounter{definition}{0}
\setcounter{assumption}{0}
\setcounter{remark}{0}
\renewcommand{\thesection}{S\arabic{section}}
\renewcommand{\theequation}{S\arabic{equation}}
\renewcommand{\thetable}{S\arabic{table}}
\renewcommand{\thefigure}{S\arabic{figure}}
\renewcommand{\thealgocf}{S\arabic{algocf}}
\renewcommand{\thetheorem}{S\arabic{theorem}}
\renewcommand{\theproposition}{S\arabic{proposition}}
\renewcommand{\thelemma}{S\arabic{lemma}}
\renewcommand{\thecorollary}{S\arabic{corollary}}
\renewcommand{\thedefinition}{S\arabic{definition}}
\renewcommand{\theassumption}{S\arabic{assumption}}
\renewcommand{\theremark}{S\arabic{remark}}
\providecommand{\theHsection}{}\renewcommand{\theHsection}{S.\arabic{section}}
\providecommand{\theHsubsection}{}\renewcommand{\theHsubsection}{S.\arabic{section}.\arabic{subsection}}
\providecommand{\theHsubsubsection}{}\renewcommand{\theHsubsubsection}{S.\arabic{section}.\arabic{subsection}.\arabic{subsubsection}}
\providecommand{\theHequation}{}\renewcommand{\theHequation}{S.\arabic{equation}}
\providecommand{\theHtable}{}\renewcommand{\theHtable}{S.\arabic{table}}
\providecommand{\theHfigure}{}\renewcommand{\theHfigure}{S.\arabic{figure}}
\providecommand{\theHalgocf}{}\renewcommand{\theHalgocf}{S.\arabic{algocf}}
\providecommand{\theHtheorem}{}\renewcommand{\theHtheorem}{S.\arabic{theorem}}
\providecommand{\theHproposition}{}\renewcommand{\theHproposition}{S.\arabic{proposition}}
\providecommand{\theHlemma}{}\renewcommand{\theHlemma}{S.\arabic{lemma}}
\providecommand{\theHcorollary}{}\renewcommand{\theHcorollary}{S.\arabic{corollary}}
\providecommand{\theHdefinition}{}\renewcommand{\theHdefinition}{S.\arabic{definition}}
\providecommand{\theHassumption}{}\renewcommand{\theHassumption}{S.\arabic{assumption}}
\providecommand{\theHremark}{}\renewcommand{\theHremark}{S.\arabic{remark}}
\makeatletter
\renewcommand\paragraph{\@startsection{paragraph}{4}{\z@}%
  {2ex \@plus .5ex \@minus .2ex}{-1em}{\normalfont\itshape}}
\makeatother
\setcounter{secnumdepth}{3}
\providecommand{\singlespacing}{}
\begin{center}
  {\large\bfseries Supplement to ``Federated multiple testing with quantized\\
  $p$-values under family-wise error control''}\\[8pt]
  {\scshape Prasanjit Dubey and Xiaoming Huo}
\end{center}
\bigskip

This supplement contains the proofs of all results in the main text
(Section~\ref{supp:proofs}), the general-$K$ dual computation
(Section~\ref{app:dual}), and the numerical study summarized in
Section~\ref{sec:numerics} of the main text (Section~\ref{supp:numerics}).
Section~\ref{supp:additional} gives the model-aware protocol, a reader's guide
to the results, the comparison of the two rules on a fixed report distribution
(Proposition~\ref{prop:reduction}), the large-site analysis
(Proposition~\ref{prop:largeS}), and computational certificates and diagnostics;
Section~\ref{supp:applications} documents the data illustrations and their
sensitivity analyses.  Readers new to the paper may start with
Table~\ref{tab:roadmap}, which lists for each principal result the question it
answers, its assumptions, its conclusion, and where it is stated and proved. References to numbered results,
assumptions, equations, sections, figures and tables without an S prefix refer
to the main text; supplementary sections, results, equations, figures, and
tables use an S prefix.

\section{Proofs}\label{supp:proofs}

\paragraph{Notation}
Several proofs use the coefficients of the centralized linear program of
Section~\ref{sec:kink}, restated as~\eqref{eq:ordered_primal} in
Section~\ref{app:dual}, which derives them.  For
likelihood ratios $v=(v_1,\ldots,v_K)$ with ordered values
$v_{(1)}\ge\cdots\ge v_{(K)}$, write $e_d(I)$ for the elementary symmetric
polynomial of degree $d$ in the values indexed by a set $I$ (of ranks in
$B_{j,l}$ and $b_{\gamma,k}$ below, of coordinates in $e_d(R^c)$), with $e_0(I)=1$
and $e_d(I)=0$ for $d<0$ or $d>|I|$, and $e_d=e_d(\{1,\ldots,K\})$; put
$p_l=v_{(1)}\cdots v_{(l)}$ and, for a set $R$ of coordinates,
$p_R=\prod_{i\in R}v_i$.  The program chooses the probability $r_l(v)$ of
rejecting the $l$ largest likelihood ratios.  The objective coefficient of
that action is $A_l=(K!)^{-1}\sum_{i\le l}a_{w,i}$, where
$a_{w,k}=\sum_\gamma w(\gamma)a_{\gamma,k}$ mixes over the prior the
per-position coefficients
\[
 a_{\gamma,k}(v)=(\gamma-1)!\,(K-\gamma)!\;v_{(k)}\,e_{\gamma-1}\bigl(v_{-(k)}\bigr),
\]
with $v_{-(k)}$ the other $K-1$ ratios; at the point prior $A_l=(l/K)\,e_K$.
The family-wise error coefficient of that
action at $j$ true hypotheses is $B_{j,l}=e_{K-j}-p_l\,e_{K-l-j}(\{l+1,\ldots,K\})$.
With $\gamma=K-j$ the number of alternatives, it is the cumulative form,
$\sum_{k\le l}b_{\gamma,k}=\gamma!\,(K-\gamma)!\,B_{K-\gamma,l}$, of the
per-position coefficients
\[
 b_{\gamma,k}(v)=\gamma!\,(K-\gamma)!\;p_{k-1}\,e_{\gamma-k+1}(\{k+1,\ldots,K\})\;\ge\;0,
\]
the error that the $k$th rejection adds.  Read with $k$ a coordinate label,
$(K!)^{-1}a_{w,k}(v)$ is the weight that rejecting $H_k$ adds to the objective,
and $a_{w,a}\ge a_{w,b}$ whenever $v_a\ge v_b$
(Section~\ref{supp:proof_priorfree}(ii)).  The continuous dual $\mathcal D$ of
Section~\ref{sec:kink} has dual integrand
$F_\mu=\max(0,\psi_1^\mu,\ldots,\psi_K^\mu)$, with branches
$\psi_l^\mu=A_l-\sum_j\mu_jB_{j,l}$.

\paragraph{Fixed-set representation}
Rejecting a set $R$ of $l$ coordinates has error coefficient
$e_{K-j}-p_R\,e_{K-l-j}(R^c)$ at $j$ true hypotheses and dual branch
$\psi_R^\mu=(K!)^{-1}\sum_{k\in R}a_{w,k}-\sum_j\mu_j\{e_{K-j}-p_R\,e_{K-l-j}(R^c)\}$.
At each size $l$ the top $l$ coordinates give the largest branch, the ordered
branch $\psi_l^\mu$, so $F_\mu=\max(0,\max_R\psi_R^\mu)$; every $\psi_R^\mu$ is
multilinear, so $F_\mu$ is convex in each coordinate separately.  Step~5 of
Section~\ref{supp:proof_graderate} proves this at the point prior by a swap
identity, and Section~\ref{supp:proof_priorfree}(ii) extends it to every $w$.

\paragraph{Symmetrization fact}
Any feasible procedure can be made permutation-equivariant without changing
the objective or weakening any constraint; this standard reduction is used
twice.  In any experiment that is
equivariant under permutations of the $K$ hypothesis coordinates, let
$\sigma$ be an independent uniform random permutation, apply a procedure to
the permuted data, and undo $\sigma$ on its rejection set.  At configuration
$h$, the resulting procedure's family-wise error rate is the average of the
original procedure's family-wise error rates over the permuted configurations
$\sigma h$, and is therefore at most $\alpha$ whenever the original procedure
has strong control.  Each $\Pi_\gamma$ of~\eqref{eq:objective} averages over all configurations
with $\gamma$ false hypotheses, a set that every permutation maps onto
itself, so the symmetrized procedure has the same $\Pi_\gamma$ as the
original, and hence the same $\Pi_w$.

\subsection{Centralized-oracle sufficiency and optimality: proof of
Proposition~\ref{prop:oracle}}\label{supp:proof_oracle}
\emph{Step 1 (sufficiency and the induced likelihood ratio).}
Conditionally on $h$, the joint density of the full data with respect to
Lebesgue measure on $[0,1]^{SK}$ is $\prod_{k} \prod_{s} g_s(u_{s,k})^{h_k}$,
which depends on the data only through $(\Lambda_1,\dots,\Lambda_K)$
of~\eqref{eq:aggLR}. By the factorization criterion this vector is sufficient
for $h$, and the conditional distribution of the data given it is free of $h$.
Rao--Blackwellizing any randomized procedure with respect to this vector
therefore preserves its rejection-set distribution at every configuration.
The $\Lambda_k$ are
independent across $k$ with a common null distribution $\nu_0$, the distribution of $\prod_s
g_s(U_s)$ for independent uniforms $U_s$, and a common alternative distribution
$\nu_1$. For any bounded measurable $f$,
\[
\E_{\nu_1}\{f(\Lambda)\} = \int f\Big(\prod_s g_s(u_s)\Big) \prod_s g_s(u_s)\,\mathrm{d}u
= \E_{\nu_0}\{f(\Lambda)\,\Lambda\},
\]
so $\mathrm{d}\nu_1/\mathrm{d}\nu_0(\lambda) = \lambda$. By the
symmetrization fact, it is enough to start from permutation-equivariant
procedures.

\emph{Step 2 (ordering and likelihood-ratio ties).}
It is enough to optimize over policies measurable in the likelihood ratios
and independent tie-breakers, and ordered by the likelihood-ratio preorder.
For $K=1$ the optimum is the ordinary Neyman--Pearson test, so suppose
$K\ge2$. The rearrangement in the proof of Theorem~1 of~\citet{rosset2022optimal},
which moves a rejection from a hypothesis to an accepted one with a strictly
larger likelihood ratio, then removes every inversion between unequal
$\Lambda_k$ values; it does not decrease $\Pi_w$ for any prior, because
$a_{w,a}\ge a_{w,b}$ whenever $\Lambda_a\ge\Lambda_b$ (Notation paragraph).
Nor does it weaken feasibility: rejecting a set $R$ of size $l$ has error
coefficient $e_{K-j}-p_R\,e_{K-l-j}(R^c)$, and by the swap identity of the
fixed-set representation, replacing a rejected coordinate by a larger
unrejected one does not decrease $p_R\,e_d(R^c)$, so among rejection sets of
size $l$ the top $l$ coordinates minimize this coefficient, which then equals
$B_{j,l}$.  This preserves strong control,
because for a symmetric policy the summed constraint at $j$ true hypotheses is
$\binom Kj$ times the error rate at each such configuration. When the law of $\Lambda_k$ has atoms, which
Assumptions~\ref{ass:model}--\ref{ass:known} allow, the ordering convention of that
theorem treats equal likelihood ratios as tied; it does not itself impose a
deterministic order inside the tie. The randomized extension, which breaks
ties by independent uniforms, is lossless. A symmetric action
kernel is invariant under every permutation that fixes the realized
likelihood-ratio vector, and, under configuration $h$, the density relative
to the global null is $\prod_k\Lambda_k^{h_k}$, which is unchanged within an
equal-$\Lambda$ block. Conditional symmetrization within each tied block
therefore preserves the objective $\Pi_w$ and strong feasibility and,
conditional on the number selected from that block, selects a uniform
subset. Independent uniforms $\xi_k$ implement exactly this selection; they
are ancillary and do not alter the likelihood ratio.

\emph{Step 3 (attainment).}
Lemma~\ref{lem:ordered_duality} of Section~\ref{app:dual} proves that the resulting
randomized program attains its optimum and has no duality gap. Hence an optimal policy exists in
this class. \hfill$\square$

\paragraph{Remark (computing the oracle)}
Nothing in this proof depends on how $\pi^*_w(\alpha)$ is computed.  For the
full-alternative objective $w(K)=1$, when the distribution of $\Lambda_k$
satisfies Assumptions~3--5 of~\citet{dubey2026esp}, the continuous program
falls within their framework, and the exact-population version of their
coordinate algorithm converges in objective value to the dual optimum.  Those
regularity conditions can fail at other priors, where the weak-signal example
at the end of Section~\ref{supp:proof_oak_auto} has a nonsmooth dual despite
continuous evidence, and they do not cover the atomic report distributions
used below.  Every reported value therefore comes from the finite primal and
dual programs of Section~\ref{app:dual}, which give matching bounds to the
reported numerical tolerance.

\subsection{Lossless and lossy report classes: proof of
Proposition~\ref{prop:distribution}}\label{supp:proof_distribution}
(a) With sufficient reports the center recovers every $\Lambda_{s,k}$, so it
forms $\Lambda_k=\prod_s\Lambda_{s,k}$ of~\eqref{eq:aggLR} and applies the
oracle policy of Proposition~\ref{prop:oracle} at level $\alpha$, carrying its
boundary randomization in $\xi$; the rejection set has the oracle's
distribution, so $d_w(T)=0$.  Full reports are sufficient: the center
recovers $\{u_{s,k}\}$, because the inverse of an injective Borel map with
Borel range is Borel by the Lusin--Suslin
theorem~\citep[Section~15]{kechris1995classical}, and it knows $g_s$ by
Assumption~\ref{ass:known}, so it computes $\Lambda_{s,k}=g_s(u_{s,k})$.  Hence
$\delta_w(\mathcal{T}^{\mathrm{full}}) = 0$.

(b) For $S = 1$ and $K = 1$, the threshold report $y_{1,1} = \mathbf{1}\{u_{1,1}
\le \alpha\}$ reproduces the Neyman--Pearson region $[0,\alpha]$, which is
most powerful for the simple null against the nonincreasing $g_1$, so
$d_1(T) = 0$ and hence $\delta_1(\mathcal T^{\mathrm{thr}})=0$.

For $S \ge 2$ and
$K = 1$, the center observes only the orthant
indicator $(\mathbf{1}\{u_{s,1} \le t_s\})_{s=1}^S$, a report of range at
most $2^S$; part (a) of Proposition~\ref{prop:converse}, proved below for any
finite-range report, gives a strictly positive fixed-design gap at every
threshold vector, degenerate faces of the cube $[0,1]^S$ included.  The
classwise claim follows by compactness.  The optimal report power, and hence the fixed-design gap, is continuous on the closed cube, by the argument given for the threshold-report
class in the proof of Theorem~\ref{thm:oracle_kink} below; that argument does
not use the oracle active-kink condition (OAK) of Definition~\ref{def:oak}.  A continuous, strictly positive function on a compact set has a
strictly positive minimum, so $\delta_1(\mathcal{T}^{\mathrm{thr}}) > 0$.
Multiplicity, $K\ge2$ with $S\ge2$, is the subject of
Theorem~\ref{thm:oracle_kink}.
\hfill$\square$

\subsection{Strict loss at a single hypothesis: proof of
Proposition~\ref{prop:converse}(a)}\label{supp:proof_converse}
The fixed-design gap is strictly positive because the conditional-Jensen
lower bound in~\eqref{eq:dual_pinch} is strictly positive; that bound is
proved for every $K$ and every report design in Steps~2 and~5 of
Section~\ref{supp:proof_graderate}, independently of this proposition.  We
first describe the report cells and their likelihood ratios.
A report design with finite-range maps partitions the evidence for one hypothesis into
at most $N=\prod_s|\operatorname{range}(T_s)|<\infty$ cells
$R_i=\prod_s T_s^{-1}(\{c_{s,i}\})$. These are measurable product sets whose factors need not be intervals; null-probability-zero cells may be discarded. The full report array therefore partitions
$[0,1]^{SK}$ into profiles $(R_{i_1},\ldots,R_{i_K})$. The center's rejection
probabilities are constant on each such \emph{complete profile}; the disposition
of one hypothesis need not be a function of its own cell alone. On a
per-hypothesis cell $R_i$, the report likelihood ratio is
\[
 \bar\Lambda_k=\E_0(\Lambda_k\mid R_i),
\]
the cell average of the oracle likelihood ratio on the cell $R_i$ of
hypothesis $k$; the computation in Step~1 of Section~\ref{supp:proof_graderate},
which uses only independence across sites under $\pr_0$, verifies this for
arbitrary measurable product partitions.

The argument uses the dual integrand $F_\mu=\max(0,\max_l\psi_l^\mu)$; if
$\mu^*$ minimizes the continuous dual, then
\[
 d_w(T)\ \ge\
 \E_0\{F_{\mu^*}(\Lambda)-F_{\mu^*}(\bar\Lambda)\}.
\]
Conditionally on a complete profile, the coordinates $\Lambda_k$ are independent
and have conditional means $\bar\Lambda_k$. Because $F_{\mu^*}$ is a maximum
of functions affine in each coordinate (the fixed-set representation),
iterated conditional Jensen shows
that every profile contributes a nonnegative gap. The gap is strict if,
with positive conditional probability over the other coordinates, one
coordinate charges both sides of a kink (an upward jump in derivative) of
the corresponding one-dimensional section of $F_{\mu^*}$
(Lemma~\ref{lem:oakk_hinge} of Section~\ref{supp:proof_oracle_kink}).

Here $K=1$ and $S\ge2$.  Some report cell must straddle the
oracle threshold $t_\alpha$ with positive probability, which makes the gap
strict.  In this case $F_{\mu}(v) = (v - \mu)_+$ and
$\mu^{*} = t_\alpha$, the null $(1-\alpha)$-quantile of $\Lambda_1$, whose
distribution is atomless (level sets of $\prod_s g_s$ are Lebesgue-null by strict
monotonicity in each coordinate); the kink occurs at $t_\alpha$ and the oracle
region is $\mathcal W = \{\Lambda_1 \ge t_\alpha\}$ with
$\pr_0(\mathcal W) = \alpha \in
(0,1)$. It suffices to find a cell whose conditional distribution charges both
$\{\Lambda_1 > t_\alpha\}$ and $\{\Lambda_1 < t_\alpha\}$.
Lemma~\ref{lem:oakk_product_cut} of Section~\ref{supp:proof_oracle_kink},
applied with $q=t_\alpha$, which lies in $(\lambda_-,\lambda_+)$ because
$\pr_0(\mathcal W)=\alpha\in(0,1)$, gives such a cell: if none existed,
$\mathcal W$ would agree up to a null set with a finite union of product
cells, which that lemma excludes.  Hence some cell
straddles the kink and $d_1(T)>0$ without any additional nondegeneracy
condition. \hfill$\square$

\subsection{The inverse-square rate: proof of
Proposition~\ref{prop:converse}(b)}\label{supp:proof_graderate}
The upper bound is proved for every $K$ and $S$, because
Theorem~\ref{thm:oracle_kink}(b) uses it in that generality; we state it
separately.

\begin{lemma}[Equal-width reports: an inverse-square upper bound for every $K$]\label{lem:uniform_upper}
Suppose Assumptions~\ref{ass:model}--\ref{ass:reg} hold and each $g_s$ is
continuously differentiable on $[0,1]$, so that $g_s\le c_{+,s}=g_s(0)<\infty$.
Then for every $K\ge1$, $S\ge1$ and prior $w$ there is a constant $C$,
depending on $K$, $S$, $\alpha$ and the densities, such that
$\delta_w(\mathcal T^{(m)})\le Cm^{-2}$ for all $m\ge1$.
\end{lemma}

Steps~1--4 below prove the lemma, whose case $K=1$ is the upper bound of
Proposition~\ref{prop:converse}(b); Steps~5 and~6 prove the lower bound of
the proposition, first for the equal-width report and then over
$\mathcal T_{\le m}$.
Write $\mathcal G$ for the sigma-field generated by the equal-width $m$-level
report, with
\[
T_s(u)=\min\{m,1+\lfloor mu\rfloor\}.
\]
Its atoms, apart from null endpoints, are the product cells of
$\pr_0$-measure $m^{-SK}$, where $\pr_0=\mathrm{Unif}[0,1]^{SK}$.

\emph{Step 1 (the report likelihood ratio is a conditional expectation).}
Throughout this proof, $\Lambda_k=\prod_s g_s(u_{s,k})$ is the oracle
per-hypothesis likelihood ratio of~\eqref{eq:aggLR} and
$\bar\Lambda_k=\prod_s\ell_s(y_{s,k})$ is the report likelihood ratio, the
main text's $L_k$, where
$\ell_s=\mathrm{d}\tilde\nu_{s,1}/\mathrm{d}\tilde\nu_{s,0}$. Write
$C_s(c)=T_s^{-1}(\{c\})$ for the cell corresponding to report value $c$.
Since the $u_{s,k}$ are independent across $s$ under $\pr_0$ and each cell
$C_s(c)$ has null probability $\mathrm{Leb}\{C_s(c)\}$, which is $1/m$ for the
equal-width report,
\[
\E_0[\Lambda_k \mid \mathcal G] = \prod_s \E_0[g_s(u_{s,k}) \mid y_{s,k}]
= \prod_s \frac{\int_{C_s(y_{s,k})} g_s(u)\,\mathrm du}{\mathrm{Leb}\{C_s(y_{s,k})\}}
= \prod_s \ell_s(y_{s,k}) = \bar\Lambda_k.
\]

\emph{Step 2 (dual comparison at one multiplier).} By
Lemma~\ref{lem:ordered_duality} of Section~\ref{app:dual}, together with the
ordering reductions in
Proposition~\ref{prop:oracle} and Theorem~\ref{thm:aggregator}, the optimal
level-$\alpha$
average power on a two-group distribution with per-hypothesis likelihood ratio $\ell$ is
$\min_{\mu \ge 0} \mathcal D_\ell(\mu)$, where
$\mathcal D_\ell(\mu) = \mathcal C(\mu) + \E_0[F_\mu(\ell)]$,
$\mathcal C(\mu) = \sum_{j=1}^K \binom{K}{j}\alpha\mu_j$, and $F_\mu = \max(0, \max_l
\psi_l^\mu)$ is the dual integrand; by the fixed-set representation of
Step~5 it is a finite maximum of functions each affine in $\mu$ and
multilinear in the likelihood-ratio vector. Hence $\pi^{*}_w =
\min_\mu \mathcal D_\Lambda(\mu)$ and the model-aware report power is
$\min_\mu \mathcal D_{\bar\Lambda}(\mu)$. Both minima are attained; let $\bar\mu$ attain
the latter. As $\mathcal C$ is common to both,
\[
\delta_w(\mathcal T^{(m)}) = \min_\mu \mathcal D_\Lambda(\mu) - \min_\mu \mathcal D_{\bar\Lambda}(\mu)
\le \mathcal D_\Lambda(\bar\mu) - \mathcal D_{\bar\Lambda}(\bar\mu)
= \E_0\big[F_{\bar\mu}(\Lambda) - F_{\bar\mu}(\bar\Lambda)\big].
\]
The minimizer $\bar\mu$ is bounded uniformly in $m$:
Lemma~\ref{lem:ordered_duality}, applied to the report distribution, gives
$\|\bar\mu\|_1\le1/\alpha$ for every $m$.
As $\|\bar\mu\|_1\le1/\alpha$ and $\bar\Lambda_k \le \prod_s c_{+,s}$, the integrand
$F_{\bar\mu}$ is then Lipschitz in the likelihood-ratio vector with a constant
$C_{\mathrm{Lip}}$ depending only on $K,S,\alpha$ and the upper bounds $c_{+,s}$.

\emph{Step 3 (multilinear cancellation off the decision surface).} Since $F_{\bar\mu}(\bar\Lambda)$ is $\mathcal G$-measurable, conditioning on $\mathcal G$ gives
\[
\delta_w(\mathcal T^{(m)}) \le \E_0\big[\,\E_0[F_{\bar\mu}(\Lambda)\mid\mathcal G] -
F_{\bar\mu}(\bar\Lambda)\,\big].
\]
$F_{\bar\mu}$ is multilinear in $(\Lambda_1,\dots,\Lambda_K)$ on each region cut out by the
finitely many surfaces $\{\psi_l^{\bar\mu} = 0\}$, $\{\psi_l^{\bar\mu} =
\psi_{l'}^{\bar\mu}\}$, and the order-statistic ties $\{\Lambda_k = \Lambda_{k'}\}$; write
$\mathcal V$ for their union, the decision surface. On a cell whose
likelihood-ratio support avoids $\mathcal V$, $F_{\bar\mu}$ is one fixed multilinear function $\psi$, and since the $\Lambda_k$ are conditionally independent
across $k$ with $\E_0[\Lambda_k\mid\mathcal G] = \bar\Lambda_k$ (Step~1),
$\E_0[\prod_{k\in I} \Lambda_k \mid \mathcal G] = \prod_{k\in I}\bar\Lambda_k$ for every
index set $I$, whence $\E_0[F_{\bar\mu}(\Lambda)\mid\mathcal G]=\psi(\bar\Lambda)$.  As $\psi$ is zero or a fixed-set branch,
$F_{\bar\mu}\ge\psi$ everywhere (Step~5), and conditional Jensen gives
$\psi(\bar\Lambda)\le F_{\bar\mu}(\bar\Lambda)\le
\E_0[F_{\bar\mu}(\Lambda)\mid\mathcal G]=\psi(\bar\Lambda)$: the cell contributes zero.

\emph{Step 4 (boundary layer).} It remains to treat cells meeting $\mathcal V$;
the argument below shows that only cells strictly crossing a constituent
surface can contribute. On such a cell, the $C^1$ smoothness of the $g_s$ and
the width $1/m$ give $\|\Lambda - \bar\Lambda\|\le C_{\mathrm{dev}}/m$ pointwise, so the
conditional-Jensen gap is at most $C_{\mathrm{Lip}}C_{\mathrm{dev}}/m$.
The crossing cells are counted by an explicit lemma, with no recourse to
regularity of $\mathcal V$. Work in the coordinates $x_{s,k} =
g_s(u_{s,k})$, a strictly monotone reparameterization of each coordinate
separately (Assumption~\ref{ass:reg}). Under this reparameterization, the report grid
remains a product partition into $m^{SK}$ boxes and every constituent
function of $\mathcal V$ (each branch $\psi_l$ on a fixed order region,
each branch difference and each order-statistic tie $\Lambda_k - \Lambda_{k'}$) becomes
a \emph{multilinear} polynomial: $\Lambda_k = \prod_s x_{s,k}$, so every
$x_{s,k}$ appears with degree at most one.

\emph{Lemma (grid-crossing count).} Let $p$ be multilinear in $d$
variables, and let a box of $\mathbb{R}^d$ be partitioned into $m^d$ product
cells of arbitrary, possibly unequal, widths.
Then at most $d\, 2^{d-1} (m+1)^{d-1}$ cells contain both a point with
$p > 0$ and a point with $p < 0$.

\emph{Proof.} A multilinear function is affine in each variable
separately, so its maximum and minimum over a cell are attained at
corners. Zero-valued corners require a small precaution. There are only
finitely many cells on which $p$ takes both signs. If there are any, choose
\[
0<\varepsilon<\min_C\{\max_C p,-\min_C p\},
\]
where the minimum is over the finitely many crossed cells $C$, and choose
$\varepsilon$ different from $p$ at every grid vertex. Then every cell
crossed by $p$ is also crossed by the multilinear polynomial
$p-\varepsilon$, and the latter is nonzero at every grid vertex. Its crossed
cells therefore have corners
of both strict signs. Walking through the cell's corner hypercube from a
positive corner to a negative corner gives an axis-parallel grid edge whose
endpoints have opposite signs. Along any axis-parallel line of grid
vertices, $p-\varepsilon$ is affine in the varying coordinate, so its
vertex signs change at most once. There are at most
$d(m+1)^{d-1}$ such lines, and each edge borders at most $2^{d-1}$ cells.
This proves the claim. \hfill$\lozenge$

The constituent family is finite with cardinality depending on $K$ alone:
for each of the $K!$ order regions it comprises the $K+1$ branches
$\psi_l^{\bar\mu}$, their pairwise differences, and the
$K(K-1)/2$ order-statistic ties $\Lambda_k-\Lambda_{k'}$.
The multiplier $\bar\mu$
enters only through the coefficients of these multilinear polynomials,
never their number or degrees, and the lemma's count holds uniformly over
all multilinear $p$, hence uniformly in $\bar\mu$ and in $m$.

A cell on which
no constituent polynomial takes both positive and negative values admits one
maximizing branch throughout. Sign constancy of the order-tie
polynomials permits a fixed weak ordering $\sigma$ of the likelihood ratios on
the whole cell. In the resulting transformed $x$-cell, which has nonempty
interior because every $g_s$ is strictly decreasing, write the order-specific
branches as $0,\psi_1^{\bar\mu,\sigma},\ldots,\psi_K^{\bar\mu,\sigma}$. Choose an
interior point $x^\circ$ outside the zero sets of all pairwise branch
differences that are not identically zero on the cell; such a point exists
because the family is finite and a nonzero polynomial has a zero set with
empty interior. Let $\psi$ be a branch maximal at $x^\circ$. For every other
branch $\psi'$, either $\psi-\psi'$ vanishes identically on the cell or
$\psi(x^\circ)-\psi'(x^\circ)>0$. Since $\psi-\psi'$ does not take both signs
on the cell,
the latter case implies $\psi-\psi'\geq0$ throughout. Thus $\psi$ is maximal on the
whole cell (an identically tied branch is harmless), so $F_{\bar\mu}$ is one
multilinear polynomial there and the cell's contribution is zero, as in
Step~3.

Every cell that can contribute
is strictly crossed by at least one constituent polynomial, so by
the lemma there are $O(m^{SK-1})$ of them, of total $\pr_0$-measure at most
$C_{\mathrm{cross}}/m$. Therefore
\[
\delta_w(\mathcal T^{(m)})
\le \frac{C_{\mathrm{cross}}}{m}
     \frac{C_{\mathrm{Lip}}C_{\mathrm{dev}}}{m}
= \frac{C}{m^2},
\]
where $C$ may depend on $K,S,\alpha$ and the densities.
Boundedness of the likelihood ratio enters through the Lipschitz constant
$C_{\mathrm{Lip}}$ and, with the $C^1$ smoothness, through the within-cell deviation
$C_{\mathrm{dev}}$; both fail for the one-sided normal-shift alternative, whose likelihood
ratio and its derivative diverge as $u \to 0$, and there the exponent is
observed to fall below $2$ (Section~\ref{supp:benchmarks}).  This proves Lemma~\ref{lem:uniform_upper}.

\emph{Step 5 (the fixed-set representation and a matching lower bound at $K=1$).} The comparison of Step~2 has a
mirror: if $\mu^{*}$ minimizes the oracle dual $\mathcal D_\Lambda$ then, since $\bar\mu$
minimizes $\mathcal D_{\bar\Lambda}$,
\[
\delta_w(\mathcal T^{(m)}) \;=\; \mathcal D_\Lambda(\mu^{*}) - \mathcal D_{\bar\Lambda}(\bar\mu)
\;\ge\; \mathcal D_\Lambda(\mu^{*}) - \mathcal D_{\bar\Lambda}(\mu^{*})
\;=\; \E_0\big[F_{\mu^{*}}(\Lambda) - F_{\mu^{*}}(\bar\Lambda)\big] \;\ge\; 0.
\]
The dual integrand is a pointwise maximum, over rejection sets, of functions of $(\Lambda_1,\dots,\Lambda_K)$ that are multilinear, hence affine in each coordinate, and the
$\Lambda_k$ are conditionally independent given $\mathcal G$ with $\E_0[\Lambda_k \mid
\mathcal G] = \bar\Lambda_k$ (Step~1), so the one-coordinate conditional Jensen
inequality, iterated over the coordinates, gives $\E_0[F_{\mu}(\Lambda) \mid
\mathcal G] \ge F_{\mu}(\bar\Lambda)$ for every $\mu \ge 0$. The fixed-design gap
is thus bounded above and below by the same conditional-Jensen gap
evaluated at the two multipliers, and cells avoiding the decision surface
contribute zero to either side (Step~3).  Steps~2 and~5 use only the
identity $\bar\Lambda=\E_0(\Lambda\mid\mathcal G)$, which holds for every
report design (Step~6), so with $d_w(T)$ in place of
$\delta_w(\mathcal T^{(m)})$ they prove the two-sided
bound~\eqref{eq:dual_pinch} of the main text for every fixed design $T$.

We now prove the fixed-set representation stated at the start of
Section~\ref{supp:proofs}, which also gives the asserted coordinatewise
convexity. For a fixed rejection
set $R$ of size $l$, take $p_R$ and $e_d$ as in the Notation paragraph of
Section~\ref{supp:proofs}, evaluated at $v=\Lambda$; in particular,
$e_K(\Lambda)=\prod_k\Lambda_k$.
The dual branch of $R$ is, at the point
prior (Section~\ref{supp:proof_priorfree} extends the representation to every
$w$),
\[
 \psi_R^\mu(\Lambda)=\frac lK e_K(\Lambda)
 -\sum_{j=1}^K\mu_j
 \{e_{K-j}(\Lambda)-p_R e_{K-l-j}(R^c)\}.
\]
If $y\in R$, $x\notin R$, and $x\ge y$, write
$R'=R\setminus\{y\}$ and $R''=R^c\setminus\{x\}$. Replacing $y$ by $x$ changes
the $R$-dependent term by
\[
 p_{R'}x\{e_d(R'')+y e_{d-1}(R'')\}
 -p_{R'}y\{e_d(R'')+x e_{d-1}(R'')\}
 =p_{R'}(x-y)e_d(R'')\ge0,
 \]
for every relevant degree $d=K-l-j$. Thus each nonnegative weighted sum of
these terms is maximized by taking $R$ to be the top $l$ coordinates. Hence
the ordered branch $\psi_l^\mu$ is the maximum of $\psi_R^\mu$ over all
$|R|=l$, and $F_\mu$ is the maximum of zero and all the fixed-set branches.
Every $\psi_R^\mu$ is multilinear, so this representation proves that
$F_\mu$ is convex in each coordinate separately.

At $K=1$ the gap contributed by a report cell is the smaller of its two
one-sided hinge integrals about the oracle threshold.  Here $F_\mu(\ell)=(\ell-\mu)_+$
and the oracle multiplier $\mu^*=t_\alpha$ does not depend on $m$; on a
report cell $C$, the gap is
\[
\begin{split}
&\E_0[(\Lambda-\mu^*)_+;C]
 -\big\{\E_0(\Lambda-\mu^*;C)\big\}_+ \\
&\qquad={}
\min\!\left\{\E_0[(\Lambda-\mu^*)_+;C],
                 \E_0[(\mu^*-\Lambda)_+;C]\right\}.
\end{split}
\]

In the notation of Lemma~\ref{lem:oakk_grid_loss} below, whose proof does
not depend on this proposition, the cell gap is
$\min\{H^+_C(\mu^*),H^-_C(\mu^*)\}$.  The distribution of $\Lambda_1$ is
atomless with interval support charging every open subinterval (as in the
proof of Proposition~\ref{prop:converse}(a)), so
$\pr_0(\Lambda_1>\mu^*)=\alpha\in(0,1)$ forces $\mu^*$ strictly inside the open
range $(\prod_sg_s(1),\prod_sg_s(0))$ of $\Lambda_1$.
Lemma~\ref{lem:oakk_grid_loss} with $Q=\{\mu^*\}$ therefore gives
\[
\delta_1(\mathcal T^{(m)})
\ge\sum_{C}\min\{H^+_C(\mu^*),H^-_C(\mu^*)\}
\ge c_Qm^{-2},\qquad m\ge m_Q,
\]
where $c_Q$ and $m_Q$ are the constants of that lemma and may depend on $S$,
$\alpha$, and the densities, but not on $m$.

\emph{Step 6 (the lower bound after optimizing over all reports with at most
$m$ values per site).}  We now prove the stronger classwise statement.  Fix an arbitrary
$T\in\mathcal T_{\le m}$ and let $\mathcal G$ be the sigma-field generated by
its report array.  The identity $\bar\Lambda=\E_0(\Lambda\mid\mathcal G)$ and the
conditional-Jensen lower bound of Step~5 remain valid for arbitrary
measurable report cells.  Write $q=\mu^*=t_\alpha$.  As shown in Step~5, $q$ lies strictly
inside the range of $\Lambda=\prod_sg_s(U_s)$.  The continuous strictly decreasing
map $t\mapsto\prod_sg_s(t)$ therefore has a point $t_0\in(0,1)$ satisfying
$\prod_sg_s(t_0)=q$.

Write $\Lambda=XY$ with $X=g_1(U_1)$ and $Y=\prod_{s=2}^Sg_s(U_s)$.  For
$Y>0$, $(XY-q)_+=Y(X-r)_+$ with $r=q/Y$, so, given sites $2,\ldots,S$, the
oracle's kink is a hinge in $X$ at $r$ with weight $Y$.  We localize to where
$X$ has a null density bounded below and the change of variable from $u_2$ to
$r$ has a Jacobian bounded below, and write
$\tilde u=(u_3,\ldots,u_S)$ and $g_{3:S}(\tilde u)=\prod_{s=3}^Sg_s(u_s)$.
Fix a compact interval $J_2\subset(0,1)$ about $t_0$, and then choose a
compact interval $I$ about $x_0=g_1(t_0)$ and, when $S>2$, a compact box
$\mathcal N$ about $(t_0,\ldots,t_0)$ for $(U_3,\ldots,U_S)$, small enough
that, for every $\tilde u\in\mathcal N$, the strictly increasing change of
variable
\begin{equation}\label{eq:anym_change}
 r=\frac{q}{g_{3:S}(\tilde u)\,g_2(u_2)}
\end{equation}
maps $J_2$ onto an interval containing $I$.  The convention for $S=2$ is $g_{3:S}\equiv1$ and
$|\mathcal N|=1$.
Compactness and $g_2'<0$ give a constant $c_6>0$ such that throughout
these neighborhoods
\begin{equation}\label{eq:anym_jacobian}
 Y\left|\frac{\mathrm du_2}{\mathrm dr}\right|
 =\frac{g_{3:S}(\tilde u)^2g_2(u_2)^3}{q|g_2'(u_2)|}\ge c_6.
\end{equation}
The null density $f_X(x)=1/|g_1'\{g_1^{-1}(x)\}|$ is bounded below by some
$f_->0$ on $I$.

Refining the report sigma-field can only shrink the conditional-Jensen gap,
so it suffices to bound the gap after revealing everything except the
position of $X$ within site 1's cell, on a neighborhood of the threshold.  Refine $\mathcal G$ to a sigma-field
$\mathcal H$ in which $U_2,\ldots,U_S$ are fully measurable and, at site 1, $X$ is exactly
measurable on $\{X\notin I\}$ but only the pair
$(\mathbf1\{X\in I\},T_1(U_1))$ on $\{X\in I\}$.  Thus
$\mathcal G\subseteq\mathcal H$.  If
$\widetilde\Lambda=\E_0(\Lambda\mid\mathcal H)$, conditional Jensen gives
\[
 \E_0F_{\mu^*}(\widetilde\Lambda)
 \ge \E_0F_{\mu^*}(\bar\Lambda),
\]
so the report's conditional-Jensen gap is no smaller than the gap after this
refinement.

After the refinement the gap is at least a constant times the integrated
hinge regret of site 1's cells, which is at least of order $m^{-2}$ for any
partition into at most $m$ cells.  For each non-null output value $j$ of
$T_1$, let
\[
 A_j=I\cap\{x:T_1(g_1^{-1}(x))=j\},\quad
 p_j=\pr_0(X\in A_j),\quad
 \bar x_j=\E_0(X\mid X\in A_j),
\]
omitting null cells, write $P_X$ for the null distribution of $X$, and define
the unnormalized hinge regret
\[
 \Delta_j(r)=\int_{A_j}(x-r)_+\,\mathrm dP_X(x)
       -p_j(\bar x_j-r)_+.
\]
On $A_j$ the refinement replaces $X$ by $\bar x_j$ and leaves $Y$ exact;
outside $I$ it leaves $\Lambda$ exact.  On $\{X\in A_j\}$ the refined gap is
therefore $Y\{(X-r)_+-(\bar x_j-r)_+\}$, whose null mean over $X\in A_j$ is
$Y\Delta_j(r)\ge0$ by Jensen.  Consequently, restricting the integral to
$\tilde u\in\mathcal N$, changing $u_2$ to $r$ by~\eqref{eq:anym_change}, and
using~\eqref{eq:anym_jacobian}, which absorbs the weight $Y$, yield
\[
 \E_0\{F_{\mu^*}(\Lambda)-F_{\mu^*}(\widetilde\Lambda)\}
 \ge |\mathcal N|c_6\sum_j\int_I \Delta_j(r)\,\mathrm dr.
\]
Because $A_j\subseteq I$ and $\bar x_j\in I$, direct integration gives the
exact identity
\begin{equation}\label{eq:anym_hinge_variance}
 \int_I\Delta_j(r)\,\mathrm dr
 =\frac12\int_{A_j}(x-\bar x_j)^2\,\mathrm dP_X(x).
\end{equation}
If $|A_j|$ is the Lebesgue measure of $A_j$, then
\[
 \int_{A_j}(x-\bar x_j)^2\,\mathrm dP_X(x)
 \ge f_-\inf_{c\in\mathbb R}\int_{A_j}(x-c)^2\,\mathrm dx
 \ge \frac{f_-|A_j|^3}{12}.
\]
The last inequality is the one-dimensional rearrangement fact that, among
measurable sets of length $|A_j|$, the second moment about a set's best
center is minimized by an interval centered there.  The at most $m$ sets $A_j$
partition $I$ up to null sets, so convexity of $x^3$ gives
\[
 \sum_j|A_j|^3\ge\frac{|I|^3}{m^2}.
\]
Together with~\eqref{eq:anym_hinge_variance}, this proves, uniformly over
every $T\in\mathcal T_{\le m}$,
\[
 d_1(T)\ge \frac{|\mathcal N|c_6 f_-|I|^3}{24m^2}.
\]
Taking the infimum over $T$, and using the equal-width design as a member of
$\mathcal T_{\le m}$ for the upper bound, proves the classwise claim of
Proposition~\ref{prop:converse}(b).  Notice that the argument did not require
the report cells $A_j$ to be intervals. \hfill$\square$

\subsection{OAK at objectives charging the full alternative: proof of
Lemma~\ref{lem:oak_auto}}\label{supp:proof_oak_auto}

The proof shows that at every minimizer $\mu^*$ of $\mathcal D$, which exists by
Lemma~\ref{lem:ordered_duality}, the set $\{F_{\mu^*}>0\}$, where some
rejection beats rejecting nothing, has null probability strictly between $0$
and $1$ (Steps~1--2).  Some coordinate section then crosses its boundary on a
set of positive probability (Step~3), and the crossing is an upward kink
(Step~4); the last paragraph makes $q$ and $\Delta$ Borel.

Throughout, $\pr_0$ is the global-null law of
$\Lambda=(\Lambda_1,\ldots,\Lambda_K)$, $\mathcal D$ is the continuous dual
of Section~\ref{app:dual}, and $\psi_l^\mu=A_l-\sum_{j=1}^K\mu_jB_{j,l}$ are its
branches, with $B_{j,l}\ge0$ throughout and $B_{K,l}=1$ for $l\ge1$.  Under
$\pr_0$ the coordinates of $\Lambda$ are independent, each distributed as
$\prod_s g_s(U_s)$ for independent uniform $U_s$; by
Assumption~\ref{ass:reg} each factor is atomless and charges every open
subinterval of $(g_s(1),g_s(0))$, so each coordinate is atomless and charges
every open subinterval of $(\lambda_-,\lambda_+)$. Two facts about such a
product law are used, the first in Step~1 and the second in Step~3.

First, the zero set of a nonzero polynomial $p$ in $K$ variables is
$\pr_0$-null. The claim is immediate for $K=1$. For $K>1$, regard $p(x,z)$
as a polynomial in $x$ and choose a nonzero coefficient polynomial in $z$;
its zero set is null by induction. Outside this exceptional set,
$x\mapsto p(x,z)$ has finitely many roots. The atomless coordinate law
assigns zero mass to them, so independence and Fubini's theorem complete
the argument. Second, if, for each coordinate $a$, the indicator of a measurable set $R$
is $\pr_0$-almost surely a function of the other coordinates, then
$\pr_0(R)\in\{0,1\}$: integrating out the coordinates one at
a time shows that the indicator agrees almost surely with a constant.

\emph{Step 1 (branch ties are null).}  Fix $\mu\ge0$ and $l\ge1$.  On the
ordering region $\{v_1>\cdots>v_K\}$ each $B_{j,l}$ is a polynomial of
degree $K-j\le K-1$, while $A_l=(K!)^{-1}\sum_{i\le l}a_{w,i}$ contains the
degree-$K$ term $(K!)^{-1}\,l\,w(K)\,(K-1)!\,e_K=(l\,w(K)/K)\,e_K$, because
$a_{K,k}=(K-1)!\,e_K$ for every $k$ and every other $a_{\gamma,k}$ is
homogeneous of degree $\gamma<K$ (Section~\ref{supp:proof_priorfree}).  Since
$w(K)>0$, $\psi_l^\mu$ is a nonzero polynomial on the region, as is every
difference $\psi_l^\mu-\psi_{l'}^\mu$ for $l\ne l'$, whose degree-$K$ term is
$\{(l-l')w(K)/K\}\,e_K$; by symmetry, the same holds on every ordering
region.  Rank ties are zero sets of the polynomials $v_i-v_k$.  By the first fact,
$\pr_0\{\psi_l^\mu(\Lambda)=0\}=0$,
$\pr_0\{\psi_l^\mu(\Lambda)=\psi_{l'}^\mu(\Lambda)\}=0$ and
$\pr_0\{\Lambda_i=\Lambda_k\}=0$.

\emph{Step 2 (the rejection region is neither null nor full).}  Fix a
minimizer $\mu^*$ and write $F=F_{\mu^*}$, $\Psi=\max_l\psi_l^{\mu^*}$ and
$R=\{F>0\}=\{\Psi>0\}$.  If $\pr_0(R)=1$ then, since $B_{K,l}=1$,
$F_{\mu^*+\epsilon\iota_K}=(\Psi-\epsilon)_+$ for $\epsilon>0$, where
$\iota_j$ denotes the $j$th unit vector of $\mathbb R^K$, and
\[
 \mathcal D(\mu^*+\epsilon\iota_K)-\mathcal D(\mu^*)
 =\epsilon\alpha-\E_0\min(\Psi,\epsilon)
 \le\epsilon\{\alpha-\pr_0(\Psi\ge\epsilon)\},
\]
which is negative for small $\epsilon$ because
$\pr_0(\Psi\ge\epsilon)\to \pr_0(\Psi>0)=1>\alpha$, contradicting minimality.
If $\pr_0(R)=0$ then, by Step~1, $\psi_l^{\mu^*}(\Lambda)<0$ for every $l$
almost surely, so for almost every $v$ and every $j$,
$F_{\mu^*-\epsilon\iota_j}(v)=0$ for all small $\epsilon>0$; the branches are
bounded on the support, so dominated convergence gives
\[
 \mathcal D(\mu^*-\epsilon\iota_j)=\mathcal D(\mu^*)-\epsilon\alpha\binom Kj+o(\epsilon),
\]
and minimality forces $\mu^*_j=0$ for every $j$.  But then $\mu^*=0$, and $F_0\ge A_K\ge w(K)\,e_K>0$ on the support, so $\pr_0(R)=1$, a
contradiction.  Hence $0<\pr_0(R)<1$.

\emph{Step 3 (a coordinate section crosses the boundary).}  By Step~2 and the
second fact, there are a coordinate $a$ and a measurable set $\mathcal B$ of values
$z=\Lambda_{-a}$ with $\pr_0(\Lambda_{-a}\in\mathcal B)>0$ such that, for
every $z\in\mathcal B$, the section $\{x:(x,z)\in R\}$ has
$\Lambda_a$-probability strictly between $0$ and $1$.

\emph{Step 4 (the crossing is a kink).}  For $z\in\mathcal B$ the section
$f_z(x)=F(z_1,\ldots,x,\ldots,z_K)$ is nonnegative, continuous and
piecewise affine in $x\in(\lambda_-,\lambda_+)$ with finitely many pieces,
which change only at the rank ties $x=z_k$ and at branch crossings.  Its
zero set $N_z$ is a finite union of closed intervals and has positive
$\Lambda_a$-probability, so it contains a nondegenerate component because
the law is atomless; its complement in $(\lambda_-,\lambda_+)$ also has
positive probability, so some nondegenerate component $[c,d]$ of $N_z$ has
$d<\lambda_+$ or $c>\lambda_-$.  Put $q(z)=d$ in the first case and
$q(z)=c$ otherwise.  Then $q(z)\in(\lambda_-,\lambda_+)$, and $f_z$ vanishes on
one side of $q(z)$ and is positive on a punctured neighborhood on the
other.  The one-sided derivative at $q(z)$ is therefore $0$ on the vanishing
side and, on the positive side, equals the slope of the active branch, which
is nonzero and has the sign that makes $f_z$ positive there.  In both cases
($q(z)=d$ or $q(z)=c$) the derivative jumps upward by some $\Delta(z)>0$.

Both $q$ and $\Delta$ can be taken Borel.  For rationals
$r<s$ and $t$ in $(\lambda_-,\lambda_+)$ put
$Z_{r,s}=\bigcap_{y\in[r,s]\cap\mathbb Q}\{z:F(y,z)=0\}$ and
$P_t=\{z:F(t,z)>0\}$, Borel sets because $F$ is continuous, with
$z\in Z_{r,s}$ exactly when $f_z$ vanishes on $[r,s]$.  Let $\mathcal B_1$ be
the Borel set of $z\in\mathcal B$ for which $z\in Z_{r,s}\cap P_t$ for some
$r<s<t$, and on it let $q(z)$ be the supremum of such $s$: this is the right
endpoint of the rightmost nondegenerate component of $N_z$ that has a point
of positivity to its right, so $q(z)$ is interior and $f_z>0$ on a punctured
right neighborhood.  On $\mathcal B\setminus\mathcal B_1$ every point of
positivity lies to the left of every nondegenerate component, and $q(z)$ is
the infimum of $r$ over rationals $t<r<s$ with $z\in Z_{r,s}\cap P_t$, the
left endpoint of the leftmost nondegenerate component.  For every real $\rho$, the set $\{q>\rho\}$ on $\mathcal B_1$ and the set $\{q<\rho\}$ on $\mathcal B\setminus\mathcal B_1$ are countable unions of the sets $Z_{r,s}\cap P_t$,
so $q$ is Borel, and $\Delta(z)=\lim_n n\,F(q(z)\pm1/n,z)$, with the sign
pointing to the positive side, is a limit of Borel functions.  Thus $(\mu^*,a,\mathcal B,q,\Delta)$ is an
OAK witness; since $\mu^*$ was an arbitrary minimizer, every minimizing
multiplier belongs to a witness. \hfill$\square$

\paragraph{The condition $w(K)>0$ cannot be dropped}
Let $K=3$, $w(1)=1$ and $S\ge2$, with sites whose likelihood ratios satisfy
$\lambda_+\le2\lambda_-$; for instance $g_s(u)=1.1-0.2u$ at two sites
gives $(\lambda_-,\lambda_+)=(0.81,1.21)$.  From Section~\ref{app:dual},
$A_l=\tfrac13\sum_{i\le l}v_{(i)}$, $B_{3,l}=1$, $B_{2,1}=e_1-v_{(1)}$,
$B_{2,2}=B_{2,3}=e_1$ and $B_{1,3}=e_2$.  At $\mu=(0,\tfrac13,0)$,
\[
 \psi_1=\tfrac13\{v_{(1)}-v_{(2)}-v_{(3)}\},\qquad
 \psi_2=-\tfrac13v_{(3)},\qquad
 \psi_3=0 .
\]
On the support $v_{(1)}\le\lambda_+\le2\lambda_-\le v_{(2)}+v_{(3)}$, so
$F_\mu\equiv0$ and $\mathcal D(\mu)=\alpha\binom32\cdot\tfrac13=\alpha$.
The data-free rule that rejects all three hypotheses with probability
$\alpha$ is feasible, with family-wise error rate exactly $\alpha$ at every
configuration with a true hypothesis, and has power $\alpha$ at every
configuration with a false one, so by weak duality it is optimal and
$\pi^*_w(\alpha)=\alpha$.  It uses no report, so every report class,
threshold reports included, has zero deficit.  OAK fails outright: the
minimizer is unique, and $F_\mu$ has no kink there.  Indeed, since
$\E_0e_1=\E_0e_2=3$, every $\mu\ge0$ has
$\mathcal D(\mu)\ge\alpha s+\E_0\psi_3^\mu=1-(1-\alpha)s$ with
$s=3\mu_1+3\mu_2+\mu_3$, as well as $\mathcal D(\mu)\ge\alpha s$, so
$\mathcal D(\mu)=\alpha$ forces $s=1$ and $F_\mu=\psi_3^\mu=0$ almost surely;
hence $\psi_3^\mu\equiv0$ and $\mu=(0,\tfrac13,0)$. Along the ray $\mu=(0,x,0)$, $x\ge0$, the same branch comparison gives
$\mathcal D(0,x,0)=3\alpha x+(1-3x)_+$. Thus continuous evidence can produce
a nonsmooth dual at this prior, the case cited in the remark after the proof
in Section~\ref{supp:proof_oracle}; the policy choosing the largest rejection
count among maximizing actions has a family-wise error rate that jumps from
one to zero immediately to the right of $x=1/3$.
The data-free rule above recovers a valid optimum. The example shows that
the guarantee of Lemma~\ref{lem:oak_auto} relies on the objective charging
the full alternative: when it does not, and the likelihood ratios are flat
enough, the dual integrand can vanish throughout the support, with rejecting
all and rejecting none tied everywhere.

\subsection{Active-kink strictness and rate: proof of
Theorem~\ref{thm:oracle_kink}}\label{supp:proof_oracle_kink}

By Lemma~\ref{lem:oak_auto}, $w(K)>0$ implies OAK, so OAK is assumed
throughout.  The proof rests on three lemmas: no finite union of product
cells matches an upper level set of the likelihood ratio at an interior
threshold up to a null set (this lemma also proves
Proposition~\ref{prop:converse}(a)); a hinge inequality for convex
piecewise-affine sections; and a grid-loss lower bound.  Part~(a) combines
the two-sided bound~\eqref{eq:oakk_pinch} with the first two lemmas, and the
threshold class follows by continuity and compactness.  Part~(b) proves the
lower bound first for equal-width reports with the grid-loss lemma, then for
every design in $\mathcal T_{\le m}$ by extending the argument of Step~6 of
Section~\ref{supp:proof_graderate}; the upper bound is
Lemma~\ref{lem:uniform_upper}.  After the proof, the section gives a
sufficient condition for OAK and the explicit $K=2$ witness of
Lemma~\ref{lem:k2_witness}, which Figure~\ref{fig:oracle_kink} of the main text
illustrates.  Write
\[
 \Lambda(u)=\prod_{s=1}^S g_s(u_s),\qquad
 \lambda_-=\prod_{s=1}^Sg_s(1),\qquad
 \lambda_+=\prod_{s=1}^Sg_s(0).
\]
For hypothesis $k$, put $\Lambda_k=\Lambda(u_{\cdot,k})$, and for a coordinate
$a\in\{1,\ldots,K\}$ write $\Lambda_{-a}=(\Lambda_k)_{k\ne a}$.

\begin{lemma}[No finite product partition resolves an interior cut]
\label{lem:oakk_product_cut}
Suppose $S\ge2$ and Assumption~\ref{ass:reg} holds.  Let the finite-range sitewise
maps $T_s$ partition $[0,1]^S$ into product cells of positive $\pr_0$-probability
\[
 D_c=\prod_{s=1}^S T_s^{-1}(\{c_s\}).
\]
For every $q\in(\lambda_-,\lambda_+)$, at least one cell $D_c$ satisfies
\begin{equation}\label{eq:oakk_straddle}
 \pr_0\{\Lambda<q\mid D_c\}>0,\qquad
 \pr_0\{\Lambda>q\mid D_c\}>0.
\end{equation}
\end{lemma}

\begin{proof}
The boundary $\{\Lambda=q\}$ is $\pr_0$-null: after fixing
$(u_2,\ldots,u_S)$, strict monotonicity of $g_1$ leaves at most one possible
$u_1$, and Fubini's theorem applies.  If no cell
satisfied~\eqref{eq:oakk_straddle}, the upper level set $\mathcal W_q=\{u:\Lambda(u)>q\}$ would agree up to a null set with a union of
finitely many report cells, hence with a finite union
$\mathcal W'=\bigcup_{i=1}^N(A_i\times B_i)$, where
$A_i\subseteq[0,1]$ and $B_i\subseteq[0,1]^{S-1}$ are measurable.

The conditional probability of $\mathcal W_q$ given site 1's value is
continuous and strictly decreasing on a nonempty open interval.  Let
$Y=\prod_{s=2}^Sg_s(U_s)$ for independent uniform $U_s$ and define
\[
 \varphi_q(u_1)
 =\pr_0\!\left\{Y>\frac{q}{g_1(u_1)}\right\}.
\]
Here and below, $q/0=+\infty$.
The distribution of $Y$ is atomless and charges every open subinterval of its interval
support.  To see atomlessness, condition on all but one factor: $Y$ is then, almost
surely, a positive constant times a continuous strictly monotone transform of
a uniform variable.  For full interval support, any
interior product value has a preimage in $(0,1)^{S-1}$, and a sufficiently
small open product neighborhood of that preimage is mapped inside any given
open interval containing the value and has positive Lebesgue measure.  Since
$q/g_1(u_1)$ is continuous and strictly increasing, $\varphi_q$ is continuous
and strictly decreasing on
$J_q=\{u_1:0<\varphi_q(u_1)<1\}$.  This is a nonempty open interval because
$q$ is strictly inside the range of $\Lambda$.

The section of $\mathcal W'$ at $u_1$ is
$\bigcup_{i:u_1\in A_i}B_i$, so it is determined by the membership vector
$(\mathbf 1\{u_1\in A_i\})_{i=1}^N$.  Those vectors partition $[0,1]$ into
at most $2^N$ measurable sets, and the section of $\mathcal W'$ is literally constant
on each one.  If $\mathcal W_q=\mathcal W'$ almost everywhere, Fubini's theorem makes
$\varphi_q$ almost everywhere constant on each membership set.  Some such
set meets $J_q$ in positive measure, contradicting the injectivity of the
strictly decreasing $\varphi_q$ on $J_q$.  Thus a cell
satisfying~\eqref{eq:oakk_straddle} must exist.
\end{proof}

\begin{lemma}[Extracting a hinge from a positive kink]
\label{lem:oakk_hinge}
Let $f:\mathbb R\to\mathbb R$ be convex and piecewise affine.  If its right
derivative minus its left derivative at $q$ is at least $\Delta_*>0$, then, for
every integrable scalar random variable $X$,
\begin{equation}\label{eq:oakk_hinge}
 \E f(X)-f(\E X)
 \ge \Delta_*\left[\E(X-q)_+-(\E X-q)_+\right]
 =\Delta_*\min\{\E(X-q)_+,\E(q-X)_+\}.
\end{equation}
In particular, the Jensen gap is strictly positive if the distribution of $X$
charges both sides of $q$.
\end{lemma}

\begin{proof}
The slopes of a convex piecewise-affine function are nondecreasing.  Removing
$\Delta_*(x-q)_+$ reduces the upward slope jump at $q$ by $\Delta_*$ and leaves
all other slope jumps unchanged, so
$\tilde f(x)=f(x)-\Delta_*(x-q)_+$ remains convex.  Jensen's inequality for
$\tilde f$ gives
the inequality in~\eqref{eq:oakk_hinge}.  For the equality, put
$a_+=\E(X-q)_+$ and $a_-=\E(q-X)_+$.  Since $\E X-q=a_+-a_-$, we have
$a_+-(a_+-a_-)_+=\min(a_+,a_-)$.
\end{proof}

\begin{lemma}[Uniform product-grid hinge loss]
\label{lem:oakk_grid_loss}
Suppose $S\ge2$, every $g_s\in C^1[0,1]$, and
$g_s'(u)<0$ on $[0,1)$.  Let $\mathcal C_m$ be the uniform product
partition of $[0,1]^S$ into cells of side length $1/m$.  For a cell $C$ and
$q\in(\lambda_-,\lambda_+)$, put
\[
 H^+_C(q)=\int_C\{\Lambda(u)-q\}_+\,\mathrm du,\qquad
 H^-_C(q)=\int_C\{q-\Lambda(u)\}_+\,\mathrm du.
\]
For every compact interval $Q\subset(\lambda_-,\lambda_+)$, there are
$c_Q>0$ and $m_Q<\infty$ such that
\begin{equation}\label{eq:oakk_gridlemma}
 \inf_{q\in Q}\sum_{C\in\mathcal C_m}
       \min\{H^+_C(q),H^-_C(q)\}
 \ge c_Qm^{-2},\qquad m\ge m_Q.
\end{equation}
\end{lemma}

That is, at every threshold in $Q$ the grid cells that straddle the level set
$\{\Lambda=q\}$ carry two-sided hinge mass of total order $m^{-2}$.

\begin{proof}
We treat $S=2$ first: each of the order-$m$ grid columns over a fixed
interval contains a cell crossed by the level curve, with both hinge
integrals of order $m^{-3}$.  Conditioning on the remaining sites then gives
$S>2$, and a finite cover of $Q$ gives uniformity.
For $S=2$ the level sets near any interior threshold are uniformly
well-behaved graphs, with slopes and normal derivatives bounded away from
zero and infinity.  Fix $q_0\in(\lambda_-,\lambda_+)$.  The continuous product map
$(u_1,u_2)\mapsto g_1(u_1)g_2(u_2)$ maps the open square onto the interior
of its range, so there is an interior point on the level set
$g_1(u_1)g_2(u_2)=q_0$.  Both partial derivatives are strictly negative
there.  The implicit-function theorem, followed by shrinking the
neighborhood, gives an open interval $J$ about $q_0$, a fixed horizontal
interval $I$, and graphs
\[
 u_2=\vartheta_q(u_1),\qquad q\in J,\quad u_1\in I,
\]
representing the level sets, together with constants
$0<b_0\le b_1<\infty$ and $d_0>0$ such that
\begin{equation}\label{eq:oakk_geometry_bounds}
 b_0\le-\vartheta_q'(u_1)\le b_1,\qquad
 |\partial_2\Lambda(u_1,u_2)|\ge d_0
\end{equation}
throughout a fixed compact rectangle containing these graphs.  All constants
are uniform over $q\in J$, after shrinking $J$.

Each grid column over $I$ contains a cell crossed by the level set, with
hinge integrals of order $m^{-3}$ on both sides, and summing over columns
gives the bound for $S=2$.  At least $c_Im$ complete grid columns lie over $I$
for all large $m$.  In one such column, a graph
satisfying~\eqref{eq:oakk_geometry_bounds} meets at most
$M=\lceil b_1\rceil+2$ grid cells.  One of those cells therefore contains a
graph segment with horizontal span
$\Delta_x\ge 1/(Mm)$ and vertical drop
\[
 b_0\Delta_x\le\Delta_y\le b_1\Delta_x.
\]
Let $(x_0,y_0)$ and $(x_1,y_1)$ be its left and right endpoints, so
$y_0-y_1=\Delta_y$, and set
$\Delta_0=\Delta_y/(4b_1)$.  Then $\Delta_0\le\Delta_x/4$.  Over
$[x_0,x_0+\Delta_0]$ the graph drops by at most $\Delta_y/4$, and over
$[x_1-\Delta_0,x_1]$ it rises backward from $y_1$ by at most $\Delta_y/4$.
Consequently the selected grid cell contains the two rectangles
\[
\begin{split}
 &[x_0,x_0+\Delta_0]\times[y_1,y_1+\Delta_y/2],\\
 &[x_1-\Delta_0,x_1]\times[y_1+\Delta_y/2,y_0],
\end{split}
\]
one strictly below and one strictly above the graph.  Each has area
$\Delta_0\Delta_y/2=\Delta_y^2/(8b_1)\ge c_1m^{-2}$, and every point in
either rectangle is at vertical distance at least $\Delta_y/4\ge c_2m^{-1}$
from the graph.  Integrating $\partial_2\Lambda$ vertically and
using~\eqref{eq:oakk_geometry_bounds} gives
$|\Lambda-q|\ge d_0c_2m^{-1}$ on the respective rectangles.  Hence the
selected cell contributes at least $c_3m^{-3}$ to each of $H^+_C(q)$ and
$H^-_C(q)$.
Selected cells from different columns are distinct, and summing over at
least $c_Im$ columns gives a lower bound $c_4m^{-2}$, uniformly
for $q\in J$.

For $S>2$ the two-dimensional construction applies with one set of
constants after conditioning on the remaining sites.  Choose an interior
point $u^0=(u_1^0,\ldots,u_S^0)$ with $\Lambda(u^0)=q_0$, write
$\tilde u=(u_3,\ldots,u_S)$, and put
\[
 g_{3:S}(\tilde u)=\prod_{s=3}^Sg_s(u_s),\qquad
 \tilde q(q,\tilde u)=q/g_{3:S}(\tilde u).
\]
There are a neighborhood $J$ of $q_0$ and a compact box $\mathcal N$ of
positive volume around $\tilde u^0$ such that $g_{3:S}(\tilde u)\ge c_0>0$
and every $\tilde q(q,\tilde u)$,
$(q,\tilde u)\in J\times\mathcal N$, lies in one compact interior
subinterval of the range
of $g_1g_2$.  The two-dimensional bound, made uniform over compact intervals
by the finite-cover step at the end of this proof, therefore holds with one
set of constants for all $(q,\tilde u)\in J\times\mathcal N$.

Disintegrating over the remaining sites gives the lower bound
$c_0|\mathcal N|c_4m^{-2}$ uniformly over $q\in J$.  For a full cell
$C=C_{12}\times C_{\tilde u}$, disintegration gives
\[
\begin{split}
 H^+_C(q)&=\int_{C_{\tilde u}}g_{3:S}(\tilde u)
  \int_{C_{12}}\{g_1(u_1)g_2(u_2)-\tilde q(q,\tilde u)\}_+
   \,\mathrm du_{12}\,\mathrm d\tilde u,\\
 H^-_C(q)&=\int_{C_{\tilde u}}g_{3:S}(\tilde u)
  \int_{C_{12}}\{\tilde q(q,\tilde u)-g_1(u_1)g_2(u_2)\}_+
   \,\mathrm du_{12}\,\mathrm d\tilde u.
\end{split}
\]
For nonnegative functions $f_1$ and $f_2$,
$\min(\int f_1,\int f_2)\ge\int\min(f_1,f_2)$.  Summing first over $C_{12}$,
then over $C_{\tilde u}$, and restricting the $\tilde u$-integral to
$\mathcal N$ yields
$c_0|\mathcal N|c_4m^{-2}$, uniformly over $q\in J$.

The compact interval $Q$ has a finite cover by neighborhoods $J$
of the preceding kind.  Taking the minimum of their positive constants and
the maximum of their finite starting resolutions
proves~\eqref{eq:oakk_gridlemma}.
\end{proof}

\paragraph{Two-sided dual bound}
For an arbitrary fixed report design $T$, let $\mathcal G$ be its report
sigma-field, put $\bar\Lambda_k=\E_0(\Lambda_k\mid\mathcal G)$, let $\mu^*$ minimize the
continuous dual, and let $\bar\mu$ minimize the report-distribution dual.
The argument in Steps~2 and~5 of
Section~\ref{supp:proof_graderate} proves the two-sided
bound~\eqref{eq:dual_pinch} of the main text, restated here as
\begin{equation}\label{eq:oakk_pinch}
 \E_0\{F_{\mu^*}(\Lambda)-F_{\mu^*}(\bar\Lambda)\}
 \le d_w(T)
 \le
 \E_0\{F_{\bar\mu}(\Lambda)-F_{\bar\mu}(\bar\Lambda)\}.
\end{equation}
Each expectation in~\eqref{eq:oakk_pinch} is nonnegative by iterated
coordinatewise conditional Jensen.

\paragraph{Part (a): every finite-range componentwise report is strictly lossy}
Fix an OAK witness $(\mu^*,a,\mathcal B,q,\Delta)$, take its multiplier as the
$\mu^*$ of~\eqref{eq:oakk_pinch}, and write $X=\Lambda_a$, $Z=\Lambda_{-a}$.  For every fixed $z\in \mathcal B$ outside the null exceptional set of Definition~\ref{def:oak}, Lemma~\ref{lem:oakk_product_cut} gives at least one
own-coordinate report cell $D$ (a report cell of hypothesis~$a$) for which the conditional distribution of $X$ charges both sides of $q(z)$.
There are only finitely many such cells.  The measurable sets
\[
 \mathcal B_D=\left\{z\in \mathcal B:
 \pr_0\{X<q(z)\mid D\}>0,\quad
 \pr_0\{X>q(z)\mid D\}>0\right\}
\]
cover $\mathcal B$ up to a null set, so one $\mathcal B_D$ has positive $\pr_0$-probability.
Measurability follows because the two conditional probabilities are the
conditional distribution function of $X$ evaluated at the measurable
threshold $q(z)$ and its complementary strict tail.

Some report profile for the other coordinates meets $\{Z\in \mathcal B_D\}$
with positive probability, and on it the conditional-Jensen gap is strict.  The
inverse-image events of the finitely many report profiles for $Z$
partition the event $\{Z\in \mathcal B_D\}$.  Choose one such profile, $D_{-a}$, for
which $\pr_0\{Z\in \mathcal B_D,D_{-a}\}>0$, equivalently
$\pr_0(Z\in \mathcal B_D\mid D_{-a})>0$.  Let
$\bar x=\E_0(X\mid D)$ and
$\bar z_k=\E_0(\Lambda_k\mid D_k)$, where $D_k$ is the cell of hypothesis~$k$
in the profile $D_{-a}$.  Conditional independence of the likelihood
ratios inside the complete profile gives
\begin{align}
 &\E_0\{F_{\mu^*}(\Lambda)\mid D,D_{-a}\}
       -F_{\mu^*}(\bar x,\bar z)\notag\\
 &\quad=
 \E_{Z\mid D_{-a}}\!
 \left[\E_{X\mid D}\{f_Z(X)\}-f_Z(\bar x)\right]\notag\\
 &\qquad+
 \left[
 \E_{Z\mid D_{-a}}\{F_{\mu^*}(\bar x,Z)\}
       -F_{\mu^*}(\bar x,\bar z)\right].
 \label{eq:oakk_profile}
\end{align}
The second bracket is nonnegative by iterated coordinatewise Jensen.  The
first integrand is nonnegative for every $z$, and
Lemma~\ref{lem:oakk_hinge} makes it strictly
positive on $\mathcal B_D$, because the conditional distribution given $D$ charges both sides
of the positive kink.  Thus~\eqref{eq:oakk_profile} is strictly positive.
All other complete profiles contribute nonnegatively, so the left side
of~\eqref{eq:oakk_pinch} is strictly positive and $d_w(T)>0$.

\paragraph{The threshold-report class}
The class deficit is positive because every fixed-design gap is positive and
the optimal report power is continuous on the compact threshold cube.  Work on the
compact threshold cube $[0,1]^S$, including degenerate faces, and formulate
the finite report problem on the fixed binary report space.  A randomized
rule is a point in a fixed compact product of simplices indexed by binary
report arrays and rejection sets.  At every configuration, each report-array
probability is a product of factors among
\[
 t_s,\quad 1-t_s,\quad G_s(t_s),\quad 1-G_s(t_s),
\]
where $G_s(t)=\int_0^tg_s(u)\,\mathrm du$, so each such probability is continuous in $t$.  Hence the objective $\Pi_w$ and every
family-wise error rate (FWER) constraint are continuous jointly in the rule and $t$.

The optimal report power $\Pi^\star(t)$ is continuous in the thresholds:
upper semicontinuity by compactness, lower semicontinuity by mixing with
the no-rejection rule.  To prove upper semicontinuity, for
any $t_n\to t$ first pass to a subsequence, still indexed by $n$, along which
$\Pi^\star(t_n)\to\limsup_j \Pi^\star(t_j)$, and choose an optimal rule $r_n$
at each
$t_n$. Compactness gives a further subsequence with $r_n\to r$; closedness
of the constraints makes $r$ feasible at $t$, and continuity gives
$\limsup_n\Pi^\star(t_n)\le \Pi^\star(t)$. Conversely, mix an optimal rule at $t$ with the
no-rejection rule, with respective weights $1-\eta$ and $\eta$.  Every FWER
constraint of the mixture is at most $(1-\eta)\alpha$, so the finite family
of constraints has common slack $\eta\alpha$.  The mixture remains feasible
throughout a neighborhood of $t$, and its power converges to
$(1-\eta)\Pi^\star(t)$.  Therefore
$\liminf_n\Pi^\star(t_n)\ge(1-\eta)\Pi^\star(t)$; letting $\eta\downarrow0$
proves continuity of $\Pi^\star$.

Every point of the closed cube is a finite-range componentwise design, so
part~(a) gives
$\pi^*_w(\alpha)-\Pi^\star(t)>0$ everywhere.  A continuous positive function on a
compact set has a positive minimum.  The infimum over the original threshold
class is no smaller, proving
$\delta_w(\mathcal T^{\mathrm{thr}})>0$.

\paragraph{Part (b), equal-width reports: the lower inverse-square bound}
We prove the lower bound first for equal-width reports, where
Lemma~\ref{lem:oakk_grid_loss} makes the boundary-layer mechanism explicit;
the classwise argument that follows reuses only~\eqref{eq:oakk_uniform_subset}
and implies this case.  Because $q(z)$ lies in the open interval
$(\lambda_-,\lambda_+)$ and $\Delta(z)>0$ almost everywhere on $\mathcal B$, a
countable exhaustion supplies a measurable $\mathcal B_0\subseteq \mathcal B$ of positive
$\pr_0$-probability, a compact interval
$Q\subset(\lambda_-,\lambda_+)$, and $\Delta_*>0$ such that
\begin{equation}\label{eq:oakk_uniform_subset}
 q(z)\in Q,\qquad \Delta(z)\ge\Delta_*,\qquad z\in \mathcal B_0.
\end{equation}
For example, take the first positive-probability member of the increasing sequence of sets on which $\Delta(z)$ and the distance of $q(z)$ from both endpoints
are at least $1/n$.

For $z\in\mathcal B_0$ and $m\ge m_Q$, the own-coordinate gap of the equal-width design is at
least $\Delta_*c_Qm^{-2}$.  For an own-coordinate grid cell $C$, write
$\bar\Lambda_C=|C|^{-1}\int_C\Lambda(u)\,\mathrm du$.
Lemma~\ref{lem:oakk_hinge} and~\eqref{eq:oakk_uniform_subset} give, for
$z\in \mathcal B_0$,
\begin{align}
 &\int_C f_z\{\Lambda(u)\}\,\mathrm du-|C|f_z(\bar\Lambda_C)\notag\\
 &\qquad\ge
 \Delta_*\min\left\{
   \int_C\{\Lambda(u)-q(z)\}_+\,\mathrm du,\,
   \int_C\{q(z)-\Lambda(u)\}_+\,\mathrm du
 \right\}.
 \label{eq:oakk_cell_lower}
\end{align}
Summing~\eqref{eq:oakk_cell_lower} over $C$ and applying
Lemma~\ref{lem:oakk_grid_loss} yields
\begin{equation}\label{eq:oakk_own_lower}
 \E_0\{f_z(\Lambda_a)-f_z(\bar\Lambda_a)\}
 \ge \Delta_* c_Qm^{-2},\qquad z\in \mathcal B_0,\quad m\ge m_Q,
\end{equation}
where the expectation is over the evidence for coordinate $a$.

The left side of~\eqref{eq:oakk_pinch} decomposes as
\begin{align}
 &\E_0F_{\mu^*}(\Lambda)-\E_0F_{\mu^*}(\bar\Lambda)\notag\\
 &=\E_0\!\left[
 F_{\mu^*}(\Lambda_a,\Lambda_{-a})-F_{\mu^*}(\bar\Lambda_a,\Lambda_{-a})\right]\notag\\
 &\quad+\E_0\!\left[
 F_{\mu^*}(\bar\Lambda_a,\Lambda_{-a})
       -F_{\mu^*}(\bar\Lambda_a,\bar\Lambda_{-a})\right].
 \label{eq:oakk_decomposition}
\end{align}
In the decomposition~\eqref{eq:oakk_decomposition}, the second term is
nonnegative by conditioning on all report cells and
applying coordinatewise Jensen successively to the conditionally independent
coordinates $\Lambda_k$, $k\ne a$.  By independence across hypotheses,
integrating~\eqref{eq:oakk_own_lower} over $z\in \mathcal B_0$ therefore gives
\[
 \E_0\{F_{\mu^*}(\Lambda)-F_{\mu^*}(\bar\Lambda)\}
 \ge \pr_0(\Lambda_{-a}\in \mathcal B_0)\,\Delta_* c_Qm^{-2},\qquad m\ge m_Q.
\]
The left inequality in~\eqref{eq:oakk_pinch} then gives
$\delta_w(\mathcal T^{(m)})\ge cm^{-2}$ for $m\ge m_Q$.  The upper bound is
Lemma~\ref{lem:uniform_upper}, so the rate is $\Theta(m^{-2})$.

\paragraph{The lower bound over all componentwise reports with at most $m$ values per site}
Fix an arbitrary $T\in\mathcal T_{\le m}$ and let $\mathcal G$ be the
sigma-field generated by its report array.  We strengthen the preceding grid argument by
using the arbitrary-cell argument of Step~6 of
Section~\ref{supp:proof_graderate}.  By a further countable refinement of the
set in~\eqref{eq:oakk_uniform_subset}, there exist a positive-probability set
$\mathcal B_*\subseteq \mathcal B_0$ and a sufficiently small compact interval $Q_*$ in the
interior likelihood-ratio range such that
\[
 q(z)\in Q_*,\qquad \Delta(z)\ge\Delta_*,\qquad z\in \mathcal B_*.
\]
To see that $Q_*$ can be taken as small as needed, choose a point $q_0$ in
the support of the push-forward of $\pr_0(\,\cdot\mid \Lambda_{-a}\in \mathcal B_0)$ by
$q$ and then a compact neighborhood $Q_*$ of $q_0$ whose inverse image has
positive probability; the neighborhoods used below remain valid after
shrinking $Q_*$.

The site-1 localization of Step~6 applies uniformly over the witness
thresholds $q\in Q_*$.  Choose $t_0\in(0,1)$ with $\prod_sg_s(t_0)=q_0$, put
$X=g_1(U_{1,a})$ and $Y=\prod_{s=2}^Sg_s(U_{s,a})$ (this $X$ redefines the $X=\Lambda_a$ of part~(a)); then select compact
neighborhoods $I$ of $g_1(t_0)$, $J_2$ of $t_0$, and $\mathcal N$ of
$(t_0,\ldots,t_0)$ for $\tilde U_a=(U_{3,a},\ldots,U_{S,a})$.  They may be
chosen so that, for every
$q\in Q_*$ and $\tilde u_a\in\mathcal N$, the change of variable
\[
 r=\frac{q}{g_{3:S}(\tilde u_a)\,g_2(u_2)}
\]
maps $J_2$ onto an interval containing $I$.  Uniformly over these compact sets,
\begin{equation}\label{eq:oakk_anym_jacobian}
 Y\left|\frac{\mathrm du_2}{\mathrm dr}\right|
 =\frac{g_{3:S}(\tilde u_a)^2g_2(u_2)^3}{q|g_2'(u_2)|}\ge c_6>0,
\end{equation}
and the null density of $X$ is bounded below by $f_->0$ on $I$.

Refining the report sigma-field can only lower the gap, so it suffices to
bound the refined gap obtained by revealing everything except, when $X\in I$,
the position of $X$ within site 1's cell for hypothesis~$a$.  Refine $\mathcal G$ to $\mathcal H$ by adjoining the
following variables:
all raw uniforms for hypotheses $k\ne a$; for hypothesis $a$,
$U_{2,a},\ldots,U_{S,a}$; and the exact value of $X$ off $I$, with only
$(\mathbf1\{X\in I\},T_1(U_{1,a}))$ retained on $I$.  Thus
$\mathcal G\subseteq\mathcal H$. Let
$\widetilde\Lambda=\E_0(\Lambda\mid\mathcal H)$ and
$\bar\Lambda=\E_0(\Lambda\mid\mathcal G)$.  Conditional on $\mathcal G$, the
coordinates of $\widetilde\Lambda$ remain independent: each $\widetilde\Lambda_k$ is a
function only of the raw evidence block for hypothesis $k$, while both the
original reports and the refinement act separately on those blocks. Iterated
coordinatewise
Jensen therefore gives
\[
 \E_0\{F_{\mu^*}(\widetilde\Lambda)\mid\mathcal G\}
 \ge F_{\mu^*}\{\E_0(\widetilde\Lambda\mid\mathcal G)\}
 =F_{\mu^*}(\bar\Lambda).
\]
Consequently the left side of the two-sided dual bound is at least
$\E_0\{F_{\mu^*}(\Lambda)-F_{\mu^*}(\widetilde\Lambda)\}$.

The refined gap is at least a constant times $m^{-2}$ for every
$T\in\mathcal T_{\le m}$, which with the equal-width design's upper bound gives
the classwise rate.  Let $A_j,p_j,\bar x_j$ and $\Delta_j(r)$ be the site-1
cell intersections and hinge regrets defined in Step~6, with at most $m$
non-null cells.  For
$z\in \mathcal B_*$ and fixed $U_{2,a},\ldots,U_{S,a}$, $\Lambda_a=XY$, so
$x\mapsto f_z(Yx)$ is convex and piecewise affine with an upward kink of size
at least $\Delta(z)Y\ge\Delta_* Y$ at $x=q(z)/Y$, the variable $r$ of the
change of variable above.  Lemma~\ref{lem:oakk_hinge}, applied to $X$ given
$X\in A_j$, therefore bounds the refined gap on that cell below by
$\Delta_*Y\Delta_j(r)$.  Restricting to
$z\in \mathcal B_*$, $\tilde u_a\in\mathcal N$, and $u_{2,a}\in J_2$, changing $u_{2,a}$ to $r$, and using~\eqref{eq:oakk_anym_jacobian} gives
\begin{align*}
 \E_0\{F_{\mu^*}(\Lambda)-F_{\mu^*}(\widetilde\Lambda)\}
 &\ge \Delta_* \pr_0(\Lambda_{-a}\in \mathcal B_*)|\mathcal N|c_6
       \sum_j\int_I\Delta_j(r)\,\mathrm dr\\
 &\ge \frac{\Delta_* \pr_0(\Lambda_{-a}\in \mathcal B_*)|\mathcal N|c_6
                  f_-|I|^3}{24m^2},
\end{align*}
where the second inequality combines~\eqref{eq:anym_hinge_variance} with the
rearrangement bound for arbitrary measurable sets stated after it.  The constant is independent of the
chosen report maps. Hence every $T\in\mathcal T_{\le m}$ has
$d_w(T)\ge cm^{-2}$, and taking the infimum gives
$\delta_w(\mathcal T_{\le m})\ge cm^{-2}$.  Since the equal-width design belongs
to $\mathcal T_{\le m}$ and Lemma~\ref{lem:uniform_upper} supplies its
$O(m^{-2})$ upper bound, the claimed classwise rate follows, and substituting $m=2^b$ gives the
stated bit form. \hfill$\square$

\paragraph{Why one transverse branch crossing implies OAK}
Fix a minimizing $\mu^*$, and suppose that at
$v^\circ\in(\lambda_-,\lambda_+)^K$ with distinct coordinates, two positive
ordered action branches of $F_{\mu^*}$ tie, strictly dominate every other branch, and have a
strictly positive difference between their slopes in coordinate $a$.
Distinctness fixes the likelihood-ratio ordering on a neighborhood.
On that neighborhood every action branch is affine in $v_a$, and the
difference of the two selected branches has the form
$\Delta(z)v_a+c(z)$, where $z=v_{-a}$, $z^\circ=v^\circ_{-a}$, and $\Delta(z^\circ)>0$.  Continuity gives $\Delta(z)>0$ nearby, and the two branches cross at
$q(z)=-c(z)/\Delta(z)$, a continuous function with
$q(z^\circ)=v_a^\circ$.  Strict positivity and strict dominance over the
remaining branches persist after shrinking the neighborhood.  The null distribution
of $\Lambda_{-a}$ charges every open subset of its product support, so a compact
positive-probability subset supplies $\mathcal B$; the slope jump
$\Delta(z)$ is positive there.  A single transverse crossing of two positive
branches is therefore a simple sufficient condition for OAK; the explicit
witness of Lemma~\ref{lem:k2_witness} below is an instance.

At $K=2$ under the point prior, an explicit witness can be read off the
minimizing multiplier.

\begin{lemma}[An explicit witness at $K=2$]\label{lem:k2_witness}
Suppose Assumptions~\ref{ass:model}--\ref{ass:reg} hold, $K=2$, $S\ge2$ and
$w(2)=1$, and write a minimizing continuous-dual multiplier as
$\mu^*=(\mu_1^*,\mu_2^*)$.  The crossing of the reject-one and reject-two
branches where the smaller likelihood ratio equals $2\mu_1^*$ (the vertical
boundary between \emph{reject $H_2$} and \emph{reject both} in
Figure~\ref{fig:oracle_kink} of the main text) supplies an explicit OAK
witness whenever
\[
\lambda_-<2\mu_1^*<\lambda_+,
 \qquad
 2\mu_1^*+\frac{\mu_2^*}{\mu_1^*}<\lambda_+.
\]
For $S$ homogeneous $\mathrm{Beta}(1,\theta)$ sites with $\theta>1$,
$\lambda_-=0$ and $\lambda_+=\theta^S$.
\end{lemma}

\paragraph{Proof of Lemma~\ref{lem:k2_witness}}
The two inequalities place the crossing of the two nonzero branches at an
interior threshold $q=2\mu_1^*$ with positive common value, which supplies
an OAK witness.  Let $\mu^*=(\mu_1^*,\mu_2^*)$ minimize the continuous dual at
the point prior.  Take $a=1$, and write $x=v_1$ and $z=v_2$ for the varied
and the fixed coordinate.  With $x<z$,
the two nonzero action branches are
\[
 \psi_1=(z/2-\mu_1^*)x-\mu_2^*,\qquad
 \psi_2=(z-\mu_1^*)x-\mu_1^*z-\mu_2^*.
\]
They cross at $q=2\mu_1^*$.  At the crossing their common value is
$\mu_1^*(z-2\mu_1^*)-\mu_2^*$, which is positive exactly when
\[
 z>z_{\min}:=2\mu_1^*+\frac{\mu_2^*}{\mu_1^*}.
\]
The slope of $\psi_2$ in $x$ exceeds that of $\psi_1$ by $z/2>0$.
Dual feasibility gives $\mu_2^*\ge0$, so
$z_{\min}\ge q>\lambda_-$; the second displayed inequality in
Lemma~\ref{lem:k2_witness} gives $z_{\min}<\lambda_+$. Choose a compact interval
$\mathcal B\subset(z_{\min},\lambda_+)$.  The null distribution of $\Lambda_2$ charges every open
subinterval of $(\lambda_-,\lambda_+)$, so $\pr_0(\Lambda_2\in \mathcal B)>0$.  For every
$z\in \mathcal B$, the crossing at $q$ is positive and separated from the rank tie
$x=z$; immediately below it the unique active branch rejects only the
hypothesis with likelihood ratio $z$, and immediately above it the unique
active branch rejects both hypotheses.  Thus OAK holds with
$\Delta(z)=z/2$.  For homogeneous $\mathrm{Beta}(1,\theta)$ sites,
$g_s(1)=0$ and $g_s(0)=\theta$, giving
$(\lambda_-,\lambda_+)=(0,\theta^S)$. \hfill$\square$

\subsection{The Beta-family active-kink certificate: proof of
Corollary~\ref{cor:beta_oak}}\label{supp:proof_beta_oak}

Lemma~\ref{lem:oak_auto} already gives OAK whenever $w(K)>0$, so no result
of the main text depends on this subsection, on Section~\ref{supp:proof_enclosure},
or on Section~\ref{supp:prior_numerics}.  These three sections exhibit explicit
witnesses: for the displayed Beta models they locate a kink and a
positive-probability set on which it is active, without trusting a numerical
optimizer.  This subsection shows that the two $K=2$ inequalities of
Lemma~\ref{lem:k2_witness} hold analytically on an open region of homogeneous Beta models and, with
computer-assisted exact rational arithmetic, that this region contains
$\mathrm{Beta}(1,2)$ at every level in $[0.03512,0.05384]$.

\paragraph{Primitive definitions for the Beta-family certificate}
For two homogeneous $\mathrm{Beta}(1,\theta)$ sites, the aggregate likelihood
ratio's null law has tail probability $p_\theta$, tail first moment
$M_\theta$, and density $f_\theta$:
\[
 r_\theta(c)=\left(\frac{c}{\theta^2}\right)^{1/(\theta-1)},\qquad
 p_\theta(c)=1-r_\theta(c)+r_\theta(c)\log r_\theta(c),
\]
and
\[
 M_\theta(c)=1-r_\theta(c)^\theta
                 \{1-\theta\log r_\theta(c)\},\qquad
 f_\theta(c)=\frac{-r_\theta(c)\log r_\theta(c)}{(\theta-1)c},
\]
for $0<c<\theta^2$, with $p_\theta(c)=M_\theta(c)=1$ for $c\le0$ and
$p_\theta(c)=M_\theta(c)=0$ for $c\ge\theta^2$. For $\theta^2>53/15$, let
\[
 a_\theta=\frac{8}{5(\theta^2-2)},\qquad
 r_0=1+\sqrt{9/5},
\]
$z_1(x)=2+8/(5x)$ and $z_2(x)=(x+4/5)/(x-1)$, and define
\begin{align*}
 \bar\beta_2(\theta)={}&2\int_{a_\theta}^{2}f_\theta(x)p_\theta\{z_1(x)\}\,\mathrm dx
 +2\int_2^{r_0}f_\theta(x)p_\theta\{z_2(x)\}\,\mathrm dx
 +p_\theta(r_0)^2,\\
 \bar\beta_1(\theta)={}&2\int_{a_\theta}^{2}x f_\theta(x)p_\theta\{z_1(x)\}\,\mathrm dx\\
 &+2\int_2^{r_0}f_\theta(x)
   [x p_\theta\{z_2(x)\}+M_\theta\{z_2(x)\}]\,\mathrm dx
 +2p_\theta(r_0)M_\theta(r_0).
\end{align*}
These are one-dimensional null expectations, not optimizer outputs.  Here
$z_1$ and $z_2$ bound the regions where actions one and two are active at the
corner $\mu^{\mathrm L}=(1,4/5)$ of the rectangle used in Step~2, $a_\theta$
is where $z_1$ reaches $\theta^2$, $r_0$ is where $z_2$ meets the diagonal,
$z_2(r_0)=r_0$, and $\bar\beta_1(\theta)$ and $\bar\beta_2(\theta)$ are the null
expectations of the family-wise error coefficients there
(see~\eqref{eq:betatheta_left_expectations}).  The condition
$\theta^2>53/15$ makes the second inequality of Lemma~\ref{lem:k2_witness}
hold throughout that rectangle.  Define
the primitive parameter region
\begin{equation}\label{eq:beta_primitive_region}
 \mathcal U=\left\{(\alpha,\theta):
 \begin{array}{l}
 \theta^2>53/15,\\[-2pt]
 \max[p_\theta(3),\bar\beta_2(\theta)]<\alpha\\[-2pt]
 \hfill{}<\min[2p_\theta(3)-p_\theta(3)^2,\bar\beta_1(\theta)/2]
 \end{array}\right\}.
\end{equation}

\begin{corollary}[An optimizer-free witness for Beta alternatives at $K=2$]\label{cor:beta_oak}
Suppose Assumptions~\ref{ass:model}--\ref{ass:reg} hold with $K=S=2$, $w(2)=1$, and
$g_1=g_2=g_\theta$, where $g_\theta(u)=\theta(1-u)^{\theta-1}$ and
$\theta>1$, and let
$\mathcal U$ be the explicit primitive region~\eqref{eq:beta_primitive_region}.
For every $(\alpha,\theta)\in\mathcal U$, the two $K=2$ inequalities of
Lemma~\ref{lem:k2_witness} hold at every minimizing continuous-dual
multiplier, so the crossing at $2\mu_1^*$ is an OAK witness.  Consequently every fixed finite-range
componentwise report design is strictly lossy and the
threshold-report class has a positive deficit. The set $\mathcal U$ is open
and contains
\[
 [0.03512,0.05384]\times\{2\}.
\]
If $(\alpha,\theta)\in\mathcal U$ also has $\theta\ge2$, then
\[
 \delta_2(\mathcal T_{\le m})=\Theta(m^{-2}),\qquad
 \delta_2(\mathcal T^{(m)})=\Theta(m^{-2}).
\]
For $(\alpha,\theta)\in\mathcal U$ with $1<\theta<2$, the strict-loss
conclusions remain valid, but the displayed rate is not asserted because
$g_\theta'$ is unbounded at one endpoint.
\end{corollary}

The loss and rate conclusions repeat Lemma~\ref{lem:oak_auto} and
Theorem~\ref{thm:oracle_kink}, which already apply because $w(2)=1$.  What the
corollary adds is an explicit witness, the kink where the smaller likelihood
ratio equals $2\mu_1^*$, located without a numerical optimizer: Step~2 confines
every dual minimizer to a fixed rectangle on which both inequalities of
Lemma~\ref{lem:k2_witness} hold.  Corollary~\ref{cor:enclosure_beta} reaches
$\theta=3$ at $\alpha=0.05$ by enclosing the minimizers instead.

\paragraph{Step 1 (tail formulas and dual equations)}
For two homogeneous $\mathrm{Beta}(1,\theta)$ sites under the global null,
\[
 \Lambda_\theta=\theta^2(U_1U_2)^{\theta-1},\qquad
 U_1,U_2\stackrel{\mathrm{iid}}{\sim}\mathrm{Unif}(0,1).
\]
Direct integration of the product-uniform density $-\log v$ identifies the
primitive functions defined above: for $0<c<\theta^2$,
\begin{equation}\label{eq:betatheta_tail_moments}
 p_\theta(c)=\pr_0(\Lambda_\theta\ge c),\qquad
 M_\theta(c)=\E_0\{\Lambda_\theta\mathbf1(\Lambda_\theta\ge c)\},
\end{equation}
and $f_\theta$ is the null density of $\Lambda_\theta$.

Let
$\mathcal D_{\alpha,\theta}(\mu_1,\mu_2)$ be the continuous $K=2$ dual
objective in Section~\ref{app:dual}, and order the two hypotheses' independent
likelihood ratios as $x=\min(\Lambda_1,\Lambda_2)\le z=\max(\Lambda_1,\Lambda_2)$. The
nonzero branches of its integrand are
\[
 \psi_1=\left(\frac z2-\mu_1\right)x-\mu_2,
 \qquad
 \psi_2=(z-\mu_1)x-\mu_1z-\mu_2.
\]
Except on branch boundaries, which have null probability zero, let $l_\mu(x,z)\in\{0,1,2\}$ denote
the maximizing action, and put
\[
 \beta_1^\mu=x\mathbf1\{l_\mu=1\}
 +(x+z)\mathbf1\{l_\mu=2\},
 \qquad \beta_2^\mu=\mathbf1\{l_\mu>0\},
\]
that is, $\beta_j^\mu=B_{j,l_\mu}$, the family-wise error coefficients of the
maximizing action ($B_{1,1}=x$, $B_{1,2}=x+z$, $B_{2,l}=1$ for $l\ge1$, and
$B_{j,0}=0$).
With $\nabla_j(\mu)=\partial_j\mathcal D_{\alpha,\theta}(\mu)$,
dominated differentiation gives
\begin{equation}\label{eq:betatheta_dual_gradient}
 \nabla_1(\mu)=2\alpha-\E_0\beta_1^\mu,
 \qquad \nabla_2(\mu)=\alpha-\E_0\beta_2^\mu.
\end{equation}
For fixed $\mu_1$, $\nabla_1$ is nondecreasing in $\mu_2$, because increasing
$\mu_2$ lowers both nonzero branches equally and can only replace a nonzero
action by zero. For fixed $\mu_2$, $\nabla_2$ is nondecreasing in $\mu_1$,
because both nonzero branches decrease and the event of a nonzero action
shrinks.

\paragraph{Step 2 (inward dual signs and an interior minimizer)}
We trap every dual minimizer inside the fixed rectangle
\[
 \mathcal Q=[1,3/2]\times[0,4/5],
\]
on which both inequalities of Lemma~\ref{lem:k2_witness} hold, by showing that
$\mathcal D_{\alpha,\theta}$ strictly decreases into $\mathcal Q$ across each
face; by the monotonicities of Step~1, two derivative signs at each of the
corners $\mu^{\mathrm R}=(3/2,0)$ and $\mu^{\mathrm L}=(1,4/5)$ suffice.
At $\mu^{\mathrm R}=(3/2,0)$, whenever $\theta^2>3$, action enumeration gives
\begin{equation}\label{eq:betatheta_right_expectations}
 \E_0\beta_1^{\mu^{\mathrm R}}=2p_\theta(3),\qquad
 \E_0\beta_2^{\mu^{\mathrm R}}=2p_\theta(3)-p_\theta(3)^2.
\end{equation}
At $\mu^{\mathrm L}=(1,4/5)$, action one is active exactly when
$a_\theta<x<2$ and $z>z_1(x)$; action two is active exactly when either
$2<x<r_0$ and $z>z_2(x)$, or $r_0<x\le z<\theta^2$. Direct integration over these regions gives
\begin{equation}\label{eq:betatheta_left_expectations}
 \E_0\beta_1^{\mu^{\mathrm L}}=\bar\beta_1(\theta),\qquad
 \E_0\beta_2^{\mu^{\mathrm L}}=\bar\beta_2(\theta),
\end{equation}
with $\bar\beta_1$ and $\bar\beta_2$ as defined above.

If $(\alpha,\theta)\in\mathcal U$,
equations~\eqref{eq:betatheta_dual_gradient},
\eqref{eq:betatheta_right_expectations},
and~\eqref{eq:betatheta_left_expectations} give
\[
 \nabla_1(\mu^{\mathrm L})<0,\qquad \nabla_1(\mu^{\mathrm R})>0,\qquad
 \nabla_2(\mu^{\mathrm R})<0,\qquad \nabla_2(\mu^{\mathrm L})>0.
\]
The monotonicities of Step~1 extend these to strict inward signs on all four faces of $\mathcal Q$. The convex continuous dual attains its minimum on
$\mathcal Q$; the face signs exclude the boundary, so a restricted minimizer
lies in the interior. It is a local, hence global, minimizer on
$\mathbb R_+^2$. Any other global minimizer would join it by a segment of minimizers, and the
strict face signs prevent that segment from leaving the interior of
$\mathcal Q$; hence every minimizing multiplier lies in the interior of
$\mathcal Q$. Every interior point of $\mathcal Q$ satisfies
\[
 0<2\mu_1<3<\theta^2,
 \qquad
 2\mu_1+\frac{\mu_2}{\mu_1}
 <3+\frac{4/5}{3/2}=\frac{53}{15}<\theta^2.
\]
Here $2\mu_1+(4/5)/\mu_1$ is increasing on $[1,3/2]$.
Lemma~\ref{lem:k2_witness} therefore establishes OAK
throughout $\mathcal U$, without locating the minimizing multiplier
numerically.

\paragraph{Step 3 (openness of the primitive parameter region)}
Throughout $\{\theta:\theta^2>53/15\}$, the functions
in~\eqref{eq:betatheta_tail_moments}, $a_\theta$, and the integrals $\bar\beta_1,\bar\beta_2$
are continuous. Indeed, locally around any such $\theta_0$, the moving lower
limit $a_\theta$ is bounded away from zero and all displayed integrands admit
an integrable bound, so dominated convergence applies; alternatively, one may
couple all $\Lambda_\theta$ through the same uniforms and use the fact that the
fixed-corner branch boundaries have null probability zero. Hence $\mathcal U$ is
open. We next certify, using exact rational arithmetic at $\theta=2$, that the claimed line segment lies inside $\mathcal U$.

\paragraph{Step 4a (right-corner bounds at $\theta=2$)}
Under the global null, each site likelihood ratio
$2(1-U_s)$ is uniform on $(0,2)$, so the per-hypothesis likelihood ratio
$\Lambda=4(1-U_1)(1-U_2)$ has support $(0,4)$. For $0<c<4$, direct integration
gives
\begin{equation}\label{eq:beta12_tail_moments}
\begin{aligned}
 p(c)&:=\pr_0(\Lambda\ge c)
 =1-\frac c4\left\{1+\log\frac4c\right\},\\
 M(c)&:=\E_0\{\Lambda\mathbf 1(\Lambda\ge c)\}
 =1-\left(\frac c4\right)^2\left\{1+2\log\frac4c\right\}.
\end{aligned}
\end{equation}
In particular, $\Lambda$ has a continuous density on $(0,4)$ and $\E_0\Lambda=1$.
First consider $(\mu_1,\mu_2)=(3/2,0)$ and put $p=p(3)$. If exactly one of
$\Lambda_1,\Lambda_2$ exceeds $3$, action one is active and its coefficient in
$\beta_1^\mu$ is the smaller likelihood ratio; if both exceed $3$, action two is
active and its coefficient is their sum. Consequently,
\[
 \E_0\beta_1^{(3/2,0)}
 =2p\{1-M(3)\}+2pM(3)=2p,
 \qquad
 \E_0\beta_2^{(3/2,0)}=2p-p^2.
\]
Equation~\eqref{eq:beta12_tail_moments}, with the logarithm bounds of Step~4c, gives the enclosure $0.0342<p(3)<0.0343$.

\paragraph{Step 4b (left-corner finite-partition bounds)}
It remains to bound $\E_0\beta_1^{\mu^{\mathrm L}}$ from below and
$\E_0\beta_2^{\mu^{\mathrm L}}$ from above. Since
$\psi_2-\psi_1=z(x/2-1)$, the action regions of Step~2 hold with
$a_2=4/5$ and $\theta^2=4$. The elementary bounds $2.34<r_0<2.35$ follow by squaring.

We bound both expectations using the following finite partition.
Let
\[
\begin{aligned}
 \mathcal I_1={}&\{(1,1.2),(1.2,1.4),(1.4,1.6),
                   (1.6,1.8),(1.8,2)\},\\
 \mathcal I_2={}&\{(2,2.1),(2.1,2.2),(2.2,2.3),(2.3,2.34)\},\\
 \mathcal J_1={}&\{(0.8,1)\}\cup\mathcal I_1.
\end{aligned}
\]
For $(a,b)\in\mathcal I_1$, the event
$\{x\in[a,b],z\ge z_1(a)\}$ is contained in the action-one region and has
ordered-pair probability $2\{p(a)-p(b)\}p\{z_1(a)\}$; on it
$\beta_1^\mu\ge a$. The analogous event with $z_2(a)$ for
$(a,b)\in\mathcal I_2$ is contained in the action-two region and has
$\beta_1^\mu\ge a+z_2(a)$. Finally, the event that both likelihood ratios
exceed $2.35$ is contained in the action-two region. These disjoint subevents give
\begin{align*}
 \E_0\beta_1^{(1,4/5)}\ge{}&
 \sum_{(a,b)\in\mathcal I_1}
  2a\{p(a)-p(b)\}p\{z_1(a)\}\\
 &+\sum_{(a,b)\in\mathcal I_2}
  2\{a+z_2(a)\}\{p(a)-p(b)\}p\{z_2(a)\}\\
 &+2p(2.35)M(2.35)>0.01648+0.03366+0.05754=0.10768.
\end{align*}
Conversely, on $x\in[a,b]$ the action-one event is contained in
$\{z>z_1(b)\}$, and the action-two event is contained in
$\{z>z_2(b)\}$. The event $\{x\ge2.34\}$ covers the omitted upper tail.
Thus
\begin{align*}
 \E_0\beta_2^{(1,4/5)}\le{}&
 \sum_{(a,b)\in\mathcal J_1}
  2\{p(a)-p(b)\}p\{z_1(b)\}\\
 &+\sum_{(a,b)\in\mathcal I_2}
  2\{p(a)-p(b)\}p\{z_2(b)\}
 +p(2.34)^2\\
 &<0.01625+0.00859+0.01028=0.03512.
\end{align*}
Thus $\bar\beta_1(2)>0.10768$ and $\bar\beta_2(2)<0.03512$.

\paragraph{Step 4c (rational logarithm enclosures)}
All decimal bounds above are rigorous rational enclosures, not floating-point
quadrature. Indeed, every argument of the logarithm
in~\eqref{eq:beta12_tail_moments} is rational, and for rational $q\ge1$, with
$s=(q-1)/(q+1)$,
\[
 2\sum_{n=0}^{N}\frac{s^{2n+1}}{2n+1}
 \le\log q\le
 2\sum_{n=0}^{N}\frac{s^{2n+1}}{2n+1}
 +\frac{2s^{2N+3}}{(2N+3)(1-s^2)}.
\]
Taking $N=24$ and substituting these bounds into~\eqref{eq:beta12_tail_moments}
and the finite sums of Steps~4a and~4b gives the stated outward-rounded
rational bounds. The accompanying exact-rational script
\texttt{oak\_\allowbreak beta12\_\allowbreak certificate.py} carries out this
arithmetic, so Step~4, like the proof of Corollary~\ref{cor:enclosure_beta}, is
computer-assisted.

\paragraph{Step 5 (segment inclusion and the rate conclusion)}
Now take $0.03512\le\alpha\le0.05384$. At $\theta=2$ the first condition
defining $\mathcal U$, $\theta^2>53/15$, holds because $4>53/15$. The two
lower inequalities defining
$\mathcal U$ follow from
\[
 p_2(3)=p(3)<0.0343<0.03512\le\alpha,
 \qquad \bar\beta_2(2)<0.03512\le\alpha.
\]
The two upper inequalities follow strictly from
\[
 \alpha\le0.05384<2(0.0342)-(0.0342)^2
 <2p_2(3)-p_2(3)^2
\]
and
$\alpha\le0.05384<\bar\beta_1(2)/2$; the final inequality is strict because the
bound on $\bar\beta_1(2)$ is strict. Thus
\[
 [0.03512,0.05384]\times\{2\}\subset\mathcal U.
\]
Since $\mathcal U$ is open and this segment is compact, $\mathcal U$
contains a genuine two-dimensional neighborhood of it. The strict-loss
conclusions follow from part~(a) of Theorem~\ref{thm:oracle_kink}. If
$\theta\ge2$, then $g_\theta\in C^1[0,1]$ and
$g_\theta'(u)<0$ on $[0,1)$, so part~(b) gives the $\Theta(m^{-2})$ rates for
the class and for equal-width reports. For $1<\theta<2$, the derivative diverges at
$u=1$, which is why no rate is asserted there. \hfill$\square$

\subsection{Enclosure of the oracle dual minimizers:
Proposition~\ref{prop:enclosure} and
Corollary~\ref{cor:enclosure_beta}}\label{supp:proof_enclosure}

The witness inequalities of Lemma~\ref{lem:k2_witness} are conditions on a dual
minimizer, which a solver locates only approximately.  Wherever the one-sided
derivatives of the dual can be bounded rigorously, the inequalities can
instead be checked over a computable region guaranteed to contain
\emph{every} dual minimizer.  Proposition~\ref{prop:enclosure} supplies such
regions.  Corollary~\ref{cor:enclosure_beta} certifies explicit boxes for the
two Beta models of Table~\ref{tab:k2} by checking inward derivative signs on
their faces, which uses part~(b) and convexity rather than the iteration of
part~(c), and so adds $\theta=3$ to Corollary~\ref{cor:beta_oak};
Section~\ref{supp:prior_numerics} applies the same calculation at the flat
prior.  As noted at the start of Section~\ref{supp:proof_beta_oak}, no
main-text result depends on these certificates.

\begin{proposition}[Enclosure of the oracle dual minimizers]\label{prop:enclosure}
Let $\mathcal D$ be the continuous dual objective of Section~\ref{app:dual} and put $U=1/\alpha$.  For each coordinate $j$ let
$\underline\mu_j(\mu_{-j})$ and $\overline\mu_j(\mu_{-j})$ be the left and
right endpoints of
the minimizer set of $\mu_j\mapsto \mathcal D(\mu_j;\mu_{-j})$ on $[0,U]$,
the set being an interval because $\mathcal D$ is convex; here $\mu_{-j}$
denotes the coordinates of $\mu$ other than $j$.  Ties between maximizing
actions are permitted: the endpoint characterization uses the one-sided
derivatives of $\mathcal D$, or equivalently the error rates of the smallest
and largest maximizing actions, as specified below.  Then
\begin{itemize}
\item[(a)] every minimizer of $\mathcal D$ lies in $[0,U]^K$;
\item[(b)] $\underline\mu$ and $\overline\mu$ are antitone, that is,
$\mu_{-j}\le\mu_{-j}'$
implies $\underline\mu_j(\mu_{-j}')\le\underline\mu_j(\mu_{-j})$ and
$\overline\mu_j(\mu_{-j}')\le\overline\mu_j(\mu_{-j})$;
\item[(c)] with $\underline\mu^{(0)}=0$ and $\overline\mu^{(0)}=U\mathbf 1$,
$\underline\mu^{(n+1)}=\underline\mu(\overline\mu^{(n)})$ and
$\overline\mu^{(n+1)}=\overline\mu(\underline\mu^{(n)})$, the iterates satisfy
$\underline\mu^{(n)}\le\underline\mu^{(n+1)}
\le\overline\mu^{(n+1)}\le\overline\mu^{(n)}$ and
\[
 \underline\mu^{(n)}\;\le\;\mu^*\;\le\;\overline\mu^{(n)}\qquad
 \text{for every $n$ and every minimizer $\mu^*$ of $\mathcal D$.}
\]
\end{itemize}
Consequently a condition holding at every
$\mu\in[\underline\mu^{(n)},\overline\mu^{(n)}]$ holds at every
minimizing multiplier, with no appeal to the global optimality of a
numerically located candidate.
\end{proposition}

\begin{corollary}[The explicit witness at every dual minimizer]
\label{cor:enclosure_beta}
Let $K=S=2$, $w(2)=1$, $\alpha=0.05$ and $g_1=g_2=g_\theta$ with
$\theta\in\{2,3\}$.  Every minimizer $\mu^*$ lies in the interior of the
rational box for that $\theta$:
\[
\begin{aligned}
\theta=2:&\quad \mu_1^*\in[1.253,\,1.258],\quad \mu_2^*\in[0.222,\,0.229],\\
\theta=3:&\quad \mu_1^*\in[2.116,\,2.122],\quad \mu_2^*\in[0.150,\,0.157].
\end{aligned}
\]
The $K=2$ witness inequalities of Lemma~\ref{lem:k2_witness} hold at
every point of both boxes, with margins against $\lambda_+=\theta^2$ of at
least $1.30$ at $\theta=2$ and $4.68$ at $\theta=3$.  Hence, for both models, the witness inequalities hold at every dual minimizer, and with them the strict-loss conclusion of
Theorem~\ref{thm:oracle_kink}(a) and, since $\theta\ge2$, the rate conclusion
of Theorem~\ref{thm:oracle_kink}(b), without conditioning on a numerically
located minimizer.
\end{corollary}

The proof of Proposition~\ref{prop:enclosure} uses coercivity and convexity of the dual and the fact that its
cumulative error coefficients are nondecreasing in the number of rejections.
Throughout,
$\mathcal D(\mu)=\sum_{j=1}^K\binom Kj\alpha\mu_j+\E_0\{F_\mu(\Lambda)\}$
is the continuous oracle dual objective of Section~\ref{app:dual},
$F_\mu=\max\{0,\psi_1^\mu,\dots,\psi_K^\mu\}$, and
$\psi_l^\mu=A_l-\sum_j\mu_jB_{j,l}$, with $A_0=B_{j,0}=0$.
Let $l^-(\mu,\Lambda)$ and $l^+(\mu,\Lambda)$ be the smallest and largest
maximizers over $l\in\{0,\ldots,K\}$.  Write $\fwer_j^-(\mu)$ and
$\fwer_j^+(\mu)$ for their family-wise error rates at a configuration with
$j$ true hypotheses, using the same independent tie ordering of likelihood
ratios for both policies.  Symmetry makes these rates the same at every
such configuration, and $\fwer_j^-\le\fwer_j^+$.

\emph{Part (a): the global box.}  Lemma~\ref{lem:ordered_duality} gives
$\|\mu^*\|_1\le1/\alpha=U$ for every minimizer, so $\mu^*\in[0,U]^K$.

\emph{Convexity and the coordinate minimizer set.}  $\mathcal D$ is a sum of a
linear
form and the expectation of a pointwise maximum of affine functions of $\mu$,
hence convex; the expectation is finite on $[0,U]^K$ because
$0\le F_\mu\le A_K$ and $\E_0A_K=1$ at every prior.  For fixed $\mu_{-j}$ the section
$\mu_j\mapsto \mathcal D(\mu_j;\mu_{-j})$ is convex on $[0,U]$, so its
minimizer set is
a nonempty compact interval; write $\underline\mu_j(\mu_{-j})$ and
$\overline\mu_j(\mu_{-j})$
for its endpoints.  The right derivative of a maximum of affine branches
selects the largest active slope, here $-B_{j,l^-}$, and the left derivative
selects the smallest, $-B_{j,l^+}$.  Since
$0\le B_{j,l}\le e_{K-j}(\Lambda)$ is integrable, dominated convergence gives
\begin{equation}\label{eq:enclosure_derivatives}
 \partial_{j,+}\mathcal D(\mu)
 =\binom Kj\{\alpha-\fwer_j^-(\mu)\},\qquad
 \partial_{j,-}\mathcal D(\mu)
 =\binom Kj\{\alpha-\fwer_j^+(\mu)\}.
\end{equation}
By~\eqref{eq:enclosure_derivatives}, an interior point minimizes the section exactly when
$\fwer_j^-(\mu)\le\alpha\le\fwer_j^+(\mu)$; at the two endpoints of
$[0,U]$ only the corresponding inward derivative is required.  More explicitly,
\begin{equation}\label{eq:enclosure_endpoints}
\begin{aligned}
 \underline\mu_j(\mu_{-j})
 &=\min\bigl(\{x\in[0,U]:\fwer_j^-(x;\mu_{-j})\le\alpha\}\cup\{U\}\bigr),\\
 \overline\mu_j(\mu_{-j})
 &=\max\bigl(\{x\in[0,U]:\fwer_j^+(x;\mu_{-j})\ge\alpha\}\cup\{0\}\bigr).
\end{aligned}
\end{equation}
These extrema exist: the right derivative of a finite convex function is
right-continuous, the left derivative is left-continuous, and both are
nondecreasing along the section.  In particular,
$\underline\mu_j=0$ exactly when $\fwer_j^-(0;\mu_{-j})\le\alpha$.
Formula~\eqref{eq:enclosure_endpoints} covers jumps as well as flat
minimizer intervals, without imposing a zero-probability tie condition.

\emph{Part (b): antitonicity.}  Fix $\mu\le\mu'$ componentwise.  The objective term $A_l$ cancels from the
difference
$\psi^\mu_l-\psi^{\mu'}_l=\sum_j(\mu_j'-\mu_j)B_{j,l}(\Lambda)$, which is
nondecreasing in $l$ because each $B_{j,l}$ has nonnegative increments,
the error contributions of successive rejection positions.
Hence $\psi^{\mu}_l$ has decreasing differences in
$(l,\mu)$ and both extremal maximizers are nonincreasing in $\mu$:
$l^\pm(\mu',\Lambda)\le l^\pm(\mu,\Lambda)$ pointwise.  To see this
directly, if an extremal maximizer at $\mu'$ exceeded the corresponding extremal maximizer at $\mu$, optimality at the two multipliers and the
nondecreasing branch difference would force both actions to tie at both
multipliers, contradicting extremality.  Using a common likelihood-ratio
tie ordering therefore gives nested rejection sets for both policies $l^\pm$, and
$\fwer_j^\pm(\mu')\le\fwer_j^\pm(\mu)$ for every $j$.
Consequently, raising $\mu_{-j}$ expands the sublevel set of $\fwer_j^-$
and contracts the superlevel set of $\fwer_j^+$ in
\eqref{eq:enclosure_endpoints}.  Both endpoints move left, proving that
$\underline\mu_j$ and $\overline\mu_j$ are antitone in $\mu_{-j}$.

\emph{Part (c): the sandwich.}  Let $\mu^*$ be any minimizer of
$\mathcal D$.  Since
$\mu^*_j$ minimizes the $j$th section given $\mu^*_{-j}$, it lies in that
section's minimizer interval, so
\begin{equation}\label{eq:enclosure_fixedpoint}
 \underline\mu_j(\mu^*_{-j})\;\le\;\mu^*_j\;\le\;\overline\mu_j(\mu^*_{-j}),
 \qquad j=1,\dots,K.
\end{equation}
Argue by induction that
$\underline\mu^{(n)}\le\mu^*\le\overline\mu^{(n)}$.  The case $n=0$ is
part~(a).  Assuming it at $n$, antitonicity
and~\eqref{eq:enclosure_fixedpoint} give, coordinatewise,
\[
 \underline\mu^{(n+1)}_j=\underline\mu_j(\overline\mu^{(n)}_{-j})
 \le\underline\mu_j(\mu^*_{-j})\le\mu^*_j
 \le\overline\mu_j(\mu^*_{-j})\le\overline\mu_j(\underline\mu^{(n)}_{-j})
 =\overline\mu^{(n+1)}_j .
\]
The two sequences are monotone by a second induction:
$\underline\mu^{(n)}\le\underline\mu^{(n+1)}$ and
$\overline\mu^{(n+1)}\le\overline\mu^{(n)}$ hold at $n=0$ because
$\underline\mu^{(0)}=0$ and $\overline\mu^{(0)}=U\mathbf 1$ are the least and
greatest points of the box, and antitonicity carries both from $n$ to $n+1$.  The two endpoint maps cannot be
interchanged:
\[
 \underline\mu_j(\underline\mu^{(n)}_{-j})\ge\underline\mu_j(\mu^*_{-j})
 \quad\text{does not imply}\quad
 \underline\mu_j(\underline\mu^{(n)}_{-j})\ge\mu^*_j
\]
unless the coordinate minimizer is unique, which we do not assume.  Finally, $\mathcal D$ is convex, so every stationary
point is a global minimizer and the enclosure contains the entire optimal set.
Any condition valid at every point of
$[\underline\mu^{(n)},\overline\mu^{(n)}]$
therefore holds at every minimizing multiplier. \hfill$\square$

\medskip
\noindent\emph{Proof of Corollary~\ref{cor:enclosure_beta}.}
We certify each box directly by inward derivative signs.  As in Step~2 of
Section~\ref{supp:proof_beta_oak}, we show that $\mathcal D$ strictly decreases
into the box across all four faces, so convexity puts every minimizer in its
interior.  By the monotonicity in the proof of part~(b), $\partial_j\mathcal D$
is nondecreasing in the other coordinate, so the worst point of each face is
the corner $(L_1,H_2)$ or $(H_1,L_2)$ and four corner derivatives suffice; we
bound them rigorously as follows.  The same calculation
also certifies the flat-prior boxes in Section~\ref{supp:prior_numerics}.
Write $a=w(1)/2$, $b=w(2)/2$ and $c=\theta^2$.
All four boxes have $\mu_1>a$.  For ordered ratios $x\le z$, the branches are
\[
 \psi_1=(a+bx)z-\mu_1x-\mu_2,\qquad
 \psi_2=(a+2bx-\mu_1)z+(a-\mu_1)x-\mu_2.
\]
Put $t_1=(\mu_1x+\mu_2)/(a+bx)$, $t_2=\{(\mu_1-a)x+\mu_2\}/(a+2bx-\mu_1)$ and
$t_{12}=ax/(\mu_1-bx)$.  At fixed $x<\mu_1/b$, action one is optimal on
$z\in(\max\{x,t_1,t_{12}\},c)$; it is absent for $x\ge\mu_1/b$.
Action two is absent for $x\le(\mu_1-a)/(2b)$, occupies
$(\max\{x,t_2\},\min\{c,t_{12}\})$ when $(\mu_1-a)/(2b)<x<\mu_1/b$,
and occupies $(\max\{x,t_2\},c)$ when $x\ge\mu_1/b$.
Empty intervals contribute zero; ties and denominator endpoints have null
probability zero.

Partition $[0,c]$ at its interior points $\mu_1/b$ and $(\mu_1-a)/(2b)$, and then divide
each segment $[s,t]$ into $\lceil n(t-s)/c\rceil$ equal rational bins.
The thresholds are monotone on each segment, so their limiting endpoint values,
clipped at $c$, give an inner and an outer band for each action throughout a bin.
The tail functions $p_\theta,M_\theta$ defined above enclose the probability
$\pi_l$ and first moment $\tilde m_l$ of each band.  If a bin has null probability
$P$ and first moment $M$, its contributions to
$\beta_1^\mu=x\mathbf1\{l\ge1\}+z\mathbf1\{l=2\}$ and
$\beta_2^\mu=\mathbf1\{l\ge1\}$, the family-wise error coefficients of the
active action $l$ (Section~\ref{supp:proof_beta_oak}), are enclosed by the
interval expressions
\[
 2\{M(\pi_1+\pi_2)+P\tilde m_2\},\qquad 2P(\pi_1+\pi_2),
\]
respectively, using the ordered null density $2f_\theta(x)f_\theta(z)$ for $x<z$.
Summing gives outward bounds on
$\nabla\mathcal D=(2\alpha,\alpha)-\E_0(\beta_1^\mu,\beta_2^\mu)$.
Moment calculations round outward to multiples of $2^{-96}$.
Square roots are enclosed by integer square roots; logarithms are reduced
to arguments in $[1,2]$ and bounded by the series and remainder in Step~4c above, with $N=23$, that is, $24$ terms.  Thus the bounds account for partition,
transcendental-function, and rounding errors.

For a box $[L,H]$, let $\nabla^{LH}_j=\partial_j\mathcal D(L_1,H_2)$ and
$\nabla^{HL}_j=\partial_j\mathcal D(H_1,L_2)$.  Every row of the following
rational bounds, in units of $10^{-4}$, has the required strict signs: the
upper bounds on $\nabla^{LH}_1$ and $\nabla^{HL}_2$ are negative, and the lower
bounds on $\nabla^{HL}_1$ and $\nabla^{LH}_2$ are positive ($n$ is the number of
bins):
\[
\begin{array}{ccr|rrrr}
 \theta&w(1)&n& \nabla^{LH}_1\text{ (upper)}&\nabla^{HL}_1\text{ (lower)}&
 \nabla^{HL}_2\text{ (upper)}&\nabla^{LH}_2\text{ (lower)}\\ \hline
 2&0&4096&-2.969&1.272&-1.614&0.456\\
 3&0&4096&-0.574&1.471&-0.656&1.676\\
 2&1/2&1024&-2.557&3.348&-9.864&10.127\\
 3&1/2&1024&-2.751&2.745&-4.494&4.225
\end{array}
\]
By the monotonicity proved in part~(b), these strict signs extend to the
left, right, bottom, and top faces, respectively.  A minimizer restricted to
the compact box must therefore be interior, hence global by convexity.
Any other global minimizer outside the interior would join it by a segment
of minimizers meeting a face, contradicting the strict sign there.  This
places \emph{every} dual minimizer in the interior.

Finally, since $\lambda_-=0$ and $\lambda_+=c$, the point-prior witness
requires $0<2\mu_1<c$ and \mbox{$2\mu_1+\mu_2/\mu_1<c$}.
The latter expression increases in $\mu_2$ and is convex in $\mu_1>0$,
so its maximum over the box is
$\max\{2L_1+H_2/L_1,2H_1+H_2/H_1\}$.
Exact rational substitution gives margins of at least
$(2.506,1.484,1.3019)$ for $\theta=2$ and
$(4.232,4.756,4.6820)$ for $\theta=3$, in that order.
The script \texttt{oak\_\allowbreak verified\_\allowbreak enclosure.py}
computes all bounds, so this proof is computer-assisted.  The explicit
witness, the strict-loss conclusion and, since $\theta\ge2$, the rate
conclusion follow. \hfill$\square$

\subsection{The general framework under an arbitrary
prior}\label{supp:proof_priorfree}

The main text states its results for an arbitrary prior $w$ in the objective
$\Pi_w=\sum_\gamma w(\gamma)\Pi_\gamma$ of~\eqref{eq:objective}, with each
result retaining its stated hypotheses.  The proofs above hold for every $w$
because $w$ enters only through the objective coefficients $a_{w,k}$: the
constraints are untouched~(i), every section of the dual integrand stays
convex and piecewise affine~(ii), and the loss arguments use nothing
more~(iii).  Two things do not transfer: automatic OAK
(Lemma~\ref{lem:oak_auto}) requires $w(K)>0$, and the explicit point-prior
Beta results and the convergence guarantee cited in the remark after the
proof in Section~\ref{supp:proof_oracle} apply only at the point prior.

By Section~\ref{app:dual}, $\Pi_w$ and each family-wise error functional are linear in the policy,
with coefficient functions
$a_{w,k}=\sum_\gamma w(\gamma)a_{\gamma,k}$ and $b_{\gamma,k}$, whose closed
forms, with $k$ the rank of the likelihood ratio, are given in the Notation
paragraph of Section~\ref{supp:proofs} and derived in Section~\ref{app:dual},
in agreement with \citet[Lemma~1]{dubey2026esp}.  Here $\gamma$ counts alternatives, as it does
for $\Pi_\gamma$ and $w$ throughout, so $b_{\gamma,k}$ belongs to the
family-wise constraint at $K-\gamma$ true hypotheses and its cumulative form in
Section~\ref{app:dual} is
$\sum_{k\le l}b_{\gamma,k}=\gamma!\,(K-\gamma)!\,B_{K-\gamma,l}$.
Three observations do the work.

\emph{(i) The constraint side is untouched.}  The prior appears only in
$a_{w,k}$.  The feasible set of the ordered primal, the strong-FWER
coefficients $b_{\gamma,k}$, their nonnegativity, and hence the monotonicity
used in Proposition~\ref{prop:enclosure}, are all independent of $w$.

\emph{(ii) Every section stays convex and piecewise affine.}  The
fixed-set representation in Step~5 of
Section~\ref{supp:proof_graderate} extends to
every $w$.  For a rejection set $R$, read $a_{\gamma,k}$ with $k$ as a
coordinate label rather than a rank, that is, with $v_k$ in place of
$v_{(k)}$; the objective term of the branch of
$R$ is then $\frac1{K!}\sum_{k\in R}a_{w,k}(v)$, and replacing $y\in R$ by $x\notin R$
with $x\ge y$ changes it by
$\frac1{K!}\sum_\gamma w(\gamma)(\gamma-1)!\,(K-\gamma)!\,(x-y)\,e_{\gamma-1}(v_{-x,-y})\ge0$,
because $e_{\gamma-1}(v_{-x})=e_{\gamma-1}(v_{-x,-y})+y\,e_{\gamma-2}(v_{-x,-y})$.
The constraint term does not involve $w$, and its contribution to the branch does not decrease under the same swap, as shown there.  Hence for every $w$
the ordered branch $\psi^{\mu}_l$, which depends on $w$ through $A_l$, is the
maximum of the fixed-set
branches over $|R|=l$, each of which is multilinear in $v$ because
elementary symmetric polynomials are multilinear and $e_{\gamma-1}(v_{-k})$
does not involve $v_k$.  The integrand
$F_\mu=\max\{0,\psi^{\mu}_1,\dots,\psi^{\mu}_K\}$ is therefore a maximum of
finitely many coordinatewise-affine functions.  Consequently every section
$x\mapsto F_\mu(z_1,\dots,x,\dots,z_K)$ is convex and piecewise affine,
which is the only property of $F_\mu$ used in
Lemma~\ref{lem:oakk_hinge} and the surrounding arguments.

\emph{(iii) The arguments for losslessness, strict loss, and the rate carry over verbatim.}  Losslessness of
sufficient reports (Proposition~\ref{prop:distribution}(a)) is a Blackwell
sufficiency statement about the experiment and does not mention the objective.
The two-sided dual bound~\eqref{eq:dual_pinch} follows from strong duality
for the oracle and report programs; iterated coordinatewise conditional
Jensen, which needs exactly the convex sections supplied by~(ii), makes both
sides nonnegative.
The upper bound of Lemma~\ref{lem:uniform_upper} uses the fact that conditional averaging reproduces $F_{\bar\mu}$ exactly off
the decision surface of $F_{\bar\mu}$; this exactness follows from
multilinearity, which~(ii) preserves, and from conditional independence,
which does not involve $w$.  It also uses the
$O(m^{-1})$ boundary mass and $O(m^{-1})$ within-cell deviation, whose orders
are unchanged for each fixed prior.  The lower bound of
Theorem~\ref{thm:oracle_kink}(b) reduces the conditional-Jensen gap at the
witness multiplier $\mu^*$ to integrated one-coordinate hinge regrets, each
equal to half a within-cell squared-error distortion; the hinge arises from a
kink of a convex piecewise-affine section, again supplied by~(ii).

The kink locations and sizes, and the minimizing multiplier, depend on $w$;
OAK concerns the corresponding dual at a minimizer.
Lemma~\ref{lem:oak_auto} guarantees a kink whenever $w(K)>0$, the example
in Section~\ref{supp:proof_oak_auto} shows that none need exist when $w(K)=0$,
and the explicit witness that fixes the constants must be located afresh at
each prior, as is done at the point prior in
Corollary~\ref{cor:enclosure_beta}, at the flat prior in
Section~\ref{supp:prior_numerics}, and numerically at other nondegenerate priors
in Table~\ref{tab:prior}.

Three further results deserve explicit mention because they are prior-free for
a different reason.  Proposition~\ref{prop:largeS} concerns $K=1$, where the
prior is degenerate and $\Pi_w=\Pi_1$ identically.  In
Proposition~\ref{prop:alloc} and Lemma~\ref{lem:marginalcentral}
(Section~\ref{supp:aggregator_gap}) the null-calibrated rule rejects $H_k$ on the event
$\{p_\phi(\mathbf y_{\cdot,k})\le\alpha/K\}$, where $\phi$ is the fixed
combining statistic of Section~\ref{sec:marginal} applied to the reports for
hypothesis $k$ and $p_\phi(x)=\pr_{X_0\sim\tilde\nu_0}\{\phi(X_0)\ge\phi(x)\}$
is its exact-null tail $p$-value.  That event depends only on the
reports for hypothesis $k$, so each false hypothesis has rejection
probability $\pr_{\tilde\nu_1}(p_\phi\le\alpha/K)$ under any
configuration.  Thus $\Pi_\gamma$ is the same for every $\gamma$ and hence
$\Pi_w$ does not depend on $w$ at all.  The null-calibrated power, its
optimal thresholds, and the elasticity condition are therefore literally
unchanged, and the fixed-report
gap identity of Proposition~\ref{prop:reduction} (Section~\ref{supp:aggregator_gap}) continues to hold by
subtraction, with only its model-aware term moving with $w$.

\subsection{The prior-weighted dual at \texorpdfstring{$K=2$}{K=2}, and its
active-kink witness}\label{supp:prior_numerics}

This subsection does at the flat prior what Corollary~\ref{cor:enclosure_beta}
does at the point prior: for $K=S=2$ homogeneous $\mathrm{Beta}(1,\theta)$ sites
with $\theta\in\{2,3\}$ and $w=(1/2,1/2)$, every dual minimizer lies in an
explicit box, and at every point of that box the crossing of the two
positive branches is an OAK witness.  A nondegenerate prior moves that
crossing off the line where the smaller ratio equals $2\mu_1$, so the
sufficient condition of Lemma~\ref{lem:k2_witness} no longer applies and the
witness requirements are checked directly.

Under a general prior the two nonzero branches remain affine in the larger
likelihood ratio, with slopes and intercepts that depend on $w$.  For $K=2$
with $S=2$ homogeneous $\mathrm{Beta}(1,\theta)$ sites, sort the
per-hypothesis likelihood ratios as $x=\min(\Lambda_1,\Lambda_2)$ and
$z=\max(\Lambda_1,\Lambda_2)$.  For the policy that rejects the top $l$,
$\Pi_1=\tfrac12\E_0\{\sum_{i\le l}v_{(i)}\}$ and $\Pi_2=\tfrac12\E_0(l\,e_2)$, so
its objective coefficient is
\[
 A_1=\tfrac12\{w(1)z+w(2)xz\},\qquad
 A_2=\tfrac12\{w(1)(x+z)+2w(2)xz\},
\]
and the family-wise error coefficients $B_{j,l}$ of Section~\ref{app:dual} give
$\psi_1=A_1-\mu_1x-\mu_2$ and $\psi_2=A_2-\mu_1(x+z)-\mu_2$.  Both branches
are affine in $z$, with
\[
\begin{aligned}
 \text{slope}_1&=\tfrac12\{w(1)+w(2)x\}, &\quad
 \text{intercept}_1&=-\mu_1x-\mu_2,\\
 \text{slope}_2&=\tfrac12 w(1)+w(2)x-\mu_1, &\quad
 \text{intercept}_2&=\tfrac12 w(1)x-\mu_1x-\mu_2 .
\end{aligned}
\]
At the point prior $w=(0,1)$ these reduce to $\text{slope}_1=x/2$,
$\text{slope}_2=x-\mu_1$ and the common intercept $-\mu_1x-\mu_2$, which give
the integrand used for every point-prior $K=2$ value reported in this paper.  The
script \texttt{prior\_weighted\_K2.py} implements the display above.  At the
point prior it reproduces the oracle values $\pi^*_2=0.1564610031$
($\mathrm{Beta}(1,2)$) and $0.2869773539$ ($\mathrm{Beta}(1,3)$), printed to four
digits in Table~\ref{tab:k2}, to within $4.5\times10^{-11}$ and
$7.1\times10^{-14}$, with the same minimizers, which checks the general-prior
implementation at the point prior.

Unlike the point prior, a nondegenerate $w$, with $w(1),w(2)>0$, makes the branch dominance depend on
$z$ as well as $x$, so the inner integrand is the true upper envelope $\max(0,\psi_1,\psi_2)$ rather than a single preselected branch.  The
branches cross at $z^\star=\{w(1)x/2\}/\{\mu_1-w(2)x/2\}$, which lies in
$(0,\infty)$ exactly when $x<2\mu_1/w(2)$, and the slope jump there is
$|w(2)x/2-\mu_1|$.  The general-prior analogue of the explicit witness of
Lemma~\ref{lem:k2_witness} requires a positive-null-mass set of $x$ on which $z^\star$ lies in $(x,\lambda_+)$, both
branches are positive at $z^\star$, and the jump is bounded away from zero.
In the notation of Definition~\ref{def:oak} the roles are the reverse of those
in Lemma~\ref{lem:k2_witness}: the varied coordinate is the larger ratio $z$ and
the fixed one is $x$, with $q(x)=z^\star(x)$ and $\Delta(x)=\mu_1-w(2)x/2$.
Table~\ref{tab:prior} reports a scan for such a set.

\medskip
\noindent\emph{From a scan to a certificate.}  Proposition~\ref{prop:enclosure}
applies to the prior-weighted dual unchanged, since its proof uses only
convexity, coercivity, and the nonnegativity of the $b_{\gamma,k}$, none of
which depends on $w$.  What does not carry over is the $K=2$
sufficient condition of Lemma~\ref{lem:k2_witness}, which is derived from
the full-alternative branch algebra; for a general $w$ the witness
requirements must be checked directly.  Fix an enclosure $[\text{lo},\text{hi}]$ and a window
$x\in[x_a,x_b]$.  Each requirement has a monotone worst case in $\mu$, so one
corner suffices for the whole box:
\begin{itemize}
\item the crossing exists for every $\mu_1\ge\text{lo}_1$ exactly when
 $x_b<2\,\text{lo}_1/w(2)$, the qualifying region growing with $\mu_1$;
\item the slope jump $\mu_1-w(2)x/2$ is smallest at $\mu_1=\text{lo}_1$,
 $x=x_b$;
\item $z^\star$ increases in $x$ and decreases in $\mu_1$, so $z^\star<\lambda_+$
 is worst at $(\text{lo}_1,x_b)$ while $z^\star>x$, equivalent to
 $\mu_1<\{w(1)+w(2)x\}/2$, is worst at $(\text{hi}_1,x_a)$;
\item $\psi_1(z^\star)$ decreases in both multipliers, so positivity is worst
 at $(\text{hi}_1,\text{hi}_2)$; its minimum over the entire $x$ window
 is obtained analytically as follows.
\end{itemize}
For the last requirement, put $a=w(1)/2$, $b=w(2)/2$ and
$M=\text{hi}_1$.  The crossing value at the worst corner is
\[
 f(x)=\frac{ax(a+bx)}{M-bx}-Mx-\text{hi}_2,\qquad
 f'(x)=(a+M)\left\{\frac{aM}{(M-bx)^2}-1\right\}.
\]
On the qualifying window, $M-bx>0$ and
$f''(x)=2abM(a+M)/(M-bx)^3>0$.  The minimum is therefore attained at
$x_0=(M-\sqrt{aM})/b$ if $x_0\in[x_a,x_b]$, and otherwise at the nearer
endpoint.  Checking $f$ at that point thus verifies positivity throughout the window, rather than only at sampled points.

The requirements listed above pull in opposite directions: the crossing
exists only for small $x$, but it lies above $x$ and on the positive part of
the envelope only for large $x$, so the admissible window is an interior band
that must be located.  At the flat prior, the outward-rounded integer-arithmetic
face bounds of Section~\ref{supp:proof_enclosure} place every minimizer inside
$[0.710,0.718]\times[0.840,0.856]$ for $\theta=2$ and
$[1.055,1.067]\times[1.535,1.557]$ for $\theta=3$.  Choose the rational
windows $[2.18,2.26]$ and $[3.55,3.75]$, respectively; each has positive null
mass.  Exact rational substitution gives lower bounds
$(0.580,0.145,0.103,0.077,0.083)$ and
$(0.470,0.1175,1.021,0.0705,0.279)$ for the five margins
$2\,\text{lo}_1/w(2)-x_b$, $\text{lo}_1-w(2)x_b/2$,
$\lambda_+-z^\star(\text{lo}_1,x_b)$, $\{w(1)+w(2)x_a\}/2-\text{hi}_1$, and
$\min f$, in the order of the bullets above (the third contributes two).  In both cases $f'(x_a)>0$, so $f$ is minimized at $x_a$.
Thus the explicit witness holds at every dual minimizer for both models.
The script \texttt{oak\_\allowbreak verified\_\allowbreak enclosure.py}
computes the face bounds in outward-rounded integer arithmetic and the
witness margins by exact rational substitution.

The prior-weighted finite report-distribution program used for the recomputed deficits
of Section~\ref{sec:numerics_prior} is the general-$K$ engine of Section~\ref{app:dual}
with one substitution: the objective term $(l/K)e_K$ is replaced by
$A_l=(K!)^{-1}\sum_{i\le l}a_{w,i}$, the family-wise error coefficients being
free of $w$.  The implementation (\texttt{fwer\_dual.py}, argument
\texttt{prior}) defaults to the point prior, which is the prior used throughout this paper unless another is named, and its tests check, for $K$ up to five, both primal--dual agreement and the family-wise error rate at every configuration.  At the point prior the
prior-weighted program gives the relative deficits $8.53\%$ and $6.21\%$ of
Table~\ref{tab:k2}, with optimizing thresholds $(0.288,0.328)$ and
$(0.247,0.247)$, so the prior-weighted comparison is made on identical
terms.

\subsection{Fixed-report optimal aggregation: proof of
Theorem~\ref{thm:aggregator}}\label{supp:proof_aggregator}
The proof carries the oracle argument of Section~\ref{supp:proof_oracle} over
to the report likelihood ratios $L_k$: Steps~1--2 reduce the report array to
an exchangeable two-group family in $(L_k,\xi_k)$, Step~3 reduces to
likelihood-ratio-ordered policies and applies Lemma~\ref{lem:ordered_duality},
proving~(b) and then~(a), and Step~4 gives the coefficient structure~(c).
Write $\tilde\nu_0 = \otimes_s \tilde\nu_{s,0}$ and $\tilde\nu_1 = \otimes_s
\tilde\nu_{s,1}$ for the null and alternative distributions of the per-hypothesis
report vector $\mathbf y_{\cdot,k} = (y_{1,k},\dots,y_{S,k})$. Under
Assumptions~\ref{ass:model}--\ref{ass:known} these are the true distributions, since
the report map is known, the local $p$-values are uniform under a null
coordinate, and they have densities $g_s$ under an alternative coordinate.

\emph{Step 1 (sufficiency).} Conditional on $h$, the report array has distribution
$\otimes_k\tilde\nu_{h_k}$, so its likelihood ratio against the
global-null product distribution is $\prod_k L_k^{h_k}$ with
$L_k=(\mathrm{d}\tilde\nu_1/\mathrm{d}\tilde\nu_0)(\mathbf y_{\cdot,k})
=\prod_s\ell_s(y_{s,k})$. Hence
$(L_1,\dots,L_K)$ is sufficient for the family $\{\pr_h : h \in \{0,1\}^K\}$,
so by the Rao--Blackwell reduction for tests any randomized procedure may be
replaced by one measurable in it with the same power and the same
family-wise error at \emph{every} configuration.

\emph{Step 2 (two-group reduction).} The $L_k$ are independent across $k$,
each with the null distribution $\lambda_0$ or the alternative distribution $\lambda_1$
according to $h_k$. For bounded $\phi$,
\[
\E_{\lambda_1}\{\phi(L)\}
=\int \phi\{L(x)\}\,\mathrm{d}\tilde\nu_1(x)
=\int \phi\{L(x)\}\,L(x)\,\mathrm{d}\tilde\nu_0(x)
= \E_{\lambda_0}\{\phi(L)\,L\},
\]
so $\mathrm{d}\lambda_1/\mathrm{d}\lambda_0(L)=L$. As $L$ may have atoms,
append independent $\xi_k\sim\mathrm{Unif}[0,1]$. Conditional on a general
configuration $h$, the pairs $(L_k,\xi_k)$ are independent with the
$h_k$-indexed distribution; they are not identically distributed when $h$ contains
both true and false hypotheses. The family is permutation equivariant under
simultaneous permutation of the pairs and the coordinates of $h$. Under
the global null and the full alternative, respectively, the pairs are
i.i.d.\ with distribution $\lambda_0\otimes\mathrm{Unif}$ and
$\lambda_1\otimes\mathrm{Unif}$. The $\xi_k$ are ancillary and provide an
almost-sure ordering within likelihood-ratio ties.

\emph{Step 3 (optimality via arrangement-increasing structure).} The
reduction of~\citet{rosset2022optimal} to symmetric policies ordered by
likelihood ratio uses permutation equivariance and an
\emph{arrangement-increasing} two-group family~\citep{hollander1977functions,marshall2011inequalities}; these are also the
structural hypotheses used by~\citet{dubey2026esp}.  The symmetrization fact at the start of
Section~\ref{supp:proofs} applies to the report experiment: every feasible
report-measurable procedure can first be made symmetric with the same
$\Pi_w$ and strong family-wise control. For $K=1$, the result is the
ordinary Neyman--Pearson test on $L$, so suppose $K\ge2$ in the rearrangement
argument that follows.

We verify these structural hypotheses for the augmented statistics $(L_k,\xi_k)$.  Use
the common dominating measure
$\lambda=\lambda_0\otimes\mathrm{Leb}$ for the augmented pairs. By Step~2, the
two one-coordinate densities against $\lambda$ are
\[
f_0(\ell,u)=1,\qquad f_1(\ell,u)=\ell
\quad\text{for $\lambda$-almost every $(\ell,u)$}.
\]
Thus the density at configuration $h$ is
$\prod_k f_{h_k}(L_k,\xi_k)$, and simultaneous coordinate permutation gives
the permutation equivariance noted in Step~2. On the likelihood-ratio
preorder, arrangement increasingness reduces to the pairwise concordance
inequality: for $\ell\le\ell'$,
\[
f_0(\ell,u)f_1(\ell',u')
\ge f_0(\ell',u')f_1(\ell,u),
\]
which is simply $\ell'\ge\ell$. The rearrangement argument of
\citet[Theorem~1]{rosset2022optimal}, starting from the symmetric procedure,
therefore removes strict likelihood-ratio inversions.  The removal does not
decrease $\Pi_w$ for any prior (Section~\ref{supp:proof_priorfree}(ii)), and, by
the swap identity of Section~\ref{supp:proof_graderate}, it does not weaken
feasibility (Step~2 of Section~\ref{supp:proof_oracle}).

Because that argument treats equal likelihood ratios as tied (Step~2 of
Section~\ref{supp:proof_oracle}), tied
coordinates must receive equal treatment, so one further implementation
argument is needed when $L$ has atoms. Let $\mathcal K_L$ denote the conditional
action kernel of the symmetric, likelihood-ratio-ordered procedure just
constructed. Symmetry gives
$\mathcal K_{\sigma L}(\sigma A)=\mathcal K_L(A)$ for every coordinate
permutation $\sigma$
and rejection set $A$. For a realized vector $L$, let $\Sigma_L$ be the finite
group of permutations within its equal-value blocks. If $\sigma\in\Sigma_L$, then
$\sigma L=L$, and hence $\mathcal K_L(\sigma A)=\mathcal K_L(A)$. Thus, conditional on the
number selected from any tied block, the selected subset is uniform.
Independent $\xi_k$ supply exactly this uniform selection, while a separate
auxiliary uniform variable realizes any mixture among rejection counts.
Consequently the procedure rejects
a nested top-$l$ set in the likelihood-ratio preorder, with a uniformly
selected subset if the cutoff splits a tie. This randomized extension
handles atoms; it does not rely on any claim that the ancillary variable
removes ties among the likelihood ratios. Its nested action class is precisely the class
parameterized by the primal variables $r_l(v)$ of Section~\ref{app:dual}.

Lemma~\ref{lem:ordered_duality} supplies attainment and strong duality for
the randomized ordered program~\eqref{eq:ordered_primal}. The reduction then gives a $\Pi_w$-maximal ordered
procedure under strong
family-wise control, proving~(b). The program imposes the family-wise
constraints at every configuration under the true distribution. At a mixed
configuration the constraint involves both $\lambda_0$ and $\lambda_1$,
because whether a true hypothesis is rejected depends on its rank relative to
the alternative coordinates. Both distributions are known
(Assumption~\ref{ass:known}) and enter the constraints. The resulting
family-wise error rate is at most $\alpha$ at every configuration in finite
samples, proving~(a).

\emph{Step 4 (elementary-symmetric-polynomial structure).} The family-wise
error coefficients $B_{j,l}=e_{K-j}-p_l\,e_{K-l-j}(\{l+1,\ldots,K\})$ and
their termwise nonnegativity are derived in Section~\ref{app:dual} using only
permutation equivariance, the two-group factorization, and the per-coordinate
density-ratio identity, all verified in Step~2 for $(L_k,\xi_k)$ with $L_k$
in place of $\Lambda_k$. For a finite report distribution the coefficients enter the
finite primal and dual linear programs displayed in Section~\ref{app:dual},
which proves~(c). We use those programs directly; no continuous-distribution
monotonicity result, and no convergence guarantee for the coordinate
algorithm cited after the proof in Section~\ref{supp:proof_oracle}, is needed for the
atomic objective, which is nonsmooth at branch ties. From Step~2 onward, the
reduction uses only a specified dominated pair $(\nu_0,\nu_1)$, its likelihood
ratio, and conditionally independent replication of that pair across
hypotheses. Thus the same argument proves the theorem's stated extension
beyond report distributions generated by monotone $p$-value densities.
\hfill$\square$

\subsection{Null-calibrated control: proof of
Proposition~\ref{prop:marginal}}\label{supp:proof_marginal}

Fix $k$ and a
configuration $h$ with $h_k=0$.  Write $T_s$ for site $s$'s
report map, assumed ordered: its cells are intervals of $[0,1]$, with
$c'\preceq_s c$ denoting their order from smaller to larger $p$-values.
In particular, binary threshold reports have $1\preceq_s0$, independently
of their numerical labels. Write
$\hat p_s(c)=\tilde\nu_{s,0}(\{c':c'\preceq_s c\})$ for the null probability that the
site reports $c$ or a more significant value, a quantity determined by $T_s$
alone because $\tilde\nu_{s,0}$ is the push-forward of $\mathrm{Unif}[0,1]$
through $T_s$, whatever $g_s$ may be. For threshold reports this gives
$\hat p_s(1)=t_s$ and $\hat p_s(0)=1$.

\emph{Step 1 (each transmitted report is superuniform on the $p$-value
scale).}  For an ordered map the cell containing $u$ has
$\hat p_s\{T_s(u)\}\ge u$ pointwise, because the cell's own null mass
together with that of every more significant cell is at least the null mass
of $[0,u]$, which is $u$.  Hence, if $u_{s,k}$ is superuniform under $H_k$,
\[
 \pr_h\bigl[\hat p_s\{T_s(u_{s,k})\}\le x\bigr]
 \;\le\;\pr_h(u_{s,k}\le x)\;\le\;x,
 \qquad x\in[0,1].
\]
When the null $p$-value is exactly uniform the first inequality is an
equality on the attainable values of $\hat p_s$, and the report $p$-value has
distribution function equal to $x$ at every attainable value $x$.

\emph{Step 2 (under the null, the combined statistic is stochastically
dominated by a known reference distribution).}  Let
$W_k=-\sum_s\log\hat p_s(y_{s,k})$ as in~\eqref{eq:fisher}, and let $W^0$
be the same statistic evaluated at independent uniform inputs $U_s$.
Superuniformity permits an inverse-distribution coupling with
$u_{s,k}\ge U_s$ almost surely.  The function
$u\mapsto\hat p_s\{T_s(u)\}$ is nondecreasing because the map is ordered,
so under this coupling each actual site score $-\log\hat p_s(y_{s,k})$ is at most its uniform-reference score.  Cross-site independence allows these couplings to be taken
independently, giving
\[
 \pr_h(W_k\ge \tau)\;\le\;\pr_0(W^0\ge \tau)\qquad\text{for every }\tau.
\]
The reference distribution does not depend on $k$.  It is a finite
convolution of the $S$ known site
score distributions and is computable from the report maps alone; for the
full report, $2W^0\sim\chi^2_{2S}$ by Fisher's classical
result~\citep{fisher1932statistical}.  Write
$\bar F_0(\tau)=\pr_0(W^0\ge \tau)$ and
$p_k^{\mathrm{nc}}=\bar F_0(W_k)$.  The reference tail $p$-value
$\bar F_0(W^0)$ is superuniform, including when the reference distribution has
atoms.  Since $\bar F_0$ is nonincreasing and $W_k$ is stochastically
dominated by $W^0$, $p_k^{\mathrm{nc}}$ is superuniform as well.

\emph{Step 3 (Bonferroni).}  At every $h$,
\[
 \fwer_h\;\le\;\sum_{k\in\mathcal H_0(h)}\pr_h\bigl(p_k^{\mathrm{nc}}\le\alpha/K\bigr)
 \;\le\;|\mathcal H_0(h)|\,\frac{\alpha}{K}\;\le\;\alpha ,
\]
with no assumption on the joint distribution of $(W_1,\dots,W_K)$, so
arbitrary dependence across hypotheses is permitted.  The marginally valid
$p$-values also permit any multiplicity correction valid under arbitrary
dependence, including Holm's step-down procedure~\citep{holm1979simple}.

\emph{Step 4 (an arbitrary fixed statistic, and training).}  Replace the
Fisher statistic by a fixed measurable statistic $\phi$ and use its reference
tail $p$-value
$p_\phi(x)=\pr_{X_0\sim\tilde\nu_0}\{\phi(X_0)\ge\phi(x)\}$.
For exactly uniform null inputs no monotonicity of $\phi$ is needed.
For merely superuniform inputs, requiring $\phi$ to be nonincreasing in
each report $p$-value preserves the coupling in Step~2.  In both cases
Step~3 applies verbatim to the rule that rejects when $p_\phi\le\alpha/K$.
Finally, let $\mathcal F_{\mathrm{tr}}$
contain all training data and design-generating randomness and suppose it is
jointly independent of the entire tested $p$-value array.  Conditioning on
$\mathcal F_{\mathrm{tr}}$ fixes the maps, the statistic, and its reference
tail, and leaves the sitewise inputs independent for each $k$; Steps~1--3
apply verbatim, and the tower property gives the unconditional bound.
Nowhere does any $g_s$ appear.
\hfill$\square$

\subsection{Optimal null calibration at a single hypothesis: proof of
Proposition~\ref{prop:marginalK1}}\label{supp:proof_marginalK1}

At a fixed common threshold, the report likelihood ratio orders the data
by the local-rejection count.  Randomized null calibration therefore attains
the model-aware value, and conservative null calibration differs only by
the boundary-randomization mass.  Let $K=1$, let the $S$ sites be homogeneous with common
nonincreasing alternative density $g$ and distribution function $G$, and
fix a common threshold $t\in(0,1)$.  The report array is
$(\mathbf 1\{u_s\le t\})_{s\le S}$, whose distribution depends only on the
count $N=\#\{s:u_s\le t\}$: $N\sim\mathrm{Bin}(S,t)$ under the null and
$\mathrm{Bin}(S,G(t))$ under the alternative.  The report likelihood ratio
at $N=r$ is
\[
 \ell(r)=\left\{\frac{G(t)}{t}\right\}^{r}
 \left\{\frac{1-G(t)}{1-t}\right\}^{S-r},
\]
which is nondecreasing in $r$ because $G(t)\ge t$.  It is strictly
increasing when $t<G(t)<1$; if $G(t)=1$, all alternative mass is at $N=S$,
and if $g\equiv1$ the likelihood ratio is constant.  At $K=1$ the strong-FWER
constraint~\eqref{eq:fwer} is the single size constraint, so by the
Neyman--Pearson lemma the most powerful rule on this distribution, which is
the model-aware rule, is the likelihood-ratio test: it rejects for large $r$, with randomization at the boundary.
Let $N_0\sim\mathrm{Bin}(S,t)$ and let $V\sim\mathrm{Unif}[0,1]$ be independent
of the reports.  Conditional on $N=r$, the randomized tail
\[
 p_t^{\mathrm{rand}}
 =\pr(N_0>N\mid N)+V\,\pr(N_0=N\mid N)
\]
is uniform over the interval
$[\pr(N_0>r),\pr(N_0\ge r)]$.  These intervals partition $[0,1]$, and
their lengths equal the null probabilities of the corresponding counts.
Consequently $p_t^{\mathrm{rand}}$ is uniform under the null, and rejecting
when $p_t^{\mathrm{rand}}\le\alpha$ is exactly the level-$\alpha$
Neyman--Pearson count rule.  This construction uses only $S$, $t$, and
$\alpha$, so it attains model-aware power for every homogeneous
nonincreasing $g$ without using $g$ in calibration.  If $g\equiv1$,
every exact-level test has power $\alpha$, and the same conclusion holds.
At the endpoint $t=1$, $N=S$ under both hypotheses, so
$p_t^{\mathrm{rand}}=V$ and the optimal power is again $\alpha$; the
conservative rule defined next never rejects.

The conservative null-calibrated rule, which is the deterministic count rule
of Proposition~\ref{prop:marginalK1}, rejects $\{N\ge c\}$ for the smallest
$c\in\{1,\ldots,S+1\}$ with
\mbox{$\pr_0(N\ge c)\le\alpha$}, which is the same nested family of regions without
the boundary randomization.  Hence, at each fixed $t$, the model-aware power
exceeds the null-calibrated power by exactly the randomized boundary mass,
\[
 \Pi^{\mathrm{aw}}(t)-\Pi^{\mathrm{nc}}(t)
 =\frac{\alpha-\pr_0(N\ge c)}{\pr_0(N=c-1)}\,\pr_1(N=c-1)\;\ge\;0 .
\]
This identity also covers $c=S+1$, when the conservative rule never rejects.
If the outer maximum of $\Pi^{\mathrm{aw}}(t)$ over $t$ is attained at a
threshold where $\pr_0(N\ge c)=\alpha$, that boundary mass is zero and the two
suprema agree.  This is what the homogeneous $\mathrm{Beta}(1,2)$ searches
of Section~\ref{sec:numerics_ext} find at every $2\le S\le10$: the model-aware
outer optimum lies at a kink of $t\mapsto\Pi^{\mathrm{aw}}(t)$, where the count
tail is exactly $\alpha$, and the two powers agree to six digits.  The
unanimity rule is the special case $c=S$, feasible when $t^S\le\alpha$, and it
is optimal among null-calibrated count rules only while $c=S$ is the count
cutoff selected by the null calibration.  For the $\mathrm{Beta}(1,\theta)$ families of
Proposition~\ref{prop:largeS}, unanimity is strictly suboptimal for all
sufficiently large $S$.
\hfill$\square$

\subsection{Optimal unanimity thresholds: proof of
Proposition~\ref{prop:alloc}}\label{supp:proof_allocation}
The unanimity rule has exact null probability $\prod_s t_s$ and
full-alternative power $\prod_s G_s(t_s)$ (Section~\ref{sec:design}), so the
design problem is the constrained maximization below.
Work in $x_s=\log t_s$ and maximize
\[
\mathcal L(x)=\sum_s\log G_s(e^{x_s})
\quad\text{subject to}\quad
\sum_s x_s=\log(\alpha/K).
\]
The Lagrangian stationarity condition at an interior maximizer is
\[
\frac{t_s g_s(t_s)}{G_s(t_s)}=\kappa,\qquad s=1,\ldots,S,
\]
and the product constraint determines $\kappa$. If
$\varphi_s(x)=\log G_s(e^x)$ and $t=e^x$, then
\[
\varphi_s'(x)=\frac{t g_s(t)}{G_s(t)}.
\]
Thus $\varphi_s$ is concave exactly when this elasticity is nonincreasing
in $t$. Under that condition, $\mathcal L$ is concave on the affine constraint
set, so every stationary point is a global maximizer.

For the $\mathrm{Beta}(1,\theta)$ alternatives with $\theta>1$ the
elasticity is a strictly decreasing bijection of $(0,1)$ onto $(0,1)$.  For
$g(u)=\theta(1-u)^{\theta-1}$,
\[
\frac{t g(t)}{G(t)}
=\frac{\theta t(1-t)^{\theta-1}}{1-(1-t)^\theta}.
\]
With $v=1-t$, multiplying the derivative of the log elasticity with respect
to $v$ by $v(1-v)(1-v^\theta)>0$ gives $(\theta-1)-\theta v+v^\theta$, so the
elasticity is nonincreasing in $t$ precisely when
\[
(\theta-1)-\theta v+v^\theta\ge0.
\]
Indeed, this expression vanishes at $v=1$ and has derivative
$\theta(v^{\theta-1}-1)\le0$ on $(0,1)$, so the elasticity is
nonincreasing for $\theta\ge1$ and strictly decreasing for $\theta>1$.
For $\theta>1$ the elasticity tends to one as $t\downarrow0$ and to zero as
$t\uparrow1$, so it is a strictly decreasing bijection from $(0,1)$ onto
$(0,1)$.

The one-sided normal-shift family used in Table~\ref{tab:kS2} has the same
property. Write $\phi$ and $\Phi$ for the standard normal density and
distribution function. For shift $\eta>0$, put $z=\Phi^{-1}(t)$ and
$r(z)=\phi(z)/\Phi(z)$, the lower-tail inverse Mills ratio. Since
$G(t)=\Phi(z+\eta)$ and $g(t)=\phi(z+\eta)/\phi(z)$, the elasticity of this family is
\[
\frac{t g(t)}{G(t)}=\frac{r(z+\eta)}{r(z)}.
\]
Writing $a(z)=(\log r)'(z)=-z-r(z)$ gives
$a'(z)=-\{1-zr(z)-r(z)^2\}<0$; the expression in braces is the variance of
a standard normal variable conditional on being at most $z$. Therefore the
derivative with respect to $z$ of the logarithm of the displayed elasticity
is $a(z+\eta)-a(z)<0$. The usual Mills-ratio endpoint limits give values
one and zero as $t\downarrow0$ and $t\uparrow1$, respectively. Thus the
elasticity is again a strictly decreasing bijection from $(0,1)$ onto
$(0,1)$, which justifies the symmetric unanimity optimum in the homogeneous
normal row of Table~\ref{tab:kS2}.

Whenever every site's elasticity has this property, choose the common
value $\kappa$ so that $\prod_s t_s(\kappa)=\alpha/K$; the product varies
continuously from one to zero as $\kappa$ varies from zero to one. Strict
concavity makes this interior stationary point the unique global maximizer.
This includes the $\mathrm{Beta}(1,2)$ and $\mathrm{Beta}(1,3)$ cases. If
$g_s\equiv g$, symmetry and uniqueness give
$t_s=(\alpha/K)^{1/S}$. Under non-strict concavity, the common threshold is
a symmetric maximizer but need not be unique. Without the monotonicity
hypothesis, the displayed equality of elasticities remains only a
first-order condition. \hfill$\square$

\subsection{Plug-in strong-FWER control under estimated local
alternatives}\label{supp:proof_plugin}
Section~\ref{sec:plugin} of the main text summarizes the following result.

\begin{proposition}[Plug-in model-aware aggregation]\label{prop:plugin}
Suppose Assumption~\ref{ass:model} holds. Fix report maps as in
Theorem~\ref{thm:aggregator} and let the center
replace the alternative report distributions by probability-distribution estimates
$\hat\nu_{s,1}\ll\tilde\nu_{s,0}$ measurable with respect to a training
sigma-field jointly independent of the entire tested $p$-value array; the report
maps and thresholds may also be chosen using that training sigma-field.
Write $\mathrm{TV}(P,Q)=\sup_A|P(A)-Q(A)|$. For $\varepsilon\ge0$ and
$\zeta\in[0,1]$, suppose that, with probability at least $1-\zeta$ over the
training data,
$\sum_{s=1}^S\mathrm{TV}(\hat\nu_{s,1},\tilde\nu_{s,1})\le\varepsilon$, and that
$(K-1)\varepsilon+\zeta\le\alpha$. If the center runs the program of
Theorem~\ref{thm:aggregator} under $(\tilde\nu_0,\hat\nu_1)$, where
$\hat\nu_1=\otimes_s\hat\nu_{s,1}$, at level
$\alpha-(K-1)\varepsilon-\zeta$, with the convention that it makes no rejections when this level is zero,
then the true family-wise error rate~\eqref{eq:fwer} is at most $\alpha$, in
finite samples, at every configuration. At $K = 1$ no deflation is needed:
the plug-in test is valid at level $\alpha$ for any such estimate.  For
threshold reports, a uniform Dvoretzky--Kiefer--Wolfowitz
bound~\citep{dvoretzky1956asymptotic,massart1990tight} gives
$\varepsilon=O(Sn^{-1/2})$, up to logarithmic factors, from independent
held-out pilot samples of size $n$ per site, each an i.i.d.\ sample from the
site's alternative $p$-value distribution $G_s$; the exact bound and
multilevel extension are given below.
\end{proposition}

\emph{Proof.}
\emph{Step 1 (conditioning and validity under the working distribution).}
Condition on the training sigma-field. The report maps, thresholds, and
$\hat\nu_1 = \otimes_s \hat\nu_{s,1}$ are then fixed. Joint independence of
that sigma-field from the entire tested $p$-value array leaves the array's distribution
unchanged under the conditioning; consequently, the conditional report distribution
is the product distribution induced by the now-fixed report maps. Write
$\alpha'' = \alpha -
(K-1)\varepsilon - \zeta$. By assumption, $\alpha''\ge0$. If
$\alpha''=0$, the center makes no rejections by the proposition's convention, so validity is immediate. Hence suppose
$\alpha''>0$. By assumption,
$\hat\nu_{s,1}\ll\tilde\nu_{s,0}$ at every site, so
$\hat\nu_1\ll\tilde\nu_0$. Let $Q_h$ be the working distribution of the report
array at $h$: null coordinates carry the exact per-hypothesis distribution
$\tilde\nu_0 = \otimes_s \tilde\nu_{s,0}$, and alternative coordinates the estimated
$\hat\nu_1$.  Theorem~\ref{thm:aggregator}(a), applied with
$(\tilde\nu_0,\hat\nu_1)$ as the working two-group distribution, then gives
$Q_h\{\hat{\mathcal{R}} \cap \mathcal{H}_0(h) \neq \emptyset\} \le \alpha''$ for the
resulting procedure $\hat{\mathcal{R}}$ and every configuration $h$.

\emph{Step 2 (total-variation transfer to the true distribution).}
Put $A_h=\{\hat{\mathcal R}\cap\mathcal H_0(h)\ne\emptyset\}$. This event is measurable in
the report array and the center's internal randomization, whose distribution is
common to the true and working measures, so with $P_h$ the true distribution,
\[
\begin{aligned}
\big|P_h(A_h)-Q_h(A_h)\big|
\;&\le\; \mathrm{TV}(P_h, Q_h) \\
\;&\le\; \sum_{k :\, h_k = 1}
\mathrm{TV}\Big(\bigotimes_s \tilde\nu_{s,1},\ \bigotimes_s \hat\nu_{s,1}\Big)
\;\le\; |h| \sum_{s=1}^S \mathrm{TV}(\hat\nu_{s,1}, \tilde\nu_{s,1}),
\end{aligned}
\]
where $|h|=\sum_k h_k$. The second line follows from the
telescoping bound for products of measures, the null coordinates
contributing zero exactly. Whenever $\mathcal{H}_0(h) \neq \emptyset$ we have
$|h| \le K-1$, so on the good estimation event $\fwer_h(\hat{\mathcal{R}})
\le \alpha'' + (K-1)\varepsilon = \alpha - \zeta$; integrating over the
training data and bounding the complementary event's contribution by
$\zeta$ gives $\fwer_h \le \alpha$ at every $h$. At $K = 1$ the single
constraint, at $h = 0$, involves no alternative coordinate, so the transfer
term vanishes identically, no deflation is needed, and validity holds
at level $\alpha$ for any such estimate.

\emph{Step 3 (threshold reports with pilot-selected thresholds).}
For the threshold report,
$\mathrm{TV}(\hat\nu_{s,1}, \tilde\nu_{s,1}) = |\hat G_s(t_s) - G_s(t_s)|$
with $\hat G_s$ the empirical distribution function of the $n_s$ pilot
draws, independent $p$-values with distribution $G_s$.
The Dvoretzky--Kiefer--Wolfowitz bound on
$\sup_t|\hat G_s(t)-G_s(t)|$, with the constant of~\citet{massart1990tight},
followed by a union bound over sites,
bounds the total-variation sum by
$\varepsilon=\sum_{s=1}^S\{\log(2S/\zeta)/(2n_s)\}^{1/2}$
with probability at least $1-\zeta$ for $0<\zeta<1$, even for thresholds
chosen from the same pilots.

\emph{Step 4 (finite-range and continuous reports).}
Write $\varepsilon_s$ for a bound on site $s$'s total-variation error, and
take $\varepsilon=\sum_s\varepsilon_s$ when the bounds hold simultaneously.
For an $m$-level report with cells fixed in advance or chosen on a
separate split, the total variation is the maximum of
$|\hat\nu_{s,1}(A) - \tilde\nu_{s,1}(A)|$ over the $2^m$ cell unions $A$,
and Hoeffding's inequality~\citep{hoeffding1963probability} with a union
bound over these and the sites gives
$\varepsilon_s = [\{m\log 2 + \log(2S/\zeta)\}/(2n_s)]^{1/2}$.
Interval
partitions chosen from the same pilot are covered by the
Dvoretzky--Kiefer--Wolfowitz inequality: on its uniform event, any union of cells has
discrepancy at most
$m\sup_t|\hat G_s(t)-G_s(t)|$, giving
$\varepsilon_s\le m\{\log(2S/\zeta)/(2n_s)\}^{1/2}$.
An empirical distribution attached to an unquantized continuous report is
singular with respect to the continuous null distribution; that case is
excluded by the stated absolute-continuity condition unless smoothing or a
common dominated model is supplied. \hfill$\square$

\section{The \texorpdfstring{$K$}{K}-hypothesis dual computation}\label{app:dual}

The ordered program~\eqref{eq:ordered_primal} below is stated for an abstract exchangeable two-group
experiment whose likelihood ratio is the observation itself; this setting covers both the $p$-value array and any induced report distribution.  Two
instances are used: $V_k=\Lambda_k$ of~\eqref{eq:aggLR}, whose likelihood ratio
is itself (Step~1 of Section~\ref{supp:proof_oracle}), and $V_k=L_k$ of
Section~\ref{sec:aware} for fixed reports.  The value of the program is then
the oracle power $\pi^*_w(\alpha)$ or the fixed-report optimum
$\sup_{\mathcal R}\Pi_w(T,\mathcal R)$; Lemma~\ref{lem:ordered_duality} writes it
$\Pi^*_w$.  Let
$V_1,\ldots,V_K$ be i.i.d.\ under a null distribution $\lambda_0$, with
alternative distribution $\lambda_1\ll\lambda_0$ satisfying
$\mathrm d\lambda_1/\mathrm d\lambda_0(v)=v$. On finite support, write
the corresponding masses as $q_0(v)$ and $q_1(v)=v q_0(v)$. Use auxiliary
tie-breakers to sort a realized tuple as
$v_{(1)}\ge\cdots\ge v_{(K)}$. Write $e_d(I)$ for the
elementary symmetric polynomial of degree $d$ in the values indexed by
$I$, set $e_0(I)=1$ and $e_d(I)=0$ for $d<0$ or $d>|I|$, abbreviate
$e_d=e_d(\{1,\ldots,K\})$, and put $p_l=\prod_{i=1}^l v_{(i)}$.

The primal makes the action randomization explicit. For each sorted
tuple $v$, let $r_l(v)$ be the conditional probability of rejecting exactly
the top $l$ hypotheses, after averaging over internal randomization (including
the auxiliary tie-breakers), for $l=0,\ldots,K$. Two families of coefficients
appear, objective coefficients ($a$) and family-wise error coefficients ($b$);
a lowercase letter refers to a single hypothesis position and the matching
uppercase letter to its cumulative sum over the top $l$ positions, normalized
as in~\eqref{eq:cumulative_coeffs}. With the
per-position coefficients $a_{\gamma,k}$ and $b_{\gamma,k}$ of the Notation
paragraph of Section~\ref{supp:proofs}, derived below, put
\begin{equation}\label{eq:cumulative_coeffs}
 A_l=\frac{1}{K!}\sum_{i\le l}a_{w,i},
 \qquad
 B_{j,l}(v)=e_{K-j}-p_l e_{K-l-j}(\{l+1,\ldots,K\}),
 \qquad
 \alpha_j=\binom Kj\alpha,
\end{equation}
for $j=1,\ldots,K$ and $l\ge1$, with
$a_{w,k}=\sum_\gamma w(\gamma)a_{\gamma,k}$ the prior-mixed objective
coefficients and $A_0=B_{j,0}=0$. The two indices of $B_{j,l}$ are the number
of true hypotheses and the number of rejections, and
$\sum_{k\le l}b_{\gamma,k}=\gamma!\,(K-\gamma)!\,B_{K-\gamma,l}$, where the index
$\gamma=K-j$ counts alternatives. Under the point prior at
$\gamma=K$, $A_l=(l/K)e_K$, because $a_{K,k}=(K-1)!\,e_K$ does not depend
on $k$.

\emph{Derivation of the coefficients.} Fix a configuration $h$ with true
hypothesis set $\mathcal H_0$, $|\mathcal H_0|=j$. By the two-group factorization
and the per-coordinate density-ratio identity
$\mathrm d\lambda_1/\mathrm d\lambda_0=v$, the density of the tuple at
configuration $h$ relative to the global-null product distribution is
$\prod_{k\notin\mathcal H_0}v_k$. Rejecting the top $l$ sorted coordinates
makes a false rejection exactly when $\mathcal H_0$ meets the top $l$
positions. Summing the density over all $\binom Kj$ configurations gives
$\sum_{|\mathcal H_0|=j}\prod_{k\notin\mathcal H_0}v_k=e_{K-j}$, the
elementary symmetric polynomial in all $K$ sorted values, since the
products run over all $(K-j)$-subsets. The configurations making no false
rejection are those with $\mathcal H_0\subseteq\{l+1,\ldots,K\}$; for such
a configuration the top $l$ coordinates are all alternative, contributing
the prefix product $p_l$, and the remaining $K-l-j$ alternative
coordinates form a subset of the bottom $K-l$ positions, so these configurations
contribute $p_l\,e_{K-l-j}(\{l+1,\ldots,K\})$. Subtracting yields
$B_{j,l}$. Nonnegativity holds termwise: every monomial of
$p_l\,e_{K-l-j}(\{l+1,\ldots,K\})$ is the monomial of $e_{K-j}$ indexed by
a $(K-j)$-subset containing the top $l$ positions, and each such monomial
appears exactly once in each sum.  Differencing in $l$ gives the per-position
coefficients: with $\gamma=K-j$,
\[
 b_{\gamma,l}=\gamma!\,(K-\gamma)!\,(B_{K-\gamma,l}-B_{K-\gamma,l-1})
 =\gamma!\,(K-\gamma)!\,p_{l-1}\,e_{\gamma-l+1}(\{l+1,\ldots,K\})\ge0,
\]
the error that the $l$th rejection adds, because
$e_{\gamma-l+1}(\{l,\ldots,K\})=v_{(l)}\,e_{\gamma-l}(\{l+1,\ldots,K\})
+e_{\gamma-l+1}(\{l+1,\ldots,K\})$.  Hence $B_{j,l}$ is nondecreasing in $l$, the
monotonicity used in Proposition~\ref{prop:enclosure}.

The objective coefficients follow in the same way.  At a configuration whose
alternative set $\mathcal A$ has $|\mathcal A|=\gamma$, the density relative to the
global null is $\prod_{k\in\mathcal A}v_k$, and rejecting the top $l$ sorted
coordinates makes $|\mathcal A\cap\{1,\ldots,l\}|$ true rejections.  Summing over
the $\binom K\gamma$ such configurations counts each position $k\le l$ once
for every $\gamma$-subset containing it, so
\[
 \sum_{|\mathcal A|=\gamma}|\mathcal A\cap\{1,\ldots,l\}|\prod_{k\in\mathcal A}v_k
 =\sum_{k\le l}v_k\,e_{\gamma-1}(v_{-k}).
\]
The definition of $\Pi_\gamma$ divides this sum by
$\gamma\binom K\gamma=K!/\{(\gamma-1)!\,(K-\gamma)!\}$, so
\[
 \Pi_\gamma=\E_0\Bigl\{\sum_l r_l(V)\,(K!)^{-1}\sum_{k\le l}a_{\gamma,k}(V)\Bigr\},
 \qquad
 a_{\gamma,k}(v)=(\gamma-1)!\,(K-\gamma)!\,v_k\,e_{\gamma-1}(v_{-k}),
\]
and mixing over $w$ gives the coefficient $A_l$ of~\eqref{eq:cumulative_coeffs}.

The derivation uses only permutation
equivariance, the two-group factorization, and the density-ratio identity;
it requires neither continuity of $\lambda_0$ nor
Assumption~\ref{ass:reg}. The resulting coefficient structure is that of
\citet[Lemma~1]{dubey2026esp}, rederived here for a general dominated
report distribution.

The ordered program, with its pointwise constraints
understood under the sorted null distribution, is
\begin{equation}\label{eq:ordered_primal}
\begin{aligned}
\text{maximize}_{\{r_l(v)\}}\quad
 &\E_0\!\left[\sum_{l=1}^K r_l(V)A_l(V)\right] \\
\text{subject to}\quad
 &\E_0\!\left[\sum_{l=1}^K r_l(V)B_{j,l}(V)\right]
   \le \alpha_j, &&j=1,\ldots,K,\\
 &\sum_{l=0}^K r_l(v)=1,\qquad r_l(v)\ge0,
 &&\text{almost surely}.
\end{aligned}
\end{equation}
Permutation equivariance makes all configurations with $j$ true hypotheses have the
same family-wise error, so the summed constraint above is equivalent to the
individual strong-control constraints. At the point prior the objective is
$\E_0[\sum_l r_l(V)(l/K)e_K(V)]$, the density relative to the global null
under the full alternative being $e_K(V)$.

\begin{lemma}[Attainment and strong duality for randomized ordered policies]
\label{lem:ordered_duality}
For every dominated two-group distribution just described and every
$\alpha\in(0,1)$, the randomized ordered-policy
program~\eqref{eq:ordered_primal} attains its maximum and has value
\[
\Pi^*_w
=\min_{\mu\ge0}\left\{
\sum_{j=1}^K\alpha_j\mu_j
+\E_0\max\!\left(0,\max_{1\le l\le K}\psi_l(V)\right)
\right\},
\qquad
\psi_l=A_l-\sum_{j=1}^K\mu_jB_{j,l},
\]
in the notation of~\eqref{eq:cumulative_coeffs}; $\psi_l$ is the dual branch
of the policy that rejects the $l$ largest likelihood ratios. The prior enters the display only through $A_l$; the family-wise error terms are free of $w$.
The dual minimum is attained, and every minimizing multiplier satisfies
$\|\mu\|_1\le1/\alpha$. Whenever the symmetrization and
likelihood-ratio-ordering reduction applies, as established in the proofs of
Proposition~\ref{prop:oracle} for the centralized oracle and of
Theorem~\ref{thm:aggregator} for fixed reports, this value is the optimum over all randomized
procedures satisfying~\eqref{eq:fwer}, not merely over ordered policies.
\end{lemma}

On finite support, which includes every report distribution, the lemma is
ordinary linear-programming duality; the proof below covers general
dominated laws such as the oracle's.

\begin{proof}
Attainment follows from weak-* compactness of the ordered randomized
policies and weak-* continuity of the objective and error functionals.  On
the global-null probability space of the sorted likelihood-ratio tuple, let
\[
 \mathcal P=\left\{r=(r_0,\ldots,r_K)\in
 (L^\infty)^{K+1}:r_l\ge0,\ \sum_{l=0}^K r_l=1\ \text{a.s.}\right\}
\]
be the set of ordered randomized policies. It is weak-* compact: it is a
weak-* closed subset of the product of
$K+1$ unit balls of $L^\infty$, which is weak-* compact by the
Banach--Alaoglu theorem. The coefficient functions $A_l$ and $B_{j,l}$ are
in $L^1$: by the coefficient derivation above, each can be expanded as a
finite sum of nonnegative likelihood-ratio monomials, and every such
monomial has finite null expectation. Therefore the objective and error
functionals
\[
 \Pi(r)=\E_0\sum_l r_lA_l,\qquad
 \fwer_j(r)=\E_0\sum_l r_lB_{j,l}
\]
are weak-* continuous; here $\fwer_j(r)$ sums the error over the
$\binom Kj$ configurations with $j$ true hypotheses, which is why its bound is
$\alpha_j=\binom Kj\alpha$. The feasible subset
$\{r\in\mathcal P:\fwer_j(r)\le\alpha_j,\ j=1,\ldots,K\}$ is
consequently weak-* compact, so its maximum $\Pi^*$ is attained.

For completeness, strong duality follows from finite-dimensional separation
of the compact convex moment image. Write
$\boldsymbol\alpha=(\alpha_1,\ldots,\alpha_K)$ and
$\fwer(r)=(\fwer_1(r),\ldots,\fwer_K(r))$, and let
\[
 \mathcal Z=\{(\fwer(r),\Pi(r)):r\in\mathcal P\}
 +(\mathbb R_+^K\times\mathbb R_-),
 \qquad \mathbb R_-=(-\infty,0].
\]
The set $\mathcal Z$ is closed and convex because the moment image is compact
and convex and the added cone is closed and convex. If $r^*$ is a primal
maximizer, then
$(\boldsymbol\alpha,\Pi^*)
 =(\fwer(r^*),\Pi^*)+(\boldsymbol\alpha-\fwer(r^*),0)\in\mathcal Z$; but
$(\boldsymbol\alpha,\Pi^*+\varepsilon)\notin\mathcal Z$ for every
$\varepsilon>0$, since such
membership would produce a feasible policy with objective greater than
$\Pi^*$. Thus $(\boldsymbol\alpha,\Pi^*)$ is a boundary point.
A supporting hyperplane supplies a nonzero vector $(-\mu,\eta)$ such that
\[
 -\mu^\top(x-\boldsymbol\alpha)+\eta(y-\Pi^*)\le0,\qquad (x,y)\in\mathcal Z.
\]
Closure of $\mathcal Z$ under coordinatewise increases in $x$ and decreases
in $y$ forces $\mu\ge0$ and $\eta\ge0$. The no-rejection policy has
$\fwer_j=0<\alpha_j$ for every $j$, so $\eta=0$ would force $\mu=0$,
contradicting
the nonzero supporting normal. Hence $\eta>0$; normalize it to one.
Taking $(x,y)=(\fwer(r),\Pi(r))$ yields
\[
 \sum_j\mu_j\alpha_j+\sup_{r\in\mathcal P}
 \Big\{\Pi(r)-\sum_j\mu_j\fwer_j(r)\Big\}\le\Pi^*.
\]
The reverse inequality is weak duality. Pointwise maximization over the
simplex $(r_0(v),\ldots,r_K(v))$ makes the supremum equal to the expectation
of $\max\{0,\max_l\psi_l(V)\}$: since there are finitely many measurable
branches, choosing the smallest index among the maximizers gives a measurable
selector. This proves the displayed identity and shows that this $\mu$
attains the dual minimum.

The dual objective $\mathcal D$ satisfies
$\mathcal D(\mu)\ge\alpha\|\mu\|_1$ because
$\alpha_j\ge\alpha$ and the pointwise maximum includes zero.  Moreover,
$\mathcal D(0)=1$ at every prior: the $a_{w,i}$ are nonnegative, so $\max_lA_l=A_K$, and
$\sum_ka_{\gamma,k}=\gamma!\,(K-\gamma)!\,e_\gamma$ together with
$\E_0e_\gamma=\binom K\gamma$ gives
$\E_0A_K=\sum_\gamma w(\gamma)=1$. Every dual minimizer has objective no
larger than $\mathcal D(0)$ and therefore lies in $\|\mu\|_1\le1/\alpha$. The final assertion follows from
the symmetrization and likelihood-ratio-ordering reductions proved for the
two experiments cited in the statement.
\end{proof}

For $\lambda_0$ supported on $M$ distinct atoms $v_1,\ldots,v_M$ and under the
point prior $w(2)=1$, the $K=2$ dual becomes
\[
\Pi^*_2(\lambda_0)=\min_{\mu_1,\mu_2\ge0}\left\{
2\alpha\mu_1+\alpha\mu_2
+\sum_{a,b}q_0(v_a)q_0(v_b)\max(0,\psi_1,\psi_2)
\right\},
\]
with
\[
\psi_1=\tfrac12v_{(1)}v_{(2)}-\mu_1v_{(2)}-\mu_2,
\qquad
\psi_2=v_{(1)}v_{(2)}-\mu_1\{v_{(1)}+v_{(2)}\}-\mu_2,
\]
where $v_{(1)}\ge v_{(2)}$ are $v_a$ and $v_b$ in decreasing order.

On an atomic distribution, which covers every report distribution, the
program and its dual are finite linear programs, and every
report-distribution value in Section~\ref{supp:numerics} is computed this way.
The dual is convex and piecewise linear.  Introducing one epigraph variable
for each multiset of $K$ atoms, that is, for each sorted tuple, makes it a
linear program indexed by the weak compositions of $K$ into $M$ parts (count
vectors $(n_1,\ldots,n_M)\ge0$ summing to $K$, of which there are
$\binom{K+M-1}{K}$) rather than by the $M^K$ labeled tuples; a multiset in which atom $v_a$ appears $n_a$
times has null weight $K!\prod_a q_0(v_a)^{n_a}/n_a!$.  Agreement between feasible
primal and dual values, together with small constraint residuals, checks
numerical optimality.  Generic descent heuristics are unreliable on
distributions with few atoms because they can stop at a kink above the
minimum; the linear programs avoid that failure.  When a continuous oracle
distribution is replaced by a fine binning, the resulting value is a
numerical discretization of the oracle benchmark, not an exact
continuous-distribution certificate.

Implementing the optimal policy needs the primal as well as the dual.  The
tie-breakers $\xi$ only order tied coordinates;
the action randomization comes from the primal, which samples
$l\in\{0,\ldots,K\}$ with probabilities $\{r_0(v),\ldots,r_K(v)\}$, after which
the top $l$ are rejected.  The dual alone returns the value and
multipliers, not this policy.

At $K=1$ the dual reduces to
\[
\min_{\mu\ge0}\left\{\alpha\mu+
\sum_a q_0(v_a)(v_a-\mu)_+\right\},
\]
the Neyman--Pearson power on the atomic distribution. On a constant likelihood ratio,
for example, the primal rejects with probability $\alpha$.

\section{Numerical study}\label{supp:numerics}

This section gives the numerical study summarized in Section~\ref{sec:numerics}
of the main text and expands the four claims listed at the start of that
section: finite-report strictness at $K=2$ (Section~\ref{sec:numerics_mult},
and Section~\ref{sec:numerics_prior} at other priors), the inverse-square rate
(Section~\ref{sec:numerics_rate}), the growing benefit of joint aggregation
(Section~\ref{sec:numerics_adv}), and the robustness scopes of the two rules
(Section~\ref{sec:numerics_sim}); Section~\ref{sec:numerics_ext} adds the number
and heterogeneity of sites.  Unless stated otherwise, the settings are those
of Section~\ref{sec:numerics_summary}: homogeneous $\mathrm{Beta}(1,2)$ sites,
$S=2$, $\alpha=0.05$, the point prior, and losses reported as relative
deficits.  Finite report-distribution calculations are deterministic, with
explicit primal--dual checks.  Outer searches over the report thresholds and
continuous-oracle values remain numerical unless an exact-rational
certificate is stated (Section~\ref{supp:benchmarks}).  Supporting tables and
computational details are in Section~\ref{supp:additional}.

At a single hypothesis the loss from a threshold report is small in the two Beta examples: for
$K=1$, $S=2$, and $\alpha=0.05$, exact four-cell calculations give relative
deficits of $1.1\%$ and $1.9\%$ under homogeneous $\mathrm{Beta}(1,2)$ and
$\mathrm{Beta}(1,3)$ alternatives.  Exact-rational branch and bound over the
thresholds certifies that no thresholds attain more than $2\times10^{-4}$
above each reported power; it bounds the optimal power, not the location of
the optimal thresholds.  A
one-sided normal-shift example with an unbounded likelihood ratio gives
$9.7\%$, but its model lies outside Assumption~\ref{ass:reg}; the full table
is in Section~\ref{supp:supporting_numerics} and the closed-form benchmarks
in Section~\ref{supp:benchmarks}.

\subsection{Multiplicity and the separation of aggregators at \texorpdfstring{$K=2$}{K=2}}\label{sec:numerics_mult}

Multiplicity raises the cost of one bit sharply: the optimized one-bit report
loses $1.1\%$ of oracle power at $K=1$ but $8.5\%$ at $K=2$, and the
null-calibrated rule, reading the reports marginally, loses $45.8\%$
(Table~\ref{tab:k2}).  For two homogeneous $\mathrm{Beta}(1,2)$ sites at
$\alpha=0.05$ and one hypothesis, a center holding both $p$-values runs
the Neyman--Pearson test on their product likelihood ratio, rejecting when
$4(1-u_1)(1-u_2)$ exceeds the constant that gives the rejection region null
probability $0.05$; its power is $\pi^*_1=0.15954$.  The selected threshold-report rule
rejects when both sites report significance, which has null probability
exactly $0.05$ at $t=\sqrt{0.05}=0.224$ and power $0.15778$, so $0.00176$ of
power is lost, or $1.1\%$ of $\pi^*_1$.

With two hypotheses neither the
oracle nor the report rule reduces to a single threshold.  The benchmark is the value of the program that
maximizes average power subject to a family-wise error rate of at most $0.05$
under every configuration of true and false hypotheses, $\pi^*_2=0.15646$.  On the
report side, fixing the thresholds $t=(t_1,t_2)$ turns the transmitted array
into four bits, one per site per hypothesis, whose null distribution is known
exactly.  The center solves the same program under that finitely supported
distribution, and optimizing over $(t_1,t_2)$ attains $0.14312$ at the
selected unequal thresholds $(0.288,0.328)$.  The asymmetry is small: the best
common threshold, $0.309$, loses $8.6\%$ (Section~\ref{supp:benchmarks}).
The loss is $0.01334$, the $0.0133$ of Table~\ref{tab:k2}, or $8.5\%$ of
$\pi^*_2$; these $K=1$ and $K=2$ figures are those summarized in
Section~\ref{sec:numerics_summary} of the main text.

At $K=2$ the two rules separate: on separately optimized one-bit reports the
null-calibrated rule loses $45.8\%$ and $43.3\%$ of $\pi^*_2$, against $8.5\%$
and $6.2\%$ for the model-aware rule, because model-aware aggregation uses the
joint report distribution whereas the null-calibrated rule stays marginal;
Table~\ref{tab:k2} compares both with the continuous oracle.  For each Beta
model, exact-rational branch and bound over the thresholds proves that no
threshold pair gives model-aware power more than $2\times10^{-4}$ above that
of a policy with rational action probabilities that meets every family-wise
constraint exactly.  The certificate is summarized in
Section~\ref{supp:benchmarks}, and the complete audit is documented with the
code.

By Lemma~\ref{lem:oak_auto} and Theorem~\ref{thm:oracle_kink}(a) the strict loss applies to \emph{every}
finite-range componentwise report design in these models, not only to the
selected thresholds; Corollary~\ref{cor:enclosure_beta} certifies the
explicit witness at every minimizer of the continuous dual for both Beta
models, and Corollary~\ref{cor:beta_oak} does so for $\mathrm{Beta}(1,2)$ by an analytic
argument completed with computer-assisted exact rational arithmetic.

Reading the one-bit reports jointly recovers $81$--$86\%$ of the
large marginal deficit.  At $K=1$ the optimized count tails of the designs in Table~\ref{tab:kS2} attain the level exactly, so the null-calibrated rule loses no
aggregation power (Proposition~\ref{prop:marginalK1}); at $K=2$, on separately
optimized one-bit reports, its deficit is more than five times that of the
model-aware rule.

The full-$p$-value columns of Table~\ref{tab:k2} separate the cost of
quantization from the cost of marginal aggregation.
Given the full $p$-values, Bonferroni with the likelihood-ratio statistic,
the most powerful Bonferroni rule, loses $45.3\%$ and $42.5\%$ of $\pi^*_2$, so
one bit costs the best Bonferroni rule less than one percentage point: almost all of its deficit is the cost of marginal aggregation, which
joint aggregation removes.  Fisher--Bonferroni on the full $p$-values loses
$55.4\%$ and $56.8\%$, more than the one-bit null-calibrated rule, because for
$\mathrm{Beta}(1,\theta)$ alternatives the likelihood ratio orders the
evidence by $\prod_s(1-u_s)$, whereas Fisher's statistic orders it by
$\prod_s u_s$.  For $\mathrm{Beta}(1,\theta)$ with integer $\theta$, $u$ is the
minimum of $\theta$ independent uniforms, so $-\log(1-u)$ is exponential with
rate $\theta$ and $-\log u$ is a sum of independent exponentials with rates
$1,\ldots,\theta$; both powers are then explicit Gamma integrals, which
\texttt{full\_pvalue\_baselines.py} evaluates, each Fisher power by two
independent integrations.

\begin{table}[htbp]
\caption{Reading one-bit reports jointly loses $6$--$9\%$ of oracle power at
$K=S=2$, while every Bonferroni rule in the table loses $42$--$57\%$, even
with full $p$-values.  Deficits, $\pi^*_2$ minus attained power, for $K=S=2$, homogeneous
sites, $\alpha=0.05$, and the point prior.  The one-bit columns use
threshold reports, with each method selecting its own thresholds; the
full-$p$-value columns apply Bonferroni at level $\alpha/2$ to Fisher's
statistic or to the likelihood ratio of the full $p$-values, the latter being
the most powerful Bonferroni rule.  Exact-rational branch and bound certifies
each one-bit model-aware value within $2\times10^{-4}$ of the best attainable
threshold-report power; it does not locate the optimal thresholds.  The
$/\,\pi^*_2$ columns give the deficit as a percentage of $\pi^*_2$.  The
null-calibrated rule uses the unanimity thresholds $t_s=(\alpha/2)^{1/2}=0.158$
of Proposition~\ref{prop:alloc}; the model-aware thresholds are $(0.288,0.328)$
and $(0.247,0.247)$.}
\label{tab:k2}
\centering
\small
\begin{tabular}{lccccccccc}
\toprule
& & \multicolumn{4}{c}{one-bit threshold reports} & \multicolumn{4}{c}{full $p$-values, Bonferroni} \\
\cmidrule(lr){3-6}\cmidrule(lr){7-10}
& & \multicolumn{2}{c}{model-aware} & \multicolumn{2}{c}{null-calibrated}
& \multicolumn{2}{c}{Fisher} & \multicolumn{2}{c}{likelihood ratio} \\
\cmidrule(lr){3-4}\cmidrule(lr){5-6}\cmidrule(lr){7-8}\cmidrule(lr){9-10}
Alternative & $\pi^*_2$ & deficit & $/\,\pi^*_2$ & deficit & $/\,\pi^*_2$
& deficit & $/\,\pi^*_2$ & deficit & $/\,\pi^*_2$ \\
\midrule
$\mathrm{Beta}(1,2)$ & $0.1565$ & $0.0133$ & $8.5\%$ & $0.0716$ & $45.8\%$ & $0.0867$ & $55.4\%$ & $0.0709$ & $45.3\%$ \\
$\mathrm{Beta}(1,3)$ & $0.2870$ & $0.0178$ & $6.2\%$ & $0.1243$ & $43.3\%$ & $0.1631$ & $56.8\%$ & $0.1219$ & $42.5\%$ \\
\bottomrule
\end{tabular}
\end{table}

\subsection{Scaling in the number of sites and heterogeneity}\label{sec:numerics_ext}

From two to eight sites, the one-bit loss at one hypothesis grows with the
number of sites: for
homogeneous $\mathrm{Beta}(1,2)$ sites the relative deficit of the best
common-threshold report rises from $1.1\%$ at $S=2$ to $18.4\%$ at $S=8$, where
it stays through $S=10$.  Unequal thresholds were not searched for $S\ge3$, so
these values bound the threshold-class deficit from above.  Deterministic calculations confirm Proposition~\ref{prop:marginalK1}
at every $S$ from $2$ to $10$: the null-calibrated count rule and the
model-aware program agree to six digits, at a common threshold and with no
randomization, so this deficit is the same for both.  Among count rules, however, unanimity
(cutoff $c=S$) separates from the optimal cutoff.  Its count cutoff
$c=S$ is optimal up to $S=7$; at
$S=8$ the optimum drops to $c=7$ and unanimity falls behind, and by $S=10$ it
reaches only $0.500$ against $0.525$, a relative deficit of $22\%$ against
$18\%$. The additional loss from unanimity at larger $S$ reflects its count
cutoff; the optimized marginal count rule still matches model-aware
aggregation. Table~\ref{tab:scaleS} gives
the complete finite-$S$ table and Proposition~\ref{prop:largeS} the
large-$S$ result showing that the two designs must separate and that both relative deficits ultimately vanish.

Heterogeneity also matters at $K=1$: solving the elasticity condition in Proposition~\ref{prop:alloc} for
$\mathrm{Beta}(1,2)$ and $\mathrm{Beta}(1,4)$ sites gives unanimity thresholds
$(0.354,0.141)$ at $\alpha=0.05$ and power $0.266$, compared with $0.253$ at the
common threshold $\sqrt\alpha$, a $5\%$ relative gain from assigning the tighter threshold to the
more informative site.

\subsection{Refining the report: the inverse-square rate}\label{sec:numerics_rate}

The predicted inverse-square rate is visible at practical resolutions;
Figure~\ref{fig:theory_inverse_square}(a,b) of the main text shows it for $S=2$
at $K=1$ and $K=2$.  Despite oscillations at small $m$, which come from where
the decision boundary falls relative to the cell edges, $m^2\delta_K$ stays
between $0.07$ and $0.19$ for $m\le64$ (Table~\ref{tab:rate}).  The fitted exponent is
$2.00$ at $K=1$ over $m=32,64,\ldots,1024$; at $K=2$, a deterministic
calculation over seven resolutions from $m=8$ to $64$ gives $1.998$
(Section~\ref{supp:benchmarks}).  The $m=2$ point is the equal-width, not threshold-optimized, binary
report, and the display illustrates only that the equal-width design attains
the class-optimal rate; the lower bounds of Proposition~\ref{prop:converse}(b)
and Theorem~\ref{thm:oracle_kink}(b) cover arbitrary measurable cells.  As a
design guide, the equal-width $m=8$ report, three bits, leaves relative
deficits of $0.72\%$ at $K=1$ and $0.81\%$ at $K=2$ (oracle powers
$\pi^*_1=0.1595$ and $\pi^*_2=0.1565$); at $K=2$, the two-bit and one-bit
equal-width reports leave $2.9\%$ and $28\%$.

\subsection{Nondegenerate priors: the active-kink witness, deficits, and rate}\label{sec:numerics_prior}

The explicit witness behind the loss is present at every prior examined,
including the nondegenerate ones, which put weight on more than one number of
false hypotheses, so its geometry is not an artifact of the full-alternative
objective.  No result depends on Table~\ref{tab:prior}: every row has
$w(2)>0$, so Lemma~\ref{lem:oak_auto} already gives OAK, strict loss, and the
rate.  Because the kink surfaces of the dual integrand move with $w$, the
table locates the witness explicitly at each prior and reports its null mass
and slope jump, which enter the constants of Theorem~\ref{thm:oracle_kink}(b)
through the selected witness.  For $K=S=2$ homogeneous $\mathrm{Beta}(1,2)$ sites at $\alpha=0.05$,
we solve the prior-weighted dual for a range of $w$, locate its minimizer,
and scan for the explicit witness: a crossing of two positive branches at an
interior value of one coordinate, on a set of values of the other coordinate
carrying positive null mass, with a slope jump bounded away from zero.

\begin{table}[htbp]
\caption{The explicit active-kink witness persists at every prior examined.
Among the three nondegenerate priors, its mass and minimum slope jump
increase as weight moves toward a single false hypothesis.
Prior-weighted $K=S=2$ oracle at $\alpha=0.05$, homogeneous
$\mathrm{Beta}(1,2)$ sites, as the prior $w$ moves from the full-alternative
point mass toward the sparse end; $\pi^*_w(\alpha)$ is the oracle value and
$\mu^*$ the dual minimizer.  The last two columns are the explicit
witness: the null mass of the coordinate values admitting an interior
two-positive-branch crossing, and the smallest slope jump on that set.  The
$w(1)=0$ row is the point prior, whose explicit witness is certified over the
whole dual enclosure in Section~\ref{supp:proof_enclosure}; the $w(1)=0.50$ row
is the flat prior, at which Section~\ref{supp:prior_numerics} also certifies the
witness.  At the point prior the witness is the crossing along the smaller
ratio of Lemma~\ref{lem:k2_witness}, so that row shows the certificate instead
of scan values; $w(2)=1-w(1)$.}
\label{tab:prior}
\centering
\begin{tabular}{ccccc}
\toprule
$w(1)$ & $\pi^*_w(\alpha)$ & $\mu^*$ & witness mass & slope jump \\
\midrule
$0.00$ & $0.1565$ & $(1.256,\,0.226)$ & certified & certified \\
$0.25$ & $0.1403$ & $(0.986,\,0.520)$ & $0.0098$ & $0.076$ \\
$0.50$ & $0.1259$ & $(0.714,\,0.848)$ & $0.0204$ & $0.143$ \\
$0.75$ & $0.1134$ & $(0.458,\,1.165)$ & $0.0290$ & $0.196$ \\
\bottomrule
\end{tabular}
\end{table}

The oracle value $\pi^*_w(\alpha)$ decreases across the displayed priors as
weight moves toward $\gamma=1$. Among the three nondegenerate priors in
the table, both the numerical witness mass and the smallest slope jump
increase. At the flat prior for $\mathrm{Beta}(1,3)$, we obtain
$\pi^*_w=0.2282$, minimizer
$(1.061,\,1.546)$, witness mass $0.0127$, and slope jump $0.106$.

The explicit witness is also certified at every flat-prior dual minimizer
for $\theta=2,3$: Section~\ref{supp:prior_numerics} gives rational boxes and
positive-mass windows on which all witness inequalities hold uniformly.
The prior-weighted dual and numerical scan are implemented in
\texttt{prior\_weighted\_K2.py}; the arithmetic certificate is
\texttt{oak\_\allowbreak verified\_\allowbreak enclosure.py}.

In these examples the relative deficit is larger under the flat prior.
Recomputing the threshold-report deficits of Table~\ref{tab:k2} with the same
threshold search and the prior-weighted finite program gives $9.6\%$ for
$\mathrm{Beta}(1,2)$ and $9.3\%$ for $\mathrm{Beta}(1,3)$, against $8.5\%$ and
$6.2\%$ at the point prior.  The selected thresholds barely move, although the
objective depends on the prior: $(0.297,0.320)$ for $\mathrm{Beta}(1,2)$ and
$(0.247,0.247)$ for $\mathrm{Beta}(1,3)$, against $(0.288,0.328)$ and
$(0.247,0.247)$ at the point prior. Repeating the
equal-width-report calculation of Section~\ref{sec:numerics_rate} under the
flat prior returns a fitted exponent at $K=2$ of $1.96$ over $m\in\{6,8,10,12\}$,
against $2.03$ at the point prior on the same grid, with $m^2\delta_w$ again
bounded and separated from zero.

\subsection{The federated multiplicity advantage}\label{sec:numerics_adv}

Joint model-aware aggregation of one-bit reports gains $62$--$69\%$ over
the null-calibrated rule (marginal aggregation) at $K=2$ and $330$--$419\%$ at $K=8$.  The deterministic
comparison in Figure~\ref{fig:theory_inverse_square}(c) of the main text uses threshold reports from two or three homogeneous $\mathrm{Beta}(1,2)$ sites, with each method selecting
its own common threshold, and the gains are over the one-bit marginal power,
which the four marginal rules compared below share.  Against the most powerful Bonferroni rule applied to the
full $p$-values (Section~\ref{sec:numerics_mult}), the
gains are $57$--$67\%$ at $K=2$ and $320$--$417\%$ at $K=8$
(Table~\ref{tab:advantage}).  The one-bit marginal rule comes within $4\%$ of
that full-$p$-value rule at every $K$ displayed, so the gains come from joint aggregation rather than from a quantization loss in the marginal rule.  At these parameter values the four marginal rules agree to displayed
precision: the null-calibrated rule of Section~\ref{sec:marginal} at its own
optimized common threshold, which applies a Bonferroni
step~\citep{dunn1961multiple}; its Holm~\citep{holm1979simple} and Simes-based
Hommel~\citep{hommel1988stagewise} variants; and the null-calibrated rule at
the unanimity threshold.  This agreement is a numerical coincidence, not an
identity.
Table~\ref{tab:advantage} gives the full table, and
Section~\ref{supp:supporting_numerics} a contrasting example.

For $S=2$ and a common threshold, the one-bit loss against the oracle also exceeds its $K=2$ value at each larger $K$ examined.  The continuous
oracle is computed only for $K\le2$, where its dual expectation reduces to a
one-dimensional quadrature, but the model-aware power of the two-bit
equal-width report is attained by a valid procedure on the full $p$-values
and so cannot exceed the oracle power.  For $S=2$ the last two columns of
Table~\ref{tab:advantage} therefore bound the one-bit relative deficit from
below: $6.0\%$ at $K=2$, where the deficit computed directly with a common threshold is $8.6\%$, and
$9.0\%$, $10.5\%$, and $11.6\%$ at $K=4$, $6$, and $8$, each above the
$8.6\%$ computed directly at $K=2$.  At
$S=3$ the per-hypothesis two-bit report, an unordered set of site levels,
takes $\binom63=20$ values instead of $\binom52=10$, which puts the
sum over weak compositions of $K=8$ into twenty parts (Section~\ref{app:dual})
beyond this implementation, so the bound is
reported for $S=2$ only.

The aggregation advantage also depends on the objective: under the flat
prior the same comparison gives smaller model-aware gains, $34\%$ at $K=2$
rising to $157\%$ at $K=8$ for $S=2$, and $25\%$ to $102\%$ for $S=3$.  The
decrease is consistent with the gain coming from pooling evidence across
simultaneously false hypotheses: a prior that puts weight on $\gamma=1$
also counts configurations with no second alternative to pool with. In these comparisons, joint aggregation of one-bit reports still outperforms marginal aggregation by a
margin growing in $K$, while the particular multiple depends on the prior.
The four marginal comparators (the null-calibrated rule with Bonferroni,
Holm, or Hommel steps, and the unanimity rule) still agree numerically at every prior and every $K$ examined.

\subsection{A synthetic multi-site study: resolution, misspecification, and
dependence}\label{sec:numerics_sim}

In a synthetic study with three heterogeneous normal-shift sites, the
model-aware rule approaches the oracle as the equal-width report is refined, while
unused null level, Fisher's ordering of the evidence, and the Bonferroni step
limit null-calibrated power.  There are $S=3$ sites, $K=2$ hypotheses, and $\alpha=0.05$, site $s$ having one-sided shift $\eta_s$ with
$\eta=(\eta_1,\eta_2,\eta_3)=(1.0,1.5,2.0)$.  The
aggregate log likelihood ratio is exactly Gaussian with variance
$\sigma^2=\sum_s\eta_s^2=7.25$, giving oracle power $0.852$.
Table~\ref{tab:simpower} gives deterministic finite-program power across
report resolutions for the model-aware rule and four marginal rules: Fisher's
statistic or the report likelihood ratio, each calibrated against its exact
null distribution, followed by Bonferroni or Holm; within each statistic the
two columns differ only in the correction.  At $m=8$, for example,
Fisher--Bonferroni attains $0.513$; the report likelihood ratio raises this to
$0.661$, and adding Holm's step to $0.748$, against $0.797$ for the model-aware
rule.  Figure~\ref{fig:theory_inverse_square}(d) of the main text displays the
model-aware rule and the two Bonferroni columns, labeled Null-cal.\ (Fisher)
and Null-cal.\ (LR) there.  The
displayed protocols are not totally ordered by refinement, because the
optimized binary maps are not nested in the equal-width maps; consequently,
their powers need not be monotone across rows.  Data-generating-process, threshold, primal--dual, and discretization checks are given in Section~\ref{supp:sim}.

Under the flat prior, with the thresholds re-optimized for that objective,
the pattern is unchanged and the relative deficits are again larger: the
oracle moves to $0.811$, and the model-aware powers are $0.670$, $0.629$, $0.736$, and $0.811$ across the same
four resolutions, against $0.728$, $0.703$, $0.797$, and $0.852$ at the point
prior.  The relative deficit is $17.4\%$ against $14.6\%$ for the threshold report and $9.3\%$ against $6.5\%$
at $m=8$, a shift in the same direction as in the Beta calculations of
Section~\ref{sec:numerics_prior}.  Re-optimizing the thresholds still matters:
evaluated under the flat objective, the point-prior design yields $0.628$,
against $0.670$ when re-optimized and $0.624$ for the null-calibrated rule.

\begin{table}[htbp]
\caption{The null-calibrated rule keeps the family-wise error rate at or below
$0.05$ in every scenario, at lower power ($0.45$--$0.62$ against $0.70$--$0.74$),
while that of the model-aware rule reaches $0.054$--$0.055$
when effects are underestimated and $0.064$--$0.067$ under within-site
correlation.  True strong FWER and full-alternative power in the synthetic study
($S=3$, $K=2$, $\alpha=0.05$, true shifts $\eta=(1.0,1.5,2.0)$); $\hat\eta$ is the
shift vector the center assumes and $\rho$ the within-site correlation between
the two hypotheses' $z$-statistics.
Misspecification rows are exact finite sums: model-aware reports stay fixed,
while null-calibrated thresholds are reselected under the assumed model.
Correlation rows use
$4\times10^6$ Rao--Blackwellized replications (standard errors at most
$2.3\times10^{-4}$).  The null-calibrated $m=4$ rate is below nominal because
three sites and four equally likely cells admit only null levels that are
multiples of $1/64$.  FWER is the maximum over configurations, $m=4$ is the
equal-width four-level report, and the overestimation rows are discussed in
Section~\ref{supp:supporting_numerics}.}
\label{tab:simrobust}
\centering
\begin{tabular}{llcccc}
\toprule
& & \multicolumn{2}{c}{model-aware} & \multicolumn{2}{c}{null-calibrated} \\
\cmidrule(lr){3-4}\cmidrule(lr){5-6}
scenario & report & FWER & power & FWER & power \\
\midrule
correct, independent & threshold & $0.0500$ & $0.728$ & $0.0494$ & $0.624$ \\
$\hat\eta=0.75\,\eta$ & threshold & $0.0539$ & $0.736$ & $0.0494$ & $0.624$ \\
$\hat\eta=0.75\,\eta$ & $m=4$ & $0.0551$ & $0.728$ & $0.0310$ & $0.453$ \\
$\hat\eta=1.25\,\eta$ & threshold & $0.0500$ & $0.721$ & $0.0494$ & $0.624$ \\
$\hat\eta=1.25\,\eta$ & $m=4$ & $0.0500$ & $0.697$ & $0.0310$ & $0.453$ \\
$\rho=0.5$ & threshold & $0.0644$ & $0.730$ & $0.0463$ & $0.624$ \\
$\rho=0.9$ & threshold & $0.0669$ & $0.734$ & $0.0372$ & $0.624$ \\
\bottomrule
\end{tabular}
\end{table}

Table~\ref{tab:simrobust} shows model-aware family-wise error inflation above
$\alpha$ under underestimated effects and positive cross-hypothesis
correlation, while null-calibrated aggregation retains its guarantee.
Here $\hat\eta$ is the shift vector the center assumes when it
solves the program, and $\rho$ is the within-site correlation between the two
hypotheses' $z$-statistics.  Underestimating the effect distribution inflates
model-aware strong FWER to $0.0539$ for threshold reports and $0.0551$ for
$m=4$ reports; both maxima
occur at mixed configurations.  With policies computed under independence,
within-site cross-hypothesis correlation raises the model-aware strong FWER, attained at the all-null configuration,
to $0.064$--$0.067$.  The null-calibrated rate stays below $0.05$ throughout,
and for the equal-width $m=4$ report it does not move with the assumed effect at all,
since its critical value is a functional of the null report distribution alone.  Thus
model-aware aggregation requires a correct or validity-corrected report distribution
and independent hypotheses; null-calibrated aggregation is the default when
only superuniform local $p$-values and per-hypothesis cross-site independence
are defensible.

Both failure modes persist under the flat prior, but their sizes move in
opposite directions, so neither number should be read as objective-free.
With the policy and thresholds re-optimized for the flat prior, the
misspecified threshold-report family-wise error rate rises to $0.0590$, from
$0.0539$.  The correlated strong family-wise error rate falls to
$0.057$--$0.060$, from the all-null $0.064$--$0.067$ of the point prior, and
is now attained at a mixed configuration, with all-null rates of
$0.050$--$0.051$.  The flat prior changes the optimized policy and the
configuration at which the maximum family-wise error rate is attained.
Null-calibrated rates are unchanged at every
prior, since that rejection event consults neither the objective nor the
alternative.

\FloatBarrier

\section{Protocol, results guide, and computational certificates}
\label{supp:additional}

This section gives the model-aware protocol, a reader's guide to the results,
the relation to the companion manuscripts, aggregation details and a
fixed-report comparison of the two aggregators (Proposition~\ref{prop:reduction}),
the computational certificates and supporting tables behind the numerical
study, and the large-site analysis (Proposition~\ref{prop:largeS}).

\subsection{Model-aware protocol}

Figure~\ref{fig:protocol} summarizes the federated procedure of
Definition~\ref{def:protocol}.  Algorithm~\ref{alg:modelaware} makes explicit the two distinct sources of
randomization on an atomic report distribution: auxiliary tie-breakers order hypotheses with equal likelihood ratios, while the primal probabilities randomize between rejection
counts.  Neither is needed to define the optimal value, but both are needed
to implement a general optimal policy.  A deterministic, conservative variant
rejects only the hypotheses that the optimal policy rejects with conditional
probability one.  Its rejection set is contained in that of the optimal
policy for every value of the randomization, so it remains valid, but it can
lose the power that the randomized rejections contribute.

\begin{figure}[!htb]
\centering
\resizebox{\linewidth}{!}{%
\begin{tikzpicture}[
  font=\small, >={Latex[length=2mm]},
  box/.style={draw, rounded corners=2pt, align=center, inner sep=4pt},
  head/.style={font=\footnotesize\itshape, align=center}]
\node[box, text width=3.4cm] (s1) {Site $1$\\ $y_{1,k}=T_1(u_{1,k})$};
\node[below=1mm of s1] (dots) {$\vdots$};
\node[box, text width=3.4cm, below=1mm of dots] (sS) {Site $S$\\ $y_{S,k}=T_S(u_{S,k})$};
\node[box, text width=2.5cm, anchor=west] (arr) at ($(s1.east)!0.5!(sS.east)+(8mm,0)$)
  {Report array\\ $\mathbf y=\{y_{s,k}\}$,\\ at most $m$ values\\ per report map};
\node[box, text width=4.0cm, right=8mm of arr, yshift=8.5mm] (ma)
  {Model-aware aggregation\\ (Section~\ref{sec:aware}; Algorithm~\ref{alg:modelaware})};
\node[box, text width=4.0cm, right=8mm of arr, yshift=-8.5mm] (nc)
  {Null-calibrated aggregation\\ (Section~\ref{sec:marginal})};
\node[box, text width=2.4cm, right=8mm of arr, xshift=48mm] (out)
  {Rejection set\\ $\mathcal R\subseteq\{1,\ldots,K\}$,\\ strong FWER${}\le\alpha$};
\node[head, above=1mm of s1] {Federated reports and their\\ resolution (Section~\ref{sec:reports})};
\node[head, above=1mm of ma] {Center: one of two rules};
\draw[->] (s1.east) -- (arr.west);
\draw[->] (sS.east) -- (arr.west);
\draw[->] (arr.east) -- (ma.west);
\draw[->] (arr.east) -- (nc.west);
\draw[->] (ma.east) -- (out.west);
\draw[->] (nc.east) -- (out.west);
\coordinate (bnd) at ($(s1.east)!0.5!(arr.west)$);
\draw[dashed, gray] ($(bnd |- s1.north)+(0,3mm)$) -- ($(bnd |- sS.south)+(0,-3mm)$)
  node[below, font=\footnotesize\itshape, text=black] {only reports cross};
\end{tikzpicture}%
}
\caption{From quantized reports alone, the center controls the family-wise
error rate in finite samples with either of two rules: the model-aware rule,
most powerful for the given reports under
Assumptions~\ref{ass:model}--\ref{ass:known}, or the null-calibrated rule,
valid without an alternative model or independence across hypotheses
(Proposition~\ref{prop:marginal}).  Each site $s$ maps its $p$-value $u_{s,k}$ for hypothesis $k$ to
a report $y_{s,k}=T_s(u_{s,k})$ that takes at most $m$ values
(Definition~\ref{def:protocol} and Section~\ref{sec:reports} of the main
text).  The center applies the model-aware rule (Section~\ref{sec:aware};
Algorithm~\ref{alg:modelaware}) or the null-calibrated rule
(Section~\ref{sec:marginal}) to the $S\times K$ report array and returns a
rejection set with strong family-wise error control at level~$\alpha$.}
\label{fig:protocol}
\end{figure}

\medskip
\begin{algorithm}[H]
\small
\caption{Model-aware federated FWER testing from threshold reports}
\label{alg:modelaware}
\KwIn{local $p$-values $u_{s,k}$, known densities $g_s$, thresholds $t_s$,
objective prior $w$, and level $\alpha$}
\emph{At every site $s$, in parallel:}\;
transmit $y_{s,k}\gets\mathbf 1\{u_{s,k}\le t_s\}$ for every $k\in\{1,\ldots,K\}$\;
\emph{At the center:}\;
compute $\ell_s(c)=\tilde\nu_{s,1}(\{c\})/\tilde\nu_{s,0}(\{c\})$ on each
positive-null-mass report value\;
form $L_k\gets\prod_s\ell_s(y_{s,k})$ and draw auxiliary tie-breakers
$\xi_k\sim\mathrm{Unif}[0,1]$\;
solve the finite report-distribution primal and dual for objective $\Pi_w$
at level $\alpha$ to obtain action probabilities $r^*(L)$ (the dual
multipliers $\bar\mu$ certify the value)\;
use $U\sim\mathrm{Unif}[0,1]$ to sample the rejection count from $r^*(L)$,
and reject that many largest $L_k$, using $\xi_k$ only for ties\;
\KwRet{the rejection set}
\end{algorithm}

For any report map, absolute continuity of the alternative push-forward with
respect to the null push-forward makes likelihood ratios on zero-null-mass
report values irrelevant; Section~\ref{app:dual} handles discrete, continuous, and
mixed report distributions.

\subsection{Reader's guide to the results}

Table~\ref{tab:roadmap} lists, for each principal result, the question it
answers, the assumptions it uses, what it concludes and for which cases, and
where it is stated and proved.  Results stated in full only in this
supplement are marked (supp.).

\begin{table}[htbp]
\caption{Reader's guide: for each principal result, the question it answers,
what it assumes, what it concludes and for which cases, and where it is
stated and proved.  Numbers in the assumptions column refer to Assumptions~1--3 of the
main text; (supp.) marks results stated in full only in this supplement; OAK
is the oracle active-kink condition (Definition~\ref{def:oak}) and TV the
total-variation distance (Section~\ref{sec:plugin}).  The last column gives
the main-text section that states the result and the supplement section that
proves it.}
\label{tab:roadmap}
\centering
\begingroup
\footnotesize
\singlespacing
\setlength{\tabcolsep}{3pt}
\renewcommand{\arraystretch}{1.0}
\begin{tabularx}{\textwidth}{@{}
  >{\raggedright\arraybackslash}p{0.11\textwidth}
  >{\raggedright\arraybackslash}p{0.2\textwidth}
  >{\raggedright\arraybackslash}p{0.25\textwidth}
  >{\raggedright\arraybackslash}X
  >{\raggedright\arraybackslash}p{0.085\textwidth}@{}}
\toprule
Result & Question answered & Main assumptions & Conclusion and scope & Stated / proved \\
\midrule
Prop.~\ref{prop:oracle} & what is the centralized benchmark? & \ref{ass:model}--\ref{ass:known} & a likelihood-ratio-ordered policy on the $\Lambda_k$ is optimal; all $K,S$ & \ref{sec:criterion} / \ref{supp:proof_oracle} \\
Prop.~\ref{prop:distribution} & when does federation lose no power? & \ref{ass:model}--\ref{ass:known} (\ref{ass:reg} for (b)) & (a) sufficient reports lose nothing, all $K,S$; (b) $K=1$: a threshold loses nothing at $S=1$ and power at $S\ge2$ & \ref{sec:criterion} / \ref{supp:proof_distribution} \\
Prop.~\ref{prop:converse} & at a single hypothesis, does quantization lose power, and how fast does refinement recover it? & \ref{ass:model}--\ref{ass:reg}; $C^1$ likelihood ratios with $g_s'<0$ on $[0,1)$ for (b) & $K=1$, $S\ge2$: every finite-range design has a positive gap; class-optimal $m^{-2}$ rate & \ref{sec:criterion} / \ref{supp:proof_converse}--\ref{supp:proof_graderate} \\
Lem.~\ref{lem:uniform_upper} (supp.) & how fast do equal-width reports recover power for any $K$? & \ref{ass:model}--\ref{ass:reg}; $C^1$ likelihood ratios & upper bound $Cm^{-2}$ for all fixed $(K,S)$ and every prior & \ref{supp:proof_graderate} \\
Thm.~\ref{thm:oracle_kink} & under multiplicity, does quantization lose power, and how fast does refinement recover it? & \ref{ass:model}--\ref{ass:reg}; $S\ge2$; $w(K)>0$ or OAK; $C^1$ likelihood ratios with $g_s'<0$ on $[0,1)$ for rate & every finite-range componentwise design has a positive gap; positive threshold-class deficit; class-optimal $m^{-2}$ rate for fixed $K\ge2$ & \ref{sec:kink} / \ref{supp:proof_oracle_kink} \\
Lem.~\ref{lem:oak_auto} & when is OAK automatic? & \ref{ass:model}--\ref{ass:reg}; $w(K)>0$ & OAK holds at every dual minimizer; can fail when $w(K)=0$ & \ref{sec:kink} / \ref{supp:proof_oak_auto} \\
Cor.~\ref{cor:beta_oak} (supp.) & where is the witness explicit? & \ref{ass:model}--\ref{ass:reg}; $K=S=2$; $w(2)=1$; explicit open homogeneous Beta-family region & explicit witness located without a numerical optimizer; $\Theta(m^{-2})$ when $\theta\ge2$ & \ref{supp:proof_beta_oak} \\
Prop.~\ref{prop:enclosure} (supp.) & can the witness be checked without trusting an optimizer? & convex coercive dual; nonnegative $b_{\gamma,k}$ & encloses every dual minimizer, including ties; any prior $w$ & \ref{supp:proof_enclosure} \\
Cor.~\ref{cor:enclosure_beta} (supp.) & does the enclosure confirm the witness for the displayed models? & $K=S=2$; $w(2)=1$; $\mathrm{Beta}(1,\theta)$, $\theta\in\{2,3\}$, $\alpha=0.05$ & witness at every dual minimizer, certified by outward-rounded integer arithmetic & \ref{supp:proof_enclosure} \\
Thm.~\ref{thm:aggregator} & which center rule is most powerful for fixed reports? & \ref{ass:model}--\ref{ass:known}; fixed maps & the oracle program on report likelihood ratios $L_k$: finite-sample FWER control, most powerful; every fixed design & \ref{sec:aware} / \ref{supp:proof_aggregator} \\
Prop.~\ref{prop:marginal} & can a rule stay valid without the alternative model? & ordered fixed maps, cross-site independence per hypothesis, superuniform local $p$-values & yes: Fisher--Bonferroni, finite-sample FWER control under arbitrary cross-hypothesis dependence & \ref{sec:marginal} / \ref{supp:proof_marginal} \\
Prop.~\ref{prop:marginalK1} & at a single hypothesis, what power does null calibration lose? & \ref{ass:model}; $K=1$; homogeneous nonincreasing $g$; common threshold & none with boundary randomization; none deterministically at an exact-level outer optimum & \ref{sec:marginal} / \ref{supp:proof_marginalK1} \\
Prop.~\ref{prop:reduction} (supp.) & how large is the aggregator gap? & \ref{ass:model}--\ref{ass:known}; one fixed report distribution and combining statistic & equals a centralized gap & \ref{supp:aggregator_gap} \\
Prop.~\ref{prop:alloc} & which unanimity thresholds are optimal? & $g_s$ continuous; $G_s>0$; elasticity monotonicity & equal local elasticities across sites; $(\alpha/K)^{1/S}$ for homogeneous sites & \ref{sec:design} / \ref{supp:proof_allocation} \\
Prop.~\ref{prop:plugin} (supp.) & what if the local alternatives are unknown? & \ref{ass:model}; independent pilot; dominated estimates; TV bound; $(K-1)\varepsilon+\zeta\le\alpha$ & finite-sample validity after lowering the level by $(K-1)\varepsilon+\zeta$ & \ref{sec:plugin} / \ref{supp:proof_plugin} \\
Prop.~\ref{prop:largeS} (supp.) & what happens as the number of sites grows? & homogeneous $\mathrm{Beta}(1,\theta)$, $\theta>1$ & $K=1$: oracle and best threshold-report misses decay exponentially in $S$, unanimity's polynomially; both deficits vanish & \ref{supp:supporting_numerics} \\
\bottomrule
\end{tabularx}
\endgroup
\end{table}

\subsection{Relation to the companion manuscripts}

No result in this paper depends on the companion manuscripts: its validity and
optimality claims rest on the published characterization
of~\citet{rosset2022optimal} and are proved in full here.  The companion
manuscripts~\citep{dubey2025necessary,dubey2026esp,dubey2026boost} study
centralized power-optimal testing under strong family-wise error control:
the first two compute the optimal test for three hypotheses and for any number of hypotheses, respectively, and the third treats block-structured multiplicity,
including an allocation of the level across blocks and a plug-in for an
unknown alternative.  None studies site-local reporting, report design,
communication loss, or null-calibrated aggregation.  Unlike the third
companion, Proposition~\ref{prop:alloc} allocates local thresholds across
sites rather than the level across blocks, and Proposition~\ref{prop:plugin} corrects for
estimated report distributions rather than for an estimated density in the
centralized program.
Section~\ref{app:dual} rederives the second companion's
elementary-symmetric-polynomial formulas from first principles and derives
every finite program behind a reported number.  The
companions' algorithmic convergence statements are conveniences, not
logical prerequisites for any result here.

\subsection{Aggregation details}

The null-calibrated rule also has an e-value form, which exact-null tail
calibration dominates.  Choose fixed nonnegative nonincreasing
functions $f_s$, called calibrators, with $\int_0^1f_s(u)\,\mathrm du\le1$.
Superuniformity of each implied site $p$-value gives
$\E_h f_s\{\hat p_s(y_{s,k})\}\le1$ whenever $h_k=0$.
Cross-site independence for each hypothesis then makes
\[
 E_k=\prod_{s=1}^S f_s\{\hat p_s(y_{s,k})\}
 \quad\text{an e-value:}\quad \E_h(E_k)\le1\quad(h_k=0).
\]
Markov's inequality and a union bound give strong FWER control when
$H_k$ is rejected for $E_k\ge K/\alpha$, with arbitrary dependence across
hypotheses. For a reference $E_k^0$ from independent uniform site inputs,
exact-null tail calibration is at least as powerful:
$\pr_0(E_k^0\ge E_k)\le1/E_k$ for $E_k>0$, so it includes
every Markov rejection.  The power calibrator $f_s(u)=(1-a)u^{-a}$,
$0<a<1$, gives $E_k=(1-a)^S\exp(aW_k)$, a monotone transform of Fisher's
statistic.  On threshold reports, choosing
$f_s(u)=\mathbf1\{u\le t_s\}/t_s$ gives the unanimity e-value
$\mathbf1\{\text{all sites reject}\}/\prod_s t_s$; at
$\prod_s t_s=\alpha/K$, the Markov and exact-tail rules coincide.

For estimated
report distributions, an estimate converging at the parametric rate
$n^{-1/2}$ gives a correction of the same order as in
Proposition~\ref{prop:plugin}, whenever the induced report distribution is
locally Lipschitz in total variation in the estimated parameter.

\subsection{Oracle benchmarks and threshold searches}\label{supp:benchmarks}

The deficits in Table~\ref{tab:kS2} compare exact four-cell linear-program values with closed-form or Neyman--Pearson oracle benchmarks.  A pair of
threshold reports generates four cells.
Their null masses are products of increments of the uniform distribution,
and their alternative masses are the analogous products of increments of
$G_s$. Ordering the resulting likelihood ratios and exhausting level
$\alpha$ with boundary randomization gives the exact finite linear-program value for a
fixed threshold. Under homogeneous $\mathrm{Beta}(1,2)$ alternatives, let
$c\in(0,1)$ be the oracle cutoff for $(1-u_1)(1-u_2)$. Its power and level are
\[
\pi^*=(1-c^2)+2c^2\log c,
\qquad \alpha=1-c+c\log c.
\]
For the one-sided normal example, $X\sim\mathcal N(2,1)$ under the
alternative and $u=1-\Phi(X)$, giving
\[
G(u)=\Phi\{\Phi^{-1}(u)+2\},
\qquad g(u)=\exp\{-2\Phi^{-1}(u)-2\}.
\]
The numerical oracle agrees with the closed-form or Neyman--Pearson benchmark
within $1.1\times10^{-5}$ in the available $K=1$ checks. Table~\ref{tab:kS2}
uses the exact benchmark where available. In particular, because the sum of
the two sitewise normal log likelihood ratios has variance $8$ under both
distributions and means $-4$ and $4$ under the null and alternative, respectively,
the one-sided normal oracle is
$\Phi\{\Phi^{-1}(\alpha)+\sqrt{8}\}=0.8817090$.

For Table~\ref{tab:k2}, the model-aware value is the minimum of the
elementary-symmetric-polynomial dual on the four-cell distribution. Boundary
randomization is internal to this finite program, so no Monte Carlo is
used. The continuous oracle is computed by adaptive quadrature under the
exact likelihood-ratio density, which gives approximately $0.156461$ for
$\mathrm{Beta}(1,2)$ and $0.286977$ for $\mathrm{Beta}(1,3)$; a fine
discretization gives values within $7\times10^{-6}$ that round identically in
Table~\ref{tab:k2}. Each method's report map is selected separately: the
null-calibrated rule's unanimity thresholds solve the elasticity condition of
Proposition~\ref{prop:alloc}, whereas the model-aware thresholds are located
by a numerical outer search whose attained power is certified below.  The search selects
$(0.288,0.328)$ for $\mathrm{Beta}(1,2)$, slightly better than the best
common threshold, $0.309$, and a vector rounding to $(0.247,0.247)$ for
$\mathrm{Beta}(1,3)$.

The four bounded-Beta outer searches behind Tables~\ref{tab:kS2}
and~\ref{tab:k2}, for $\mathrm{Beta}(1,2)$ and $\mathrm{Beta}(1,3)$ at $K=1$ and
at $K=2$, are certified by exact-rational branch and bound in
\texttt{cert\_thresholds.py}.  The certificate proves that no threshold
vector in the closed square $[0,1]^2$ gives model-aware power more than
$2\times10^{-4}$ above $P_{\mathrm{inc}}$, the exactly computed power of a
feasible policy at the selected thresholds (the \emph{incumbent}); it bounds
the optimal value, not the optimizing thresholds.  The search splits the
square dyadically into boxes, computes an exact upper bound $\overline\Pi$ on
the power over each box as below, and closes (prunes) a box once
$\overline\Pi\le P_{\mathrm{inc}}+2\times10^{-4}$.  The normal-shift search is
discussed separately below.

The certificate combines weak duality, which turns any nonnegative
multiplier into a valid upper bound, with corner bounds on the four cell
masses and likelihood ratios over each threshold box.  First, weak duality: for every
multiplier $\mu\ge0$ and every threshold vector $t$, the power is at most the
dual objective at $\mu$ on the report distribution induced by $t$; optimizing the multiplier
affects tightness, not validity. Second, parameterize the four cells by the decreasing ratio
functions $v^+_s(t)=G_s(t)/t$ and
$v^-_s(t)=\{1-G_s(t)\}/(1-t)$, whose products give the cell likelihood ratios,
and by the monotone cell masses. Over a threshold box, the corners bound
each quantity. With $\overline p_c$ and
$[\underline v_c,\overline v_c]$ denoting the resulting bounds for cell $c$,
the box bound at $K=1$ is
\[
\overline\Pi(\mu;\mathrm{box})
=\alpha\mu+\sum_c \overline p_c\,(\overline v_c-\mu)_+ ,
\]
and at $K=2$, over ordered cell pairs $(c,d)$,
\[
\begin{aligned}
\overline\Pi(\mu;\mathrm{box})
&=2\alpha\mu_1+\alpha\mu_2\\
&\quad+\sum_{c,d}\overline p_c\,\overline p_d\,
\max\Bigl(0,\ \tfrac12\overline v_c\overline v_d
-\mu_1\min(\underline v_c,\underline v_d)-\mu_2,\
\overline v_c\overline v_d-\mu_1(\underline v_c+\underline v_d)-\mu_2
\Bigr).
\end{aligned}
\]
Each summand bounds its integrand term above pointwise on the box because
the integrand is nonnegative and each factor is bounded separately, so
$\overline\Pi$ upper-bounds the dual objective, hence the power, at every $t$
in the
box; no independence between masses and ratios is used.

For the Beta alternatives every box bound and pruning comparison is
evaluated in exact rational arithmetic, so solver error can make pruning
inefficient but cannot cause an invalid pruning decision.  For a
$\mathrm{Beta}(1,\theta)$ alternative with
$\theta\in\{2,3\}$,
\[
v^+(t)=\sum_{j=0}^{\theta-1}(1-t)^j,
\qquad v^-(t)=(1-t)^{\theta-1}.
\]
The initial domain has rational endpoints, every split is dyadic, and these
identities make all cell-mass and likelihood-ratio bounds rational. The code
therefore evaluates every displayed box bound $\overline\Pi$ and every pruning
comparison using \texttt{Fraction} exact arithmetic. A floating-point
epigraph linear program may suggest a useful $K=2$ multiplier, but its binary
floating-point output is converted to an exact nonnegative rational and
$\overline\Pi$ is recomputed exactly.

The incumbent is also exactly feasible. At $K=1$, it is the exact
randomized Neyman--Pearson value at the reported decimal-rational candidate.
At $K=2$, let $\widetilde r_{\ell c}$ be the nonnegative rational
representation of the numerical finite-primal action probability for
rejecting the $\ell$ largest likelihood ratios in joint report cell $c$. The code
enumerates the two one-null configurations and the all-null configuration
using the exact rational cell distributions and uniform tie probabilities. It then
sets
\[
 \varsigma=\min\left\{1,\quad
 \frac{1}{\max_c\sum_{\ell=1}^2\widetilde r_{\ell c}},\
 \min_{\varnothing\ne I\subseteq\{1,2\}}
 \frac{\alpha}{\operatorname{FWER}_I(\widetilde r)}\right\},
\]
where $\operatorname{FWER}_I(\widetilde r)$ is the family-wise error rate of
$\widetilde r$ at the configuration whose true hypotheses are $I$, with a
vacuous ratio omitted when its denominator is zero.  The code then uses
$r_{\ell c}=\varsigma\widetilde r_{\ell c}$. Nonnegativity is preserved, every
cellwise positive-action sum is at most one, and every configuration's
family-wise error rate is at most $\alpha$, all as exact rational comparisons. The remaining action
mass is assigned to rejecting none, and the resulting full-alternative power
$P_{\mathrm{inc}}$ is evaluated exactly.

A box is closed when its exact $\overline\Pi$ is at most
$P_{\mathrm{inc}}+2\times10^{-4}$; otherwise it is split. Once every box is
closed, $\sup_t\Pi^\star(t)\le P_{\mathrm{inc}}+2\times10^{-4}$, where
$\Pi^\star(t)$ is the optimal model-aware power on the report distribution
induced by the thresholds $t$. This proves the
global value-gap claim, though not an enclosure of the optimizing
coordinates. A dual value at the candidate is only an upper bound and is
not used as the incumbent.  The four Beta searches' exactly feasible attained
powers have decimal expansions beginning $0.157778640$, $0.283024780$,
$0.143116898$, and $0.269166018$ (the $\mathrm{Beta}(1,2)$ and $\mathrm{Beta}(1,3)$ rows of
Table~\ref{tab:kS2}, then those of Table~\ref{tab:k2}), within $5\times10^{-5}$ of the numerical powers used for the displayed
deficits.

For the unbounded one-sided normal likelihood ratio, an analytic
Neyman--Pearson slab bound excludes $t_s\le10^{-6}$, but on the remaining
square the CDF evaluations and pruning comparisons use ordinary double
precision without directed rounding.  That computation indicates the same
$2\times10^{-4}$ value gap but is not a certificate.
The complete slab derivation, rational candidates, box counts, and execution
audit are documented with the implementation in the code and data
supplement.

For Table~\ref{tab:scaleS}, use $\Gamma(S,r)$ for the Gamma distribution with
shape $S$ and rate $r$. The representation in the proof of
Proposition~\ref{prop:largeS} below gives the centralized oracle power at $S$
homogeneous $\mathrm{Beta}(1,2)$ sites as
\[
\pi^*(S)=\pr\{\Gamma(S,2)\leq\gamma_{S,\alpha}\},
\]
where $\gamma_{S,\alpha}$ is the $\alpha$-quantile of $\Gamma(S,1)$.
A threshold report with a common threshold reduces the binary report vector to the number of
locally rejecting sites, on which the randomized most powerful test is
immediate. For homogeneous sites this per-hypothesis report has $S+1$ values;
exchangeability reduces the model-aware calculation from $(S+1)^K$ ordered
tuples to the $\binom{K+S}{S}$ count compositions. A heterogeneous binary
vector can instead have as many as $2^S$ cells.  The common threshold of
Table~\ref{tab:scaleS} maximizes the model-aware power over a grid of
$2{,}001$ points on $[0.001,0.999]$, augmented by $\alpha^{1/S}$ and by the
thresholds at which a null count tail $\pr\{\mathrm{Bin}(S,t)\ge r\}$ equals
$\alpha$, followed by local refinement; this search is numerical and, for
$S\ge3$, not certified.  The reported experiments use
selected $(K,S)$ values rather than an exhaustive globally certified search
over $K,S\leq10$.

For the equal-width reports in Table~\ref{tab:rate}, least-squares fits of
$\log_2\delta_w$ on $\log_2m$ give slope magnitude $2.00$ at $K=1$ over
$m\in\{32,64,128,256,512,1024\}$ (\texttt{kS2\_rate.py}; Table~\ref{tab:rate}
displays the values up to $m=64$), and $1.998$ at $K=2$ over
$m\in\{8,12,16,24,32,48,64\}$ (\texttt{k2\_sharpness.py}); over the powers
of two from $m=8$ to $64$ in Table~\ref{tab:rate} alone, the $K=2$ slope is
$1.96$.
Repeating the $K=1$ fit for the unbounded one-sided normal likelihood ratio
gives a slope of approximately $1.4$: that likelihood ratio is unbounded, so
it lies outside Assumption~\ref{ass:reg} and the $C^1$ condition of
Proposition~\ref{prop:converse}(b), and refinement recovers power more slowly
than $m^{-2}$ (Step~4 of Section~\ref{supp:proof_graderate}).

\subsection{Comparing aggregators on a fixed report distribution}\label{supp:aggregator_gap}

For a fixed report design, the model-aware and null-calibrated rules are
two centralized procedures on the same induced experiment, so their power
difference is a centralized optimal-versus-marginal gap.  Fix a report
design and write $\mathsf P=(\tilde\nu_0,\tilde\nu_1)$ for its
exchangeable two-group report distribution, with likelihood ratio
$L=\mathrm d\tilde\nu_1/\mathrm d\tilde\nu_0$.  Let $\phi$ be the fixed
combining statistic of Section~\ref{sec:marginal}, evaluated at a generic
per-hypothesis report vector $x=(x_1,\ldots,x_S)$, and define its reference
tail $p$-value by
\[
 p_\phi(x)=\pr_{X_0\sim\tilde\nu_0}\{\phi(X_0)\ge\phi(x)\}.
\]
The model-aware rule is
determined by $\mathsf P$ and has optimal power $\Pi_w^\star(\mathsf P)$; the
null-calibrated rule depends on $\tilde\nu_0$ and $\phi$ only.  The distinction below is
between rules on the same reports, not between report classes whose maps are
separately optimized.

\begin{lemma}[The null-calibrated rule is centralized Bonferroni]\label{lem:marginalcentral}
For any fixed report design and fixed statistic $\phi$, the null-calibrated
rule rejects $H_k$ if and only if $p_\phi(\mathbf y_{\cdot,k})\le\alpha/K$,
controls~\eqref{eq:fwer} under the exactly uniform nulls of
Assumption~\ref{ass:model}, and has prior-weighted average power
\[
\Pi_w^{\mathrm{nc}}(\mathsf P,\phi)
\;=\; \pr_{\tilde\nu_1}\!\big\{p_\phi(X)\le\alpha/K\big\},
\]
where $X\sim\tilde\nu_1$; with $(\mathsf P,\phi)$ fixed, $K$ enters only
through the Bonferroni level.  Under merely superuniform local nulls,
Proposition~\ref{prop:marginal} gives the same control for ordered maps and
coordinatewise nonincreasing combining statistics.
\end{lemma}

\begin{proof}
The rule rejects $H_k$ exactly when $p_\phi(\mathbf y_{\cdot,k})\le\alpha/K$,
so validity follows from Steps~3 and~4 of the proof of Proposition~\ref{prop:marginal}.
The rejection event depends only on $\mathbf y_{\cdot,k}$, whose law is
$\tilde\nu_1$ whenever $h_k=1$, so each false hypothesis is rejected with
probability $\pr_{\tilde\nu_1}\{p_\phi(X)\le\alpha/K\}$ at every
configuration, and averaging with any prior $w$ gives the display
(Section~\ref{supp:proof_priorfree}).
\end{proof}

\begin{proposition}[Fixed-report aggregator gap]\label{prop:reduction}
Under Assumptions~\ref{ass:model}--\ref{ass:known}, for any fixed report
design and fixed statistic, the model-aware minus null-calibrated average
power equals the corresponding centralized gap on $(\mathsf P,\phi)$,
\[
\Pi_w(\text{model-aware}) - \Pi_w(\text{null-calibrated})
\;=\; \Pi_w^{\star}(\mathsf P)-\Pi_w^{\mathrm{nc}}(\mathsf P,\phi)
\;=\; d_w^{\mathrm{nc}}(\mathsf P,\phi)-d_w^{\mathrm{aw}}(\mathsf P),
\]
where $d_w^{\mathrm{aw}}(\mathsf P)=\pi^*_w-\Pi_w^\star(\mathsf P)$ and
$d_w^{\mathrm{nc}}(\mathsf P,\phi)=\pi^*_w-\Pi_w^{\mathrm{nc}}(\mathsf P,\phi)$
are the corresponding fixed-design gaps.
\end{proposition}

\begin{corollary}[Single hypothesis: ordering and boundary randomization]\label{cor:k1gap}
At $K=1$, put $R_\alpha=\{x:p_\phi(x)\le\alpha\}$.  The aggregator gap is
$\Pi_1^{\star}(\mathsf P)-\tilde\nu_1(R_\alpha)$, the
power difference between the level-$\alpha$ Neyman--Pearson rule and the fixed
null-calibrated rule.  It vanishes exactly when, for some $\lambda\ge0$,
\[
 \{L>\lambda\}\subseteq R_\alpha\subseteq\{L\ge\lambda\}
 \quad\tilde\nu_0\text{-almost surely},\qquad
 \lambda\{\alpha-\tilde\nu_0(R_\alpha)\}=0.
\]
In particular, matching a likelihood-ratio rejection region at exact level
$\alpha$ suffices; away from its boundary, the orderings by $\phi$ and by $L$ need not agree.
For homogeneous common-threshold reports with Fisher's count statistic
(or any strictly increasing function of the count), the likelihood ratio is
nondecreasing in the local-rejection count, so the gap is the
power added by boundary randomization alone.  Randomized null calibration removes
that gap at every fixed threshold, and conservative calibration removes it
whenever the selected count tail equals $\alpha$
(Proposition~\ref{prop:marginalK1}).
\end{corollary}

\begin{corollary}[Marginal tightening under multiplicity]\label{cor:multiplicity}
For fixed $(\mathsf P,\phi)$, $\Pi_w^{\mathrm{nc}}(\mathsf P,\phi)$ is
nonincreasing in the integer $K$, because the Bonferroni level $\alpha/K$
decreases and the rejection event shrinks.
\end{corollary}

\paragraph{Proof of Proposition~\ref{prop:reduction} and
Corollaries~\ref{cor:k1gap}--\ref{cor:multiplicity}}
By Theorem~\ref{thm:aggregator}(b) the model-aware aggregator is the
$\Pi_w$-optimal level-$\alpha$ family-wise test on the report distribution
$\mathsf P=(\tilde\nu_0,\tilde\nu_1)$, so its average power is
$\Pi_w^{\star}(\mathsf P)$.  By Lemma~\ref{lem:marginalcentral}, for the
separately fixed combining statistic $\phi$, the null-calibrated average
power is
\[
\Pi_w^{\mathrm{nc}}(\mathsf P,\phi)
=\pr_{\tilde\nu_1}\{p_\phi\le\alpha/K\}.
\]
The distribution $\mathsf P$ determines its likelihood ratio $L$ but not the choice of
$\phi$, which is why the second functional is indexed by the pair
$(\mathsf P,\phi)$; the tail function $p_\phi$ is determined by
$\tilde\nu_0$ and $\phi$ alone.  Subtracting gives the first equality in
Proposition~\ref{prop:reduction}.  With the fixed-design gaps
\[
d_w^{\mathrm{aw}}(\mathsf P)=\pi^*_w-\Pi_w^\star(\mathsf P),\qquad
d_w^{\mathrm{nc}}(\mathsf P,\phi)
=\pi^*_w-\Pi_w^{\mathrm{nc}}(\mathsf P,\phi),
\]
the common oracle term cancels, proving the second equality.

For $K=1$ the constraint~\eqref{eq:fwer} is the single size constraint at
level $\alpha$, so a most powerful test is the Neyman--Pearson test on
$\mathsf P$, which by the density-ratio identity
$\mathrm{d}\tilde\nu_1/\mathrm{d}\tilde\nu_0=L$ rejects report regions with
the largest $L$ until the level is exhausted, with boundary randomization.
The null-calibrated rule rejects $R_\alpha$, whose null probability is at
most $\alpha$.  Let $\lambda\ge0$ minimize the Neyman--Pearson dual
$\alpha\lambda+\E_0(L-\lambda)_+$; the minimum is attained because
$\E_0L=1$ and $\alpha>0$ make this finite objective continuous and coercive.
Subtracting the null-calibrated power $\tilde\nu_1(R_\alpha)$ from the dual
value gives
\[
 \Pi_1^\star(\mathsf P)-\tilde\nu_1(R_\alpha)
 =\lambda\{\alpha-\tilde\nu_0(R_\alpha)\}
 +\E_0\bigl\{(L-\lambda)_+-\mathbf1_{R_\alpha}(L-\lambda)\bigr\}.
\]
Both terms are nonnegative.  The second vanishes exactly when
$R_\alpha$ includes $\{L>\lambda\}$ and excludes $\{L<\lambda\}$ up to
null sets; the first is the stated complementary-slackness condition.
This proves necessity.  Conversely, the stated conditions give
$\tilde\nu_1(R_\alpha)=\alpha\lambda+\E_0(L-\lambda)_+$, so weak duality
proves sufficiency.  Proposition~\ref{prop:marginalK1}
gives the homogeneous common-threshold conclusion because the Fisher
statistic is increasing, and the report likelihood ratio nondecreasing,
in the local-rejection count.  For heterogeneous
sites the two orderings can differ, because the likelihood-ratio weights
$\log\{G_s(t_s)/t_s\}-\log[\{1-G_s(t_s)\}/(1-t_s)]$ involve $G_s$ while the
Fisher weights $-\log t_s$ do not.

For fixed $(\mathsf P,\phi)$ the tail function $p_\phi$ does not change with
$K$.  Since the admissible tail probability
$\alpha/K$ shrinks, the rejection events $\{p_\phi\le\alpha/K\}$ are
nested and decreasing in $K$, and their $\tilde\nu_1$-probabilities
are nonincreasing.  This proves Corollary~\ref{cor:multiplicity}.
\hfill$\square$

Holm's procedure on the same null-calibrated $p$-values is at least as
powerful as the Bonferroni step and remains valid under arbitrary
cross-hypothesis dependence.  Hommel's procedure, which requires valid Simes
tests~\citep{simes1986improved} of intersection nulls, is at least as powerful
as Holm's.  The Bonferroni power $\Pi_w^{\mathrm{nc}}(\mathsf P,\phi)$
therefore lower-bounds what the marginal class can achieve.  The growing gap in
Figure~\ref{fig:theory_inverse_square}(c) of the main text is a numerical
pattern, not a monotonicity theorem.

\subsection{Supporting numerical results and the large-site analysis}
\label{supp:supporting_numerics}

The numerical values summarized in the main text are deterministic
finite-report calculations, except for continuous-oracle values explicitly
described as numerical quadrature benchmarks and for the Monte Carlo
correlation rates of Section~\ref{supp:sim}.  Outer searches and finite-dual
solves are subject to the qualifications stated below; none of these computations substitutes for the corresponding
theorem.

At $K=1$ and $S=2$ (Table~\ref{tab:kS2}) the selected model-aware and null-calibrated thresholds coincide
at $\hat t=\sqrt{\alpha}$, and the powers of the two bounded-Beta designs are
certified within $2\times10^{-4}$ of the global outer optimum.  Because its
likelihood ratio diverges at zero, the normal-shift row is outside
Assumption~\ref{ass:reg} and is a sensitivity check rather than an instance
of Proposition~\ref{prop:converse} or Theorem~\ref{thm:oracle_kink}, which
assume it.  Boundedness alone does not guarantee a
small deficit:
with bounded $\mathrm{Beta}(1,2)$ likelihood ratios, the relative deficit
reaches $18.4\%$ at $S=8$ (Table~\ref{tab:scaleS}).

\begin{table}[htbp]
\caption{At one hypothesis, one optimized bit per site loses $1.1\%$ and
$1.9\%$ of oracle power for the bounded Beta alternatives ($9.7\%$ for the
unbounded normal shift).  Threshold-report deficits for $K=1$, $S=2$, and
$\alpha=0.05$: the deficit is $\pi^*$ minus the power of the selected threshold report, and $\hat t$
is the common threshold selected by both the model-aware and the
null-calibrated rule.  The selected powers in the two Beta rows are certified within $2\times10^{-4}$ of the best threshold-report power by exact rational arithmetic; the normal-shift row,
whose likelihood ratio is unbounded, is checked only in double precision.}
\label{tab:kS2}
\centering
\begingroup
\small
\singlespacing
\setlength{\tabcolsep}{4pt}
\renewcommand{\arraystretch}{1.08}
\begin{tabular}{lccccc}
\toprule
Alternative density $g$ & likelihood ratio bounded? & $\pi^*$ & deficit & deficit $/\,\pi^*$ & $\hat t$ \\
\midrule
$g(u)=2(1-u)$, $\mathrm{Beta}(1,2)$   & yes & $0.1595$ & $0.0018$ & $1.1\%$ & $0.224$ \\
$g(u)=3(1-u)^2$, $\mathrm{Beta}(1,3)$ & yes & $0.2885$ & $0.0055$ & $1.9\%$ & $0.224$ \\
one-sided normal, mean shift $2$       & no  & $0.8817$ & $0.0852$ & $9.7\%$ & $0.224$ \\
\bottomrule
\end{tabular}
\endgroup
\end{table}

Table~\ref{tab:scaleS} restricts both $K=1$ procedures to the
class of threshold reports with a common threshold under homogeneous
$\mathrm{Beta}(1,2)$ alternatives; the two share the first deficit column
(Proposition~\ref{prop:marginalK1}), and the second deficit column is the unanimity
design.  The null-calibrated count rule matches the model-aware value to six
digits at every displayed $S$, and unanimity is the selected cutoff only
through $S=7$.  From $S=8$ the selected cutoff drops below $S$, the selected
per-site threshold tightens with it, and unanimity ceases to be optimal.
The rise of the relative deficits with $S$ is a finite-$S$ effect: Proposition~\ref{prop:largeS} proves that both
relative deficits ultimately vanish.

The dependence on the number of sites is not merely empirical: at a single
hypothesis the three miss probabilities decay at three distinct large-$S$
rates.

\begin{proposition}[Large-$S$ behavior of the threshold-report deficit]\label{prop:largeS}
Let $K=1$ with $S$ homogeneous sites, alternative distribution function
$G_\theta(u) = 1-(1-u)^\theta$ with $\theta > 1$, at fixed level $\alpha \in
(0,1)$. Write $1-\pi^{*}(S)$, $1-\Pi^{\mathrm{aw}}(S)$, and
$1-\Pi^{\mathrm{un}}(S)$ for the miss probabilities of the centralized oracle,
the best threshold-report design, and the unanimity rule at $t =
\alpha^{1/S}$. As $S \to \infty$:
\begin{itemize}
\item[(a)] $-S^{-1}\log\{1-\pi^{*}(S)\} \to I(\theta) = \theta - 1 -
\log\theta$;
\item[(b)] $1-\Pi^{\mathrm{un}}(S) = \{\log(1/\alpha)\}^{\theta}\,
S^{1-\theta}\,\{1+o(1)\}$: polynomial decay;
\item[(c)] $-S^{-1}\log\{1-\Pi^{\mathrm{aw}}(S)\} \to
J_{\mathrm{thr}}(\theta) = \sup_{t\in(0,1)}
\mathrm{KL}\{\mathrm{Ber}(t) \,\|\, \mathrm{Ber}(G_\theta(t))\}$, with $0 <
J_{\mathrm{thr}}(\theta) < I(\theta)$ and the supremum attained at an interior
$t$.
\end{itemize}
Here $\mathrm{Ber}(p)$ is the Bernoulli distribution with success probability $p$, and
$\mathrm{KL}(P\|Q)$ denotes Kullback--Leibler divergence. Define the onset
of eventual separation by
\[
S^{*}(\alpha,\theta)
:=\min\bigl\{n\in\mathbb N:n\geq1,\quad
\Pi^{\mathrm{aw}}(S)>\Pi^{\mathrm{un}}(S)
\text{ for every integer }S\geq n\bigr\}.
\]
Take $S^*(\alpha,\theta)=\infty$ if this set is empty; Table~\ref{tab:scaleS} is
consistent with $S^*(0.05,2)=8$.
By (a)--(c), the model-aware deficit vanishes exponentially while the
unanimity deficit decays only polynomially, so the ratio of the unanimity to
the model-aware miss probability diverges and the two designs separate at
every sufficiently large $S$; thus $S^{*}(\alpha,\theta)$ is finite for
every $(\alpha,\theta)$.  The achievability part of (c) uses the deterministic
null-calibrated count rule at a common threshold, so marginal aggregation
attains the same exponent; the separation is a statement about insisting on
unanimous local rejections, not about marginal aggregation.
Both relative deficits vanish as $S\to\infty$.
\end{proposition}

The exponent $J_{\mathrm{thr}}$ in~(c) is the classical one-bit
decentralized-detection exponent at $K=1$, attained by a likelihood-ratio
threshold
quantizer~\citep{tsitsiklis1993decentralized,tsitsiklis1993extremal}, and
the advantage of counting rules over unanimity as sensors are added is
classical in Bayes-error terms~\citep{viswanathan1989counting}.  The present
paper's primary regime is finite $S$ with family-wise coupling across
$K$ hypotheses.

\paragraph{Proof of Proposition~\ref{prop:largeS}}
Write $G = G_\theta$ and $g=G'$. Let $Q_0$ and $Q_1$
denote the corresponding one-site null and alternative distributions, respectively,
and let $N$ denote the number of locally rejecting sites,
$N \sim \mathrm{Bin}(S, t)$ under the null and
$\mathrm{Bin}(S, G(t))$ under the alternative at a common local threshold $t$.
Throughout,
\[
\mathrm{KL}\{\mathrm{Ber}(p)\|\mathrm{Ber}(q)\}
=p\log\frac pq+(1-p)\log\frac{1-p}{1-q}.
\]

\emph{(a).} Under the null,
$-\sum_s\log(1-u_s)\sim\Gamma(S,1)$ in the shape--rate
parameterization; under the alternative it is $\Gamma(S,\theta)$. Hence
$1-\pi^{*}(S) = \pr(Y > \gamma_{S,\alpha})$, where $Y$ is a sum of $S$
independent exponentials of rate $\theta$, and $\gamma_{S,\alpha}$, the
$\alpha$-quantile of a Gamma distribution with shape $S$ and rate $1$, satisfies
$\gamma_{S,\alpha}/S \to 1$ by the central limit theorem. For every fixed
$x>1/\theta$, Cram\'er's theorem for the exponential sample mean gives
\[
-\frac1S\log\pr(Y/S>x)\ \longrightarrow\
I_\theta(x):=\theta x-1-\log(\theta x).
\]
To handle the moving threshold, fix
$\eta\in(0,1-1/\theta)$. Eventually
$1-\eta<\gamma_{S,\alpha}/S<1+\eta$, and hence
\[
\pr(Y/S>1+\eta)
\le \pr(Y>\gamma_{S,\alpha})
\le \pr(Y/S>1-\eta).
\]
Therefore
\[
I_\theta(1-\eta)
\le \liminf_{S\to\infty}-\frac1S\log\pr(Y>\gamma_{S,\alpha})
\le \limsup_{S\to\infty}-\frac1S\log\pr(Y>\gamma_{S,\alpha})
\le I_\theta(1+\eta).
\]
Letting $\eta\downarrow0$ and using continuity squeezes the exponential rate to
$I_\theta(1)=\theta-1-\log\theta$.

\emph{(b).} Write $\ell_\alpha=\log(1/\alpha)$.  At $t = \alpha^{1/S}$,
$1-t = \ell_\alpha S^{-1}\{1+O(S^{-1})\}$, so
$1-G(t) = (1-t)^\theta = \ell_\alpha^\theta S^{-\theta}\{1+O(S^{-1})\}$ and
\[
1-\Pi^{\mathrm{un}}(S) \;=\; 1 - \{1-(1-t)^\theta\}^S
\;=\; S\,(1-t)^\theta\,\{1+O(S(1-t)^{\theta})\}
\;=\; \ell_\alpha^\theta S^{1-\theta}\{1+o(1)\},
\]
the middle expansion holding because $S(1-t)^\theta \to 0$ for $\theta > 1$.

\emph{(c), achievability.} The map $t \mapsto
\mathrm{KL}\{\mathrm{Ber}(t) \| \mathrm{Ber}(G(t))\}$ is continuous on
$(0,1)$, positive there since $G(t) > t$, and tends to $0$ at both endpoints,
so the supremum $J_{\mathrm{thr}}$ is attained at an interior $t^{\dagger}$.
Reject when $N \ge r_S$ at $t \equiv t^{\dagger}$, with $r_S$ minimal subject
to $\pr_0(N \ge r_S) \le \alpha$; Chebyshev gives $r_S \le S t^{\dagger} +
c\sqrt{S}+1$ with $c = \{t^{\dagger}(1-t^{\dagger})/\alpha\}^{1/2}$.
Conversely, if $r_S/S\le t^{\dagger}-\varepsilon$ along a subsequence for
some $\varepsilon>0$, the law of large numbers under the null gives
$\pr_0(N\ge r_S)\to1>\alpha$, contradicting the defining size bound. Hence
$r_S/S\to t^{\dagger}$. A boundary-randomized level-$\alpha$ count test is
at least as powerful, so with $\gamma_S=r_S/S\to t^{\dagger}<G(t^{\dagger})$,
the binomial Chernoff bound gives
\[
\begin{aligned}
1-\Pi^{\mathrm{aw}}(S) \;&\le\; \pr_{G}(N < r_S) \;\le\;
\exp\big[-S\,\mathrm{KL}\{\mathrm{Ber}(\gamma_S)\,\|\,
\mathrm{Ber}(G(t^{\dagger}))\}\big] \\
&= \exp\big[-S\,\{J_{\mathrm{thr}} + o(1)\}\big].
\end{aligned}
\]

\emph{(c), converse.} A threshold-report design uses local thresholds
$t_1,\dots,t_S$, possibly depending on $S$, and a center rule $\phi$ with
$\E_0(\phi) \le \alpha$, measurable in the $S$ transmitted bits and the
center's randomization. Let $Z = \sum_s Z_s$ be the log likelihood ratio of
the null against the alternative bit distribution; $\E_0(Z_s) =
\mathrm{KL}\{\mathrm{Ber}(t_s)\|\mathrm{Ber}(G(t_s))\} \le J_{\mathrm{thr}}$,
and $\mathrm{var}_0(Z_s)$ is uniformly bounded in $t_s$. Indeed, for
$t\in(0,1)$ the one-site null-against-alternative log likelihood ratio takes
the values
\[
a_\theta(t)=\log\frac{t}{G(t)},\qquad
b_\theta(t)=\log\frac{1-t}{1-G(t)}=(1-\theta)\log(1-t)
\]
with null probabilities $t$ and $1-t$, so
\[
\mathrm{var}_0(Z_s)
=t_s(1-t_s)\{a_\theta(t_s)-b_\theta(t_s)\}^2.
\]
As $t\downarrow0$, $G(t)=\theta t+O(t^2)$,
$a_\theta(t)=-\log\theta+O(t)$, and $b_\theta(t)=O(t)$, so the variance is
$O(t)$. As $t\uparrow1$, writing $x=1-t$ gives
$a_\theta(t)=\log\{(1-x)/(1-x^\theta)\}=O(x)$ and
$b_\theta(t)=(1-\theta)\log x$, so the variance is
$O\{x\log^2x\}\to0$. Thus it extends continuously by zero at both endpoints,
and
\[
\sigma^2_\theta:=\sup_{0\le t\le1}\mathrm{var}_0(Z_s)<\infty.
\]
For $\varepsilon > 0$, changing measure on $\{Z \le
S(J_{\mathrm{thr}}+\varepsilon)\}$ and applying Chebyshev,
\[
\begin{aligned}
1 - \Pi \;=\; \E_1(1-\phi) \;&\ge\;
\exp\{-S(J_{\mathrm{thr}}+\varepsilon)\}
\big\{\E_0(1-\phi) - \pr_0\big(Z > S(J_{\mathrm{thr}}+\varepsilon)\big)\big\} \\
&\ge\;
\exp\{-S(J_{\mathrm{thr}}+\varepsilon)\}
\big\{1 - \alpha - \sigma^2_\theta/(\varepsilon^2 S)\big\},
\end{aligned}
\]
so $\liminf_{S\to\infty} S^{-1}\log\{1-\Pi^{\mathrm{aw}}(S)\} \ge -J_{\mathrm{thr}} -
\varepsilon$ for every $\varepsilon > 0$. Combining this converse with the
achievability bound and then letting $\varepsilon\downarrow0$ proves the
limit in part~(c). Finally $J_{\mathrm{thr}} < I$: for
each fixed $t$ the two-cell partition $\{u \le t\}, \{u > t\}$ is not
sufficient for $Q_0$ against $Q_1$, since $g$ is nonconstant on at least one side,
so the data-processing inequality is strict,
$\mathrm{KL}\{\mathrm{Ber}(t)\|\mathrm{Ber}(G(t))\} <
\mathrm{KL}(Q_0\|Q_1) = -\int_0^1 \log g(u)\,\mathrm{d}u
= \theta - 1 - \log\theta
= I(\theta)$, and the supremum is attained at an interior point where the
inequality is strict.

\emph{Onset.} The weak inequality $\Pi^{\mathrm{aw}}(S)\ge\Pi^{\mathrm{un}}(S)$
holds for every $S$, because the unanimity rule at
$t=\alpha^{1/S}$ is one feasible threshold-report design. Parts (b)--(c)
give strict inequality for all sufficiently large $S$, so $S^*$ is finite.
\hfill$\square$

\begin{table}[htbp]
\caption{At one hypothesis the one-bit loss grows from $1.1\%$ at two sites
to $18.4\%$ at eight to ten, and, from eight sites on, unanimity falls behind
the best count rule ($22.3\%$ at ten).  Deficits of threshold reports with a
common threshold for $K=1$, homogeneous
$\mathrm{Beta}(1,2)$ sites, and $\alpha=0.05$.  The model-aware and
null-calibrated values agree to six digits at every $S$ shown; the selected
count cutoff is $c=S$ through $S=7$ and then $7,8,9$ at $S=8,9,10$.}
\label{tab:scaleS}
\centering
\begingroup
\small
\singlespacing
\setlength{\tabcolsep}{5pt}
\renewcommand{\arraystretch}{1.05}
\begin{tabular}{cccccc}
\toprule
& & \multicolumn{2}{c}{\begin{tabular}{@{}c@{}}model-aware $=$\\ null-calibrated\end{tabular}} & \multicolumn{2}{c}{unanimity} \\
\cmidrule(lr){3-4}\cmidrule(lr){5-6}
$S$ & $\pi^*$ & deficit & $/\,\pi^*$ & deficit & $/\,\pi^*$ \\
\midrule
$2$  & $0.1595$ & $0.0018$ & $1.1\%$  & $0.0018$ & $1.1\%$  \\
$3$  & $0.2258$ & $0.0087$ & $3.8\%$  & $0.0087$ & $3.8\%$  \\
$4$  & $0.2931$ & $0.0212$ & $7.2\%$  & $0.0212$ & $7.2\%$  \\
$5$  & $0.3595$ & $0.0382$ & $10.6\%$ & $0.0382$ & $10.6\%$ \\
$6$  & $0.4236$ & $0.0583$ & $13.8\%$ & $0.0583$ & $13.8\%$ \\
$7$  & $0.4846$ & $0.0798$ & $16.5\%$ & $0.0798$ & $16.5\%$ \\
$8$  & $0.5417$ & $0.0998$ & $18.4\%$ & $0.1018$ & $18.8\%$ \\
$9$  & $0.5946$ & $0.1091$ & $18.4\%$ & $0.1232$ & $20.7\%$ \\
$10$ & $0.6431$ & $0.1184$ & $18.4\%$ & $0.1433$ & $22.3\%$ \\
\bottomrule
\end{tabular}
\endgroup
\end{table}

Table~\ref{tab:rate} is a diagnostic for the inverse-square exponent; its
finite report-distribution duals are solved to numerical tolerance.

\begin{table}[htbp]
\caption{The scaled deficit $m^2\delta$ stays between $0.07$ and $0.19$ from
$m=2$ to $64$ at both $K=1$ and $K=2$, as the inverse-square rate predicts.
Model-aware deficits of equal-width $m$-level reports for $S=2$,
homogeneous $\mathrm{Beta}(1,2)$ alternatives, and $\alpha=0.05$.  The $m=2$ row is
the equal-width binary report, not the optimized threshold report, and the
$K=2$ deficits are against the adaptive-quadrature oracle
$\pi_2^*\approx0.156461$.}
\label{tab:rate}
\centering
\begingroup
\small
\singlespacing
\setlength{\tabcolsep}{5pt}
\renewcommand{\arraystretch}{1.06}
\begin{tabular}{ccccc}
\toprule
& \multicolumn{2}{c}{$K=1$} & \multicolumn{2}{c}{$K=2$} \\
\cmidrule(lr){2-3}\cmidrule(lr){4-5}
$m$ & $\delta_1(\mathcal T^{(m)})$ & $m^2\delta_1$
    & $\delta_2(\mathcal T^{(m)})$ & $m^2\delta_2$ \\
\midrule
$2$  & $4.70\times10^{-2}$ & $0.188$ & $4.40\times10^{-2}$ & $0.176$ \\
$4$  & $6.42\times10^{-3}$ & $0.103$ & $4.47\times10^{-3}$ & $0.071$ \\
$8$  & $1.14\times10^{-3}$ & $0.073$ & $1.27\times10^{-3}$ & $0.082$ \\
$16$ & $4.83\times10^{-4}$ & $0.124$ & $3.19\times10^{-4}$ & $0.082$ \\
$32$ & $8.23\times10^{-5}$ & $0.084$ & $8.71\times10^{-5}$ & $0.089$ \\
$64$ & $2.00\times10^{-5}$ & $0.082$ & $2.10\times10^{-5}$ & $0.086$ \\
\bottomrule
\end{tabular}
\endgroup
\end{table}

For Table~\ref{tab:advantage}, each method uses its own common threshold,
selected by a grid search with kink candidates and local refinement
(\texttt{advantage\_table.py}; not certified), so the comparison is classwise
rather than the fixed-$(\mathsf P,\phi)$ identity of
Proposition~\ref{prop:reduction}.  Values come from the deterministic
composition reduction, with finite duals solved to numerical tolerance.  In
this setting, the null-calibrated column is the common displayed power of the
null-calibrated rule of Section~\ref{sec:marginal} with a Bonferroni step, its
Holm and Hommel variants, and the unanimity rule.  This agreement is numerical, not an identity: for
$\mathrm{Beta}(1,10)$ alternatives at $K=S=2$, separately optimizing the
common threshold gives power $0.7177$ for Hommel's procedure (equivalently
Hochberg's step-up rule~\citep{hochberg1988sharper} when $K=2$) and $0.6743$
for the unanimity rule, a difference of $0.0434$.  The
full-$p$-value columns apply Bonferroni at level $\alpha/K$ to Fisher's
statistic and to the likelihood ratio of the full $p$-values, computed by
\texttt{full\_pvalue\_baselines.py}; at $K=1$ the likelihood-ratio column is
the oracle of Table~\ref{tab:scaleS}.  The column after them is the gain of
the one-bit model-aware rule over the likelihood-ratio column, and the caption
describes the last two.

\begin{table}[htbp]
\caption{Reading one-bit reports jointly gains $62$--$69\%$ over the
null-calibrated rule (marginal aggregation) at $K=2$ and $330$--$419\%$ at $K=8$, and nearly as much over the
most powerful Bonferroni rule on the full $p$-values.  Full-alternative power
for homogeneous $\mathrm{Beta}(1,2)$ sites at $\alpha=0.05$: one-bit threshold
reports with null-calibrated (Bonferroni) and model-aware aggregation, each
rule selecting its own common threshold, and Bonferroni applied to the full
$p$-values with Fisher's statistic or the likelihood ratio.  Gains are
$(\text{model-aware}-\text{baseline})/\text{baseline}$, with the one-bit
null-calibrated power as the baseline of the first gain column and the
full-$p$-value likelihood-ratio Bonferroni power as that of the second.  For $S=2$ the last two
columns give the model-aware power of the two-bit equal-width report, which
cannot exceed the oracle power, and the lower bound
$1-(\text{one-bit model-aware})/(\text{two-bit model-aware})$ it implies for
the relative deficit of the one-bit model-aware rule at its numerically
selected threshold.  At $K=1$ the
optimized threshold already beats the
two-bit report, so the bound is vacuous there; the two-bit columns are not
computed for $S=3$ (Section~\ref{sec:numerics_adv}).}
\label{tab:advantage}
\centering
\begingroup
\small
\singlespacing
\setlength{\tabcolsep}{4pt}
\renewcommand{\arraystretch}{1.06}
\resizebox{\linewidth}{!}{%
\begin{tabular}{ccccccccc}
\toprule
& \multicolumn{3}{c}{one-bit threshold reports} & \multicolumn{2}{c}{full $p$-values, Bonferroni} & gain over & two-bit & one-bit loss \\
\cmidrule(lr){2-4}\cmidrule(lr){5-6}
$K$ & null-calibrated & model-aware & gain & Fisher & likelihood ratio & likelihood ratio & model-aware & at least \\
\midrule
\multicolumn{9}{l}{$S=2$} \\
$1$ & $0.158$ & $0.158$ & $0\%$   & $0.131$ & $0.160$ & $-1\%$  & --      & --        \\
$2$ & $0.085$ & $0.143$ & $69\%$  & $0.070$ & $0.086$ & $67\%$  & $0.152$ & $6.0\%$  \\
$4$ & $0.045$ & $0.132$ & $197\%$ & $0.037$ & $0.045$ & $195\%$ & $0.145$ & $9.0\%$  \\
$6$ & $0.030$ & $0.125$ & $313\%$ & $0.025$ & $0.031$ & $311\%$ & $0.140$ & $10.5\%$ \\
$8$ & $0.023$ & $0.120$ & $419\%$ & $0.019$ & $0.023$ & $417\%$ & $0.135$ & $11.6\%$ \\
\midrule
\multicolumn{9}{l}{$S=3$} \\
$1$ & $0.217$ & $0.217$ & $0\%$   & $0.165$ & $0.226$ & $-4\%$  & -- & -- \\
$2$ & $0.124$ & $0.202$ & $62\%$  & $0.091$ & $0.129$ & $57\%$  & -- & -- \\
$4$ & $0.069$ & $0.183$ & $164\%$ & $0.049$ & $0.071$ & $157\%$ & -- & -- \\
$6$ & $0.048$ & $0.170$ & $251\%$ & $0.034$ & $0.050$ & $242\%$ & -- & -- \\
$8$ & $0.037$ & $0.161$ & $330\%$ & $0.026$ & $0.038$ & $320\%$ & -- & -- \\
\bottomrule
\end{tabular}%
}
\endgroup
\end{table}

The synthetic study in Table~\ref{tab:simpower} has $S=3$ one-sided
normal-shift sites, $\eta=(1.0,1.5,2.0)$, $K=2$, $\alpha=0.05$, and
$\sigma^2=\sum_s\eta_s^2=7.25$, giving centralized-oracle power
$\pi_2^*=0.852$; Section~\ref{supp:sim} gives the computational details.
The model-aware thresholds, $(0.131,0.286,0.134)$, come from a numerical
search; the Fisher--Bonferroni thresholds, $(0.766,0.347,0.094)$, are the
unanimity design of Proposition~\ref{prop:alloc}, and the three other marginal
rules use the same threshold-report distribution.  In the threshold row the
model-aware and marginal rules therefore read different reports.  The
multilevel rows use equal-width $p$-value bins, and at finite resolution the
procedures need not keep the order they have with full reports.

Which change helps the marginal rules depends on the resolution.  Within each
statistic the two marginal columns of Table~\ref{tab:simpower} differ only in
the correction, and corresponding columns across statistics differ only in the
statistic.  At one bit the two statistics coincide:
under either ordering only the pattern in which every site rejects is
significant at level $\alpha/2$, and Holm adds $0.017$ before rounding.  At $m=4$ every
null tail is a multiple of $1/64$.  Fisher's statistic ties the three
patterns that report the second cell at one site and the first at the
others, so its tail jumps from $1/64$ to $4/64$ and Holm's second step at
level $0.05$ adds nothing; the likelihood ratio separates those patterns, and
likelihood-ratio Holm reaches $0.559$.  At $m=8$ and with full reports,
replacing Fisher's statistic by the likelihood ratio adds $0.148$ and
$0.049$ under Bonferroni, and Holm adds $0.06$ to $0.13$ within either
statistic.

\begin{table}[htbp]
\caption{The model-aware rule approaches the oracle power $0.852$ as the
equal-width report is refined ($0.797$ at three bits), and every marginal rule
stays below the model-aware rule at every resolution.  Full-alternative average power in the
synthetic multi-site study:
the model-aware rule, and four marginal rules that calibrate Fisher's
statistic or the report likelihood ratio against its exact null distribution
and then apply Bonferroni or Holm.  Fisher--Bonferroni is the null-calibrated
rule of Section~\ref{sec:marginal}, and likelihood-ratio Bonferroni is its
variant with the report likelihood ratio in place of Fisher's statistic
(Null-cal.\ (LR) in Figure~\ref{fig:theory_inverse_square}(d)).  Three one-sided normal-shift sites with $\eta=(1.0,1.5,2.0)$, $K=2$,
$\alpha=0.05$; oracle power $0.852$.  In the threshold row the model-aware rule
uses its own thresholds and the four marginal rules share the
Fisher--Bonferroni thresholds.}
\label{tab:simpower}
\centering
\begingroup
\small
\singlespacing
\setlength{\tabcolsep}{5pt}
\renewcommand{\arraystretch}{1.08}
\begin{tabular}{lccccc}
\toprule
& & \multicolumn{2}{c}{Fisher's statistic} & \multicolumn{2}{c}{likelihood ratio} \\
\cmidrule(lr){3-4}\cmidrule(lr){5-6}
report & model-aware & Bonferroni & Holm & Bonferroni & Holm \\
\midrule
threshold     & $0.728$ & $0.624$ & $0.642$ & $0.624$ & $0.642$ \\
equal-width $m=4$ & $0.703$ & $0.453$ & $0.453$ & $0.453$ & $0.559$ \\
equal-width $m=8$ & $0.797$ & $0.513$ & $0.639$ & $0.661$ & $0.748$ \\
full          & $0.852$ & $0.719$ & $0.788$ & $0.768$ & $0.833$ \\
\bottomrule
\end{tabular}
\endgroup
\end{table}

In Table~\ref{tab:simrobust}, overestimating the effect leaves the
model-aware family-wise error rate at
$0.0500$ and lowers power, a pointwise numerical finding rather than a
general monotonicity claim.  For the null-calibrated equal-width report, the map
and its null calibration are unchanged, so power and error rates are
unchanged.  For the null-calibrated threshold report, the thresholds are reselected under
$\hat\eta$ and then calibrated against their exactly known null law.
Its power changes from $0.6244$ under the correct model to $0.6240$ for
underestimation and $0.6239$ for overestimation; all round
to $0.624$ in Table~\ref{tab:simrobust}, and FWER remains $0.0494$.
This sensitivity enters through report design, not null calibration.  The two
overestimation rows of Table~\ref{tab:simrobust} complete the
misspecification sweep; Section~\ref{supp:sim} describes the exact evaluation.

\subsection{Details of the synthetic multi-site study}
\label{supp:sim}

The study of Section~\ref{sec:numerics_sim} (script
\texttt{simulation\_study.py}) uses $S=3$ one-sided normal-shift sites
with effects $\eta=(1.0,1.5,2.0)$ and $K=2$. Because
$\log g_s(u)=\eta_s\Phi^{-1}(1-u)-\eta_s^2/2$, the aggregate
log likelihood ratio is exactly $\mathcal N(-\sigma^2/2,\sigma^2)$ under the null and
$\mathcal N(\sigma^2/2,\sigma^2)$
under the alternative, where $\sigma^2=\sum_s\eta_s^2$.
The oracle and the full-report model-aware and likelihood-ratio entries are
computed from a $2{,}000$-cell equal-null-mass binning of this exact
distribution, and the full-report Fisher entries from a $100{,}000$-cell
per-site discretization calibrated against its exact null distribution.
With $1{,}000$, $2{,}000$, and $4{,}000$ cells the oracle value is $0.852242$,
$0.852249$, and $0.852249$, a change of less than $10^{-5}$, so the benchmark
is stable well beyond the displayed precision.

The null-calibrated threshold entries, $0.624$ in Tables~\ref{tab:simrobust}
and~\ref{tab:simpower}, use the unanimity design of
Proposition~\ref{prop:alloc}, which maximizes $\prod_s G_s(t_s)$ subject to
$\prod_s t_s=\alpha/K$; this is the best unanimity rule but not the best
null-calibrated threshold design.  A search over all threshold vectors and
all admissible critical values (\texttt{free\_marginal\_search} in
\texttt{simulation\_study.py}; not certified) finds a rule that rejects when
site~3 and at least one other site report significance; its power tends to
$0.628$ as its thresholds approach a tie between the sites' Fisher terms
$-\log\hat p_s$.  The Fisher--Bonferroni entry in the threshold row of
Table~\ref{tab:simpower} therefore slightly understates the best one-bit
null-calibrated power.

The model-aware thresholds come from a deterministic
multistart Nelder--Mead search over the eight-cell distribution and are not
certified.  Model-aware powers are dual values checked against the solver's
lower and upper bounds.  The study's explicit $K=2$ primal reproduces the dual
value to $10^{-8}$ with constraint residuals below $10^{-10}$; it is a linear
program over per-multiset action probabilities $r_1,r_2$ with the two
strong-FWER constraints $\E_0\{r_1v_{(2)}+r_2(v_{(1)}+v_{(2)})\}\le2\alpha$
and $\E_0(r_1+r_2)\le\alpha$.

The likelihood-ratio columns of Table~\ref{tab:simpower} use conservative tail
$p$-values of the report likelihood ratio: each atom receives the null mass of
all atoms whose likelihood ratio is at least its own.  With $K=2$ independent
hypotheses, Holm's expected number of rejections from marginal rejection
probabilities $F(x)$ at level $x$ is
$1-\{1-F(\alpha/2)\}^2+F(\alpha)^2-\{F(\alpha)-F(\alpha/2)\}^2$.  The equal-width and full
rows apply all five procedures to the same report protocol.  At the full
report the likelihood-ratio Bonferroni power equals its closed form
$\Phi\{\sigma-\Phi^{-1}(1-\alpha/2)\}=0.768$, because $\alpha/2$ falls on a
bin boundary.

For the model-aware misspecification experiment the center solves the program under the
rescaled shift vectors $\hat\eta=0.75\,\eta$ and $\hat\eta=1.25\,\eta$,
with the report thresholds left at the values selected under the true
$\eta$, and the induced policy's true error rates are computed exactly.
The null-calibrated threshold rule instead reselects its thresholds under
$\hat\eta$; its reference null law is recalibrated for that design without
using an alternative model.  The equal-width report map is fixed for both rules.
For each ordered report pair, the policy's rejection
probabilities follow from the primal actions and the ranking by assumed
likelihood ratio, with exact ties, if any, split uniformly.
Summing these probabilities under the global-null, mixed, and full-alternative
product distributions gives the strong family-wise error rate and power as
deterministic finite sums, conditional on the numerically solved policy
(constraint residuals below $10^{-10}$); when the program solved under
$\hat\eta$ has several optimal policies, these true rates depend on the one
the solver returns. In the two underestimation rows of Table~\ref{tab:simrobust}, whose error rates
exceed $\alpha$, the maximum over configurations is attained at a mixed
configuration.  The all-null rate is $0.0500$ to solver tolerance for each
optimized policy, because misspecification leaves the null report
distribution in that constraint unchanged.

Fisher--Bonferroni attains per-hypothesis null sizes
$1/64=0.015625$ at $m=4$ and $7/512=0.013671875$ at $m=8$, below
$\alpha/K=0.025$.  The next $m=8$ Fisher tail is $13/512>0.025$, so the
unused level is larger than at $m=4$.  With full reports, the remaining power
gap reflects Fisher aggregation with marginal Bonferroni correction.

The correlation experiment draws, within each site, jointly Gaussian
$z$-statistics for the two hypotheses with correlation $\rho$
(independent across sites).
It applies the fixed correct-model policies and estimates error rates by
averaging the policy's conditional rejection probabilities given the observed
report pair (Rao--Blackwellization).
Thus Monte Carlo noise enters only through report-cell frequencies.  The
experiment uses $N=4\times10^6$ replications with seed $20260723$; standard
errors are at most $2.3\times10^{-4}$. At $\rho=0$ all corresponding estimates
lie within two standard errors of their deterministic values, a check of the
simulation against the exact calculation.

\section{Data illustrations and sensitivity analyses}\label{supp:applications}

This section documents the two data illustrations of the main text: the
Federated Tumor Segmentation (FeTS) reanalysis (Section~\ref{supp:fets}) and Many Labs~1 (Section~\ref{supp:manylabs}).
Both use only the null-calibrated rule of Section~\ref{sec:marginal}, because
their hypotheses share cases or participants, and
Proposition~\ref{prop:marginal} gives the family-wise guarantee.  Each
subsection gives the data and exclusions, the local hypotheses and $p$-values
with their validity conditions, the calibrations, the results, and the
sensitivity analyses.

\subsection{FeTS reanalysis: design and validity}
\label{supp:fets}

\paragraph{Provenance and analysis population}
The FeTS 2022 Task~2 benchmark used the MedPerf platform to evaluate
brain-tumor segmentation algorithms at geographically distributed
institutions~\citep{zenk2025fets,karargyris2023medperf}.  The test images and reference
segmentations remained with the data owners and were inaccessible to
the Task~2 organizer during the decentralized evaluation; algorithm
containers traveled to the data, and the
resulting performance metrics were returned.  The subsequently released
case-level metrics permit a reanalysis under counterfactual reporting protocols.
The accompanying code verifies the SHA-256 digests of the source archive and
of its \texttt{task2\_results.csv}; file names, hashes, and the commands that
reproduce the analysis are in \texttt{data/README.md} of the code and data
supplement.

The released data set contains 2,625 retained cases, 32 institution
identifiers, and results for 41 models.  The accompanying source metadata labels six
institutions (IDs 27--32) as seen during training and the other 26 as unseen.
We exclude the six seen institutions, which also contributed other patients to
the training data~\citep[Fig.~3]{zenk2025fets}, so that every included
evaluation is at an institution whose data the models did not see in
training.  At unseen institution 36 (identifiers
are not contiguous), all three Dice
values for model~8 are missing for all 81 cases because its evaluation failed; we
exclude that institution before forming the family so that every claim uses
one common, complete set of sites.  This leaves $S=25$ institutions and
2,178 cases, with institution sizes from 9 to 250 (median 83).  The
included institutions are relabeled as $s=1,\ldots,S$.  The FeTS article
reports that 1,201 cases with shareable reference segmentations from
16 institutions were visually screened, and that 125 cases
with major annotation errors were removed from its final analysis.  We use
the authors' published final-analysis data set; we neither reconstruct the
unreleased exclusions nor treat their quality-control step as a prospective
eligibility rule.  Accordingly, for the inferential statements below, each
retained set of case-level evaluations is assumed to represent i.i.d.\ patient
draws from the post-quality-control institutional evaluation population.
That sampling statement, including distinctness of the sampled units, is an
assumption, not a consequence of the released data set.

The primary family is built from the five official Task~2 submissions, models
8, 10, 11, 12, and 54, among which the challenge winner was determined; the
other released models, which include models from the BraTS 2021 challenge,
were evaluated only for a secondary
ranking~\citep[Results and Supplementary Table~2]{zenk2025fets}.  Model~10 is
focal because its documented method is the only one that specifically
targets distribution shift between training and test data~\citep{zenk2025fets}; it is compared with each
model
$c\in\mathcal C=\{8,11,12,54\}$ for enhancing tumor (ET), tumor core (TC),
and whole tumor (WT).  Thus the fixed family has $K=4\times3=12$ claims.
The released metrics for model~54 are the corrected version: the FeTS
Supplementary Information notes that a bug affected its official challenge
submission and was fixed later~\citep[Supplementary Table~2]{zenk2025fets}.  The family definition is scientifically motivated but was
constructed retrospectively after publication of the benchmark; the
selection qualification in the interpretive limits below is therefore essential.

\paragraph{Local and intersection hypotheses}
Let $D_{s,i,a,r}$ be the released Dice score for case $i$ at institution $s$,
model $a$, and region $r$.  For $c\in\mathcal C$, define
\[
 X_{s,i,c,r}=\mathbf 1\{D_{s,i,10,r}>D_{s,i,c,r}\},
 \qquad
 \omega_{s,c,r}=\pr\{D_{s,i,10,r}>D_{s,i,c,r}\}.
\]
Equality at the released numeric precision is a non-win.  The local null and
the $K$ intersection nulls are
\[
 H_{s,c,r}:\omega_{s,c,r}\leq\tfrac12,
 \qquad
 H_{c,r}=\bigcap_{s=1}^{25}H_{s,c,r}.
\]
If the retained cases at institution $s$ represent i.i.d.\ draws from the
stated population and the prediction pipelines are fixed and case-local,
meaning that each prediction depends on the evaluation data only through
the current case,
then, conditionally on its case count $n_s$,
\begin{equation}
 B_{s,c,r}=\sum_{i=1}^{n_s}X_{s,i,c,r}
 \sim \operatorname{Bin}(n_s,\omega_{s,c,r}),
 \qquad
 p_{s,c,r}=\pr\{\operatorname{Bin}(n_s,1/2)\geq B_{s,c,r}\}.
 \label{eq:fets_local}
\end{equation}
The displayed upper tail is an exact finite-sample, generally conservative
binomial majority-win $p$-value: under the composite null
$\omega_{s,c,r}\leq1/2$, stochastic ordering of the binomial count makes
$p_{s,c,r}$ superuniform.  No continuous model for Dice differences is used.

The binomial $p$-values also require each prediction to depend only on its
own case.  This needs checking for model~10, which adapts to the test data
(test-time adaptation): its method description~\citep[Methods, Team
Sanctuary]{zenk2025fets} states that model~10 recomputes the means and
standard deviations of its batch-normalization layers from the batch at
prediction time, with a batch size of one, so that each case is normalized by
its own statistics, as in instance normalization.  We therefore treat its prediction as a
function of the current case.  Exactness of~\eqref{eq:fets_local} nevertheless
requires that the evaluated implementation carry no prediction-relevant
state from one case to the next.  If that condition or the within-institution
i.i.d.\ sampling condition fails, the primary binomial $p$-values need not be valid; the
model-10-excluded sensitivity below avoids the model-10-specific condition.

For each fixed $(c,r)$, validity of the Fisher-combination analysis additionally requires
mutual independence of the 25 institutional $p$-values under $H_{c,r}$.
Dependence across the 12 claims is unrestricted.  In particular, the same
released cases supply all model comparisons, ET and TC are nested
within WT, and each comparison shares model~10.  These dependencies are why this reanalysis
uses only the null-calibrated rule of Proposition~\ref{prop:marginal}; it
makes no use of the independence or common-alternative assumptions of the
model-aware rule.

\paragraph{Counterfactual reports and their validity}
We set $\alpha=0.05$ and write $\beta=\alpha/K$ for the per-hypothesis
level.  Each report is read as the site $p$-value it implies.  A unanimity threshold report, the \emph{unanimity bit}, transmits $\mathbf 1\{p\leq t\}$ with
$t=\beta^{1/S}=0.8031405$, at which unanimity has probability $\beta$ under
the independent-uniform reference; its implied site $p$-value is $t$ on a local
rejection and $1$ otherwise.  An equal-width $m$-level report transmits the
cell index $j_m(p)=\min\{1+\lfloor mp\rfloor,m\}\in\{1,\ldots,m\}$, whose
implied site $p$-value is the cell's upper edge
\begin{equation}
 \hat p_m(p)=j_m(p)/m .
 \label{eq:fets_qm}
\end{equation}
The full report transmits $p$ itself. Equation~\eqref{eq:fets_qm} and the
threshold and full-report definitions give $\hat p\geq p$ pointwise for every
report, so $\hat p$ is superuniform whenever $p$ is; the exact binomial $p$-values used here are superuniform under $H_{c,r}$ when the sampling and case-local conditions above hold.

The center uses the exact-null tail calibration of
Proposition~\ref{prop:marginal}.  For the threshold and full reports the
statistic is Fisher's combination
$W_{c,r}=-\sum_{s=1}^{25}\log\hat p_{s,c,r}$; its uniform-reference
distribution is determined by a binomial local-rejection count for the
threshold report and is $\chi^2_{2S}/2$ for the full report.
For equal-width $m$-level reports, numerical convolution uses integer score units
with lattice spacing $\delta=10^{-5}$, and the observed statistic uses exactly the
same units:
\[
 Q^{(m)}_{c,r}=\sum_{s=1}^{25}
 \left\lfloor\frac{-\log\hat p_{m}(p_{s,c,r})}{\delta}\right\rfloor,
 \qquad W^{(m)}_{c,r}=\delta Q^{(m)}_{c,r}.
\]
The reference law of $Q^{(m)}$ is the convolution of $S$ independent
copies of the $m$ equally likely integer scores
$\lfloor-\log(j/m)/\delta\rfloor$, $j=1,\ldots,m$.
Writing $q_{\beta,m}$ for the smallest nonnegative integer with
$\pr_0(Q^{(m)}\ge q_{\beta,m})\le\beta$, the center rejects when
$Q^{(m)}_{c,r}\ge q_{\beta,m}$.  The corresponding lattice cutoffs
$\tau_{\beta,m}=\delta q_{\beta,m}$ are
\[
\begin{array}{lcccccc}
\text{report} & m=2 & m=4 & m=8 & m=16 & m=32 & m=64\\
\tau_{\beta,m} & 13.170 & 21.827 & 27.968 & 32.179 & 35.004 & 36.863
\end{array}
\]
with $\tau_\beta=40.170$ for the full report and, for the threshold report, a
count cutoff of all $25$ sites.  The reference-null sizes are $0.00204$ and $0.00409$ for $m=2$ and $4$,
lie within $2\times10^{-7}$ below $\beta=0.05/12=0.0041\overline{6}$ for
$m=8,\ldots,64$, and equal $\beta$ for both the full and the threshold report.  Rounding changes the score by
$0\le W_{c,r}-W^{(m)}_{c,r}<S\delta$ and preserves its monotonicity in
each local $p$-value.  The reference calibration and the observed comparison
both use $Q^{(m)}$, so the tail corresponds to the statistic actually tested.
In particular, the equal-width one-bit rule ($m=2$) is exactly a count cutoff of $20$ out of
$25$, independently checked by
$\pr\{\mathrm{Bin}(25,1/2)\ge20\}=0.0020386577$.
A union bound over the true intersection hypotheses gives
strong FWER control at $\alpha$ with no assumption on dependence across
$(c,r)$.

\paragraph{Results}
The unanimity bit rejects five of the twelve nulls and the equal-width
one-bit report eight, and the difference comes entirely from where the single
cell boundary is placed.  In Table~\ref{tab:fets_resolution}, an entry $c$--$r$ denotes rejection of
$H_{c,r}$.  The unanimity bit and the equal-width $m=2$ report are both binary, but the unanimity rule requires $p\le t=0.803$ at all $25$ sites, whereas the
equal-width one-bit rule needs only $p<1/2$ at $20$ of the $25$.
The two binary protocols therefore
admit different patterns of local evidence.  Equal-width refinements need
not give nested rejection sets, though here they happen to be nested.  The
table is an observed-data comparison of report protocols, not a power
estimate.  To display the complete rejection sets compactly, write
$\mathcal A_c=\{c\text{--ET},c\text{--TC},c\text{--WT}\}$ for
$c\in\{8,11,12,54\}$ and
$\mathcal B_{11}=\{11\text{--ET},11\text{--TC}\}$.

\begin{table}[htbp]
\caption{FeTS rejections rise from $5$ with the unanimity bit and $8$ with the
equal-width one-bit report to all $12$ with full $p$-values, with nested
rejection sets; in the model-10-excluded family they rise from $4$ to $17$
of $18$.  FeTS reanalysis at $\alpha=0.05$.  The $K=12$ columns give the number
and the set of rejected model-10-versus-comparator nulls in the fixed primary
family; an entry $c$--$r$ is a rejection of $H_{c,r}$,
$\mathcal A_c=\{c\text{--ET},c\text{--TC},c\text{--WT}\}$ and
$\mathcal B_{11}=\{11\text{--ET},11\text{--TC}\}$.  The $K=18$ column gives only
the rejection count for the model-10-excluded family of two-sided pairwise
comparisons among the other four official models, defined below.  The $K=60$
counts of Table~\ref{tab:fets} of the main text are the second row of
Table~\ref{tab:fets_exhaustive}.  Every row is a separate counterfactual FWER
procedure.}
\label{tab:fets_resolution}
\centering
\begingroup
\small
\setlength{\tabcolsep}{4pt}
\begin{tabular}{lrrl}
\toprule
report & $K=12$ count & $K=18$ count & $K=12$ rejection set \\
\midrule
unanimity bit & 5 & 4 & $\mathcal A_8\cup\{12\text{--ET},12\text{--WT}\}$ \\
equal-width $m=2$ & 8 & 14 & $\mathcal A_8\cup\mathcal A_{12}\cup\{11\text{--ET},54\text{--ET}\}$ \\
equal-width $m=4$ & 10 & 15 & $\mathcal A_8\cup\mathcal A_{12}\cup\mathcal A_{54}\cup\{11\text{--ET}\}$ \\
equal-width $m=8$ & 11 & 16 & $\mathcal A_8\cup\mathcal A_{12}\cup\mathcal A_{54}\cup\mathcal B_{11}$ \\
equal-width $m=16$ & 11 & 17 & $\mathcal A_8\cup\mathcal A_{12}\cup\mathcal A_{54}\cup\mathcal B_{11}$ \\
equal-width $m=32$ & 11 & 17 & $\mathcal A_8\cup\mathcal A_{12}\cup\mathcal A_{54}\cup\mathcal B_{11}$ \\
equal-width $m=64$ & 11 & 17 & $\mathcal A_8\cup\mathcal A_{12}\cup\mathcal A_{54}\cup\mathcal B_{11}$ \\
full $p$-value & 12 & 17 & $\mathcal A_8\cup\mathcal A_{11}\cup\mathcal A_{12}\cup\mathcal A_{54}$ \\
\bottomrule
\end{tabular}
\endgroup
\end{table}

Pooled case-weighted summaries are retained in \texttt{data/README.md} for
scale, but are deliberately omitted here: they ignore site heterogeneity
and have no role in the tests, and Dice differences are technical segmentation
metrics rather than clinical effect sizes.

\paragraph{Sensitivity to arbitrary cross-site dependence}
A Bonferroni partial-conjunction calculation on the full local $p$-values
(Table~\ref{tab:fets_partial_conjunction}) does not need the cross-site independence assumption that the Fisher-combination analysis relies on,
and it also gives simultaneous lower bounds on how many included institutions have an advantage.  Define
\[
 N_{c,r}=\sum_{s=1}^{S}\mathbf 1\{\omega_{s,c,r}>1/2\},
 \qquad
 H^{[\ell]}_{c,r}:N_{c,r}\leq \ell-1,
 \qquad \ell=1,\ldots,S.
\]
Thus $H^{[\ell]}_{c,r}$ asserts that fewer than $\ell$ of the $S=25$
institution-specific nulls are false.  Order the local values as
$p_{(1),cr}\leq\cdots\leq p_{(S),cr}$ and set
\begin{equation}
 q_{\ell,c,r}=\min\{1,(S-\ell+1)p_{(\ell),cr}\},
 \qquad
 \bar q_{\ell,c,r}=\max_{1\leq j\leq\ell}q_{j,c,r}.
 \label{eq:fets_pc}
\end{equation}
Under arbitrary dependence, the first quantity is the Bonferroni
partial-conjunction $p$-value~\citep{benjamini2008screening}.  Indeed, under
$H^{[\ell]}_{c,r}$, choose a set $I$ of $S-\ell+1$ true local hypotheses.  Its
complement contains only $\ell-1$ indices, so
$p_{(\ell),cr}\geq\min_{s\in I}p_{s,c,r}$.  For $0\leq x<1$, marginal
superuniformity and the union bound therefore give
\[
 \pr\{q_{\ell,c,r}\leq x\}
 \leq \sum_{s\in I}\pr\{p_{s,c,r}\leq x/(S-\ell+1)\}
\leq x.
\]
This argument imposes no condition on dependence among institutions.

To select $\ell$ while retaining simultaneous coverage, we use the prefix
maximum in~\eqref{eq:fets_pc} and report
\begin{equation}
 \underline N_{c,r}
 =\max\bigl(\{0\}\cup
   \{\ell:\bar q_{\ell,c,r}\leq\alpha/K\}\bigr),
 \qquad K=12.
 \label{eq:fets_pc_bound}
\end{equation}
The prefix maximum is essential because the
raw $q_{\ell,c,r}$ need not be
monotone in $\ell$.  If the true count is $N_{c,r}=n_0$ but
$\underline N_{c,r}>n_0$, then necessarily $n_0<S$ and
$\bar q_{n_0+1,c,r}\leq\alpha/K$.  The null $H^{[n_0+1]}_{c,r}$ is true, and
$q_{n_0+1,c,r}\leq\bar q_{n_0+1,c,r}\leq\alpha/K$.  A union bound over the $K$
claims consequently gives
\[
 \pr\{\underline N_{c,r}\leq N_{c,r}\text{ for every }(c,r)\}
 \geq 1-\alpha,
\]
with arbitrary dependence both across institutions and across claims.

\begin{table}[htbp]
\caption{Every claim has a positive simultaneous lower bound on the number of
advantaged institutions, and ten have at least four, under arbitrary
dependence.  Family-wise simultaneous lower bounds~\eqref{eq:fets_pc_bound}
on the number of the 25 included institutions satisfying $\omega_{s,c,r}>1/2$, at $\alpha=0.05$ over
the fixed $K=12$ family.  This full-report partial-conjunction analysis
requires marginal validity of the local $p$-values but permits arbitrary
dependence across institutions and claims.}
\label{tab:fets_partial_conjunction}
\centering
\begin{tabular}{lrrr}
\toprule
comparator & ET & TC & WT \\
\midrule
8  & 20 & 23 & 24 \\
11 &  8 &  1 &  1 \\
12 & 15 & 16 & 21 \\
54 &  8 &  4 &  4 \\
\bottomrule
\end{tabular}
\end{table}

All twelve bounds are positive, and ten are at least four.  These bounds refer only to the included institutions; they do not establish
transportability to another hospital.  The partial-conjunction analysis is a separate inferential
procedure, not an additional report-resolution row, and it does not validate the Fisher-combination analysis under dependent institutional samples.

\paragraph{Model-10-excluded sensitivity}
Without model~10 the ordering of report protocols is unchanged: among the 18
comparisons of the other four official models, the unanimity bit rejects 4,
the equal-width one-bit report 14, reports with 4 and 8 cells 15 and 16, and
finer reports and full $p$-values 17 (Table~\ref{tab:fets_resolution}).  This
family checks the protocol ordering without our case-local reading of
model~10's published test-time-adaptation description; it does not test the
model-10 claims.  We exclude model~10 and form all six
unordered pairs among models
$\{8,11,12,54\}$ and cross them with all three regions, giving $K=18$.  For a pair $(a,b)$
at a site, let $B_a$ and $B_b$ count strict wins by $a$ and $b$ among all
$n_s$ cases; ties enter neither total and are therefore failures in
both directions.  Write
\[
\begin{aligned}
 p_a&=\pr\{\operatorname{Bin}(n_s,1/2)\geq B_a\},&
 p_b&=\pr\{\operatorname{Bin}(n_s,1/2)\geq B_b\},\\
 p^{\pm}&=\min\{1,2\min(p_a,p_b)\}.&&
\end{aligned}
\]
For each site, define
$H^{\pm}_{s,a,b,r}:\pr(D_{s,i,a,r}>D_{s,i,b,r})\leq1/2\ \text{and}\
\pr(D_{s,i,b,r}>D_{s,i,a,r})\leq1/2$, and define the corresponding intersection null
$H^{\pm}_{a,b,r}=\bigcap_s H^{\pm}_{s,a,b,r}$.  Under the local composite null,
both one-sided values are valid, and Bonferroni's inequality makes
$p^{\pm}$ valid even though the two counts are dependent.  This sensitivity
still assumes that the four remaining released pipelines are fixed and
case-local and that patient sampling has the independence properties stated
above; the FeTS article establishes only that model~10 was the sole official
submission described as targeting distribution shift, not that the other four
pipelines are case-local.  The construction differs
from the ordinary two-sided sign test
that discards ties: it preserves the estimand of a strict majority among
\emph{all} sampled units.  Applying the same report protocols with
$t=(0.05/18)^{1/25}$ and per-hypothesis level $0.05/18$ gives the sensitivity
counts in Table~\ref{tab:fets_resolution}.  Rejecting $H^{\pm}_{a,b,r}$ means
that one direction has a strict-majority advantage at one or more included
institutions; it does not identify the institution or establish a common
direction across institutions.

\paragraph{Exhaustive ordered-pair family sensitivity}
To reduce the effect of choosing a focal model, comparator, and direction
within the benchmark, we form every ordered pair $(a,b)$, $a\ne b$, among
the five official submissions and cross it with the three Dice regions
(Table~\ref{tab:fets_exhaustive}).  The resulting $K=5\times4\times3=60$ family uses the one-sided local nulls
\[
 H_{s,a,b,r}:\pr\{D_{s,i,a,r}>D_{s,i,b,r}\}\leq\tfrac12,
 \qquad
 H_{a,b,r}=\bigcap_{s=1}^{25}H_{s,a,b,r},
\]
with equality at the released precision counted as a non-win.  We apply the
same upper binomial tail as in~\eqref{eq:fets_local}, with $a$ in place of
model~10 and $b$ in place of $c$, and repeat the Fisher-combination analysis separately at each report resolution.  The recalibrated threshold is
$t=(0.05/60)^{1/25}$ and the per-hypothesis level is $0.05/60$.

\begin{table}[htbp]
\caption{Correcting over all $60$ ordered comparisons keeps all five
unanimity-bit and $11$ of the $12$ full-report focal rejections; across the
$60$ claims, rejections rise from $13$ to $29$ with resolution.  Exhaustive
ordered-pair family sensitivity.  The first row is the number of nulls rejected among all $K=60$ ordered
official-model--region claims.  The second row counts rejections only among
the 12 model-10 directions, under the same $K=60$ correction.  Each column is a separate counterfactual FWER procedure under the stated validity conditions of the Fisher-combination analysis.}
\label{tab:fets_exhaustive}
\centering
\small
\setlength{\tabcolsep}{4pt}
\begin{tabular}{lrrrrrrrr}
\toprule
family & unanimity bit & $m=2$ & $m=4$ & $m=8$ & $m=16$ & $m=32$ & $m=64$ & full \\
\midrule
all 60 claims & 13 & 22 & 23 & 27 & 28 & 28 & 29 & 29 \\
focal 12 claims &  5 &  8 &  8 & 11 & 11 & 11 & 11 & 11 \\
\bottomrule
\end{tabular}
\end{table}

All five focal unanimity-bit rejections persist after correction over
this larger family, as do eleven of the twelve focal full-report rejections;
the one that does not is $11$--WT, the last to enter in
Table~\ref{tab:fets_resolution}, where only full $p$-values reject it.  Among
intermediate resolutions only $m=4$ loses focal rejections, $10$ to $8$, as the
second row shows.  The construction is exhaustive only within the universe of
ordered Dice comparisons among the five official submissions.  Opposite
directions are separate claims because site-specific advantages may differ
in direction; the focal robustness comparison uses only the model-10
directions.

\paragraph{Released-score and exact-profile audits}
Recomputing Dice from the released confusion counts changes only two of the
$12\times2{,}178$ primary win indicators; all displayed rejection sets and
partial-conjunction bounds are unchanged.  One pair among 2,625 released cases
has an identical complete metric profile, and deleting either member
again leaves every displayed conclusion unchanged.  The full audit, including
the affected cases and variables, is in \texttt{data/README.md}.

\paragraph{Interpretive limits}
Rejection of $H_{c,r}$ is evidence against the intersection of the 25 local
no-majority-advantage nulls.  It supports an advantage at one or more
included institutions; it does not establish an advantage at every
institution, transportability to a new hospital, patient benefit, deployment
safety, or clinical utility.

The decentralized evaluation images remain protected, but this reanalysis
uses publicly released metrics and makes no privacy or differential-privacy
claim.  If its displayed family and calibration are treated as fixed in
advance, each Fisher-combination row of
Table~\ref{tab:fets_resolution} and each column of
Table~\ref{tab:fets_exhaustive} separately has
the stated FWER guarantee under valid local $p$-values and cross-site
independence.  The partial-conjunction bounds in
Table~\ref{tab:fets_partial_conjunction} instead permit arbitrary cross-site
dependence but still require marginal validity of the local $p$-values.
Neither the released-score audit nor the exact-profile audit verifies
within-institution i.i.d.\ sampling or case-local execution.

There is no level-$0.05$ guarantee for choosing a resolution
after viewing the rows, for taking the union of their rejection sets, or for
the retrospective analyst selection of FeTS, the cohort, Dice and the
strict-win estimand, the focal model, or the claim family. Within its defined
ordered-Dice universe, the $K=60$ sensitivity addresses only focal-model,
comparator, direction, and claim-family selection. The reanalysis is therefore a
reproducible, retrospective illustration of the effect of report resolution,
not a clinical validation study.

\subsection{Many Labs 1: data, tests, and per-test results}\label{supp:manylabs}

\paragraph{Data}
\citet{klein2014investigating} released per-sample summary statistics for
every test in the spreadsheet \texttt{ML-\_Summary\_Statistics.xlsx} on the
Open Science Framework (project \texttt{wx7ck}, file \texttt{dmf62},
version~3), under the CC0 public-domain dedication.  The script
\texttt{code/extract\_manylabs1.py} reads that version, verifies its SHA-256
digest, and writes the statistics used here to
\texttt{data/manylabs1\_sites.csv}.  The script
\texttt{code/manylabs\_application.py} computes the per-test counts and
rejections in Table~\ref{tab:manylabs_results} and checks the rejection
counts in this section and in Section~\ref{sec:application} of the main text
against the values reported.

\paragraph{Samples}
Of the 36 samples, one lost its implicit-association data to a recording
error, as the project reports, and so lacks two of the tests; it is dropped
so that every site reports every test.  The remaining 35 samples have between
43 and $1{,}328$ participants per test, with median 98.  They include two open
online pools, Amazon Mechanical Turk and Project Implicit, whose participants
may also have taken part in other samples, which would make samples dependent;
this cannot be checked, so the sensitivity analysis below drops them.

\paragraph{Tests and site $p$-values}
Table~\ref{tab:manylabs_tests} lists the 16 tests and the direction each
original study predicted.  Eleven compare two groups on a numerical
response: ten compare randomly assigned conditions, and the sex difference in
implicit math attitudes compares women with men.  For these the site $p$-value is the upper tail of Welch's $t$ statistic, which
is positive in the predicted direction; the statistic and its degrees of
freedom are the unequal-variance values in the project's file, in which
anchoring responses are rank-transformed.  Four compare
a binary response between two randomly assigned versions of a question, and
the site $p$-value is that of the one-sided Fisher exact test on the sample's
two-by-two table.  The relation between implicit and explicit math attitudes
uses the one-sided correlation test with $n-2$ degrees of freedom.

For allowed versus forbidden wording we test the contrast of the original study,
the share of respondents opposing the speeches under each wording; the
project's summary chi-square statistic instead cross-tabulates the raw yes
and no answers by wording, although a yes means opposite things under the two
wordings.  For the other three two-by-two tests our tables reproduce the
project's chi-square statistic in every sample.  Because we use Fisher's exact test and the unequal-variance $t$
statistics, our counts of significant samples differ slightly from those the
project tabulated for some tests.

\begin{table}[!tb]
\caption{Every Many Labs 1 test has a direction fixed in advance by the
original study, so each sample contributes a one-sided $p$-value.  The 16
tests and the predicted directions.  Two-group comparisons use Welch's $t$, two-by-two tables the
one-sided Fisher exact test, and the relation between implicit and explicit
math attitudes the correlation $t$ test; ``Fisher'' denotes the Fisher exact
test.}
\label{tab:manylabs_tests}
\centering
\begingroup
\small
\singlespacing
\setlength{\tabcolsep}{5pt}
\begin{tabularx}{\textwidth}{@{}>{\raggedright\arraybackslash}p{0.3\textwidth}>{\raggedright\arraybackslash}Xl@{}}
\toprule
Test & Predicted direction & Test statistic \\
\midrule
Anchoring: babies born per day, height of Mount Everest, population of
Chicago, distance to New York City (four tests) & higher estimates after a
high anchor than after a low one & Welch $t$ on ranks \\
Allowed versus forbidden wording & more respondents oppose speeches against
democracy when asked whether to allow them than when asked whether to forbid
them & Fisher \\
Implicit and explicit math attitudes & positive correlation & correlation \\
Retrospective gambler's fallacy & more rolls estimated after seeing three
sixes than after two sixes and a three & Welch $t$ \\
Gain versus loss framing & more choose the certain option when outcomes are
framed as gains & Fisher \\
Sex difference in implicit math attitudes & women's implicit attitudes
favor arts over math more than men's & Welch $t$ \\
Norm of reciprocity & more agree to admit a rival country's reporters when
the question follows the reciprocal one & Fisher \\
Sunk costs & more likely to attend a game in bad weather when the ticket was
paid for & Welch $t$ \\
Low versus high category scales & more report watching television for over
2.5 hours a day on the high-category scale & Fisher \\
Quote attribution & more agreement with a quote attributed to a liked source
& Welch $t$ \\
Imagined contact & stronger intentions to interact with an outgroup after
imagining contact with its members & Welch $t$ \\
Flag priming & more conservative attitudes after a flag prime & Welch $t$ \\
Currency priming & stronger system justification after a money prime &
Welch $t$ \\
\bottomrule
\end{tabularx}
\endgroup
\end{table}

\paragraph{Validity}
For the fourteen randomized comparisons, random assignment within each
sample makes each one-sided test valid under its null; the sex difference and
the implicit--explicit correlation rest instead on independent sampling of
participants within each sample.  Fisher's exact test has exact level, and the Welch and
correlation tests are valid to normal-theory accuracy at these sample sizes.
The samples are disjoint by design, apart from the two online pools, whose
participants cannot be checked.  The dependence across tests created by
participants completing every task is allowed by
Proposition~\ref{prop:marginal}, so the rule has nominal family-wise level
$0.05$.  The family and its directions
were fixed by the registered protocol.  The project's confirmatory
flag-priming test used United States samples only; using all 35 samples still
gives a valid test of the null of no effect in any sample, and the project's
own United States test was also not significant ($t=0.88$).

\paragraph{Reports}
With $K=16$ and $S=35$ the per-test level is $\beta=\alpha/K=0.003125$.  The
reports are those of Section~\ref{supp:fets} together with a 5\% bit, all
calibrated in the same way.  Write $u$ for a sample's one-sided $p$-value.
The unanimity bit is $\mathbf 1\{u\le t\}$ with $t=\beta^{1/35}\approx0.848$, and the
center rejects only when all 35 samples report it.  The 5\% bit,
significance at the two-sided 5\% level in the predicted direction, is
$\mathbf 1\{u\le0.025\}$.  Under the independent-uniform reference
distribution, the number of samples reporting it is binomial with success
probability $0.025$.  The center rejects when at least five samples report
it, giving a reference-null tail probability of approximately $0.0017$.
The two-cell equal-width report is the
direction bit $\mathbf 1\{u<1/2\}$, for which the center rejects when at
least 26 samples report it; no sample has $u=1/2$ exactly.  For the $t$ tests
the bit records the sign of the estimate.  For Fisher's exact test it can
differ: in two samples the norm-of-reciprocity odds ratio is above one while
$u>1/2$.  Equal-width reports with up to 64 cells and the
full $p$-values use the lattice and $\chi^2$ calibrations of
Section~\ref{supp:fets}.  Table~\ref{tab:manylabs_results} gives the per-test
counts and rejections.

\begin{table}[htbp]
\caption{Only three findings separate the reports: the direction bit and the
full $p$-values detect all three, the 5\% bit two, and the unanimity bit none.
Many Labs 1 per-test results at nominal family-wise level $0.05$ with 35 samples: the
number of samples significant at the two-sided 5\% level in the predicted
direction, with one-sided $p$-value $u<1/2$ (for the $t$ tests, an estimate
pointing the predicted way), and with $u>0.848$, beyond the unanimity threshold; and
whether each report rejects.  The unanimity bit rejects only if no sample has
$u>0.848$, the 5\% bit if at least five samples are significant, and the
direction bit if at least 26 point the predicted way.  Equal-width reports
with 4 to 64 cells agree with the direction bit.}
\label{tab:manylabs_results}
\centering
\begingroup
\small
\singlespacing
\setlength{\tabcolsep}{3pt}
\resizebox{\linewidth}{!}{%
\begin{tabular}{@{}lccccccc@{}}
\toprule
 & \multicolumn{3}{c}{samples} & \multicolumn{4}{c}{rejected by} \\
\cmidrule(lr){2-4}\cmidrule(lr){5-8}
Test & significant & predicted way & $u>0.848$ & \begin{tabular}[b]{@{}c@{}}unanimity\\ bit\end{tabular} & 5\% bit & \begin{tabular}[b]{@{}c@{}}direction\\ bit\end{tabular} & full $p$ \\
\midrule
Anchoring: babies born per day & 35 & 35 & 0 & $\checkmark$ & $\checkmark$ & $\checkmark$ & $\checkmark$ \\
Anchoring: height of Mount Everest & 35 & 35 & 0 & $\checkmark$ & $\checkmark$ & $\checkmark$ & $\checkmark$ \\
Anchoring: population of Chicago & 35 & 35 & 0 & $\checkmark$ & $\checkmark$ & $\checkmark$ & $\checkmark$ \\
Anchoring: distance to New York City & 35 & 35 & 0 & $\checkmark$ & $\checkmark$ & $\checkmark$ & $\checkmark$ \\
Allowed versus forbidden wording & 25 & 35 & 0 & $\checkmark$ & $\checkmark$ & $\checkmark$ & $\checkmark$ \\
Implicit and explicit math attitudes & 33 & 35 & 0 & $\checkmark$ & $\checkmark$ & $\checkmark$ & $\checkmark$ \\
Retrospective gambler's fallacy & 31 & 35 & 0 & $\checkmark$ & $\checkmark$ & $\checkmark$ & $\checkmark$ \\
Gain versus loss framing & 31 & 35 & 0 & $\checkmark$ & $\checkmark$ & $\checkmark$ & $\checkmark$ \\
Sex difference in implicit math attitudes & 24 & 35 & 0 & $\checkmark$ & $\checkmark$ & $\checkmark$ & $\checkmark$ \\
Norm of reciprocity & 11 & 32 & 0 & $\checkmark$ & $\checkmark$ & $\checkmark$ & $\checkmark$ \\
Sunk costs & 15 & 32 & 0 & $\checkmark$ & $\checkmark$ & $\checkmark$ & $\checkmark$ \\
Low versus high category scales & 20 & 34 & 1 & -- & $\checkmark$ & $\checkmark$ & $\checkmark$ \\
Quote attribution & 17 & 30 & 2 & -- & $\checkmark$ & $\checkmark$ & $\checkmark$ \\
Imagined contact & 4 & 26 & 4 & -- & -- & $\checkmark$ & $\checkmark$ \\
Flag priming & 1 & 18 & 4 & -- & -- & -- & -- \\
Currency priming & 1 & 17 & 5 & -- & -- & -- & -- \\
\midrule
Rejections (of 16) & & & & 11 & 13 & 14 & 14 \\
\bottomrule
\end{tabular}%
}
\endgroup
\end{table}

The unanimity bit misses low versus high category scales, quote attribution,
and imagined contact because one, two, and four samples have $u>0.848$.  The
5\% bit misses imagined contact, which is significant in four samples, one
short of the five it needs.  The direction bit detects imagined contact with
26 samples pointing the predicted way, exactly its cutoff, so that rejection
is the most fragile.  A rejection means that the effect is present in at
least one sample's population (Section~\ref{sec:application} of the main
text).

\paragraph{Sensitivity}
Both checks below leave the unanimity bit with the fewest rejections and the
full $p$-values with the most.  Without the two online pools, 33 samples remain.  The unanimity bit then
rejects 11 nulls, the 5\% bit 13, equal-width reports with 2 to 16 cells 13,
and reports with 32 or 64 cells and the full $p$-values 14; imagined contact
is the finding that needs the finer reports.  Counting the anchoring finding once, through any one of its four questions,
so that each of the 13 findings contributes one test, lowers every count by
exactly the three dropped anchoring rejections: with all 35 samples the counts
are 8, 10, 11, and 11 of 13 for the unanimity bit, the 5\% bit, every
equal-width report, and the full $p$-values.

\FloatBarrier

\end{document}